\documentclass[12pt]{article}
\usepackage{amsmath} 
\usepackage{amsthm}
\usepackage{amstext}
\usepackage{amssymb}
\usepackage{enumerate}
\usepackage{array}
\usepackage{color}
\usepackage{bbm}
\usepackage{natbib}
\usepackage{subcaption}
\usepackage{float}
\usepackage{graphicx}
\usepackage{caption}
\usepackage{booktabs}
\usepackage[title]{appendix}
\usepackage{url} % not crucial - just used below for the URL 
\usepackage{fancyhdr}
\usepackage{authblk}% not crucial - just used below for the URL 
\usepackage{xr}
\usepackage{placeins}
\newcommand{\blind}{1}
\renewcommand{\floatpagefraction}{0.85}

\newtheorem{remark}{Remark}
\newtheorem{assumption}{Assumption}
\newtheorem{theorem}{Theorem}
\newtheorem{lemma}{Lemma}

\usepackage{algorithm, algorithmicx, algpseudocode}

\floatname{algorithm}{Algorithm}

\usepackage{xr}
\makeatletter
\newcommand*{\addFileDependency}[1]{% argument=file name and extension
	\typeout{(#1)}
	\@addtofilelist{#1}
	\IfFileExists{#1}{}{\typeout{No file #1.}}
}
\makeatother

\begin{document}

\def\spacingset#1{\renewcommand{\baselinestretch}%
{#1}\small\normalsize} \spacingset{1}

%%%%%%%%%%%%%%%%%%%%%%%%%%%%%%%%%%%%%%%%%%%%%%%%%%%%%%%%%%%%%%%%%%%%%%%%%%%%%%
 
\if1\blind
{
  \title{\bf Enhanced localized conformal prediction with imperfect auxiliary information}
  \author{%\thanks{The authors gratefully acknowledge \textit{please remember to list all relevant funding sources in the unblinded version}}\hspace{.2cm}
    Yinjie Min\thanks{The authors contribute equally and are listed in alphabetical order.} \\
    School of Statistics and Data Science, Nankai University\\
    Liuhua Peng\thanks{Corresponding author.}\\ 
    School of Mathematics \& Statistics, the University of Melbourne\\
    and \\
    Changliang Zou\\
    School of Statistics and Data Science, LEBPS, LPMC and KLMDASR, Nankai University}
  \maketitle
} \fi

\if0\blind
{
  \bigskip
  \bigskip
  \bigskip
  \begin{center}
    {\LARGE\bf Enhanced localized conformal prediction with imperfect auxiliary information}
\end{center}
  \medskip
} \fi

\bigskip
\begin{abstract}
There is growing interest in constructing conformal prediction sets that provide approximate or asymptotic conditional coverage guarantees, capturing local data heterogeneity. However, methods like localized conformal prediction (LCP) may face challenges in ensuring reliable prediction sets in regions with sparse calibration data. This paper introduces Enhanced Localized Conformal Prediction (ELCP), a novel approach that incorporates auxiliary data to refine localized prediction sets while preserving finite-sample marginal coverage guarantees. By utilizing a density-ratio-weighted kernel estimator, ELCP seamlessly integrates auxiliary and calibration data, accommodating potential distributional shifts and improving the local reliability of prediction sets. Theoretical analysis confirms that ELCP maintains marginal coverage and enhances asymptotic test-conditional coverage. Simulation results demonstrate its superior local coverage and smaller prediction sets compared to standard LCP, highlighting its effectiveness in settings with limited calibration data but available auxiliary information from related tasks.
\end{abstract}

\noindent%
{\it Keywords:}  Exchangeability; Fused estimates; Kernel methods;  Non-asymptotic bound; Test-conditional coverage; Weakly(semi)-supervised learning.
\vfill

\clearpage
\spacingset{1.9} % DON'T change the spacing!
\section{Introduction}
    Conformal prediction is an elegant yet powerful framework for quantifying prediction uncertainty, applicable to any machine learning model. 
    It generates prediction sets that are guaranteed to contain the unknown outcome with a specified level of confidence.
    Let $Z_i=(X_i,Y_i)$, $i=1,\ldots,n$, be a sequence of independent and identically distributed (i.i.d.)~\textit{calibration} data from some joint distribution $P$, where $X_i \in\mathcal{X}\subset\mathbb{R}^d$ represents the features, and $Y_i \in\mathcal{Y}\subset\mathbb{R}$ is the response.
    Given a test point $Z_{n+1}=(X_{n+1},Y_{n+1})$ independently drawn from $P$, with $X_{n+1}$ observed but $Y_{n+1}$ unobserved, conformal prediction constructs a prediction set $\widehat{C}_{\alpha}(X_{n+1})$ such that for a specified confidence level $1-\alpha\in(0,1)$
    \begin{equation}\label{eq:CP_marginal_coverage}
        \mathrm{pr}\big( Y_{n+1}\in\widehat{C}_\alpha(X_{n+1})\big) \geq 1-\alpha\,,
    \end{equation}
    without making any distributional assumptions about $P$ \citep{shafer2008tutorial}.

    Split conformal prediction (SCP)  is a common variant of conformal prediction \citep{vovk2005algorithmic}. 
    In SCP, a fitted model $\hat{\mu}(\cdot): \mathcal{X}\rightarrow\mathcal{Y}$ is pretrained using a {\it training} dataset $\mathcal{D}_{\rm tr}$ independent of $Z_1,\ldots, Z_{n+1}$. 
    Let $S(\cdot,\cdot): \mathcal{X}\times\mathcal{Y}\rightarrow\mathbb{R}$ be a conformity score function that quantifies the discrepancy between a hypothetical value $y\in\mathcal{Y}$ and the model's prediction $\hat{\mu}(x)$.
    For example, a common choice is the absolute residuals, defined as $S(x,y)=|y-\hat{\mu}(x)|$.
    Denote $S_i=S(X_i, Y_i)$ for $i\in [n+1]$, and let $(n+1)^{-1}(\sum_{i=1}^{n}\delta_{S_i}+\delta_{\infty})$ be the empirical distribution of $S_1,\ldots,S_{n}$ and $\infty$, where $\delta_{s}$ denotes the point mass at $s$.
    The level $(1-\alpha)$ SCP set for $Y_{n+1}$ is defined as
    \begin{equation}\label{eq:SCP}
        \widehat{C}_\alpha^{\mathrm{SCP}}(X_{n+1}) = \left\{ y : S(X_{n+1}, y) \leq Q\left(1-\alpha; (n+1)^{-1}\left(\sum_{i=1}^{n}\delta_{S_i}+\delta_{\infty}\right)\right) \right\}\,,
    \end{equation}
    where $Q(1-\alpha; \cdot)$ denotes the $(1-\alpha)$-quantile of the distribution in the second argument.
    The finite-sample marginal coverage, as defined in \eqref{eq:CP_marginal_coverage}, is guaranteed by $\widehat{C}_\alpha^{\mathrm{SCP}}(X_{n+1})$ as long as $Z_1,\ldots, Z_{n+1}$ are exchangeable \citep{vovk2005algorithmic,lei2018distribution}.

    However, marginal coverage alone is not sufficient for an efficient prediction set since a marginally valid prediction set can exhibit a local miscoverage rate significantly higher than $\alpha$ in certain local regions. Therefore, the conditional coverage is also important:
   \begin{align}\label{eq:conformal_cond}
        \mathrm{pr}\left( Y_{n+1}\in\widehat{C}_\alpha(X_{n+1})\mid X_{n+1}=x  \right) \geq 1-\alpha\,.
    \end{align}
    Although appealing, achieving \eqref{eq:conformal_cond} in a finite-sample and distribution-free context is impossible \citep{lei2013}. Recent works have proposed methods to construct prediction sets with approximate or asymptotic conditional coverage guarantee by either modifying the calibration step \citep{papadopoulos2011regression, lei2014distribution,guan2023localized,gibbs2023conformal} or using different score functions \citep{papadopoulos2008normalized, 
    romano2019conformalized, chernozhukov2021distributional, gupta2022nested, ding2024class}.

    Specifically, motivated by the fact that the prediction set $\widehat{C}_\alpha^{\mathrm{SCP}}(X_{n+1})$ treats all conformity scores equally regardless of whether the corresponding $X_i$ values are close to $X_{n+1}$,
    \citet{guan2023localized} proposed localized conformal prediction (LCP), addressing heterogeneity by assigning more weight to scores $S_i$ for which $X_i$ is closer to $X_{n+1}$, focusing on the local behavior of the data around the test point.
    Let $K(\cdot,\cdot;h): \mathcal{X}\times\mathcal{X}\rightarrow [0,\infty)$ be a bivariate localizer function that depends on a parameter $h$, typically serving as the bandwidth. 
    For $i\in[n+1]$, define the weighted distributions $\hat{F}_i^{y} = \sum_{j=1}^{n}\omega_{i,j}\delta_{S_j}+\omega_{i,n+1}\delta_{S(X_{n+1},y)}$, where the weights
    $\omega_{i,j}=K(X_i,X_j;h)/\{\sum_{\ell=1}^{n+1}K(X_i,X_{\ell};h)\}$ for $j\in[n+1]$.
    Then the $(1-\alpha)$ LCP set for $Y_{n+1}$ is
    \begin{gather}\label{eq:LCP_01}
        \widehat{C}_\alpha^{\mathrm{LCP}}(X_{n+1})=\left\{y:S(X_{n+1},y)\leq Q\left(1-{\alpha}(y); \hat{F}_{n+1}^{y}\right)\right\}\,,
    \end{gather}
    Here, $\alpha(y)$ is the adjusted level so that $\widehat{C}_\alpha^{\mathrm{LCP}}(X_{n+1})$ achieves the marginal coverage (\ref{eq:CP_marginal_coverage}).   
    
    The LCP demonstrates improved local coverage empirically, which better reflects the underlying heterogeneity of the data.
    Furthermore, building on LCP and the weighted conformal prediction framework of \citet{tibshirani2019conformal}, \citet{hore2023conformal} introduced randomly-localized conformal prediction (RLCP), which circumvents level adjustment and achieves certain test-conditional coverage under covariate shift.
    
    While LCP can provide better empirical local coverage, it comes at the cost of a reduced effective sample size, particularly when few points are close to a given $X_{n+1}$.
    In such cases, most of the weights $\omega_{n+1,j}$ for $j\in[n]$ approach zero, leaving only a few significant contributions, along with $\omega_{n+1,n+1}\delta_{S(X_{n+1},y)}$, to the weighted distribution $\hat{F}_{n+1}^{y}$.
    Consequently, the constructed prediction set may become unreliable, resulting in poor local performance, which contradicts the fundamental purpose of LCP.
    This issue is intensified when the calibration dataset is small.
    As an illustrative example, Fig.~\ref{fig:motivation_example} displays the prediction bands for $Y_{n+1}$ given $X_{n+1} = x$, where $X_i\sim N(0,1.5)$ and $Y_i=\{|\cos(X_i)|+0.1\}\varepsilon_i$ with $\varepsilon_i\sim N(0,1)$.
    It can be observed that using a smaller bandwidth makes the LCP set more adaptive to heterogeneity but leads to unreliable prediction sets in regions with limited calibration data.
    In contrast, increasing the bandwidth can alleviate this problem but the LCP set becomes less sensitive to heterogeneity.
    
    \begin{figure}[htbp]
        \centering
        \begin{minipage}{0.29\linewidth}
            \centering
            \includegraphics[scale=0.3]{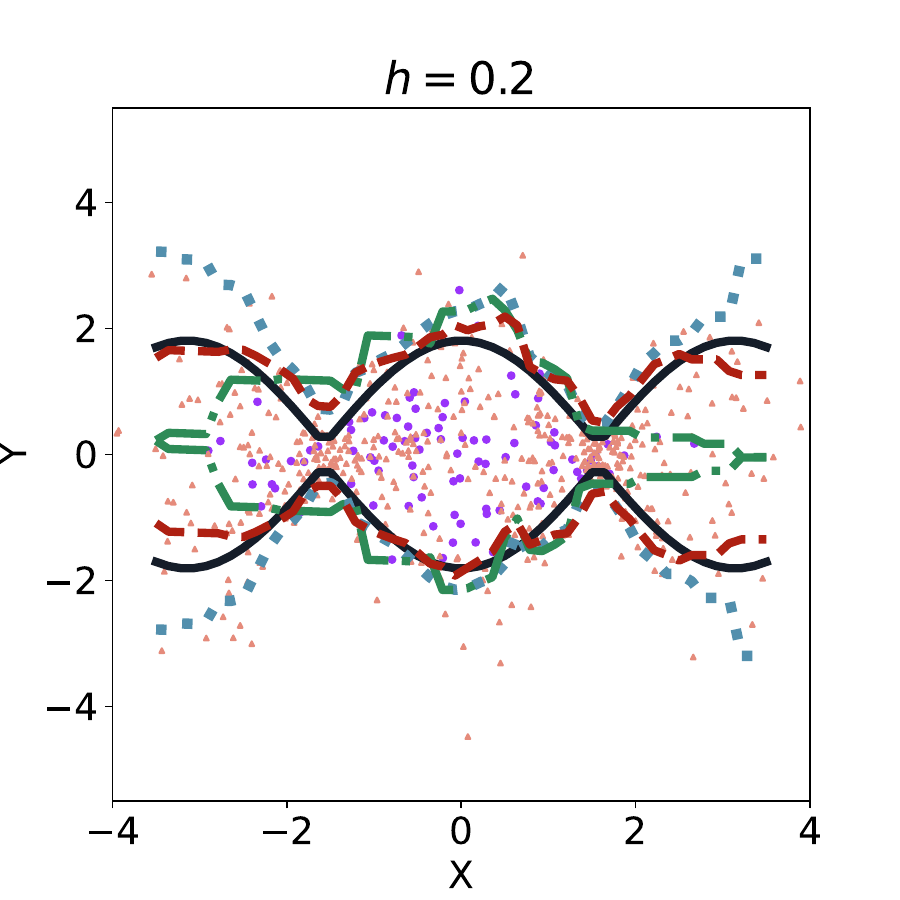}
        \end{minipage}
        \begin{minipage}{0.29\linewidth}
            \centering
            \includegraphics[scale=0.3]{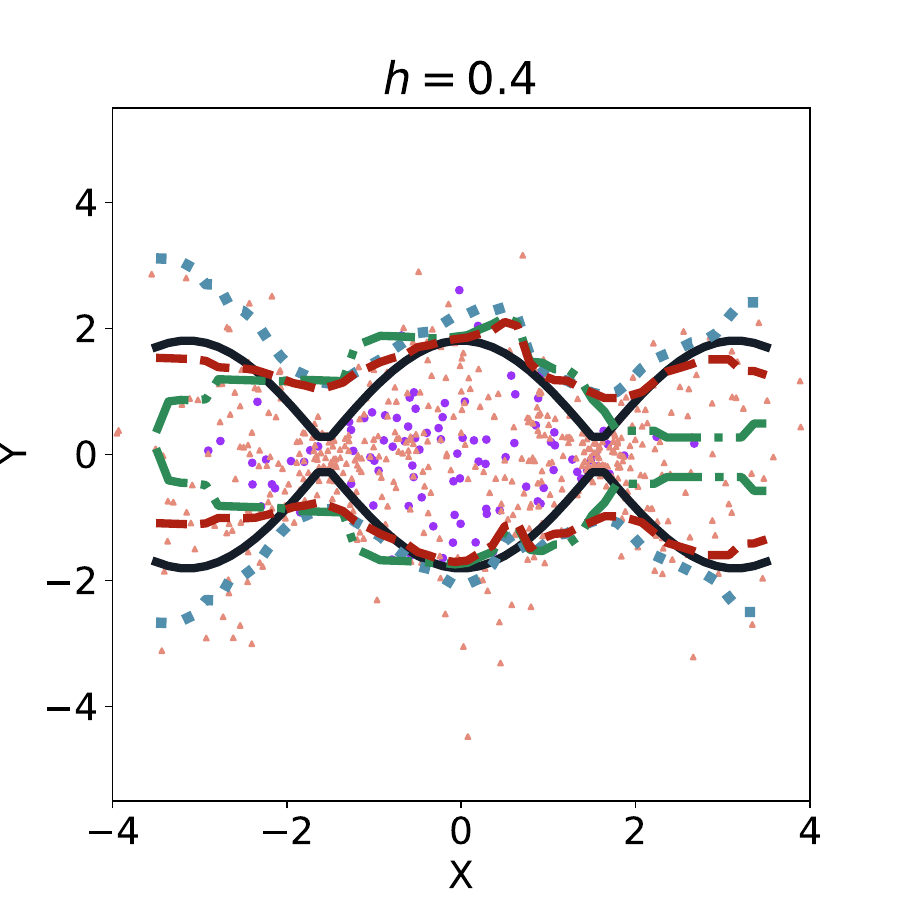}
        \end{minipage}
        \begin{minipage}{0.39\linewidth}
            \centering
            \includegraphics[scale=0.3]{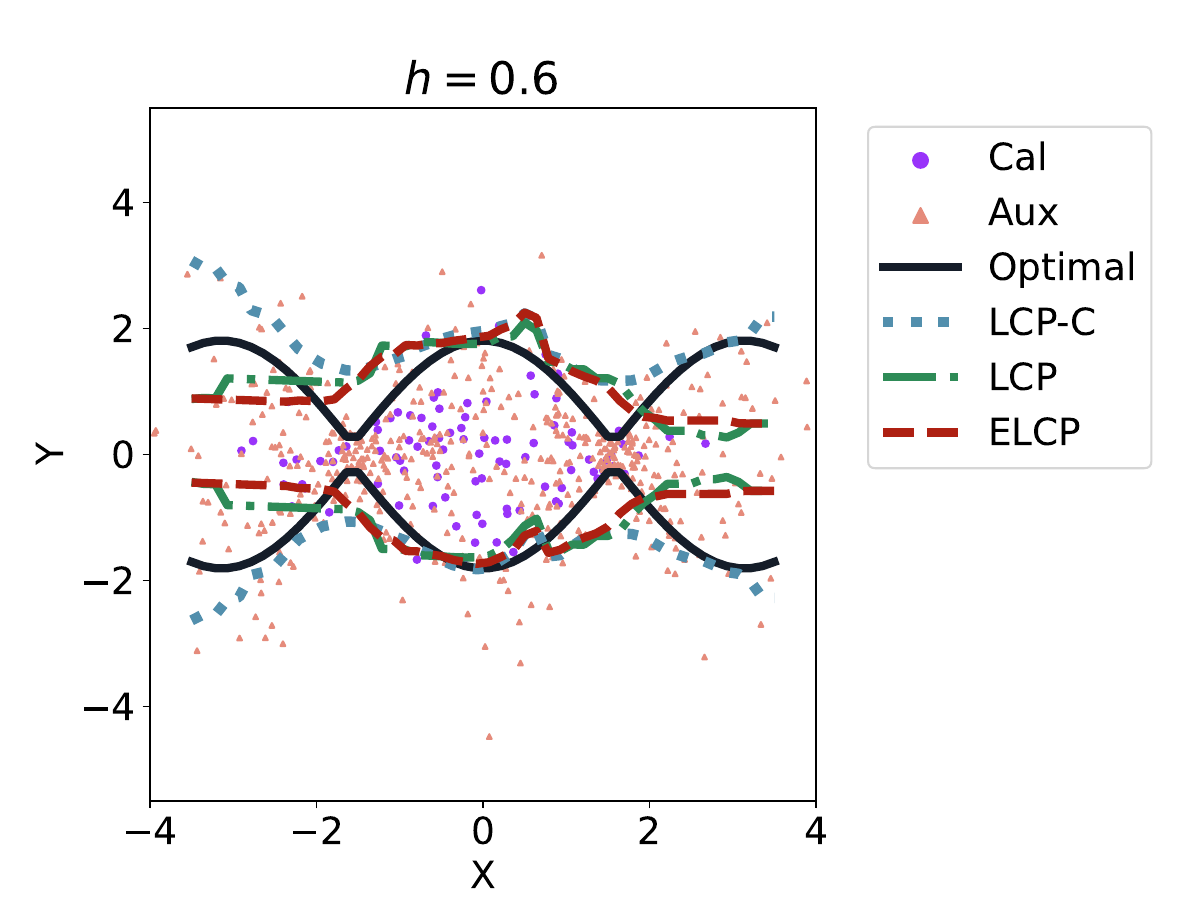}
        \end{minipage}
        \caption{\small Prediction bands by LCP (dash-dotted), LCP by direct combination (dotted), ELCP (dash) and optimal band (solid) with $h\in\{0.2,0.4,0.6\}$. The calibration data is of size $n=100$ and marked with label `Cal'. The auxiliary data is of size $m=500$ and marked with label `Aux'.}\label{fig:motivation_example}
    \end{figure}

    An intuitive solution to alleviate this issue is to increase the calibration data size, which is often impractical due to constraints such as cost, time, or feasibility. 
    However, in practical, auxiliary data from similar sources or previous studies may already be available.
    For example, data from different but related tasks as in transfer learning and meta-learning \citep{pan2009survey, hospedales2021meta}, data from other agents as in federated learning \citep{mammen2021federated}, or a large amount of unlabeled data as in semi-supervised learning \citep{van2020survey, Candes2024CrossPPI}.
    Numerous prior studies have demonstrated that properly incorporating useful information from auxiliary data can significantly enhance predictive performance \citep{angelopoulos2023prediction}.
    This motivates us to explore incorporating auxiliary data to enhance the performance of localized conformal prediction.
     
    In this paper, we consider the setting that a set of auxiliary data is available, which is drawn independently from a distribution $P^{\prime}$.
    It is highlighted that the distribution $P^{\prime}$ is allowed to differ from $P$, since auxiliary data may come from different populations with varying characteristics compared to the calibration data. We propose {\it Enhanced Localized Conformal Prediction (ELCP)}, which addresses the challenge of constructing localized conformal prediction set that incorporates auxiliary data and is able to account for potential data imperfections while ensuring marginal coverage validity.

    ELCP entails obtaining a fused estimate of the conditional distribution of the conformity score $S_i$ given features $X_i$ by leveraging auxiliary data in conjunction with calibration data, and then using only the conformity scores corresponding to $\mathcal{Z}_n\cup\{(X_{n+1},y)\}$ for calibration that guarantees marginal coverage due to exchangeability regardless of $P^{\prime}$. Compared to LCP, which uses a kernel-regression estimator solely based on $\mathcal{Z}_n\cup\{(X_{n+1},y)\}$, ELCP incorporates a density-ratio-weighted kernel estimator, combining both auxiliary and calibration data. This approach accounts for potential distributional shifts between the two data sources. The density ratio, known {\it a priori} or estimated from the data, adjusts the weighting of auxiliary data, resulting in more robust estimation of conditional distributions.    
    As shown in Fig.~\ref{fig:motivation_example},  the integration of auxiliary data significantly enhances local coverage, particularly in regions with sparse calibration data.
    
    The main contributions of this paper are: 

    (i) We propose a general framework for incorporating auxiliary data into conformal prediction set construction. By elaborately constructing fused estimates, the ELCP enjoys both finite-sample marginal coverage and empirically improved local coverage in the presence of auxiliary data. 
    
    (ii) From a theoretical standpoint, we derive a non-asymptotic bound for the test-conditional miscoverage error of ELCP which allows us to conduct a direct comparison with LCP, demonstrating that ELCP could offer improved local coverage under mild conditions. In addition, we establish asymptotic weak test-conditional coverage for ELCP.

    (iii) We provide detailed ``end-to-end'' implementations of ELCP, including data-driven parameter selection and computationally efficient deployments, both supported by theoretical justifications.

    (iv) We illustrate the easy coupling of our method with common scenarios involving auxiliary data, including multiple auxiliary datasets, and weakly or semi-supervised learning. Numerical experiments show that our method exhibits more accurate local coverage compared to existing methods, while offering a narrower prediction interval.

    The remainder of our article is structured as follows. In Section~\ref{sec:02_ELCP}, we present the basic procedure of ELCP and its application in two scenarios. Theoretical justification of ELCP is provided in Section~\ref{subsec::Local Enhancement}.
    Section~\ref{sec:implement} provides implementation details of ELCP, including parameter selection and computationally efficient deployments.
    Numerical studies are conducted in Section~\ref{sec:04_num}. Concluding remarks are provided in Section~\ref{sec:05_conclusion}. Theoretical proofs, technical details, and additional numerical results are provided in the supplementary material.
    
\section{Localized conformal prediction with auxiliary data}\label{sec:02_ELCP}

    We assume that the i.i.d.~data from distribution $P$ are split into a training set $\mathcal{D}_{\rm tr}$ and a calibration set $\mathcal{Z}_{n}=\{Z_i\}_{i\in[n]}$. A predictor $\hat{\mu}(\cdot)$ is pretrained using $\mathcal{D}_{\rm tr}$. Based on $\hat{\mu}(\cdot)$, we compute $S_i=S(X_i,Y_i):=S_i^{y}$ for $i\in[n]$ and $S_{n+1}^{y}=S(X_{n+1},y)$ for the calibration and test data.
    In addition, the i.i.d.~auxiliary data from $P^{\prime}$, independent of $\mathcal{D}_{\rm tr}$ and $\mathcal{Z}_{n}\cup\{Z_{n+1}\}$, are divided into $\mathcal{D}_{\rm tr}^{\prime}$ and $\mathcal{Z}_{m}^{\prime}=\{Z_j^{\prime}\}_{j\in[m]}$, where $Z_j^{\prime}=(X_j^{\prime},Y_j^{\prime})$ with $X_j^{\prime}\in\mathcal{X}$ and $Y_j^{\prime}\in\mathcal{Y}$ for $j\in[m]$.
    The corresponding conformity scores are $S_{j}^{\prime}=S^{\prime}(X_{j}^{\prime}, Y_{j}^{\prime})$ for $j\in [m]$.
    Here, $S^{\prime}(\cdot,\cdot)$ is defined based on a model $\hat{\mu}^{\prime}(\cdot)$, pretrained from $\mathcal{D}_{\rm tr}^{\prime}$.

\subsection{A recap of LCP}\label{subsec:recap_LCP}
    
    To motivate our method, we first provide a more intuitive explanation of how LCP ensures marginal coverage as the building block of ELCP for incorporating auxiliary data.
    Define \begin{equation}\label{eq:modified_LCP_beta}
        \hat{\beta}_i^{\mathrm{LCP}}(y) = \sum_{j=1}^{n+1}\omega_{i,j}\mathbbm{1}(S_j^{y}\leq S_i^{y}), \ i\in[n+1]\,.
    \end{equation}
    The $\hat{\beta}_{i}^{\mathrm{LCP}}(Y_{n+1})$ can be interpreted as an estimator of $F_{S\mid X}(S_i\mid X_i)$, the conditional cumulative distribution function (cdf) of $S_i$ given $X_{i}$ evaluated at $S_i$, based on $\mathcal{Z}_n\cup\{Z_{n+1}\}$.

    By the definition of $\widehat{C}_\alpha^{\mathrm{LCP}}(X_{n+1})$ in \eqref{eq:LCP_01} and the selection of $\alpha(y)$ detailed in Section~\ref{sec:LCP_detail} of the supplementary material, we have
    \begin{eqnarray}\label{eq:LCP}
        \widehat{C}_\alpha^{\mathrm{LCP}}(X_{n+1}) = \left\{y:  \hat{\beta}_{n+1}^{\mathrm{LCP}}(y) \leq Q\left( 1-\alpha; (n+1)^{-1}\overset{n+1}{\underset{i=1}\sum}\delta_{\hat{\beta}_i^{\mathrm{LCP}}(y)} \right) \right\}\,,
    \end{eqnarray}
    which takes the same form as the SCP set defined in \eqref{eq:SCP}, with the scores $S_{i}^{y}$ replaced by $\hat{\beta}_i^{\mathrm{LCP}}(y)$.
    The construction of $\widehat{C}_\alpha^{\rm LCP}(X_{n+1})$ aligns with the framework of the distributional conformal prediction set proposed by \citet{chernozhukov2021distributional} for addressing distributional heterogeneity via conditional distribution estimation. 
    However, LCP relies on $\hat{\beta}_{i}^{\mathrm{LCP}}(y)$ to estimate the conditional distribution of $S_i$ given $X_i$, whereas \citet{chernozhukov2021distributional} utilizes an estimate of the conditional distribution of $Y_i$ given $X_i$.

    It is clear that the marginal coverage of LCP is ensured when $\{\hat{\beta}_{i}^{\mathrm{LCP}}(Y_{n+1})\}_{i\in[n+1]}$ are exchangeable, which holds under the exchangeability of $\mathcal{Z}_n\cup\{Z_{n+1}\}$. It should be noted that $\widehat{C}_{\alpha}^{\mathrm{LCP}}(X_{n+1})$ is a slightly modified but conceptually similar version of the original prediction set defined by \citet{guan2023localized}, with the modifications in the definition of the adjusted level $\alpha(y)$ in \eqref{eq:LCP_01}. For a detailed discussion, see Section~\ref{sec:LCP_compare} of the supplementary material.

    When the number of data points near $X_i$ is limited, $\hat{\beta}_{i}^{\mathrm{LCP}}(Y_{n+1})$ may be inefficient and can significantly deviate from $F_{S\mid X}(S_i\mid X_i)$.
    This motivates our approach of incorporating auxiliary data $\mathcal{Z}_{m}^{\prime}$ to enhance LCP by improving the estimation of $F_{S\mid X}(S_i\mid X_i)$.

\subsection{Fused estimates of conditional distribution $F_{S\mid X}(s\mid x)$}\label{subsec:fused_estimate}
    
    Our approach to incorporating auxiliary data is guided by a key principle: accounting for the imperfect nature of the auxiliary data. Let $f_X(x)$ and $g_X(x)$ denote the probability density functions (pdfs) of $X$ and $X^\prime$, respectively.
    Define $f(x,s)$ and $g(x,s)$ as the joint pdfs of $(X_i,S_i)$ and $(X_j^{\prime},S_j^{\prime})$, respectively.
    The density ratio between these two joint distributions is denoted as $r(x,s)=f(x,s)/g(x,s)$.
    The conditional distribution of $S_i$ given $X_i=x_0$ can be reformulated as
    \begin{align}
        \mathrm{pr}\left( S_i\leq s_0\mid X_i=x_0 \right)
        & =  \int\mathbbm{1}\left( s\leq s_0 \right)g_{S\mid X}(s\mid x_0)\dfrac{f_{S\mid X}(s\mid x_0)}{g_{S\mid X}(s\mid x_0)}ds \notag\\
        & =  E\left\{ \mathbbm{1}\left( S_j^\prime\leq s_0 \right)\dfrac{f_{S\mid X}(S_j^\prime\mid X_j^\prime)}{g_{S\mid X}(S_j^\prime\mid X_j^\prime)}\mid X_j^\prime=x_0 \right\}\label{eq:reformulateauxiliary1}\\
        & =  \frac{g_{X}(x_0)}{f_{X}(x_0)}E\left\{\mathbbm{1}\left( S_j^{\prime}\leq s_0 \right)r(X_j^{\prime},S_j^{\prime})\mid X_j^{\prime}=x_0\right\},\notag
    \end{align}
    which gives the following equivalent formulations
    \begin{align}
       f_{X}(x_0)\mathrm{pr}\left( S_i\leq s_0\mid X_i=x_0 \right) %& = f_{X}(x_0)E\left\{\mathbbm{1}\left( S_i\leq s_0 \right) \mid X_j^{\prime}=x_0\right\} \\
         = g_{X}(x_0)E\left\{\mathbbm{1}\left( S_j^{\prime}\leq s_0 \right)r(X_j^{\prime},S_j^{\prime})\mid X_j^{\prime}=x_0\right\}. \label{eq:F_S_X} 
    \end{align}
    Here, $f_{S\mid X}(s\mid x)$ and $g_{S\mid X}(s\mid x)$ represent the conditional pdfs of $S_i$ and $S_j^{\prime}$ given $X_i=x$ and $X_j^{\prime}=x$, respectively.
    Similarly, $f_{X}(x_0)$ can be expressed as
    \begin{align}\label{eq:f_X}
        f_{X}(x_0) = g_{X}(x_0)E\left\{r(X_j^{\prime},S_j^{\prime})\mid X_j^{\prime}=x_0\right\}. 
    \end{align}

    The formulations in \eqref{eq:F_S_X} and \eqref{eq:f_X}, together with \eqref{eq:modified_LCP_beta}, suggest a fused estimator of $F_{S\mid X}(s\mid x)$ using both calibration and auxiliary data, defined as 
    \begin{eqnarray}\label{eq:beta}
        \hat{\beta}_{\omega,\hat{r}}^{y}(x,s) = \dfrac{\overset{n+1}{\underset{j=1}\sum}K(x, X_j;h)\mathbbm{1}(S_j^{y}\leq s)+\omega\overset{m}{\underset{j=1}\sum}K(x, X_j^\prime;h)\hat{r}(X_j^\prime,S_j^\prime)\mathbbm{1}(S_j^{\prime}\leq s)}{\overset{n+1}{\underset{j=1}\sum}K(x, X_j;h)+\omega\overset{m}{\underset{j=1}\sum}K(x, X_j^\prime;h)\hat{r}(X_j^\prime,S_j^\prime)}\,,
    \end{eqnarray}
    where $\hat{r}(x,s)$ is an estimator of $r(x,s)$ and $\omega\in[0,1]$ is a hyperparameter that controls the level of incorporation. 
    By incorporating the information of $\mathcal{Z}_{m}^{\prime}$ while adjusting potential distributional shifts, $\hat{\beta}_{\omega,\hat{r}}^{y}(X_i,S_i^{y})$ is expected to provide more efficient estimation of $F_{S\mid X}(S_i\mid X_i)$, especially in regions with sparse calibration data.
    
    It seems more intuitive to derive an estimator via \eqref{eq:reformulateauxiliary1} which will involve the estimate of the density ratio of two conditional distributions, however, we choose to formulate the estimator with \eqref{eq:F_S_X} and \eqref{eq:f_X} since estimating the density ratio $r(x, s)$ is more straightforward. There is a rich literature on density ratio estimation, with many powerful algorithms even in high-dimensional settings. For example, $r(x, s)$ can be estimated effectively using classification-based methods. Further discussion is provided in Sections~\ref{sec:method} and~\ref{subsec::Local Enhancement}. 

\subsection{Enhanced localized conformal prediction}\label{sec:method}

    Given that only the exchangeability of $\mathcal{Z}_n\cup\{Z_{n+1}\}$ is guaranteed in our setup, the key idea is to use the fused estimators $\{\hat{\beta}_{\omega,\hat{r}}^{y}(X_i,S_i^{y})\}_{i\in[n+1]}$ for calibration.
    The ELCP set, incorporating the auxiliary data $\mathcal{Z}_{m}^{\prime}$, is accordingly defined as
    \begin{equation}\label{eq:ELCP}
        \widehat{C}_{\alpha}^{\mathrm{ELCP}}(X_{n+1}) = \left\{ y: \hat{\beta}_{\omega,\hat{r}}^{y}(X_{n+1},S_{n+1}^{y})\leq Q\left(1-\alpha; (n+1)^{-1}\overset{n+1}{\underset{i=1}\sum}\delta_{\hat{\beta}_{\omega,\hat{r}}^{y}(X_{i},S_{i}^{y})}\right) \right\}\,.
    \end{equation}

    From the proof of Theorem~\ref{theo:marginal_coverage}, we see that the marginal coverage of $\widehat{C}_{\alpha}^{\mathrm{ELCP}}(X_{n+1})$ is ensured by the exchangeability of $\hat{\beta}_{\omega,\hat{r}}(X_i,S_i),i\in[n+1]$, where $\hat{\beta}_{\omega,\hat{r}}(x,s)$ is defined as $\hat{\beta}^{y}_{\omega,\hat{r}}(x,s)$ evaluated at $y=Y_{n+1}$.
    To achieve this, we make the following assumption. 
    \begin{assumption}\label{assump:permutable}
       The $\hat{r}(\cdot,\cdot)$ is invariant to permutations within $\mathcal{Z}_n \cup \{Z_{n+1}\}$ and $\mathcal{Z}^\prime_m$.
    \end{assumption}

    If additional data from $P$ and $P^{\prime}$, independent of both $\mathcal{D}_{\rm tr}\cup\mathcal{Z}_n \cup \{Z_{n+1}\}$ and $\mathcal{D}_{\rm tr}^{\prime}\cup\mathcal{Z}^\prime_m$, are available and used to estimate $r(x,s)$, then Assumption~\ref{assump:permutable} holds naturally. However, such data are often unavailable, and splitting existing data for this purpose reduces the size of the calibration set.
    Henceforth, the $\hat{r}(x,s)$ represents the estimator of $r(x,s)$ using $\mathcal{Z}_n$ and $\mathcal{Z}^\prime_m$, with data point $(X_{n+1}, y)$ included to ensure Assumption~\ref{assump:permutable} is satisfied.
    The end-to-end procedure for constructing the ELCP set is summarized in Algorithm~\ref{twoagentsProcedure}.  

    \renewcommand{\floatpagefraction}{.95}
    \begin{algorithm}[tp]
    \caption{Enhanced Localized Conformal Prediction (ELCP)}\label{twoagentsProcedure}
    {\bf Input:} Calibration and auxiliary data $\mathcal{Z}_n$, $\mathcal{Z}_m^\prime$, test point $X_{n+1}$, training data $\mathcal{D}_{\mathrm{tr}}$ and $\mathcal{D}_{\mathrm{tr}}^{\prime}$, score functions $S(\cdot,\cdot)$ and $S^{\prime}(\cdot,\cdot)$, function $K(\cdot,\cdot;\cdot)$, parameters $h$ and $\omega$, level $1-\alpha$
    \begin{algorithmic}[1]
    \State Pretrain model $\hat{\mu}(\cdot)$ from $\mathcal{D}_{\mathrm{tr}}$ and $\hat{\mu}^{\prime}(\cdot)$ from $\mathcal{D}_{\mathrm{tr}}^{\prime}$
    \State Calculate $S_i^{y}=S(X_i,Y_i)$, $i\in[n]$ using $\hat{\mu}(\cdot)$, and $S_j^\prime=S^\prime(X_j^\prime,Y_j^\prime)$, $j\in[m]$ using $\hat{\mu}^{\prime}(\cdot)$
    \For{$y\in\mathcal{Y}$}
        \State\qquad Calculate $S_{n+1}^y=S(X_{n+1},y)$;
        \State\qquad Obtain density ratio estimator $\hat{r}(\cdot,\cdot)$ using $\mathcal{Z}_n\cup\{(X_{n+1}, y)\}$ and $\mathcal{Z}^\prime_m$;
        \State\qquad Calculate $\hat{\beta}_{\omega,\hat{r}}^{y}(X_i,S_i^{y})$ for $i\in[n+1]$;
        \State\qquad Calculate $\hat{q}=Q\Big(1-\alpha; (n+1)^{-1}\sum_{i=1}^{n+1}\delta_{\hat{\beta}_{\omega,\hat{r}}^{y}(X_i,S_i^{y})}\Big)$;
        \State\qquad $y$ is included in set $\widehat{C}_\alpha^{\mathrm{ELCP}}(X_{n+1})$ as long as $\hat{\beta}_{\omega,\hat{r}}^{y}(X_{n+1},S_{n+1}^{y})\leq\hat{q}$.
    \EndFor
    \State \Return $\widehat{C}_\alpha^{\mathrm{ELCP}}(X_{n+1})$
    \end{algorithmic}
    \end{algorithm}

    \begin{remark}\label{remark:efficient_r}
        Including $(X_{n+1},y)$ into the estimation of $r(x,s)$ guarantees marginal coverage but requires updating $\hat{r}(x,s)$ for each $y$, which can be computationally intensive. A more efficient implementation is to obtain $\hat{r}(x,s)$ without using $(X_{n+1},y)$. Although this may theoretically compromise marginal coverage, the resulting change in the prediction set is often negligible under certain conditions. 
        See Section~\ref{sec:compute} for theoretical justification.
    \end{remark}
    
    The following theorem establishes the marginal coverage guarantee for %the ECLP set 
    $\widehat{C}_{\alpha}^{\mathrm{ELCP}}(X_{n+1})$.
    
    \begin{theorem}[Marginal Coverage]\label{theo:marginal_coverage}
        Suppose $\mathcal{Z}_n\cup\{Z_{n+1}\}$ are exchangeable and Assumption~\ref{assump:permutable} holds. Then for a given $\alpha\in(0,1)$, $1-\alpha\leq \mathrm{pr}\left( Y_{n+1}\in \widehat{C}_{\alpha}^{\mathrm{ELCP}}(X_{n+1}) \right)<1-\alpha+(n+1)^{-1}$.
    \end{theorem}

    When the auxiliary information is perfect, i.e., $f(x,s)=g(x,s)$ for all $x$ and $s$, the optimal approach is to combine $\mathcal{Z}_n$ and $\mathcal{Z}_m^\prime$ into a single dataset and apply LCP to the combined sample.
    Choosing $\hat{r}(x,s)\equiv1$ and $\omega=1$, we obtain $\hat{\beta}_{1,1}^{y}(X_i,S_i^{y}), i\in[n+1]$ and $\hat{\beta}_{1,1}^{y}(X_j^{\prime},S_j^{\prime}), j\in[m]$.  
    The LCP set based on the combined data is $\widehat{C}_{\alpha}^{\mathrm{LCP-C}}(X_{n+1})$ as follows
    \begin{eqnarray}\label{eq:comb_LCP}
        \left\{ y:\hat{\beta}_{1,1}^{y}(X_{n+1},S_{n+1}^{y})\leq Q\left(1-\alpha; (n+m+1)^{-1}\left\{\overset{n+1}{\underset{i=1}\sum}\delta_{\hat{\beta}_{1,1}^{y}(X_i,S_i^{y})}+\sum_{j=1}^{m}\delta_{\hat{\beta}_{1,1}^{y}(X_j^{\prime},S_j^{\prime})}\right\}\right) \right\}\,.
    \end{eqnarray}
    However, if the auxiliary dataset is imperfect, such a naive combination may invalidate the exchangeability $\{\hat{\beta}_{1,1}(X_i,S_i),i\in[n+1]\}\cup\{\hat{\beta}_{1,1}(X_j^{\prime},S_j^{\prime}),j\in[m]\}$ and would result in a failure of marginal coverage. 
    See the blue curves in Fig.~\ref{fig:motivation_example}.
    %for example. 
    In contrast, the calibration step of $\widehat{C}_{\alpha}^{\mathrm{ELCP}}(X_{n+1})$ relies {\it exclusively} on $\{\hat{\beta}_{\omega,\hat{r}}^{y}(X_i,S_i^{y})\}_{i\in[n+1]}$.
    This key difference guarantees that $\widehat{C}_{\alpha}^{\mathrm{ELCP}}(X_{n+1})$ ensures marginal coverage regardless of whether the information of $\mathcal{Z}_{m}^{\prime}$ is imperfect, yielding a safe approach.
    \iffalse
    \begin{remark}
        When the calibration dataset is too small, ELCP essentially inherits the same limitations as classical conformal methods at high coverage levels: it can only produce trivial prediction sets. To alleviate this limitation, in Section~\ref{sec: calibration augmentation} in the supplementary material, we propose combining auxiliary data with calibration data through weighting by $\hat{r}(x,s)$. However, because the estimated weights lack strict theoretical guarantees, we employ this approach only when $n$ is extremely small and ELCP fails to produce meaningful results.
    \end{remark}
    \fi
        
    \begin{remark}
        During score pre-training, incorporating the auxiliary training data $\mathcal{D}_{\rm tr}^{\prime}$ can potentially improve the pre-trained model and, in turn, the performance of the resulting prediction set.
        The marginal coverage guarantee in Theorem~\ref{theo:marginal_coverage} continues to hold. 
        Technical details, along with synthetic and real data experiments comparisons of different pre-training schemes, are provided in Sections~\ref{sec:supp_simu_score},~\ref{sec:supp_housing} and~\ref{sec:real_data} of the supplementary material.
    \end{remark}

\subsection{Examples/application scenarios}

    Section~\ref{sec:method} establishes the foundation for achieving enhanced prediction sets through the integration of auxiliary information, thereby raising questions about its practical utility: in which types of applications can the proposed ELCP be particularly beneficial? This subsection explores two scenarios that frequently arise in practice, demonstrating how ELCP can be seamlessly integrated into real-world applications.
    
    \subsubsection{ELCP with multiple auxiliary datasets}

    We explore the ELCP framework in the context of multiple auxiliary datasets. This scenario is particularly relevant in real-world applications such as federated learning \citep{peng2019federated} and multi-task learning \citep{caruana1997multitask, zhang2021survey}, where multiple auxiliary datasets are commonly available.
    These datasets typically exhibit significant distributional shifts, which may stem from differences in population demographics, measurement instruments, or data collection protocols.

    Assume there are $K$ auxiliary datasets, where the $k$-th dataset is denoted as $\mathcal{Z}^{(k)} = \big\{ Z_{i}^{(k)} = (X_{i}^{(k)}, Y_{i}^{(k)}) \big\}_{i\in[n_k]}$, consisting of $n_k$ samples drawn from the distribution $P^{(k)}$ for each $k \in [K]$. For each $\mathcal{Z}^{(k)}$, let $\{S_{i}^{(k)}\}_{i\in[n_k]}$ represent the conformity scores, computed using a model $\hat{\mu}^{(k)}(\cdot)$ pretrained on a separate training data $\mathcal{D}_{\rm tr}^{(k)}$.
    In addition, let $\hat{r}^{(k)}(\cdot, \cdot)$ be the estimator of the density ratio $r^{(k)}(x,s) = f(x,s)/g^{(k)}(x,s)$ , where $g^{(k)}(x,s)$ is the joint pdf of $(S_{i}^{(k)}, X_{i}^{(k)})$. 
    Then the ELCP set with multiple auxiliary datasets is $\widehat{C}_{\alpha}^{\mathrm{ELCP}}(X_{n+1})$ in \eqref{eq:ELCP} with $\hat{\beta}_{\omega,\hat{r}}^{y}(X_i,S_i^{y})$ replaced by
    \begin{equation}
        \notag \dfrac{\overset{n+1}{\underset{j=1}\sum}K(X_i,X_j;h)\mathbbm{1}\left( S_j^y\leq S_i^y \right)+\omega\overset{K}{\underset{k=1}\sum}\overset{n_k}{\underset{j=1}\sum}K(X_i,X_{k,j}^{\prime};h)\hat{r}^{(k)}(X_{j}^{(k)},S_{j}^{(k)})\mathbbm{1}\left( S_{k,j}^{\prime}\leq S_i^y \right)}{\overset{n+1}{\underset{j=1}\sum}K(X_i,X_j;h)+\omega\overset{K}{\underset{k=1}\sum}\overset{n_k}{\underset{j=1}\sum}K(X_i,X_{k,j}^{\prime};h)\hat{r}^{(k)}(X_{j}^{(k)},S_{j}^{(k)})}\,.
    \end{equation}
    Marginal coverage of the ELCP in this case is maintained as long as $\hat{r}^{(k)}(\cdot,\cdot)$ is invariant under permutations within $\mathcal{Z}_{n}\cup\{Z_{n+1}\}$ and $\mathcal{Z}^{(k)}$ for $k\in[K]$.

    {\subsubsection{Weakly(Semi)-supervised setting}\label{sec:semi}}
    
    In the weakly-supervised setting \citep{bilen2016weakly, zhou2018brief}, where auxiliary data is provided with only coarse-grained labels or labels that may not always represent the ground truth, our proposed ELCP can handle this scenario directly by leveraging imperfect information as auxiliary data while maintaining marginal coverage.
    
    Now we consider extending our approach to a more challenging setting, semi-supervised learning, where auxiliary data include only covariates.  This setting is common since labeled data are often expensive to obtain, while large amounts of unlabeled data are readily accessible \citep{van2020survey}. For recent development, we refer to \citet{zhang2022high}, \citet{Candes2024CrossPPI} and \citet{Wen2024Semi} therein. Suppose that $\{X_j^{\prime}\}_{j\in[m]}$ are observed while the corresponding responses are unobserved.
    A natural idea is to obtain the predicted values $\widehat{Y}_j$ for $X_j'$ from a pre-trained model $\nu(\cdot):\mathcal{X}\mapsto\mathcal{Y}$ built on the data  independent of $\{X_j^{\prime}\}_{j\in[m]}$ and $\mathcal{Z}_{n}\cup\{Z_{n+1}\}$. The $\nu(\cdot)$ is often taken as the estimate of expectation of $Y$ given $X$.
    Within our ELCP framework, we can use $\mathcal{Z}_m'=\{(X_j^{\prime},\widehat{Y}_j)\}_{j\in[m]}$ as the auxiliary data.
    However, given $\nu(\cdot)$, the distribution of the conformity score $S_{j}^{\prime}$ conditional on $X_j^{\prime}$ would become a degenerate distribution, which is substantially different from that of $S_i\mid X_i$. Although marginal coverage can still be ensured in such cases, the benefit of incorporating the auxiliary data would be limited.

    Therefore, we seek the predictions to ensure that the distribution of $S_{j}^{\prime}$ conditional on $X_j^{\prime}$ is as close to that of $S_i$ given $X_i$ as possible to make the fused estimate efficient.   
    This aim distinguishes our approach from conventional prediction-powered inference in semi-supervised settings  \citep{angelopoulos2023prediction}. 
    Assume that $X_{j}^{\prime}$ shares the same distribution as $X_i$, that is, $g_X(x)=f_X(x)$ for all $x\in\mathcal{X}$, such that the auxiliary data is informative.
    To this end, we begin with the case where an estimator of the conditional distribution of $Y_i\mid X_i$ is provided. In our framework, such a model can be trained on the dataset $\mathcal{D}_{\rm tr}$.  A rich body of literature exists on conditional distribution estimation, %including semiparametric methods employing linear projection \citep{Hall2005}, single-index projection \citep{Henzi2021}, 
    see, for example, the distributional random forests \citep{Cevid2022} and the neural network-based method \citep{shen2024engression}. Then, for each $X_{j}^{\prime}$, we randomly generate a sample from this estimated conditional distribution, denoted as ${Y}_{j}^{\prime}$, and accordingly the ELCP can be conducted with the {\it artificial} calibration data $\mathcal{Z}_m'=\{(X_j',Y_j')\}_{j\in[m]}$. 
    The maintenance of marginal coverage for the ELCP in this semi-parametric setting can readily be verified by confirming Assumption~\ref{assump:permutable} holds.
    
    In practice, a model for the conditional distribution of $Y_i\mid X_i$ is not always necessary.
    For example, consider the case of having two pre-trained models: $\nu(x)$ and $\sigma(x)$, which predict the conditional mean and standard deviation of $Y$ given $X$, respectively.
    For each $X_{j}^{\prime}$, we can generate ${Y}_{j}^{\prime}$ by the following mean-variance model:
    \begin{align*}
        \notag Y_j^{\prime} = \nu(X_j^{\prime}) +  \sigma(X_j^{\prime})\cdot\varepsilon_j^{\prime},
    \end{align*}
    where $\varepsilon_j^{\prime}$ is a random variable with zero mean and unit variance.
    In the absence of any prior knowledge, a simple choice
    is to use standard normal distribution for generating $\varepsilon_j^{\prime}$.
    The ELCP set can be constructed using the same procedures with the synthesized dataset $\mathcal{Z}_m'=\{(X_j',Y_j^{\prime})\}_{j\in[m]}$.
    This semi-supervised ELCP, as indicated by numerical results in Section~\ref{sec:simu_semi} of the supplementary material, is effective in many settings with improvements over LCP, but remains an important area for future work on theoretical justifications.

\section{Theoretical results on local performance of ELCP}
\label{subsec::Local Enhancement}
    
    We provide theoretical justification for how ELCP, which incorporates auxiliary data, improves the local performance compared to LCP which only utilizes the calibration data.
    To achieve this, we derive a non-asymptotic bound on the test-conditional miscoverage error:
    \begin{equation}\label{eq:conditional_coverage_diff}
        \left| \mathrm{pr}\left( Y_{n+1}\in\widehat{C}_{\alpha}^{\mathrm{ELCP}}(X_{n+1}) \mid X_{n+1}=x_0 \right) - (1-\alpha)\right|,
    \end{equation}
    and compare it to that of LCP.
    Here, $x_0$ is a fixed, given value in the feature space $\mathcal{X}$.
    
    As discussed in Sections~\ref{subsec:recap_LCP} and~\ref{subsec:fused_estimate}, the key benefit of incorporating auxiliary data is the improvement in estimating $F_{S\mid X}(S_i\mid X_i)$ for $i\in [n+1]$. 
    Theoretically, $\{F_{S\mid X}(S_i\mid X_i)\}_{i\in[n+1]}$ are i.i.d.~$\mathrm{Uniform}[0,1]$ random variables.
    In ELCP, $\hat{\beta}_{\omega,\hat{r}}(X_i,S_i)$ is used as an estimator of $F_{S\mid X}(S_i\mid X_i)$, whereas in the LCP, the corresponding estimator is $\hat{\beta}_{i}^{\mathrm{LCP}}(Y_{n+1})$, which is just $\hat{\beta}_{0,\hat{r}}(X_i,S_i)$ by our definition. 
    Thus, the quantities
    \begin{equation}
        \notag \Delta_i(\omega)=\left|\hat{\beta}_{\omega,\hat{r}}(X_i,S_i)-F_{S\mid X}(S_i\mid X_i)\right|, \ i\in[n+1], 
    \end{equation}
    are crucial in quantifying the test-conditional miscoverage error \eqref{eq:conditional_coverage_diff}.
    The following lemma formalizes this relationship and provides a foundation for the subsequent theorems.
    \begin{lemma}\label{lemma:conditional_coverage_bound}
        Suppose $\sup_{1\leq i\leq n+1}\Delta_i(\omega) \leq \varepsilon$ holds with probability at least $1-\delta$, where $\varepsilon$ is a nonrandom quantity which may depend on $\delta$, $n$, $m$, $\omega$ and $h$.
        Then 
        \begin{align}
            \notag & \mathrm{pr} \left( Y_{n+1}\in\widehat{C}_{\alpha}^{\mathrm{ELCP}}(X_{n+1}) \mid X_{n+1}=x_0 \right) \\
            \notag \in & \left[ \dfrac{\lceil (n+1)(1-\alpha) \rceil-1}{n+1}-2\varepsilon-\delta, \dfrac{\lceil (n+1)(1-\alpha) \rceil}{n+1}+2\varepsilon+\delta \right]\,,
        \end{align}
        where $\lceil\cdot\rceil$ is the ceiling function.
    \end{lemma}
    
    Lemma~\ref{lemma:conditional_coverage_bound} indicates that the test-conditional miscoverage error of ELCP depends on the concentration property of $\sup_{1\leq i\leq n+1}\Delta_i(\omega)$,
    assumed as a high-level condition in Lemma~\ref{lemma:conditional_coverage_bound} and will be established in Lemma~\ref{lemma:delta_concentration2}.
    Notably, Lemma~\ref{lemma:conditional_coverage_bound} also applies to the LCP set $\widehat{C}_{\alpha}^{\mathrm{LCP}}(X_{n+1})$ by setting $\omega=0$.
    Let $\|\cdot\|$ be the Euclidean norm.
    
    \begin{assumption}\label{ass2}
    \begin{enumerate}
        \item[(i)] Assume $\mathcal{X}=[0,1]^d$ without loss of generality. There exist positive constants $\underline{L}_{1}, \underline{L}_{2}, \overline{L}_{1}$ and $\overline{L}_{2}$ such that for all $x\in [0,1]^{d}$ and $s\in \mathbb{R}$,
        \begin{equation*}
            \underline{L}_{1}\leq f(x,s),g(x,s)\leq \overline{L}_{1},\ \underline{L}_{2}\leq r(x,s),\hat{r}(x,s)\leq \overline{L}_2\,.
        \end{equation*}
        
        \item[(ii)] For any $x_1,x_2\in[0,1]^d$, there exists a positive constant $L$ such that $$ \underset{s\in\mathbb{R}}{\sup}\left|F_{S\mid X}(s\mid x_1)-F_{S\mid X}(s\mid x_2)\right|\leq L\|x_1-x_2\|\,. $$ 
        
        \item[(iii)] $K(\cdot,\cdot;h)$ is in the form of $K(x_1,x_2;h)=K_0\left( \|x_1-x_2\|/h \right)$, where $K_0(\cdot)$ is a bounded univariate kernel function that is symmetric around $0$, and satisfies: (\textit{a}) $K_0(u)$ decreases when $u$ increases for $u\geq 0$; (\textit{b}) $uK_0(u)$ decreases when $u$ increases for $u>1$; (\textit{c}) $\int_{0}^{\infty} u^{d-1}K_0(u)du<\infty$. 
    \end{enumerate}
    \end{assumption}

    Assumption~\ref{ass2}-(i) ensures that $f(x,s)$ and $g(x,s)$, along with their ratio, are bounded, a standard requirement in density ratio estimation \citep{sugiyama2008direct}.
    Assumption~\ref{ass2}-(ii) imposes a smoothness condition on the conditional distribution of $S_i$ given $X_i$.
    Assumption~\ref{ass2}-(iii) holds for commonly-used kernel functions such as the Gaussian and Laplacian kernels.

    %Let $K_0^{-1}(\cdot)$ be the inverse function of $K_0(\cdot)$ on the positive real line.
    For $k\geq 2$, define
    \begin{gather}
        \notag D_k(r,\hat{r})=\left\{\int\left|\hat{r}(x,s)-r(x,s)\right|^kg(x,s)dxds\right\}^{1/k},
    \end{gather}
    which quantifies the estimation accuracy of $\hat{r}(\cdot,\cdot)$ with respect to $r(\cdot,\cdot)$. 
    Recall that $\hat{r}(x,s)$ is obtained using $\mathcal{Z}_n \cup \{Z_{n+1}\}$ and $\mathcal{Z}_{m}^{\prime}$.
    Define $\hat{r}_j^\prime(x,s)$ as the estimator of $r(x,s)$ based on $\mathcal{Z}_n\cup\{Z_{n+1}\}$ and $\mathcal{Z}_m^\prime\setminus\{Z_j^\prime\}$ for each $j\in[m]$.
    
    \begin{assumption}\label{ass3}
        There exists a positive constant $C_r$ such that $\underset{x,s}{\sup}|\hat{r}(x,s)-\hat{r}_j^\prime(x,s)|\leq C_r m^{-1}$ for all $j\in[m]$. Moreover, $\underline{L}_2\leq \hat{r}_j^\prime(x,s)\leq\overline{L}_2$ for all $x\in [0,1]^{d}$ and $s\in \mathbb{R}$.
    \end{assumption}

    \begin{remark}\label{remark:ass3}
        When viewing the density ratio estimator as a functional of the sample distributions, appropriate smoothness conditions on this functional imply that modifying a single data point alters the empirical distribution by at most $O(m^{-1})$. Consequently, the density ratio estimate changes by at most $O(m^{-1})$. This ensures the uniform stability condition \citep{feldman2018generalization} of the density ratio estimator that $\sup_{x,s}|\hat{r}(x,s)-\hat{r}_j^\prime(x,s)|\leq C_r m^{-1}$. A detailed explanation is provided in Section~\ref{sec:verify_ass3} in the supplementary material, and similar assumptions are made in \citet{lei2018distribution}.
    \end{remark}
    \begin{lemma}\label{lemma:delta_concentration2}
        Suppose Assumptions~\ref{assump:permutable}--\ref{ass3} hold. Assume that there exists some $k\geq 2$, along with nonrandom quantity $\epsilon_k(\gamma;r)$ such that $D_k(r,\hat{r})\leq\epsilon_k(\gamma;r)$ with probability at least $1-\gamma$.
        Then there exists a positive constant $C_0$ such that 
        \begin{eqnarray}
            \notag \underset{1\leq i\leq n+1}{\sup}\Delta_{i}(\omega) & \leq & C_0\left\{\dfrac{\omega m\epsilon_k(\gamma,r)}{(n+\omega m)h^{d/k}} + K_0^{-1}(h^{d})h + \dfrac{1}{(n+\omega m)^{1/2}h^{d/2}}\log^{1/2}\left(\dfrac{n}{\delta}\right)+\gamma^{1/2} \right\}\,
        \end{eqnarray}
        with probability over $1-\delta-\gamma^{1/2}$.
        %, where $\varphi=V_r\vee C_r \vee\left[\overline{L}_2\wedge\left\{  C_r m^{-1}h^{-d/k}+\epsilon_k(\gamma,r)h^{-d/k} \right\}\right]$.
    \end{lemma}

    The bound on $\sup_{1\leq i\leq n+1}\Delta_i(\omega)$ in Lemma~\ref{lemma:delta_concentration2} depends on two key components: (1) the accuracy of $\hat{r}(x, s)$ as an estimator of $r(x,s)$, and (2) the kernel estimation of $F_{S \mid X}(S_i \mid X_i)$. The first component is captured by the first term associated with $\epsilon_k(\gamma, r)$. The second component is reflected in $K_0^{-1}(h^d)h$ and $(n+\omega m)^{-1/2}h^{-d/2}$. For Gaussian kernel, $K_0^{-1}(h^{d})=O(\log^{1/2}(h^{-1}))$, while $K_0^{-1}(h^{d})=O(\log(h^{-1}))$ for Laplace kernel.
    
    We also note that all terms in the bound on $\sup_{1\leq i\leq n+1}\Delta_i(\omega)$ that involve auxiliary data are modulated by $\omega$. In practice, $\omega$ can be treated as a tuning parameter to control the level of incorporating auxiliary information. This flexibility ensures that the incorporation of auxiliary data is, at the very least, non-detrimental to the performance of ELCP. Further discussions on the choice and impact of $\omega$ are provided in subsequent sections.

    \begin{remark}\label{remark:epsilon}
        Assume $m$ is at least of order $n$ as is common in practice. The convergence rate of $\epsilon_k(\gamma; r)$ or $D_k(r, \hat{r})$ has been extensively studied for various density ratio estimators. 
        For instance, by extending the results of \citet{filipovic2025kernel}, as elaborated in Section~\ref{sec:ext_dk} of the supplementary material, we show that the kernel density machines (KDM) estimator satisfies $\epsilon_k(\gamma; r) = O\big({\log^{1/6}(2/\gamma)} n^{-1/6} \big)$ for any $k \geq 2$ under certain conditions.
        \citet{gizewski2022regularization} show that the kernelized unconstrained least-squares importance fitting (KuLSIF) approach attains $\epsilon_\infty(\gamma; r)=O(\log(1/\gamma)n^{-(2u-1)/(4u+2)})$ under certain smoothness conditions on $r(\cdot,\cdot)$, where $u>1/2$ quantifies its smoothness in reproducing kernel Hilbert space (RKHS).
        Moreover, under stronger conditions, some estimators can achieve $D_\infty(r, \hat{r}) = \sup_{x,s} \left| \hat{r}(x,s) - r(x,s) \right| = O_p(n^{-1/2})$ \citep{matsushita2023estimating}.
    \end{remark}

    The following theorem provides a non-asymptotic bound for the test-conditional miscoverage error in \eqref{eq:conditional_coverage_diff},
    which is a direct consequence of Lemma~\ref{lemma:conditional_coverage_bound} and Lemma~\ref{lemma:delta_concentration2}.
    
    \begin{theorem}[Test-conditional miscoverage error bound]\label{theo:conditional_coverage_error_bound}
       Suppose the conditions of Lemma~\ref{lemma:delta_concentration2} hold. Then there exists a positive constant $\widetilde{C}_0$ such that
        \begin{align}
            \notag&\left|\mathrm{pr}\left( Y_{n+1}\in\widehat{C}_{\alpha}^{\mathrm{ELCP}}(X_{n+1})\mid X_{n+1}=x_0 \right)-(1-\alpha)\right| - (n+1)^{-1} \\
            \leq &\widetilde{C}_0\left[\dfrac{\omega m\epsilon_k(\gamma,r)}{(n+\omega m)h^{d/k}} + K_0^{-1}(h^{d})h + \dfrac{1}{(n+\omega m)^{1/2}h^{d/2}}\left\{\log^{1/2}(n)+\log^{1/2}((n+\omega m)h^{d})\right\}+ \gamma^{1/2}\right]\,.\label{eq:test_conditional_error}
        \end{align}
    \end{theorem}

    Theorem~\ref{theo:conditional_coverage_error_bound} implies asymptotic test-conditional coverage of ELCP under reasonable conditions on $D_k(r,\hat{r})$ (or $\epsilon_k(\gamma,r)$), $h$ and $n+\omega m$.
    Setting $\omega=0$, the test-conditional miscoverage error for the LCP set satisfies the following bound:
    \begin{align*}
    & \left| \mathrm{pr}\left( Y_{n+1}\in\widehat{C}_{\alpha}^{\mathrm{LCP}}(X_{n+1}) \mid X_{n+1}=x_0 \right) - (1-\alpha)\right| \\
    \leq & (n+1)^{-1}+\widetilde{C}_0\left[ K_0^{-1}(h^d)h+(nh^d)^{-1/2}\left\{ \log^{1/2}(n)+\log^{1/2}(nh^{d}) \right\} \right]\,.
    \end{align*}
    In the ideal scenario where $r(x,s)$ is perfectly estimated such that $\epsilon_k(\gamma,r)=0$, ELCP significantly improves the test-conditional miscoverage error bound of LCP. Specifically, the term $(nh^d)^{-1/2}$ for LCP is improved to $\{(n+\omega m)h^d\}^{-1/2}$ for ELCP, reflecting an effective sample size increase to $(n+\omega m)h^d$ in estimating $F_{S\mid X}(S_i\mid X_i)$ due to the use of $\mathcal{Z}_m'$.  
    This also highlights the advantage of using a smaller $h$ for ELCP compared to LCP. Specifically, a smaller $h$ balances $K_0^{-1}(h^{-d})h$ with the improved error bound $\{(n+\omega m)h^d\}^{-1/2}$. 

    In practical situations where the estimation error of $\hat{r}(\cdot, \cdot)$ comes into play, achieving a sufficiently small order for $\epsilon_k(\gamma, r)$ or $D_k(r, \hat{r})$ can lead to an improved error bound for ELCP compared to LCP. 
    For LCP with $\omega=0$, the optimal error rate is $O(n^{-1/(d+2)})$ with optimal bandwidth $h\sim n^{-1/(d+2)}$, ignoring logarithmic factors.
    For ELCP with $\omega>0$ and $m$ at least of order $n$ as is common in practice: (1) when $r(\cdot,\cdot)$ is estimated with sufficient accuracy such that $h^{-d/k}\epsilon_k(\gamma;r)=o\big(\{(n+\omega m)h^d\}^{-1/2}\big)$, the optimal error bound of ELCP is $O((n+\omega m)^{-1/(d+2})$, outperforms LCP even for moderate $m$ thanks to the larger effective sample size. (2) when $h^{-d/k}\epsilon_k(\gamma;r)$ is not at a smaller order of $\{(n+\omega m)h^d\}^{-1/2}$, the optimal error rate is
    $O(\{\epsilon_k(\gamma;r)\}^{k/(d+k)})$ with $h\sim \{\epsilon_k(\gamma;r)\}^{k/(d+k)}$.
    If $\epsilon_k(\gamma;r)=O(n^{-c_{k}})$ for some $0<c_k<1/2$, the error bound of ELCP becomes $O(n^{-c_{k}k/(d+k)})$. 
    Hence, ELCP outperforms LCP whenever $c_{k}>(d+k)/\{k(d+2)\}$. 
    In particular, when $k=\infty$, the condition reduces to $c_{\infty}>1/(d+2)$, which guarantees a strictly smaller ELCP error bound.

    \begin{remark}
        As discussed in Remark~\ref{remark:epsilon}, the KDM estimator \citep{filipovic2025kernel} satisfies $c_{k}=1/6$ for all $k\geq 2$. Consequently, for sufficiently large $k$ and $d > 4$, the condition $c_{k}>(d+k)/\{k(d+2)\}$ is satisfied.
        For the KuLSIF approach \citep{gizewski2022regularization} with $c_{\infty}=(2u-1)/(4u+2)$, $c_{\infty}>1/(d+2)$ holds as long as $d>4$ and $u>(d+4)/(2d-8)$. 
        Furthermore, under the optimal rate $O_p(n^{-1/2})$ \citep{matsushita2023estimating}, $c_{\infty}=1/2>1/(d+2)$ always holds, implying that ELCP achieves a strictly smaller error bound than LCP.
    \end{remark}
    
    \begin{remark}
        When estimating $r(\cdot, \cdot)$ is particularly challenging, it is advisable to choose a relatively small value of $\omega$ to decrease the level of incorporating auxiliary information. Although this conservative choice may reduce the potential benefits of auxiliary data, it helps ensure that ELCP does not degrade performance. For practical guidance on selecting the parameters $\omega$ and $h$, please refer to Section~\ref{sec:implement}.
    \end{remark}

    \begin{remark}
        It is worth highlighting that our theoretical results also hold when $\omega > 1$.
        The condition $\omega \in [0,1]$ is imposed primarily for empirical and interpretational reasons. Since $\omega$ controls the relative contribution of auxiliary data to the target data, and auxiliary datasets are typically of comparable or lower relevance to the target population, this constraint reflects a realistic choice and prevents over-weighting auxiliary information.
    \end{remark}

    Now we consider weak test-conditional coverage, which extends the test-condition coverage from a specific test point to a broader subset of the feature space.

    \begin{theorem}[Weak test-conditional miscoverage error bound]\label{theo:weak_test}
    Consider a fixed set $\mathcal{B}\subset \mathcal{X}$ with $\mathrm{pr}(X_{n+1} \in \mathcal{B})=p_0$.
    Under the conditions of Lemma~\ref{lemma:delta_concentration2}, there exists a positive constant $\widetilde{C}_0$ such that
        \begin{align}
            \notag&\left|\mathrm{pr}\left( Y_{n+1}\in\widehat{C}_{\alpha}^{\mathrm{ELCP}}(X_{n+1})\mid X_{n+1}\in \mathcal{B} \right) - (1-\alpha)\right|\\
            \notag\leq& (n+1)^{-1}+\widetilde{C}_0p_0^{-1}(1-p_0)\left\{ 1-(1-p_0)^n \right\}\\
            \notag &\cdot\Big[\dfrac{\omega m\epsilon_k(\gamma,r)}{(n+\omega m)h^{d/k}}+ K_0^{-1}(h^{d})h + \dfrac{1}{(n+\omega m)^{1/2}h^{d/2}}\left\{\log^{1/2}(n)+\log^{1/2}((n+\omega m)h^{d})\right\}+\gamma^{1/2}\Big].%\label{eq:weak_test_bound}
        \end{align}
    \end{theorem}
    
    This result degenerates to marginal coverage when $\mathcal{B}$ equals the entire sample space $\mathcal{X}$.
    Furthermore, we observe that the weak test-conditional miscoverage error depends closely on the probability $p_0$ of $\mathcal{B}$.
    As $p_0$ increases such that $1-p_0$ approaches $0$, the order of the weak test-conditional miscoverage error decreases.
    Finally, ELCP achieves asymptotic weak test-conditional coverage when the estimation error $D_k(r,\hat{r})$ or $\epsilon_k(\gamma,r)$ is sufficiently small, and $n+\omega m$ is sufficiently large.

\section{Implementation details of ELCP}\label{sec:implement}

\subsection{Parameter selection}\label{sec:paramter_select}

    The parameter $\omega$ controls the level of incorporating auxiliary information. A larger $\omega$ increases reliance on the auxiliary data and reflects greater trust in the information provided by the auxiliary data. 
    The bandwidth $h$ plays a crucial role in determining the localization of the ELCP set. A smaller $h$ improves localization, allowing predictions to better adapt to distributional heterogeneity.
    Therefore, in practical applications, a criterion is needed to guide the selection of appropriate values for $h$ and $\omega$.

    Let $\hat{\beta}_{\omega,h}^y(x,s)$ and $\widehat{C}_{\alpha,\omega,h}^{\mathrm{ELCP}}(\cdot)$ be the counterparts of $\hat{\beta}_{\omega,\hat{r}}^y(x,s)$ and $\widehat{C}_\alpha^{\mathrm{ELCP}}(\cdot)$, respectively, to highlight their dependence on $h$ and $\omega$. Let $\mathcal{L}\left(\omega,h;\mathcal{Z}_n\cup\{(X_{n+1},y)\},\mathcal{Z}_m^\prime\right)$ be a general loss function which will be detailed in Section~\ref{sec:construct loss fun}. 
    For each $y\in\mathcal{Y}$, we find $\omega$ and $h$ that minimizes this loss function, i.e.,
    \begin{gather*}
        (\hat{\omega}(y),\hat{h}(y)) = \underset{(\omega,h)\in\mathcal{G}}{\arg\min}~\mathcal{L}\left(\omega,h;\mathcal{Z}_n\cup\{(X_{n+1},y)\},\mathcal{Z}_m^\prime\right)\,,
    \end{gather*}
    where $\mathcal{G}\subset [0,1]\times(0,\infty)$ is a candidate set.
    Consequently, the ELCP set with parameter selection is defined as
    \begin{align}
        \notag & \widehat{C}_\alpha^{\mathrm{ELCP-PS}}(X_{n+1}) = \left\{y:\hat{\beta}_{\hat{\omega}(y),\hat{h}(y)}^y(X_{n+1},S_{n+1}^y)\leq Q\left(1-\alpha;\dfrac{1}{n+1}\sum_{i=1}^{n+1}\delta_{\hat{\beta}_{\hat{\omega}(y),\hat{h}(y)}^y(X_i,S_i^y)}\right)\right\}\,.
    \end{align}
    Algorithm~\ref{twoagentsProcedure-parasel} in Section~\ref{sec:SM_detailed L2} of the supplementary material provides an end-to-end implementation of ELCP with parameter selection.
    
    The loss function  $\mathcal{L}\left(\omega,h;\mathcal{Z}_n\cup\{(X_{n+1},y)\},\mathcal{Z}_m^\prime\right)$ is implicitly required to be invariant under permutations of both $\mathcal{Z}_n\cup\{(X_{n+1},y)\}$ and $\mathcal{Z}_m^\prime$. This invariance ensures that $\hat{\omega}(y)$ and $\hat{h}(y)$ are likewise permutation-invariant with respect to these samples, thereby guaranteeing marginal coverage of the ELCP set with parameter selection \citep{liang2024conformal}, as established in the next theorem. The condition holds for many commonly used loss functions, including those introduced in the next section, provided that the test point $(X_{n+1},y)$ is included in the loss.
    
    \begin{theorem}\label{theo:modsel elcp}
        Suppose $\mathcal{Z}_n\cup\{Z_{n+1}\}$ are exchangeable and Assumption~\ref{assump:permutable} holds. 
        %{\color{blue} In addition, assume that $\mathcal{L}\left(\omega,h;\mathcal{Z}_n\cup\{(X_{n+1},y)\},\mathcal{Z}_m^\prime\right)$ is invariant to permutations within $\mathcal{Z}_n\cup\{(X_{n+1},y)\}$ and $\mathcal{Z}_m^\prime$.}
        Then for any given $\alpha\in(0,1)$,
        $1-\alpha\leq \mathrm{pr}\left( Y_{n+1}\in \widehat{C}_{\alpha}^{\mathrm{ELCP-PS}}(X_{n+1}) \right)<1-\alpha+(n+1)^{-1}$.
    \end{theorem}

\subsection{Construction of loss function $\mathcal{L}$}\label{sec:construct loss fun}

    A natural and widely used choice of the loss function is the average prediction set length, which directly reflects prediction efficiency \citep{liang2024conformal}. This defines the following loss function:
    \begin{gather*}
        \mathcal{L}_1\left(\omega,h;\mathcal{Z}_n\cup\{(X_{n+1},y)\},\mathcal{Z}_m^\prime\right)= (n+1)^{-1}\sum_{i=1}^{n+1}\left|\widehat{C}_{\alpha,\omega,h}^{\mathrm{ELCP}}(X_i)\right|\,.
    \end{gather*}
    However, the average length does not capture the test-conditional coverage properties of the prediction sets. To address this, we introduce a loss function that targets the test-conditional coverage. As established in Lemma~\ref{lemma:conditional_coverage_bound}, a smaller value of $\sup_{1\leq i\leq n+1}\Delta_i(\omega)$ corresponds to more precise test-conditional coverage, where $\sup_{1\leq i\leq n+1}\Delta_i(\omega)$ quantifies the accuracy of $\hat{\beta}_{\omega,h}(X_i,S_i)$ as an estimator of $F_{S\mid X}(S_i\mid X_i)$, which is $\mathrm{Uniform}[0,1]$ theoretically, across all $X_i,~i\in[n+1]$. Consequently, we design a novel loss function to explicitly measure this discrepancy between the estimated conditional distribution and the standard uniform distribution.

    We adopt the integrated conditional maximum mean discrepancy (ICMMD) proposed by \citet{yan2022distance} to construct the loss function, which quantifies the discrepancy between the conditional distribution of $\hat{\beta}_{\omega,h}(X_i,S_i^y)$ given $X_i$ and $\mathrm{Uniform}[0,1]$, as:
    \begin{align*}
        & \mathcal{L}_{2}\left(\omega,h;\mathcal{Z}_n\cup\{(X_{n+1},y)\},\mathcal{Z}_m^\prime\right) \\
        = &
        \dfrac{1}{n(n+1)}\sum_{1\leq i\neq j\leq n+1}\Big\{K_1(\hat{\beta}_{\omega,h}(X_i,S_i^y),\hat{\beta}_{\omega,h}(X_j,S_j^y)) -2\phi(\hat{\beta}_{\omega,h}(X_i,S_i^y))\Big\}K_2(X_i,X_j)\,,
    \end{align*}
    where $K_1(\cdot,\cdot)$ and $K_2(\cdot,\cdot)$ are kernel functions supported on $[0,1]$ and $\mathcal{X}$, respectively. In addition, $\phi(u)=E\left\{ K_1(u,U_1) \right\}$ with $U_1\sim\mathrm{Uniform}[0,1]$. See Section~\ref{sec:SM_detailed L2} of the supplementary material for detailed formulation.
    
    In practice, we optimize the loss function over $\omega$ and $h$ through a grid search. Let $(\omega^*,h^*) = \underset{(\omega,h)\in \mathcal{G}}{\arg\min}~\mathcal{R}_{\omega,h}^{(n,m)}$ be the optimal parameters in $\mathcal{G}$, where 
    \begin{align*}
        \mathcal{R}_{\omega,h}^{(n,m)}=E\left[ \left\{ K_1(\hat{\beta}_{\omega,h}(X_1,S_1),\hat{\beta}_{\omega,h}(X_2,S_2))-2\phi(\hat{\beta}_{\omega,h}(X_1,S_1)) \right\}K_2(X_1,X_2) \right]
    \end{align*}
    and $\mathcal{G}=\{(w,h): w\in\{w_\ell\}_{\ell\in[L]}, h\in\{h_\ell\}_{\ell\in[L]}\}$ is the candidate search set with a positive integer $L$.
    Similarly, let $(\hat{\omega},\hat{h})=\underset{(\omega,h)\in \mathcal{G}}{\arg\min}~\mathcal{L}_{2}\left(\omega,h;\mathcal{Z}_n\cup\{Z_{n+1}\},\mathcal{Z}_m^\prime\right)$. 
    Given the concentration properties of $\mathcal{L}_2$ established by Lemma~\ref{lemma: para_select_conv} in Section~\ref{proof of theo:parameter_selection_consistent} of the supplementary material, the optimal parameters can be consistently identified under suitable conditions.
    \begin{theorem}\label{theo:parameter_selection_consistent}
        Suppose Assumptions~\ref{assump:permutable}--\ref{ass3} hold. Assume that the kernel functions $K_1(\cdot,\cdot)$ and $K_2(\cdot,\cdot)$ are uniformly bounded by $D_{K,0}>0$ and that their partial derivatives are bounded in absolute value by $D_{K,1}>0$. Assume that $\inf_{(\omega,h)\in\mathcal{G}\setminus\{(\omega^*,h^*)\}}\mathcal{R}_{\omega,h}^{(n,m)}-\mathcal{R}_{\omega^*,h^*}^{(n,m)}>\zeta n^{-1/2}\mathrm{log}^{1/2}n$ for some positive constant $\zeta$. If $(n+\omega m)h_{\rm min}^d\log^{-1}(Ln)\rightarrow\infty$ with $h_{\rm min}=\inf_{\ell\in[L]}h_\ell$ and $L=o(n^{\overline{C}_1\zeta/2})$, we have $\mathrm{pr}((\hat{\omega},\hat{h})=(\omega^*,h^*))\to 1$ as $n\to\infty$.
    \end{theorem}

\subsection{Computationally efficient implementations of ELCP}\label{sec:compute}
    
    In this section, we present computationally efficient implementations of ELCP along with their theoretical justifications. First, in Algorithm~\ref{twoagentsProcedure}, we consider replacing $\hat{r}(x,s)$, which uses $(X_{n+1}, y)$, with $\tilde{r}(x,s)$, an estimator of $r(x,s)$ based on $\mathcal{Z}_n$ and $\mathcal{Z}^\prime_m$ without using $(X_{n+1}, y)$.
    This can avoid updating $\hat{r}(x,s)$ for each $y\in\mathcal{Y}$, thereby speeding up the computation of ELCP. 
    Algorithm~\ref{algo:computeeffimple ELCP} in Section~\ref{sec:sup computation} of the supplementary material presents the computationally efficient implementation of ELCP. 
    Let $\widetilde{C}_{\alpha}^{\rm ELCP}(X_{n+1})$ be the ELCP set with $\hat{r}(x,s)$ replaced by $\tilde{r}(x,s)$ in $\{\hat{\beta}_{\omega,\hat{r}}^{y}(X_i,S_i^{y})\}_{i\in[n+1]}$.
    \begin{theorem}\label{theo:marginal_coverage_r}
        Let $\tilde{\delta}^y=\sup_{1\leq i\leq n+1}\big|\hat{\beta}_{\omega,\hat{r}}^y(X_i,S_i^y)-\hat{\beta}_{\omega,\tilde{r}}^y(X_i,S_i^y)\big|$. \\
        \noindent(i) The difference between $\widehat{C}_{\alpha}^{\rm ELCP}(X_{n+1})$ and $\widetilde{C}_{\alpha}^{\rm ELCP}(X_{n+1})$ is a subset of
        \begin{align}
            \notag D_{\hat{r},\tilde{r}}(X_{n+1}):=\left\{ y:\hat{\beta}^y_{\omega,\tilde{r}}(X_{n+1},S_{n+1}^y)\in\big[\tilde{q}_\alpha^y-2\tilde{\delta}^y,\tilde{q}_\alpha^y+2\tilde{\delta}^y\big] \right\}\,,
        \end{align}
        where $\tilde{q}_\alpha^y=Q\left( 1-\alpha;(n+1)^{-1}\left\{ \sum_{i=1}^{n+1}\delta_{\hat{\beta}_{\omega,\tilde{r}}^y(X_i,S_i^y)}\right\} \right)$; \\
        \noindent(ii) Moreover, suppose $\mathcal{Z}_n\cup\{Z_{n+1}\}$ are exchangeable and Assumption~\ref{assump:permutable} holds, then
        \begin{align*}
            1-\alpha - \mathrm{pr}(Y_{n+1}\in D_{\hat{r},\tilde{r}}(X_{n+1}))&\leq \mathrm{pr}\left( Y_{n+1}\in \widetilde{C}_{\alpha}^{\mathrm{ELCP}}(X_{n+1}) \right)\\
            &<1-\alpha + (n+1)^{-1} + \mathrm{pr}(Y_{n+1}\in D_{\hat{r},\tilde{r}}(X_{n+1}))\,;
        \end{align*}
        \noindent(iii) Furthermore, assume that $\sup_{x,s}|\hat{r}(x,s)-\tilde{r}(x,s)|\leq C_r n^{-1}$ for some positive constant $ C_r $ and Assumption~\ref{ass2} holds, we have $\tilde{\delta}^y=O(n^{-1})$.
    \end{theorem}

    \begin{remark}
        Theorem~\ref{theo:marginal_coverage_r}-(iii) implies that the interval $[\tilde{q}_\alpha^y-2\tilde{\delta}^y,\tilde{q}_\alpha^y+2\tilde{\delta}^y]$ has a length of order $O(n^{-1})$ under certain conditions. Note that $\hat{\beta}^y_{\omega,\tilde{r}}(X_{n+1},S_{n+1}^y)$ converges in distribution to $\mathrm{Uniform}[0,1]$ as indicated by Lemma~\ref{lemma:delta_concentration2}. Thus, when the density of $\hat{\beta}^y_{\omega,\tilde{r}}(X_{n+1},S_{n+1}^y)$ is bounded, we find that $\mathrm{pr}(Y_{n+1}\in D_{\hat{r},\tilde{r}}(X_{n+1}))$ is also of $O(n^{-1})$.
        Consequently, the effect on the marginal coverage of the ELCP set from excluding $(X_{n+1},y)$ in the estimation of $r(x,s)$ would be negligible.
    \end{remark}

    Next, we turn to $\widehat{C}_{\alpha}^{\mathrm{ELCP\text{-}PS}}(X_{n+1})$, the ELCP prediction set with parameter selection, which can be computationally demanding, as it requires re-calculating $\hat{r}(x,s)$ and re-selecting the parameters $\hat{\omega}(y)$ and $\hat{h}(y)$ for each candidate value of $y$, along with an exhaustive search over all candidate values of $y$ to construct the final prediction set. 
    To alleviate this burden, we adopt a similar strategy: replace $\hat{r}(x,s)$ with $\tilde{r}(x,s)$, and select $\hat{\omega}(y)$ and $\hat{h}(y)$ by excluding $(X_{n+1},y)$ in the loss function. 
    This yields a close approximation to $\widehat{C}_{\alpha}^{\mathrm{ELCP\text{-}PS}}$ with substantially reduced computational cost.
    A detailed algorithm for the computationally efficient implementation of ELCP with parameter selection is provided in Section~\ref{sec:sup computation} of the supplementary material.

\section{Numerical experiments}\label{sec:04_num}

    We conduct synthetic experiments along with a real data analysis, to illustrate the marginal coverage validity and local coverage improvement of ELCP, compared to LCP and RLCP \citep{hore2023conformal} without auxiliary data, LCP using combined data (LCP-C) in \eqref{eq:comb_LCP}, and RLCP with combined data (RLCP-C, constructed analogously to LCP-C). Additional results on synthetic data, semi-supervised settings, and another real data analysis are provided in Section~\ref{sec:SM_numerical} of the supplementary material.
    ELCP is implemented using the computationally efficient procedure described in Section~\ref{sec:compute}.
    
    In the simulation study, the regression model for calibration and test data is $Y = \mu(X) + \epsilon(X)$,
    where $\mu(x)=E(Y\mid X=x)$ and the residual $\epsilon(X)$ may depend on $X$. For auxiliary data, the model is $Y^{\prime} = \mu^\prime(X^{\prime}) + \epsilon^{\prime}(X^{\prime})$.
    We evaluate all methods under a nominal coverage level of $1-\alpha=90\%$.
    For a prediction set $\widehat{C}_{\alpha}(X_{n+1})$, the following criteria are used for comparison:
    (1) marginal coverage
    $\mathrm{pr}\big( Y_{n+1}\in\widehat{C}_{\alpha}(X_{n+1})\big)$; (2) test-conditional miscoverage error
    $E\big\{ \big|\mathrm{pr}\big( Y_{n+1}\in\widehat{C}_{\alpha}(X_{n+1})\mid X_{n+1}\big)-(1-\alpha)\big| \big\}$; and (3) average size of the prediction set.

    All simulation results are based on $100$ replications, with training set sizes $n$ for $\mathcal{D}_{\rm tr}$ and $m$ for $\mathcal{D}_{\rm tr}^{\prime}$, matching the sizes of the calibration data $\mathcal{Z}_n$ and auxiliary data $\mathcal{Z}_{m}^{\prime}$, respectively.
    Notably, infinite prediction sets can be occasionally produced, particularly for small bandwidths. 
    To ensure meaningful comparisons, we report results when no more than 5\% of the prediction sets are infinite.

\subsection{Results for synthesized data}\label{sec:simu_synthesized}

    We generate covariates $X$ and $X^\prime$ from $\mathrm{Uniform}([-2,2]^d)$ with dimension $d\in\{5,10,15,20\}$. Writing $X=(X_1,\ldots,X_d)^{\top}$ and $X^\prime=(X_1^\prime,\ldots,X_d^\prime)^{\top}$, we set the regression functions as $\mu(X)=\sum_{i=1}^{d}X_i$ and $\mu^\prime(X^\prime)=2\sum_{i=1}^{d}X_i^\prime$.
    We consider the following data-generating process (DGP) settings:
    
    \noindent\textbf{DGP1:} $\epsilon(X)\sim N\left( 0, \sum_{i=1}^{5}|\arccos(X_i/2)| \right)$, $\epsilon^\prime(X^\prime)\sim N\left( 0, 1.5\sum_{i=1}^{5}|\arccos(X_i^\prime/2)| \right)$. \\
    \noindent\textbf{DGP2:} $\epsilon(X)\sim N\left( 0, \exp\big(\sum_{i=1}^{5}X_i/2\big) \right)$, $\epsilon^\prime(X^\prime)\sim N\left( 0, 1.5\exp\big(\sum_{i=1}^{5}X_i^\prime/2\big) \right)$. \\
    \noindent\textbf{DGP3:} $\epsilon(X)\sim N\left( 0, \big|\sum_{i=1}^{5}X_i/2\big|^2 \right)$, $\epsilon^\prime(X^\prime)\sim N\left( 0, 1.5\big|\sum_{i=1}^{5}X_i^\prime/2\big|^2 \right)$.

    For all methods, the conformity scores used are $S(x,y)=\left|y-\hat{\mu}(x)\right|$ and $S^{\prime}(x,y)=\left|y-\hat{\mu}^\prime(x)\right|$, where $\hat{\mu}(\cdot)$ and $\hat{\mu}^\prime(\cdot)$ are linear regression functions trained on  $\mathcal{D}_{\mathrm{tr}}$ and $\mathcal{D}_{\mathrm{tr}}^{\prime}$. Since the pre-training may also benefit from incorporating the auxiliary data, model training may alternatively be performed by merging $\mathcal{D}_{\mathrm{tr}}^\prime$ into $\mathcal{D}_{\mathrm{tr}}$ after applying importance weighting \citep{pan2009survey}. Implementation details and results for this approach are provided in Section~\ref{sec:supp_simu_score} of the supplementary material. 
    We consider three calibration data sizes, $n=100, 150$ and $200$.
    We use Gaussian kernel $K_0(u)=\exp(-u^2/2)$ for $u\in\mathbb{R}$. 
    Each method is implemented with bandwidth $h \in \{0.4, 0.6, 0.8, 1.0, 1.2, 1.4, 1.6,1.8,2.0,2.5,3.0,3.5,4\}$.

    For ELCP, we estimate the density ratio $r(\cdot,\cdot)$ using the quadratic discriminant analysis (QDA) classifier calibrated via Platt scaling \citep{niculescu2005predicting}.
    Results for alternative density ratio estimators, including the KLIEP estimator \citep{sugiyama2008direct} and random forest (RF) with Platt scaling, can be found in Section~\ref{sec:supp_simu_dre} of the supplementary material. 
    The candidate set for the parameter $\omega$ is $\{0, 0.01, 0.05, 0.1, 0.2, 0.4, 0.7, 1.0\}$.
    The auxiliary data size is set to $m/n=5$. Results for other values of $m/n$, examining the effect of auxiliary data size, are provided in Section~\ref{sec:supp_simu_auxiliary_size} of the supplementary material.

    \subsubsection{Results under fixed $\omega$ and $h$}

    In this section, we present the results for each method with parameters $\omega$ and $h$ fixed and chosen from the candidate set.
    
    \noindent\textbf{Marginal coverage:} 
    Tables~\ref{table:mardgp1}--\ref{table:mardgp3} in Section~\ref{sec:supp_marginal} of the supplementary material report the marginal coverage of ELCP with $\omega=1$, LCP, RLCP, LCP-C and RLCP-C for DGP1--DGP3, across varying $h$, $n$ and $d$. 
    Results for ELCP with other values of $\omega$ are similar and therefore are omitted.
    Overall, ELCP, LCP and RLCP achieve valid marginal coverage close to $90\%$, whereas both LCP-C and RLCP-C exhibit substantially overcoverage.
    This suggests that when the auxiliary information is imperfect, directly combining calibration and auxiliary data can lead to unreliable prediction sets.
    
    \noindent\textbf{Test-conditional miscoverage error and set size:}
    Table~\ref{table:testDGP123} shows the smallest test-conditional miscoverage error and the smallest average prediction set size for each method across all parameter values ($\omega$ and $h$ for ELCP, and $h$ for LCP and RLCP). 
    The numbers in parentheses represent the percentage improvement of ELCP compared to the better-performing method between LCP and RLCP.
    From Table~\ref{table:testDGP123}, it is evident that ELCP consistently achieves lower miscoverage errors and produces smaller prediction sets compared to LCP and RLCP across all scenarios.
    Section~\ref{sec:supp_simu_parameter1} of the supplementary material presents the effect of $\omega$ and $h$ on the test-conditional miscoverage error of ELCP.
    
    \begin{table*}[htb]
        \caption{Smallest test-conditional miscoverage error and average prediction set size across all parameter values (with ELCP reduction ratio in parentheses) for DGP1--DGP3.}
        \begin{center}
        {\fontsize{7}{7}\selectfont{\setlength{\tabcolsep}{1.5pt}\begin{tabular}{ccccccccccccccccc}
        \toprule
        & & \multicolumn{3}{c}{$d=5$} & \multicolumn{3}{c}{$d=10$} & \multicolumn{3}{c}{$d=15$} & \multicolumn{3}{c}{$d=20$}\vspace{-2pt}\\
         & $n$ & ELCP & LCP & RLCP & ELCP & LCP & RLCP & ELCP & LCP & RLCP & ELCP & LCP & RLCP \\
        \cmidrule(lr){3-5} \cmidrule(lr){6-8} \cmidrule(lr){9-11} \cmidrule(lr){12-14}
        & & \multicolumn{12}{c}{Test-conditional miscoverage error}\\
        DGP1 & $100$ & $\textbf{0.014(22.7\%)}$ & $0.018$ & $0.018$ & $\textbf{0.014(14.1\%)}$ & $0.016$ & $0.019$ & $\textbf{0.016(23.7\%)}$ & $0.020$ & $0.022$ & $\textbf{0.016(17.8\%)}$ & $0.020$ & $0.020$\\
         & $150$ & $\textbf{0.010(26.1\%)}$ & $0.014$ & $0.016$ & $\textbf{0.013(12.2\%)}$ & $0.015$ & $0.019$ & $\textbf{0.014(24.1\%)}$ & $0.018$ & $0.020$ & $\textbf{0.016(17.4\%)}$ & $0.020$ & $0.020$\\
         & $200$ & $\textbf{0.011(21.6\%)}$ & $0.014$ & $0.016$ & $\textbf{0.014(19.6\%)}$ & $0.018$ & $0.020$ & $\textbf{0.014(18.3\%)}$ & $0.017$ & $0.020$ & $\textbf{0.016(12.9\%)}$ & $0.019$ & $0.021$\\
        DGP2 & $100$ & $\textbf{0.036(46.5\%)}$ & $0.070$ & $0.068$ & $\textbf{0.044(43.5\%)}$ & $0.077$ & $0.078$ & $\textbf{0.044(45.9\%)}$ & $0.085$ & $0.081$ & $\textbf{0.047(42.5\%)}$ & $0.082$ & $0.082$\\
         & $150$ & $\textbf{0.035(47.0\%)}$ & $0.066$ & $0.067$ & $\textbf{0.046(38.3\%)}$ & $0.075$ & $0.084$ & $\textbf{0.049(38.5\%)}$ & $0.080$ & $0.083$ & $\textbf{0.052(39.5\%)}$ & $0.088$ & $0.086$\\
         & $200$ & $\textbf{0.033(49.3\%)}$ & $0.066$ & $0.068$ & $\textbf{0.043(44.3\%)}$ & $0.077$ & $0.085$ & $\textbf{0.052(32.5\%)}$ & $0.078$ & $0.087$ & $\textbf{0.049(40.9\%)}$ & $0.084$ & $0.090$\\
        DGP3 & $100$ & $\textbf{0.060(33.3\%)}$ & $0.094$ & $0.090$ & $\textbf{0.073(20.0\%)}$ & $0.094$ & $0.092$ & $\textbf{0.074(20.0\%)}$ & $0.100$ & $0.093$ & $\textbf{0.073(15.2\%)}$ & $0.092$ & $0.086$\\
         & $150$ & $\textbf{0.060(33.3\%)}$ & $0.090$ & $0.091$ & $\textbf{0.077(18.9\%)}$ & $0.095$ & $0.098$ & $\textbf{0.079(17.9\%)}$ & $0.098$ & $0.096$ & $\textbf{0.079(14.3\%)}$ & $0.098$ & $0.093$\\
         & $200$ & $\textbf{0.063(31.3\%)}$ & $0.091$ & $0.093$ & $\textbf{0.071(28.5\%)}$ & $0.099$ & $0.102$ & $\textbf{0.082(16.5\%)}$ & $0.098$ & $0.100$ & $\textbf{0.076(20.9\%)}$ & $0.096$ & $0.096$\\
         & & \multicolumn{12}{c}{Mean prediction set size}\\
         DGP1 & $100$ & $\textbf{9.525(1.2\%)}$ & $9.636$ & $9.893$ & $\textbf{10.004(0.5\%)}$ & $10.058$ & $10.478$ & $\textbf{10.439(0.8\%)}$ & $10.527$ & $10.999$ & $\textbf{10.381(0.9\%)}$ & $10.473$ & $11.004$\\
         & $150$ & $\textbf{9.411(0.7\%)}$ & $9.479$ & $9.705$ & $\textbf{9.628(0.8\%)}$ & $9.710$ & $9.995$ & $\textbf{9.769(0.7\%)}$ & $9.835$ & $10.161$ & $\textbf{9.940(0.6\%)}$ & $10.002$ & $10.427$\\
         & $200$ & $\textbf{9.352(0.9\%)}$ & $9.435$ & $9.569$ & $\textbf{9.437(0.7\%)}$ & $9.506$ & $9.687$ & $\textbf{9.686(0.4\%)}$ & $9.730$ & $9.964$ & $\textbf{9.815(0.3\%)}$ & $9.847$ & $10.164$\\
        DGP2 & $100$ & $\textbf{4.080(4.5\%)}$ & $4.273$ & $5.127$ & $\textbf{4.340(5.6\%)}$ & $4.599$ & $5.393$ & $\textbf{4.758(4.8\%)}$ & $5.000$ & $5.826$ & $\textbf{4.893(2.5\%)}$ & $5.018$ & $5.895$\\
         & $150$ & $\textbf{4.010(2.5\%)}$ & $4.113$ & $4.824$ & $\textbf{4.085(5.1\%)}$ & $4.305$ & $4.982$ & $\textbf{4.422(4.7\%)}$ & $4.639$ & $5.320$ & $\textbf{4.568(4.2\%)}$ & $4.768$ & $5.458$\\
         & $200$ & $\textbf{3.976(3.0\%)}$ & $4.098$ & $4.799$ & $\textbf{4.011(4.3\%)}$ & $4.192$ & $4.850$ & $\textbf{4.250(4.2\%)}$ & $4.436$ & $5.076$ & $\textbf{4.346(3.5\%)}$ & $4.504$ & $5.167$\\
        DGP3 & $100$ & $\textbf{3.839(8.5\%)}$ & $4.196$ & $4.696$ & $\textbf{4.312(5.9\%)}$ & $4.585$ & $4.968$ & $\textbf{4.497(3.0\%)}$ & $4.634$ & $5.123$ & $\textbf{4.674(3.5\%)}$ & $4.842$ & $5.363$\\
         & $150$ & $\textbf{3.730(6.8\%)}$ & $4.000$ & $4.442$ & $\textbf{4.076(6.3\%)}$ & $4.349$ & $4.627$ & $\textbf{4.307(3.7\%)}$ & $4.472$ & $4.809$ & $\textbf{4.496(2.1\%)}$ & $4.594$ & $4.961$\\
         & $200$ & $\textbf{3.663(6.5\%)}$ & $3.916$ & $4.398$ & $\textbf{3.941(5.5\%)}$ & $4.168$ & $4.431$ & $\textbf{4.194(3.6\%)}$ & $4.351$ & $4.647$ & $\textbf{4.338(2.8\%)}$ & $4.462$ & $4.770$\\
        \bottomrule
        \end{tabular}}}
        \end{center}
        \label{table:testDGP123}
    \end{table*}

    \subsubsection{Results under data-driven selected $\omega$ and $h$}\label{sec: exp selected wh}

    In practice, the optimal values of $\omega$ and $h$ are unknown, so we employ data-driven methods to select these parameters. For ELCP, $\omega$ and $h$ are chosen by minimizing the loss function $\mathcal{L}_2$ with Gaussian kernels for both $K_1(\cdot,\cdot)$ and $K_2(\cdot,\cdot)$, as introduced in Section~\ref{sec:construct loss fun}. 
    For LCP, $h$ is selected in the same way by setting $\omega = 0$. Parameter selection for both methods follows the efficient procedure described in Section~\ref{sec:compute}. 
    For RLCP, $h$ is chosen according to the specification in \citet{barber2023conformal}, using an effective size of 50.
    
    Table~\ref{table:mar for parasel} in Section~\ref{sec:supp_marginal_selected} of the supplementary material shows the marginal coverage of different prediction sets under data-driven parameter selection, indicating that ELCP, LCP, and RLCP all achieve valid marginal coverage.
    Table~\ref{table:testDGP123para} presents the test-conditional miscoverage error comparison and Table~\ref{table:sizeDGP123para} in Section~\ref{sec:supp_marginal_selected} presents the mean prediction set size comparison for ELCP, LCP, and RLCP. Under parameter selection, ELCP still significantly outperforms both LCP and RLCP. 
    \begin{table*}[htb]
        \caption{Test-conditional miscoverage error under data-driven selected $\omega$ and $h$ (with ELCP reduction ratio in parentheses) for DGP1--DGP3.}
        \begin{center}
        {\fontsize{7}{7}\selectfont{\setlength{\tabcolsep}{1.5pt}\begin{tabular}{ccccccccccccccccc}
        \toprule
        & & \multicolumn{3}{c}{$d=5$} & \multicolumn{3}{c}{$d=10$} & \multicolumn{3}{c}{$d=15$} & \multicolumn{3}{c}{$d=20$}\vspace{-2pt}\\
         & $n$ & ELCP & LCP & RLCP & ELCP & LCP & RLCP & ELCP & LCP & RLCP & ELCP & LCP & RLCP \\
        \cmidrule(lr){3-5} \cmidrule(lr){6-8} \cmidrule(lr){9-11} \cmidrule(lr){12-14}
        DGP1 & $100$ & $\textbf{0.015(17.5\%)}$ & $0.018$ & $0.019$ & $\textbf{0.016(16.8\%)}$ & $0.019$ & $0.020$ & $\textbf{0.016(15.2\%)}$ & $0.019$ & $0.021$ & $\textbf{0.017(10.4\%)}$ & $0.019$ & $0.020$\\
        & $150$ & $\textbf{0.013(20.7\%)}$ & $0.017$ & $0.016$ & $\textbf{0.014(7.5\%)}$ & $0.016$ & $0.018$ & $\textbf{0.014(16.8\%)}$ & $0.017$ & $0.020$ & $\textbf{0.017(10.1\%)}$ & $0.019$ & $0.020$\\
        & $200$ & $\textbf{0.013(20.3\%)}$ & $0.017$ & $0.017$ & $\textbf{0.015(10.8\%)}$ & $0.017$ & $0.019$ & $\textbf{0.015(11.5\%)}$ & $0.017$ & $0.020$ & $\textbf{0.017(8.3\%)}$ & $0.019$ & $0.021$\\
        DGP2 & $100$ & $\textbf{0.048(18.9\%)}$ & $0.059$ & $0.074$ & $\textbf{0.052(27.2\%)}$ & $0.072$ & $0.079$ & $\textbf{0.054(26.2\%)}$ & $0.073$ & $0.082$ & $\textbf{0.059(20.8\%)}$ & $0.075$ & $0.081$\\
        & $150$ & $\textbf{0.046(19.1\%)}$ & $0.056$ & $0.068$ & $\textbf{0.051(30.1\%)}$ & $0.074$ & $0.084$ & $\textbf{0.054(24.7\%)}$ & $0.072$ & $0.084$ & $\textbf{0.059(25.5\%)}$ & $0.079$ & $0.087$\\
        & $200$ & $\textbf{0.041(19.9\%)}$ & $0.051$ & $0.063$ & $\textbf{0.050(26.4\%)}$ & $0.068$ & $0.082$ & $\textbf{0.057(22.0\%)}$ & $0.073$ & $0.087$ & $\textbf{0.059(24.5\%)}$ & $0.079$ & $0.089$\\
        DGP3 & $100$ & $\textbf{0.078(4.8\%)}$ & $0.082$ & $0.094$ & $\textbf{0.081(11.6\%)}$ & $0.092$ & $0.093$ & $\textbf{0.086(8.1\%)}$ & $0.097$ & $0.093$ & $\textbf{0.082(5.7\%)}$ & $0.091$ & $0.087$\\
        & $150$ & $\textbf{0.071(6.7\%)}$ & $0.076$ & $0.093$ & $\textbf{0.083(10.2\%)}$ & $0.093$ & $0.098$ & $\textbf{0.085(9.7\%)}$ & $0.095$ & $0.096$ & $\textbf{0.087(7.0\%)}$ & $0.096$ & $0.094$\\
        & $200$ & $\textbf{0.070(8.8\%)}$ & $0.076$ & $0.091$ & $\textbf{0.080(17.3\%)}$ & $0.097$ & $0.101$ & $\textbf{0.085(10.1\%)}$ & $0.095$ & $0.100$ & $\textbf{0.088(8.9\%)}$ & $0.097$ & $0.097$\\
        \bottomrule
        \end{tabular}}}
        \end{center}
        \label{table:testDGP123para}
    \end{table*}

\subsection{Real data analysis: predicting Moscow housing price}\label{sec:real_data_housing}

    %(This dataset is available at https://www.kaggle.com/datasets/egorkainov/moscow-housing-price-dataset)

    This section presents an empirical analysis of Moscow housing market data (available on \texttt{Kaggle.com}) to compare LCP, RLCP, LCP-C, RLCP-C, and ELCP for apartment price prediction. Properties are grouped by distance to the nearest subway station via K-means clustering on station coordinates. The cluster corresponding to outlying urban areas (581 properties) is designated as the target dataset, while the two adjacent suburban clusters (1,813 properties) serve as the auxiliary dataset. Further implementation details are provided in Section~\ref{sec:supp_housing} of the supplementary material.

    The target dataset is randomly partitioned into training (193), calibration (193), and test (195) subsets, while the auxiliary dataset is split into training (906) and calibration (907) subsets. 
    This partitioning is repeated 100 times, and the results are averaged over all repetitions. 
    The features used in the prediction model are \textit{Minutes to metro}, \textit{Total area}, \textit{Living area ratio}, \textit{Number of rooms}, \textit{Floor ratio}, \textit{Number of floors}, \textit{Apartment type}, and \textit{Renovation type}. The last two features are categorical and are therefore converted into one-hot encoding during training, resulting in a 12-dimensional feature vector. For calibration, we exclude the one-hot encoded components, using only the 6 continuous features. 

    The residual score is used with a random forest regression model employed to train the point predictor. We also examine other scores and pre-training schemes detailed in Section~\ref{sec:supp_housing} of the supplementary material. The $\omega$ is chosen from $\{0.0, 0.1, 0.4, 0.7, 1.0\}$, while the bandwidth parameter $h$ is selected from $\{0.1, 0.2, 0.4, 0.6, 0.8, 1.0, 1.2, 1.4, 1.6, 1.8, 2.0\}$.
    
    The Gaussian kernel is employed for $K_0(\cdot,\cdot)$, as well as $K_1(\cdot,\cdot)$ and $K_2(\cdot,\cdot)$ in the parameter selection procedure, consistent with Section~\ref{sec: exp selected wh}. 
    Across all 100 trials, ELCP selected positive values of $\omega$. The mean selected bandwidths are $0.86$ for ELCP, $0.99$ for LCP, and $1.14$ for RLCP, corroborating the theoretical analysis in Section~\ref{subsec::Local Enhancement} that auxiliary data promotes the selection of smaller bandwidths. 
    Over 100 experiments, the average marginal coverage rates were 0.898 for ELCP, 0.901 for LCP, and 0.905 for RLCP, while LCP-C and RLCP-C achieved 0.849 and 0.882. Their failure to maintain nominal coverage underscores the distributional heterogeneity between target and auxiliary datasets, showing that direct dataset combination yields unreliable coverage performance.
    
    To further investigate how ELCP improves conditional coverage, we perform K-means clustering on the 581 target data points, partitioning the spatial domain into 10 non-overlapping subsets. 
    This allows us to compute the weak test-conditional coverage for each method within each subset over $100$ repeated experiments. 
    Table~\ref{table:houseTS} shows the deviation of the weak test-conditional coverage from the target coverage of $0.9$ for each method and the proportion of observations in each interval (first row).
    The column ``Agg'' represents the aggregated miscoverage error, computed as the sum of the deviations across all intervals, weighted by interval proportions. 
    Numbers in parentheses present corresponding average prediction set sizes. 
    Compared to LCP and RLCP, ELCP achieves the smallest weak test-conditional miscoverage error in over half of the $10$ subsets. Additionally, ELCP provides the shortest prediction set size in the majority of these subsets. The aggregated results further highlight that ELCP outperforms LCP and RLCP.
    \begin{table*}[h]
        \caption{Weak test-conditional miscoverage errors and average set sizes (in parentheses) in predicting Moscow housing price.}
        \begin{center}
        {\fontsize{10}{10}\selectfont{\setlength{\tabcolsep}{3pt}\begin{tabular}{cccccccccccc}
        \toprule
        Subset Index & 1 & 2 & 3 & 4 & 5 & 6 & 7 & 8 & 9 & 10 & Agg\vspace{4pt}\\ \midrule
        Prop. & $0.155$ & $0.120$ & $0.117$ & $0.115$ & $0.107$ & $0.106$ & $0.106$ & $0.091$ & $0.065$ & $0.018$ & $1.000$\\
        ELCP & $\textbf{0.070}$ & $\textbf{0.066}$ & $\textbf{0.023}$ & $\textbf{0.053}$ & $\textbf{0.122}$ & $\textbf{0.018}$ & $\textbf{0.001}$ & $\textbf{0.030}$ & $\textbf{0.027}$ & $\textbf{0.029}$ & $\textbf{0.0360}$\\
        & $(11.50)$ & $(6.20)$ & $(4.19)$ & $(10.75)$ & $(9.15)$ & $(7.33)$ & $(7.21)$ & $(7.93)$ & $(5.92)$ & $(7.20)$ & $(7.72)$\\
        LCP & $\textbf{0.048}$ & $\textbf{0.057}$ & $\textbf{0.068}$ & $\textbf{0.065}$ & $\textbf{0.122}$ & $\textbf{0.021}$ & $\textbf{0.007}$ & $\textbf{0.042}$ & $\textbf{0.039}$ & $\textbf{0.034}$ & $\textbf{0.0422}$\\
        & $(12.12)$ & $(5.77)$ & $(5.20)$ & $(10.48)$ & $(8.71)$ & $(7.27)$ & $(7.31)$ & $(7.68)$ & $(6.16)$ & $(7.07)$ & $(7.74)$\\
        RLCP & $\textbf{0.089}$ & $\textbf{0.071}$ & $\textbf{0.085}$ & $\textbf{0.070}$ & $\textbf{0.012}$ & $\textbf{0.035}$ & $\textbf{0.006}$ & $\textbf{0.028}$ & $\textbf{0.043}$ & $\textbf{0.036}$ & $\textbf{0.0471}$\\
        & $(13.80)$ & $(9.63)$ & $(13.14)$ & $(13.84)$ & $(28.94)$ & $(12.49)$ & $(8.52)$ & $(9.45)$ & $(7.47)$ & $(9.52)$ & $(10.97)$\\
        \bottomrule
        \end{tabular}}}
        \end{center}
        \label{table:houseTS}
    \end{table*}

\section{Concluding remarks}\label{sec:05_conclusion}

    There are several promising directions for future research. First, the computational cost of ELCP increases with the size of the auxiliary dataset. Developing strategies to select an effective subset of auxiliary data could help address this challenge. Second, both ELCP and LCP rely on kernel estimation, which can perform poorly in high-dimensional settings. Advancing techniques for more robust performance in such scenarios is essential for enhancing the applicability of these methods. Third, ELCP's framework for incorporating auxiliary data could, in principle, be combined with other localized conformal prediction methods, such as \citet{gibbs2023conformal}. We leave these directions to future work.

\section*{Disclosure Statement}
The authors report there are no competing interests to declare.

\section*{Funding}
This research was supported by the National Key R\&D Program of China (Grant No. 2022YFA1003703), the ARC (Grant No. LP240100101) and the National Natural Science Foundation of China (Grant Nos. 12231011).

\section*{Acknowledgements}
The authors thank the Editor, Associate Editor and three anonymous referees for their many helpful comments that have resulted in significant improvements in the article.

\bibliographystyle{agsm}

\bibliography{Bibliography-MM-MC}

\clearpage
\begin{appendices}
\newcommand{\Var}{\mathrm{var}}

\def\thefigure{S\arabic{figure}}
\def\thetable{S\arabic{table}}
\def\theequation{S\arabic{section}.\arabic{equation}}
\def\thesection{S\arabic{section}}
\renewcommand{\appendixname}{}
\def\thelemma{S\arabic{lemma}}
\def\thetheorem{S\arabic{theorem}}
\def\thealgorithm{S\arabic{algorithm}}
\setcounter{section}{0}  % 重置section计数器
\setcounter{subsection}{0}  % 重置subsection计数器
\setcounter{equation}{0}  % 重置公式计数器
\setcounter{figure}{0}  % 重置图片计数器
\setcounter{table}{0}  % 重置表格计数器
\setcounter{lemma}{0}
\setcounter{theorem}{0}
\setcounter{algorithm}{0}

\section*{Supplementary Materials}

This Supplement presents preliminary lemmas essential for proving the main theorems, detailed proofs of the theoretical results from the main paper in Section \ref{lemma and theorem}, and additional technical details in Section \ref{technical details}. Extended numerical experiments and detailed results are provided in Section \ref{sec:SM_numerical}.

\section{Preliminary lemmas and proofs of main theoretical results}\label{lemma and theorem}
\subsection{Preliminary lemmas}
    In this section, we introduce several preliminary lemmas that will be used in the proofs of main theoretical results.

    \begin{lemma}[Bernstein inequality]\label{lemma:Bernstein}
        Let $Y_1,\cdots,Y_n$ be conditionally independent random variables given a random variable $X$. Assume $\left|Y_i\right|\leq M_0$ almost surely for all $i\in[n]$, where $M_0$ is nonrandom. Then for any $t>0$,
        \begin{align}
            \notag \mathrm{pr}\left( \overset{n}{\underset{i=1}\sum}\left\{ Y_i-E\left( Y_i\mid X \right) \right\}\geq t\mid X \right) 
            \leq \exp\left\{ -\dfrac{t^2/2}{\overset{n}{\underset{i=1}\sum}\Var\left( Y_i\mid X \right)+M_0t/3} \right\}\,.
        \end{align}
        Furthermore, if $\Var\left( Y_i\mid X \right)\leq\sigma_i^2$ with nonrandom $\sigma_i^2$ for $i\in[n]$, then
        \begin{align}
            \notag \mathrm{pr}\left( \overset{n}{\underset{i=1}\sum}\left\{ Y_i-E\left( Y_i\mid X \right) \right\}\geq t\mid X \right) \leq \exp\left( -\dfrac{t^2/2}{\overset{n}{\underset{i=1}\sum}\sigma_i^2+M_0t/3} \right)\,
        \end{align}
        and
        \begin{equation*}
            \mathrm{pr}\left( \overset{n}{\underset{i=1}\sum}\left\{ Y_i-E\left( Y_i\mid X \right) \right\}\geq t \right)\leq\exp\left( -\dfrac{t^2/2}{\overset{n}{\underset{i=1}\sum}\sigma_i^2+M_0t/3} \right).
        \end{equation*}
    \end{lemma}

    \begin{lemma}\label{lemma:KernelFun}
        Suppose Assumption \ref{ass2}-(i) and Assumption \ref{ass2}-(iii) hold. Then there exist positive constants $\underline{L}_3$ and $\overline{L}_3$ such that
        \begin{equation}
            \notag \underline{L}_3 \leq E\left[ h^{-d}\left\{ K(X_i,x_0;h) \right\}^{\ell} \right] \leq \overline{L}_3 \text{~~and~~} \underline{L}_3 \leq E\left[ h^{-d}\left\{ K(X_j^{\prime},x_0;h) \right\}^{\ell} \right] \leq \overline{L}_3
        \end{equation}
        for any $x_0\in\mathcal{X}=[0,1]^{d}$ and any constant $\ell\geq 1$.
    \end{lemma}
    \begin{proof}
        Under Assumption \ref{ass2}-(iii), by changing to spherical coordinates for the following integral, we obtain
        \begin{equation}
            \notag \int_{\mathbb{R}^{d}}K_0(\|x\|)dx \leq \int_{0}^{\infty}u^{d-1}K_0(u)du < \infty\,.
        \end{equation}
        It follows that for any $\ell\geq 1$, 
        \begin{equation*}
            \int_{\mathbb{R}^d}\left\{ K_0(\|x\|) \right\}^{\ell}dx\leq \left\{ K_0(0) \right\}^{\ell-1}\int_{\mathbb{R}^d}K_0(\|x\|)dx<\infty\,.
        \end{equation*}

        Assumption \ref{ass2}-(i) implies that $\underline{L}_1\leq f_X(x),g_X(x)\leq \overline{L}_1$ for any $x\in\mathcal{X}$.
        Consequently, for any $x_0\in\mathcal{X}$,
        \begin{eqnarray}
            \notag E\left[ h^{-d}\left\{ K(X_i,x_0;h) \right\}^{\ell} \right] & = & \int_\mathcal{X}h^{-d}\left\{ K_0\left( \|x-x_0\|/h \right) \right\}^{\ell}f_X(x)dx\\
            \notag  & = & \int_{\{x:~x_0+hx\in\mathcal{X}\}}\left\{ K_0(\|x\|) \right\}^{\ell}f_X(x_0+hx)dx\\
            \notag & \leq & \overline{L}_1\int_{\mathbb{R}^d}\left\{ K_0(\|x\|) \right\}^{\ell}dx<\infty\,.
        \end{eqnarray}
        In addition, since $K_0(u)$ is decreasing in $u$,
        \begin{eqnarray}
            \notag E\left[ h^{-d}\left\{ K(X_i,x_0;h) \right\}^{\ell} \right] & \geq & \{K_0(1)\}^{\ell}h^{-d}\int_{\{x\in\mathcal{X}:~\|x-x_0\|\leq h\}}f_X(x)dx\,,
        \end{eqnarray}
        where $h^{-d}\int_{\{x\in\mathcal{X}:~\|x-x_0\|\leq h\}}f_X(x)dx=h^{-d}\mathrm{pr}\left(\|X_i-x_0\|\leq h\right)$ is bounded below by a positive constant due to $f_X(x)\geq \underline{L}_1$ for any $x\in\mathcal{X}$. 
        Therefore, there exist positive constants $\underline{L}_3$ and $\overline{L}_3$ such that
        \begin{equation}
            \notag \underline{L}_3 \leq E\left[ h^{-d}\left\{ K(X_i,x_0;h) \right\}^{\ell} \right] \leq \overline{L}_3\,.
        \end{equation}
        The proof for $E\big[ h^{-d}\left\{ K(X^\prime,x_0;h) \right\}^\ell \big]$ is similar and thus is omitted.
    \end{proof}

    \begin{lemma}\label{lemma:var_1}
        Suppose Assumption \ref{ass2}-(i) and Assumption \ref{ass2}-(iii) hold. Then, for any random variables $Y$, $\Var\left\{h^{-d}K(X_i,x_0;h)Y\right\}\leq \overline{L}_3h^{-d}\underset{x\in\mathcal{X}}{\sup}\ E\left( Y^2\mid X_i=x \right)$ and $$\Var\left\{h^{-d}K(X_j^{\prime},x_0;h)Y\right\}\leq \overline{L}_3h^{-d}\underset{x\in\mathcal{X}}{\sup}\ E\left( Y^2\mid X_j^{\prime}=x \right)$$ 
        for any $x_0\in\mathcal{X}=[0,1]^{d}$, where $\overline{L}_3$ is the positive constant in Lemma \ref{lemma:KernelFun}.
    \end{lemma}
    \begin{proof}
        According to Lemma \ref{lemma:KernelFun} and simple algebra,
        \begin{eqnarray}
            \notag \Var\left\{h^{-d}K(X_i,x_0;h)Y\right\}\ & \leq & h^{-2d}E\left[ \left\{ K(X_i,x_0;h) \right\}^2Y^2 \right]\\
            \notag & = & h^{-2d}\int_{\mathcal{X}}\left\{ K(x,x_0;h) \right\}^2E\left( Y^2\mid X_i=x \right)f_X(x)dx\\
            \notag & \leq & h^{-d}\underset{x\in\mathcal{X}}{\sup}E\left( Y^2\mid X=x \right)\int_{\mathcal{X}}h^{-d}\left\{ K(x,x_0;h) \right\}^2f_X(x)dx\\
            \notag & \leq & \overline{L}_3h^{-d}\underset{x\in\mathcal{X}}{\sup}\ E\left( Y^2\mid X_i=x \right)\,.
        \end{eqnarray}
        The proof for $\Var\big\{h^{-d}K(X_j^{\prime},x_0;h)Y\big\}$ follows similarly and is therefore omitted.
    \end{proof}

\subsection{Proof of Theorem \ref{theo:marginal_coverage}}\label{proofoftheo:marginal_coverage}

\begin{proof}
    Let $\beta_1,\ldots,\beta_{n+1}$ be the observed values of $\hat{\beta}_{\omega,\hat{r}}(X_1,S_1),\ldots,\hat{\beta}_{\omega,\hat{r}}(X_{n+1},S_{n+1})$ when $Z_1=z_1,\ldots,Z_{n+1}=z_{n+1}$.
    Under Assumption \ref{assump:permutable} that $\hat{r}(x,s)$ is invariant under any permutations within $\mathcal{Z}_{n}\cup \{Z_{n+1}\}$ and on the other hand within $\mathcal{Z}_{m}^{\prime}$, we can rewrite $\hat{\beta}_{\omega,\hat{r}}(X_{i},S_{i})$ as
    \begin{equation}
        \notag \hat{\beta}_{\omega,\hat{r}}(X_{i},S_{i}) = \varphi_1(Z_i; \mathcal{Z}_{n}\cup\{Z_{n+1}\}, \mathcal{Z}_{m}^\prime)\,,
    \end{equation}
    where $\varphi_1(\cdot;\cdot,\cdot)$ is invariant under permutations within $\mathcal{Z}_{n}\cup \{Z_{n+1}\}$ and also within $\mathcal{Z}_{m}^{\prime}$.
    Then, under any permutation $(\pi(1),\ldots,\pi(n+1))$ of $(1,\ldots,n+1)$, the empirical distribution $(n+1)^{-1}\sum_{i=1}^{n+1}\delta_{\beta_i}$ remains unchanged, conditional on the event $\{Z_1=z_{\pi(1)},\ldots,Z_{n+1}=z_{\pi(n+1)}\}$.
    
    Given the exchangeability of $Z_1,\ldots,Z_{n+1}$, it follows that the conditional distribution of $\hat{\beta}_{\omega,\hat{r}}(X_{n+1},S_{n+1})$ conditional on $\mathcal{A}_{Z}=\{\{Z_1,\ldots,Z_{n+1}\}=\{z_1,\ldots,z_{n+1}\}\}$ is $(n+1)^{-1}\sum_{i=1}^{n+1}\delta_{\beta_i}$.
    According to \citet{vovk2005algorithmic},
    \begin{eqnarray}
        \notag & & \mathrm{pr}\left( Y_{n+1}\in \widehat{C}_{\alpha}^{\mathrm{ELCP}}(X_{n+1}) \mid\mathcal{A}_Z \right) \\
        \notag & = & \mathrm{pr}\left( \hat{\beta}_{\omega,\hat{r}}(X_{n+1},S_{n+1})\leq Q\left( 1-\alpha;(n+1)^{-1}\overset{n+1}{\underset{i=1}\sum}\delta_{\hat{\beta}_{\omega,\hat{r}}(X_i,S_i)} \right)\mid\mathcal{A}_Z\right) \\
        \notag & \in & \left[ 1-\alpha,1-\alpha+\dfrac{1}{n+1} \right)\,.
    \end{eqnarray}
    Finally, marginalizing over $\mathcal{A}_{Z}$, we obtain
    \begin{equation*}
        \mathrm{pr}\left( Y_{n+1}\in \widehat{C}_{\alpha}^{\mathrm{ELCP}}(X_{n+1}) \right)\in\left[ 1-\alpha,1-\alpha+\dfrac{1}{n+1} \right)\,,
    \end{equation*}
    which completes the proof of this theorem.
\end{proof}

\subsection{Proof of Lemma \ref{lemma:conditional_coverage_bound}}\label{proofoflemma:conditional_coverage_bound}

\begin{proof}
    First, we claim that
    \begin{eqnarray}
        \notag \left|Q\left( 1-\alpha;(n+1)^{-1}\overset{n+1}{\underset{i=1}\sum}\delta_{\hat{\beta}_{\omega,\hat{r}}(X_{i},S_{i})} \right) - Q\left( 1-\alpha;(n+1)^{-1}\overset{n+1}{\underset{i=1}\sum}\delta_{F_{S\mid X}(S_i\mid X_i)} \right)\right| \leq \underset{1\leq i\leq n+1}{\sup}\Delta_i\,.
    \end{eqnarray}
    To see this, consider any $j\in[n+1]$ such that 
    $$F_{S\mid X}(S_j\mid X_j)\leq Q\left( 1-\alpha;(n+1)^{-1}\overset{n+1}{\underset{i=1}\sum}\delta_{F_{S\mid X}(S_i\mid X_i)} \right)\,.$$ 
    For such $j$, it follows that
    \begin{equation*}
        \hat{\beta}_{\omega,\hat{r}}(X_{j},S_{j}) \leq Q\left( 1-\alpha; (n+1)^{-1}\overset{n+1}{\underset{i=1}\sum}\delta_{F_{S\mid X}(S_i\mid X_i)} \right) + \underset{1\leq i\leq n+1}{\sup}\Delta_i \,,
    \end{equation*}
    which gives
    \begin{eqnarray}
        \notag Q\left( 1-\alpha; (n+1)^{-1}\overset{n+1}{\underset{i=1}\sum}\delta_{\hat{\beta}_{\omega,\hat{r}}(X_{i},S_{i})} \right) \leq Q\left( 1-\alpha; (n+1)^{-1}\overset{n+1}{\underset{i=1}\sum}\delta_{F_{S\mid X}(S_i\mid X_i)} \right) + \underset{1\leq i\leq n+1}{\sup}\Delta_i\,.
    \end{eqnarray} 
    By a similar argument, we can show the opposite direction.

    Let $\mathcal{A}_0=\big\{\sup_{1\leq i\leq n+1}\Delta_i\leq\varepsilon\big\}$ be the event that satisfies  $\mathrm{pr}(\mathcal{A}_0)\geq 1-\delta$, as assumed. 
    Then,
    \begin{eqnarray}
        \notag & & \mathrm{pr}\left( Y_{n+1}\in\widehat{C}_{\alpha}^{\mathrm{ELCP}}(X_{n+1}) \mid X_{n+1}=x_0 \right) \\
        \notag & = & \mathrm{pr}\left( \hat{\beta}_{\omega,\hat{r}}(X_{n+1},S_{n+1})\leq Q\left( 1-\alpha;(n+1)^{-1}\overset{n+1}{\underset{i=1}\sum}\delta_{\hat{\beta}_{\omega,\hat{r}}(X_{i},S_{i})} \right)\mid X_{n+1}=x_0 \right)\\
        \notag & \geq & \mathrm{pr}\left( \hat{\beta}_{\omega,\hat{r}}(X_{n+1},S_{n+1})\leq Q\left( 1-\alpha; (n+1)^{-1}\overset{n+1}{\underset{i=1}\sum}\delta_{\hat{\beta}_{\omega,\hat{r}}(X_{i},S_{i})} \right),\mathcal{A}_0\mid X_{n+1}=x_0 \right)\,.
    \end{eqnarray}
    
    Define the events
    \begin{eqnarray}
        \notag \mathcal{A}_1 & = & \bigg\{F_{S\mid X}(S_{n+1}\mid X_{n+1})\leq Q\bigg( 1-\alpha; (n+1)^{-1}\overset{n+1}{\underset{i=1}\sum}\delta_{F_{S\mid X}(S_i\mid X_i)} \bigg) -2\underset{1\leq i\leq n+1}{\sup}\Delta_i\bigg\}
    \end{eqnarray}
    and 
    \begin{eqnarray}
        \notag \mathcal{A}_2 = \left\{F_{S\mid X}(S_{n+1}\mid X_{n+1})\leq Q\left( 1-\alpha; (n+1)^{-1}\overset{n+1}{\underset{i=1}\sum}\delta_{F_{S\mid X}(S_i\mid X_i)} \right)-2\varepsilon\right\}\,.
    \end{eqnarray}
    It is clear that if $\mathcal{A}_1$ holds, $\hat{\beta}_{\omega,\hat{r}}(X_{n+1},S_{n+1})\leq Q\left( 1-\alpha; (n+1)^{-1}\sum_{i=1}^{n+1}\delta_{\hat{\beta}_{\omega,\hat{r}}(X_{i},S_{i})} \right)$.
    Therefore,
    \begin{eqnarray}
        \notag \mathrm{pr}\left( Y_{n+1}\in\widehat{C}_{\alpha}^{\mathrm{ELCP}}(X_{n+1}) \mid X_{n+1}=x_0 \right) & \geq &   \mathrm{pr}\left( \mathcal{A}_1\cap\mathcal{A}_0\mid X_{n+1}=x_0 \right)\\
        \notag & \geq & \mathrm{pr}\left(\mathcal{A}_2 \mid X_{n+1}=x_0 \right)-\delta\\
        \notag & \geq & \mathrm{pr}\left( \mathcal{A}_3 \mid X_{n+1}=x_0 \right) - \delta \,,
    \end{eqnarray}
    where 
    \begin{eqnarray}
        \notag \mathcal{A}_3 = \left\{F_{S\mid X}(S_{n+1}\mid X_{n+1})\leq Q\left( 1-\alpha; (n+1)^{-1}\left\{\overset{n}{\underset{i=1}\sum}\delta_{F_{S\mid X}(S_i\mid X_i)}+\delta_0 \right\} \right)-2\varepsilon\right\}\,
    \end{eqnarray}
    with $\delta_0$ denoting the point mass at $0$.
    Regarding that $F_{S\mid X}(S_{n+1}\mid X_{n+1})\sim \mathrm{Uniform}[0,1]$ conditional on $X_{n+1}=x_0$, we have
    \begin{eqnarray}
        \notag \mathrm{pr}\left( \mathcal{A}_3 \mid X_{n+1}=x_0 \right) & = & E \left\{ \mathrm{pr}\left( \mathcal{A}_3 \mid X_{n+1}=x_0; X_1,\ldots,X_n; S_1,\ldots,S_n \right) \right\} \\
        \notag & = & E\left[ Q\left( 1-\alpha; (n+1)^{-1}\left\{\overset{n}{\underset{i=1}\sum}\delta_{F_{S\mid X}(S_i\mid X_i)}+\delta_0 \right\} \right)-2\varepsilon  \right] \\
        \notag & \geq & E\left\{\tilde{U}_{(\lceil (n+1)(1-\alpha) \rceil-1)}\right\} - 2\varepsilon \\
        \notag & = & \dfrac{\lceil (n+1)(1-\alpha) \rceil-1}{n+1} - 2\varepsilon\,,
    \end{eqnarray}
    where $\tilde{U}_{(\lceil (n+1)(1-\alpha) \rceil-1)}$ is the $(\lceil (n+1)(1-\alpha) \rceil-1)$-th order statistic of $n$ i.i.d.~$\mathrm{Uniform}[0,1]$ random variables.
    Thus, we conclude that
    \begin{eqnarray}
        \notag \mathrm{pr}\left( Y_{n+1}\in\widehat{C}_{\alpha}^{\mathrm{ELCP}}(X_{n+1}) \mid X_{n+1}=x_0 \right) & \geq &   \dfrac{\lceil (n+1)(1-\alpha) \rceil-1}{n+1} - 2\varepsilon - \delta \,.
    \end{eqnarray}
    
    On the other hand, it can also be shown that
    \begin{eqnarray}
        \notag & & \mathrm{pr}\left( Y_{n+1}\in\widehat{C}_{\alpha}^{\mathrm{ELCP}}(X_{n+1}) \mid X_{n+1}=x_0 \right) \\
        \notag & \leq & \mathrm{pr}\left( \hat{\beta}_{\omega,\hat{r}}(X_{n+1},S_{n+1})\leq Q\left( 1-\alpha; (n+1)^{-1}\overset{n+1}{\underset{i=1}\sum}\delta_{\hat{\beta}_{\omega,\hat{r}}(X_{i},S_{i})} \right),\mathcal{A}_0\mid X_{n+1}=x_0 \right)+\delta \\
        \notag & \leq & \mathrm{pr}\left( \mathcal{A}_4\cap\mathcal{A}_0\mid X_{n+1}=x_0 \right)+\delta \\
        \notag & \leq & \mathrm{pr}\left( \mathcal{A}_5 \mid X_{n+1}=x_0 \right)+\delta\\
        \notag & \leq & \mathrm{pr}\left( \mathcal{A}_6 \mid X_{n+1}=x_0 \right) + \delta \,,     
    \end{eqnarray}
    where 
    \begin{eqnarray}
        \notag \mathcal{A}_{4} & = & \bigg\{F_{S\mid X}(S_{n+1}\mid X_{n+1})\leq Q\bigg( 1-\alpha; (n+1)^{-1}\overset{n+1}{\underset{i=1}\sum}\delta_{F_{S\mid X}(S_i\mid X_i)} \bigg) + 2\underset{1\leq i\leq n+1}{\sup}\Delta_i\bigg\}\,, \\
        \notag \mathcal{A}_5 & = & \left\{F_{S\mid X}(S_{n+1}\mid X_{n+1})\leq Q\left( 1-\alpha; (n+1)^{-1}\overset{n+1}{\underset{i=1}\sum}\delta_{F_{S\mid X}(S_i\mid X_i)} \right)+2\varepsilon\right\}\,,
    \end{eqnarray}
    and
    $\mathcal{A}_6 = \left\{F_{S\mid X}(S_{n+1}\mid X_{n+1})\leq Q\left( 1-\alpha; (n+1)^{-1}\left\{\overset{n}{\underset{i=1}\sum}\delta_{F_{S\mid X}(S_i\mid X_i)}+\delta_{\infty} \right\} \right)+2\varepsilon\right\}$.
    Similarly, we have
    \begin{eqnarray}
        \notag & & \mathrm{pr}\left( Y_{n+1}\in\widehat{C}_{\alpha}^{\mathrm{ELCP}}(X_{n+1}) \mid X_{n+1}=x_0 \right) \\
        \notag & \leq & E\left\{ Q\left( 1-\alpha; (n+1)^{-1}\overset{n}{\underset{i=1}\sum}\left\{\delta_{F_{S\mid X}(S_i\mid X_i)}+\delta_\infty\right\} \right) + 2\varepsilon \right\} + \delta\\
        \notag & \leq & E\left\{ \tilde{U}_{(\lceil (n+1)(1-\alpha) \rceil)} \right\}+2\varepsilon+\delta\\
        \notag & = & \dfrac{\lceil (n+1)(1-\alpha) \rceil}{n+1}+2\varepsilon+\delta\,.
    \end{eqnarray}
    We finish the proof of this lemma.
\end{proof}

\subsection{Proof of Lemma \ref{lemma:delta_concentration2}}\label{Proof of lemma:delta_concentration2}

\begin{proof}
    We prove Lemma \ref{lemma:delta_concentration2} in two steps.
    
    \textbf{Step I}: In the first part of the proof, we show the result of Lemma \ref{lemma:delta_concentration2} without Assumption \ref{ass3} while temporarily assuming that $\hat{r}(x,s)$ is independent of $\mathcal{Z}_m^\prime$. This condition can be achieved by partitioning the large auxiliary dataset $\mathcal{Z}_m^\prime$ into two independent subsets: one for training $\hat{r}(x,s)$ and the other for computing the scores $S{j}^{\prime}$ and $\hat{\beta}{\omega,\hat{r}}(X_{i},S_{i})$.
    This assumption will be relaxed in the second part of the proof to the setting where $\hat{r}(x,s)$ is estimated using $\mathcal{Z}_n \cup \{Z_{n+1}\}$ and $\mathcal{Z}_{m}^\prime$, as assumed in the main paper.

    Denote $\Lambda(x_0)=E\{K(X_i,x_0;h)\}=E\{K(X_j^\prime,x_0;h)r(X_j^\prime,S_j^\prime)\}$ for any $x_0\in\mathcal{X}=[0,1]^{d}$. According to Lemma \ref{lemma:KernelFun}, we have $\underline{L}_3\leq h^{-d}\Lambda(x_0) \leq \overline{L}_3$ for any $x_0\in\mathcal{X}$.
    Define
    \begin{eqnarray}
        \notag & J_{i}^{(0,1)}=\sum_{j\neq i}K(X_j,X_i;h),  ~J_{i}^{(0,2)}= \overset{m}{\underset{j=1}\sum}K(X_j^\prime,X_i;h)r(X_j^\prime,S_j^\prime)\,, & \\
        \notag & J_{i}^{(0,3)}= \overset{m}{\underset{j=1}\sum}K(X_j^\prime,X_i;h)\hat{r}(X_j^\prime,S_j^\prime), ~ J_{i}^{(0,4)}= \sum_{j\neq i}K(X_j,X_i;h)\mathbb{1}(S_j\leq S_i)\,, & \\
        \notag & J_{i}^{(0,5)}= \overset{m}{\underset{j=1}\sum}K(X_j^{\prime},X_i;h)r(X_j^\prime,S_j^{\prime})\mathbb{1}(S_j^{\prime}\leq S_i), ~J_{i}^{(0,6)}= \overset{m}{\underset{j=1}\sum}K(X_j^{\prime},X_i;h)\hat{r}(X_j^\prime,S_j^{\prime})\mathbb{1}(S_j^{\prime}\leq S_i)\,. &
    \end{eqnarray}
    We can decompose $\hat{\beta}_{\omega,\hat{r}}(X_{i},S_{i})=\dfrac{K_0(0)+J_{i}^{(0,4)}+ \omega J_{i}^{(0,6)}}{K_0(0)+J_{i}^{(0,1)}+ \omega J_{i}^{(0,3)}}$ as 
    \begin{equation}
        \notag \hat{\beta}_{\omega,\hat{r}}(X_{i},S_{i}) = J_i^{(1)}\left( J_i^{(2)}+J_i^{(3)}H_i^{(1)} \right)\,,
    \end{equation}
    where 
    \begin{eqnarray}
        \notag & J_i^{(1)}=\dfrac{K_0(0)+J_{i}^{(0,1)}+ \omega J_{i}^{(0,2)}}{K_0(0)+J_{i}^{(0,1)}+ \omega J_{i}^{(0,3)}}, ~ J_i^{(2)} = \dfrac{ \omega\left(J_{i}^{(0,6)}-J_{i}^{(0,5)}\right)}{K_0(0)+J_{i}^{(0,1)}+ \omega J_{i}^{(0,2)}}\,, & \\
        \notag & J_i^{(3)} = \dfrac{K_0(0)+\Lambda(X_i)(n+\omega m)}{K_0(0)+J_{i}^{(0,1)}+ \omega J_{i}^{(0,2)}}, ~ H_i^{(1)}= \dfrac{K_0(0)+J_{i}^{(0,4)}+\omega J_{i}^{(0,5)}}{K_0(0)+\Lambda(X_i)(n+\omega m)}\,. &
    \end{eqnarray}

    We analyze $J_i^{(1)}, J_i^{(2)}$ and $J_i^{(3)}$ first. For $J_i^{(1)}$, simple algebra yields that
    \begin{eqnarray}
        \notag \left| J_i^{(1)}-1 \right| & \leq & \left\{ 1 + \dfrac{K_0(0)h^{-d}+h^{-d}J_{i}^{(0,1)}+\omega h^{-d}J_{i}^{(0,2)}}{\omega h^{-d} \overset{m}{\underset{j=1}\sum}K(X_j^\prime,X_i;h)\left|\hat{r}(X_j^\prime,S_j^\prime)-r(X_j^\prime,S_j^\prime)\right|} \right\}^{-1} \\
        & \leq & \dfrac{\omega h^{-d} \overset{m}{\underset{j=1}\sum}K(X_j^\prime,X_i;h)\left|\hat{r}(X_j^\prime,S_j^\prime)-r(X_j^\prime,S_j^\prime)\right|}{K_0(0)h^{-d}+h^{-d}J_{i}^{(0,1)}+\omega h^{-d}J_{i}^{(0,2)}}\,.\label{eq:Ji1 upper1}
    \end{eqnarray}

    For $0<\tau_1<\underline{L}_3$, denote 
    \begin{eqnarray}
        \notag A_i^{(1)}(\tau_1) = \left\{h^{-d}J_{i}^{(0,1)}+\omega h^{-d}J_{i}^{(0,2)} - \Lambda(X_i)(n+\omega m)h^{-d} > -\tau_1(n+\omega m) \right\}\,.
    \end{eqnarray}
    From Lemma \ref{lemma:var_1}, $\Var\{h^{-d}K(X_j,X_i;h)\mid X_i\}\leq \overline{L}_3h^{-d}$ for $j\neq i$ and
    \begin{eqnarray}
        \notag & & \Var\{\omega h^{-d}K(X_j^{\prime},X_i;h)r(X_j^{\prime},S_j^{\prime})\mid X_i\} \\
        \notag & \leq & \overline{L}_3\omega^2h^{-d}\sup_{x\in\mathcal{X}}E\left[\{r(X_j^{\prime},S_j^{\prime})\}^2\mid X_j^{\prime}=x\right]=\overline{L}_3V_r\omega^2h^{-d}\,,
    \end{eqnarray}
    where $V_r=\sup_{x\in\mathcal{X}}E\left[\{r(X_j^{\prime},S_j^{\prime})\}^2\mid X_j^{\prime}=x\right]\leq \overline{L}_2^2$ under Assumption \ref{ass2}. 
    
    Applying Bernstein's inequality in Lemma \ref{lemma:Bernstein} with $h^{-d}K(X_j,X_i;h)\leq K_0(0)h^{-d}$ and  $\omega h^{-d}K(X_j^{\prime},X_i;h)r(X_j^{\prime},S_j^{\prime})\leq \overline{L}_2\omega K_0(0) h^{-d}$, we obtain
    \begin{equation}\label{eq:bernstein_A1}
        1- \mathrm{pr}\left(A_i^{(1)}(\tau_1)\right) \leq \exp\left\{ -\dfrac{\tau_1^2(n+\omega m)^2h^d/2}{\overline{L}_3n+\overline{L}_3V_r\omega^2m+(1+\overline{L}_2\omega)K_0(0)\tau_1(n+\omega m)/3} \right\}\,.
    \end{equation}

    Denote the event $A(\gamma,r,k)=\{D_k(r,\hat{r})\leq\epsilon_k(\gamma;r)\}$ and we have $\mathrm{pr}(A(\gamma,r,k))\geq 1-\gamma$ for $\gamma\in(0,1)$. For $\tau_2>0$, define
    \begin{align}
        \notag A_i^{(2)}(\tau_2) = & \bigg\{ \omega h^{-d} \overset{m}{\underset{j=1}\sum}K(X_j^\prime,X_i;h)\left|\hat{r}(X_j^\prime,S_j^\prime)-r(X_j^\prime,S_j^\prime)\right| \\
        \notag & \quad - \omega mh^{-d}E\left\{ K(X_j^\prime,X_i;h)\left|\hat{r}(X_j^\prime,S_j^\prime)-r(X_j^\prime,S_j^\prime)\right|\mid X_i,\hat{r} \right\} < \tau_2(n+\omega m) \bigg\}\,.
    \end{align}
    Since 
    \begin{eqnarray}
        \notag & & E\left\{ \left|\hat{r}(X_j^\prime,S_j^\prime)-r(X_j^\prime,S_j^\prime)\right|^2\mid X_j^\prime=x,\hat{r} \right\} \leq \overline{L}_2^2\,,
    \end{eqnarray}
    and utilizing Lemma \ref{lemma:var_1} again, we have
    \begin{equation}
        \Var\left\{\omega h^{-d}K(X_j^\prime,X_i;h)\left|\hat{r}(X_j^\prime,S_j^\prime)-r(X_j^\prime,S_j^\prime)\right| \mid X_i,\hat{r}  \right\}
        \leq \overline{L}_2^2\overline{L}_3\omega^2h^{-d}\,.\label{eq:varRhat1}
        % \leq \overline{L}_3h^{-d}\{V_{1}(r)+V_{2}(\hat{r})\}\,.
    \end{equation}
    On the event $A(\gamma,r,k)$, we utilize Holder's inequality and Lemma \ref{lemma:KernelFun} to obtain
    \begin{align}
        \notag & \Var\left\{\omega h^{-d}K(X_j^\prime,X_i;h)\left|\hat{r}(X_j^\prime,S_j^\prime)-r(X_j^\prime,S_j^\prime)\right| \mid X_i,\hat{r}  \right\}\\
        \notag \leq & \omega^2h^{-2d}E\left[ \{K(X_j^\prime,X_i;h)\}^2\left|\hat{r}(X_j^\prime,S_j^\prime)-r(X_j^\prime,S_j^\prime)\right|^2\mid X_i,\hat{r} \right]\\
        \notag \leq & \omega^2h^{-2d} \left( E\left[\{K(X_j^\prime,X_i;h)\}^{2k/(k-2)}\mid X_i \right] \right)^{(k-2)/k} \left[E\left\{ \left|\hat{r}(X_j^\prime,S_j^\prime)-r(X_j^\prime,S_j^\prime)\right|^k\mid \hat{r} \right\} \right]^{2/k}\\
        \notag \leq & \omega^2h^{-(k+2)d/k}\left( E\left[h^{-d}\{K(X_j^\prime,X_i;h)\}^{2k/(k-2)}\mid X_i \right] \right)^{(k-2)/k}\{\epsilon_k(\gamma;r)\}^2\\
        \leq & (\overline{L}_3\vee 1)\omega^2h^{-(k+2)d/k}\{\epsilon_k(\gamma;r)\}^2\,.\label{eq:varRhat2}
    \end{align}
    Combine \eqref{eq:varRhat1} and \eqref{eq:varRhat2}, we get that on the event $A(\gamma,r,k)$,
    \begin{align*}
        &\Var\left\{\omega h^{-d}K(X_j^\prime,X_i;h)\left|\hat{r}(X_j^\prime,S_j^\prime)-r(X_j^\prime,S_j^\prime)\right| \mid X_i,\hat{r}  \right\}\\
        \leq&(\overline{L}_3\vee 1)\left[ h^{-2d/k}\{\epsilon_k(\gamma;r)\}^2\wedge \overline{L}_2^2 \right]\omega^2h^{-d}\,.
    \end{align*}

    The independence between $\hat{r}$ and $\mathcal{Z}_m^\prime$ ensures that within event $A(\gamma,r,k)$, the samples in $\mathcal{Z}_m^\prime$ retain their i.i.d. structure. Applying Lemma \ref{lemma:Bernstein} with the condition 
    \begin{gather*}
        \omega h^{-d}K(X_j^\prime,X_i;h)\left|\hat{r}(X_j^\prime,S_j^\prime)-r(X_j^\prime,S_j^\prime)\right|\leq \overline{L}_2\omega K_0(0) h^{-d}\,,
    \end{gather*}
    we arrive at
    \begin{align}
        \notag&1 - \mathrm{pr}\left(A_i^{(2)}(\tau_2)\mid A(\gamma,r,k)\right)\\
        \leq &\exp\left( -\dfrac{\tau_2^2(n+\omega m)^2h^d/2}{(\overline{L}_3\vee 1)\left[ h^{-2d/k}\{\epsilon_k(\gamma;r)\}^2\wedge\overline{L}_2^2 \right]\omega^2m+\overline{L}_2\omega K_0(0)\tau_2 (n+\omega m)/3} \right)\,.\label{eq:bernstein_A2}
    \end{align}

    Then, on the event $A_i^{(1)}(\tau_1)\cap A_i^{(2)}(\tau_2)\cap A(\gamma,r,k)$, we have
    \begin{eqnarray}
        \notag & & \dfrac{\omega h^{-d} \overset{m}{\underset{j=1}\sum}K(X_j^\prime,X_i;h)\left|\hat{r}(X_j^\prime,S_j^\prime)-r(X_j^\prime,S_j^\prime)\right|}{K_0(0)h^{-d}+h^{-d}J_{i}^{(0,1)}+\omega h^{-d}J_{i}^{(0,2)}} \\
        \notag & \leq & \dfrac{\omega m h^{-d}E\left\{ K(X_j^\prime,X_i;h)\left|\hat{r}(X_j^\prime,S_j^\prime)-r(X_j^\prime,S_j^\prime)\right|\mid X_i,\hat{r} \right\}+\tau_2(n+\omega m)}{K_0(0)h^{-d}+\Lambda(X_i)(n+\omega m)h^{-d}-\tau_1(n+\omega m)} \\
        \notag & \leq & \dfrac{\omega m h^{-d}E\left\{ K(X_j^\prime,X_i;h)\left|\hat{r}(X_j^\prime,S_j^\prime)-r(X_j^\prime,S_j^\prime)\right|\mid X_i,\hat{r} \right\}+\tau_2(n+\omega m)}{(\underline{L}_3-\tau_1)(n+\omega m)} \,.
    \end{eqnarray}
    By Holder's inequality and some simple algebra,
    \begin{eqnarray}
        \notag & & E\left\{ h^{-d}K(X_j^\prime,X_i;h)\left|\hat{r}(X_j^\prime,S_j^\prime)-r(X_j^\prime,S_j^\prime)\right|\mid X_i,\hat{r} \right\} \\
        \notag & \leq & \left(E\left[\left\{h^{-d}K(X_j^\prime,X_i;h)\right\}^{k/(k-1)}\mid X_i \right]\right)^{(k-1)/k}\left[ E\left\{\left|\hat{r}(X_j^\prime,S_j^\prime)-r(X_j^\prime,S_j^\prime)\right|^k\mid \hat{r}\right\} \right]^{1/k} \\
        \notag & = & h^{-d/k}\left(E\left[h^{-d}\left\{K(X_j^\prime,X_i;h)\right\}^{k/(k-1)}\mid X_i \right]\right)^{(k-1)/k}D_k(r,\hat{r}) \\
        %\notag & \leq & \overline{L}_3^{(k-1)/k}h^{-d/k}D_k(r,\hat{r}) \\
        \notag & \leq & (\overline{L}_3\vee 1)h^{-d/k}D_k(r,\hat{r})\,.
    \end{eqnarray}
    It follows that
    \begin{eqnarray}
        \notag & & E\left\{ h^{-d}K(X_j^\prime,X_i;h)\left|\hat{r}(X_j^\prime,S_j^\prime)-r(X_j^\prime,S_j^\prime)\right|\mid X_i,\hat{r} \right\}\leq (\overline{L}_3\vee 1)h^{-d/k}\epsilon_k(\gamma,r)\,
    \end{eqnarray}
     on the event $A(\gamma,r,k)$.
    
    To sum up, 
    \begin{eqnarray}
        \notag \left|J_i^{(1)}-1\right| & \leq & \dfrac{\omega m h^{-d}E\left\{ K(X_j^\prime,X_i;h)\left|\hat{r}(X_j^\prime,S_j^\prime)-r(X_j^\prime,S_j^\prime)\right|\mid X_i,\hat{r} \right\}+\tau_2(n+\omega m)}{(\underline{L}_3-\tau_1)(n+\omega m)} \\
        \notag & \leq & \dfrac{(\overline{L}_3\vee 1)\omega m}{(\underline{L}_3-\tau_1)(n+\omega m)}\dfrac{\epsilon_k(\gamma,r)}{h^{d/k}}+\dfrac{\tau_2}{\underline{L}_3-\tau_1}\,
    \end{eqnarray}
    on the event $A_i^{(1)}(\tau_1)\cap A_i^{(2)}(\tau_2)\cap A(\gamma,r,k)$.
    
    For $J_i^{(2)}$, regarding that
    \begin{equation*}
        \left|J_i^{(2)}\right| \leq \dfrac{\omega h^{-d} \overset{m}{\underset{j=1}\sum}K(X_j^\prime,X_i;h)\left|\hat{r}(X_j^\prime,S_j^\prime)-r(X_j^\prime,S_j^\prime)\right|}{K_0(0)h^{-d}+h^{-d}J_{i}^{(0,1)}+\omega h^{-d}J_{i}^{(0,2)}}  \,,
    \end{equation*}
    it immediately follows that
    \begin{equation}
        \notag \left|J_i^{(2)}\right|\leq \dfrac{(\overline{L}_3\vee 1)\omega m}{(\underline{L}_3-\tau_1)(n+\omega m)}\dfrac{\epsilon_k(\gamma,r)}{h^{d/k}}+\dfrac{\tau_2}{\underline{L}_3-\tau_1}\,
    \end{equation}
    on the event $A_i^{(1)}(\tau_1)\cap A_i^{(2)}(\tau_2)\cap A(\gamma,r,k)$. In addition,
    \begin{align*}
        &\mathrm{pr}\left(\left\{ \overset{n+1}{\underset{i=1}{\cap}}A_i^{(2)}(\tau_2)\right\}\cap A(\gamma,r,k)\right)\\
        =&\mathrm{pr}\left(\left\{ \overset{n+1}{\underset{i=1}{\cap}}A_i^{(2)}(\tau_2)\right\}\mid A(\gamma,r,k)\right)\mathrm{pr}\left( A(\gamma,r,k)\right)\\
        \geq&\left[ 1-\sum_{i=1}^{n+1}\left\{ 1-\mathrm{pr}\left(A_i^{(2)}(\tau_2)\mid A(\gamma,r,k)\right) \right\} \right]\mathrm{pr}\left( A(\gamma,r,k)\right)\\
        \geq&\left[ 1-\sum_{i=1}^{n+1}\left\{ 1-\mathrm{pr}\left(A_i^{(2)}(\tau_2)\mid A(\gamma,r,k)\right) \right\} \right]-\left\{ 1-\mathrm{pr}\left( A(\gamma,r,k)\right) \right\}\\
        \geq&1-\gamma-(n+1) \exp\left[ -\dfrac{\tau_2^2(n+\omega m)^2h^d/2}{(\overline{L}_3\vee 1)\left[ h^{-2d/k}\{\epsilon_k(\gamma;r)\}^2\wedge\overline{L}_2^2 \right]\omega^2m+\overline{L}_2\omega K_0(0)\tau_2 (n+\omega m)/3} \right]\,.
    \end{align*}

    For $J_i^{(3)}$, we can bound $\big|J_i^{(3)}-1\big|$ by
    \begin{equation*}
        \big|J_i^{(3)}-1\big|\leq\dfrac{\left|J_i^{(0,1)}+\omega J_i^{(0,2)}-\Lambda(X_i)(n+\omega m)\right|}{J_i^{(0,1)}+\omega J_i^{(0,2)}}\,.
    \end{equation*}
    Denote $A_i^{(3)}(\tau_3)=\big\{\big|h^{-d}J_i^{(0,1)}+\omega h^{-d}J_i^{(0,2)}-\Lambda(X_i)(n+\omega m)h^{-d}\big|<\tau_3\Lambda(X_i)(n+\omega m)h^{-d}\big\}$ for $\tau_3>0$.
    Then $A_i^{(3)}(\tau_3)\subset \big\{\big|J_i^{(3)}-1\big|< \tau_3/(1+\tau_3) \big\}$.
    Similar to the technique in deriving \eqref{eq:bernstein_A1} for $A_i^{(1)}(\tau_1)$, we have
    \begin{equation}\label{eq:bernstein_A3}
        1- \mathrm{pr}\left(A_i^{(3)}(\tau_3)\right) \leq \exp\left\{ -\dfrac{\tau_3^2\underline{L}_3^2(n+\omega m)^2h^d/2}{\overline{L}_3n+\overline{L}_3V_r\omega^2m+(1+\overline{L}_2\omega)\overline{L}_3K_0(0)\tau_3(n+\omega m)/3} \right\}\,.
    \end{equation}
    Thus, $\big|J_i^{(3)}-1\big|<\tau_3/(1+\tau_3)<\tau_3$ with probability over 
    \begin{equation}
        \notag 1 - \exp\left\{ -\dfrac{\tau_3^2\underline{L}_3^2(n+\omega m)^2h^d/2}{\overline{L}_3n+\overline{L}_3V_r\omega^2m+(1+\overline{L}_2\omega)\overline{L}_3K_0(0)\tau_3(n+\omega m)/3} \right\}\,.
    \end{equation}

    Now, we decompose $H_i^{(1)}=\dfrac{K_0(0)+J_{i}^{(0,4)}+\omega J_{i}^{(0,5)}}{K_0(0)+\Lambda(X_i)(n+\omega m)}$ as $H_i^{(1)} = D_i^{(1)}+D_i^{(2)}+D_i^{(3)}+F_{S\mid X}(S_i\mid X_i)$, where $D_i^{(1)}=H_i^{(1)}-H_i^{(2)}$, $D_i^{(2)}=H_i^{(2)}-H_i^{(3)}$, and $D_i^{(3)}=H_i^{(3)}-F_{S\mid X}(S_i\mid X_i)$ with
    \begin{eqnarray}
        \notag H_i^{(2)} & = & \dfrac{K_0(0)F_{S\mid X}(S_i\mid X_i)+J_i^{(0,7)}+\omega J_i^{(0,9)}}{K_0(0)+\Lambda(X_i)(n+\omega m)}\,,\\
        \notag H_i^{(3)} & = & \dfrac{K_0(0)+J_i^{(0,1)}+\omega J_i^{(0,8)}}{K_0(0)+\Lambda(X_i)(n+\omega m)}F_{S\mid X}(S_i\mid X_i)\,,
    \end{eqnarray}
    and
    \begin{gather}
        \notag J_i^{(0,7)}=\overset{}{\underset{j\neq i}\sum}K(X_j,X_i;h)F_{S\mid X}(S_i\mid X_j),~J_i^{(0,8)}=\overset{m}{\underset{j=1}\sum}K(X_j^\prime,X_i;h)\dfrac{f_X(X_j^\prime)}{g_X(X_j^\prime)},\\
        \notag J_i^{(0,9)} = \overset{m}{\underset{j=1}\sum}K(X_j^\prime,X_i;h)\dfrac{f_X(X_j^\prime)}{g_X(X_j^\prime)}F_{S\mid X}(S_i\mid X_j^\prime)\,.
    \end{gather}

    For $D_i^{(1)}$, define the event
    \begin{eqnarray}
        \notag A_i^{(4,1)}(\tau_{4,1}) = \Big\{ \left|h^{-d}J_{i}^{(0,4)}+\omega h^{-d}J_{i}^{(0,5)} -h^{-d}J_i^{(0,7)}-\omega h^{-d}J_i^{(0,9)} \right| < \tau_{4,1}(n+\omega m)\Big\}\,
    \end{eqnarray}
    for $\tau_{4,1}>0$.
    For $j\neq i$, straightforward calculations show that 
    \begin{eqnarray}
        \notag \Var\left\{ h^{-d}K(X_j,X_i;h)\mathbb{1}(S_j\leq S_i)\mid X_1,\ldots,X_{n+1}; X_1^\prime,\ldots,X_m^\prime; S_i \right\}\leq h^{-2d}\left\{ K(X_j,X_i;h) \right\}^2\,,
    \end{eqnarray}
    and for $j=1,\ldots,m$,
    \begin{eqnarray}
        \notag & & 
        \Var\left\{ \omega h^{-d}K(X_j^\prime,X_i;h)r(X_j^\prime,S_j^\prime)\mathbb{1}(S_j^\prime\leq S_i)\mid X_1,\ldots,X_{n+1}; X_1^\prime,\ldots,X_m^\prime; S_i \right\} \\
        \notag & = & \omega^2 h^{-2d}\left\{ K(X_j^{\prime},X_i;h) \right\}^2\Var\left\{ r(X_j^\prime,S_j^\prime)\mathbb{1}(S_j^\prime\leq S_i)\mid  X_1,\ldots,X_{n+1}; X_1^\prime,\ldots,X_m^\prime; S_i \right\}\,,
    \end{eqnarray}
    where 
    \begin{eqnarray}
        \notag & & \Var\left\{ r(X_j^\prime,S_j^\prime)\mathbb{1}(S_j^\prime\leq S_i)\mid  X_1,\ldots,X_{n+1}; X_1^\prime,\ldots,X_m^\prime; S_i \right\} \\
        \notag & \leq & E\left[ \left\{ r(X_j^\prime,S_j^\prime) \right\}^2\mid X_1,\ldots,X_{n+1}; X_1^\prime,\ldots,X_m^\prime; S_i \right] \leq V_r\,.
    \end{eqnarray}
    Thus, applying Bernstein inequality in Lemma \ref{lemma:Bernstein}, we obtain that
    \begin{eqnarray}\label{eq:bernstein_A41}
        \notag & & 1 - \mathrm{pr}\left( A_i^{(4,1)}(\tau_{4,1})\mid X_1,\ldots,X_{n+1}; X_1^\prime,\ldots,X_m^\prime; S_i \right) \\
        & \leq & 2\exp\left\{ -\dfrac{\tau_{4,1}^2(n+\omega m)^2h^d/2}{J_i^{(4,0)}+(1+\overline{L}_2\omega)K_0(0)\tau_{4,1}(n+\omega m)/3} \right\}\,,
    \end{eqnarray}
    where $J_i^{(4,0)}=h^{-d}\overset{}{\underset{j\neq i}\sum}\{K(X_j,X_i;h)\}^2+V_r\omega^2h^{-d}\overset{m}{\underset{j=1}\sum}\{K(X_j^\prime,X_i;h)\}^2$.
    Define for $\tau_{4,2}>0$,
    $$A_i^{(4,2)}(\tau_{4,2})=\left\{J_i^{(4,0)}< (1+\tau_{4,2})\overline{L}_3\big(n+V_r\omega^2m\big)\right\}\,.$$
    As $E\big(J_i^{(4,0)}\big)\leq \overline{L}_3\big(n+V_r\omega^2m\big)$ and $\Var\left[h^{-d}\{K(X_j,X_i;h)\}^2\mid X_i\right]\leq \overline{L}_3h^{-d}$, it follows that
    \begin{eqnarray}\label{eq:bernstein_A42}
        \notag & & 1 - \mathrm{pr}\left( A_i^{(4,2)}(\tau_{4,2}) \right) \\
        & \leq & \exp\left[ -\dfrac{\tau_{4,2}^2\overline{L}_3\big(n+V_r\omega^2m\big)^2h^d/2}{n+V_r^2\omega^4m+\big(1+V_r\omega^2\big)\{K_0(0)\}^2\tau_{4,2}\big(n+V_r\omega^2m\big)/3} \right]\,.
    \end{eqnarray}
    Thus, taking expectations on both sides of \eqref{eq:bernstein_A41}, we obtain
    \begin{eqnarray}
        \notag & & 1 - \mathrm{pr}\left( A_i^{(4,1)}(\tau_{4,1})\right) \\
        \notag & \leq & 2E\left[\exp\left\{ -\dfrac{\tau_{4,1}^2(n+\omega m)^2h^d/2}{J_i^{(4,0)}+(1+\overline{L}_2\omega)K_0(0)\tau_{4,1}(n+\omega m)/3} \right\}\right] \\
        \notag & \leq & 2E\left[\exp\left\{ -\dfrac{\tau_{4,1}^2(n+\omega m)^2h^d/2}{J_i^{(4,0)}+(1+\overline{L}_2\omega)K_0(0)\tau_{4,1}(n+\omega m)/3} \right\}\mathbb{1}\left( A_i^{(4,2)}(\tau_{4,2}) \right)\right] \\
        \notag & & + 1 - \mathrm{pr}\left( A_i^{(4,2)}(\tau_{4,2}) \right) \\
        \notag & \leq & \exp\left[ -\dfrac{\tau_{4,1}^2(n+\omega m)^2h^d/2}{(1+\tau_{4,2})\overline{L}_3\big(n+V_r\omega^2m\big)+(1+\overline{L}_2\omega)K_0(0)\tau_{4,1}(n+\omega m)/3} \right] \\
        \notag & & + \exp\left[ -\dfrac{\tau_{4,2}^2\overline{L}_3\big(n+V_r\omega^2m\big)^2h^d/2}{n+V_r^2\omega^4m+\big(1+V_r\omega^2\big)\{K_0(0)\}^2\tau_{4,2}\big(n+V_r\omega^2m\big)/3} \right]\,.
    \end{eqnarray}
    
    Use the fact that
    \begin{eqnarray}
    \notag & & \left|D_i^{(1)}\right|\\
        \notag & \leq &  \dfrac{K_0(0)h^{-d}\left|1-F_{S\mid X}(S_i|X_i)\right|+\left|h^{-d}J_{i}^{(0,4)}+\omega h^{-d}J_{i}^{(0,5)} -h^{-d}J_i^{(0,7)}-\omega h^{-d}J_i^{(0,9)} \right|}{\Lambda(X_i)(n+\omega m)h^{-d}} \\
        \notag & \leq & \dfrac{K_0(0)}{\underline{L}_3(n+\omega m)h^{d}} + \dfrac{\left|h^{-d}J_{i}^{(0,4)}+\omega h^{-d}J_{i}^{(0,5)} -h^{-d}J_i^{(0,7)}-\omega h^{-d}J_i^{(0,9)} \right|}{\underline{L}_3(n+\omega m)}\,,
    \end{eqnarray}
    we conclude that 
    $\left|D_i^{(1)}\right|\leq \dfrac{K_0(0)}{\underline{L}_3(n+\omega m)h^{d}} + \dfrac{\tau_{4,1}}{\underline{L}_3}$ on $A_i^{(4,1)}(\tau_{4,1})$ with probability over
    \begin{eqnarray}
        \notag 1&- & \exp\left[ -\dfrac{\tau_{4,1}^2(n+\omega m)^2h^d/2}{(1+\tau_{4,2})\overline{L}_3\big(n+V_r\omega^2m\big)+(1+\overline{L}_2\omega)K_0(0)\tau_{4,1}(n+\omega m)/3} \right]\, \\
        \notag &- & \exp\left[ -\dfrac{\tau_{4,2}^2\overline{L}_3\big(n+V_r\omega^2m\big)^2h^d/2}{n+V_r^2\omega^4m+\big(1+V_r\omega^2\big)\{K_0(0)\}^2\tau_{4,2}\big(n+V_r\omega^2m\big)/3} \right]\,.
    \end{eqnarray}

    For $D_i^{(2)}$, denote $B_l(x_0)=\left\{ x\in\mathcal{X}:\|x-x_0\|\leq lh \right\}$ for any positive $l\in\mathbb{R}$.
    Assumption \ref{ass2} indicates that for any $l_0>1$,
    \begin{eqnarray}
        \notag & & \left|D_i^{(2)}\right|\\
        \notag  & \leq & \dfrac{L\overset{}{\underset{j\neq i}\sum}K(X_j,X_i;h)\|X_j-X_i\|+\omega L\overset{m}{\underset{j=1}\sum} K(X_j^\prime,X_i;h)\dfrac{f_X(X_j^\prime)}{g_X(X_j^\prime)}\|X_j^\prime-X_i\|}{K_0(0)+\Lambda(X_i)(n+\omega m)} \\
        \notag & = & L\dfrac{\overset{}{\underset{X_j\in B_{l_0}(X_i), j\neq i}\sum}K(X_j,X_i;h)\|X_j-X_i\|+\omega\overset{}{\underset{X_j^\prime\in B_{l_0}(X_i)}\sum} K(X_j^\prime,X_i;h)\dfrac{f_X(X_j^\prime)}{g_X(X_j^\prime)}\|X_j^\prime-X_i\|}{K_0(0)+\Lambda(X_i)(n+\omega m)}\\
        \notag & & + L\dfrac{\overset{}{\underset{X_j\notin B_{l_0}(X_i)}\sum}K(X_j,X_i;h)\|X_j-X_i\|+\omega\overset{}{\underset{X_j^\prime\notin B_{l_0}(X_i)}\sum} K(X_j^\prime,X_i;h)\dfrac{f_X(X_j^\prime)}{g_X(X_j^\prime)}\|X_j^\prime-X_i\|}{K_0(0)+\Lambda(X_i)(n+\omega m)}\\
        \notag & \leq & Ll_0h\dfrac{\overset{}{\underset{j\neq i}\sum}K(X_j,X_i;h)+\omega\overset{m}{\underset{j=1}\sum} K(X_j^\prime,X_i;h)\dfrac{f_X(X_j^\prime)}{g_X(X_j^\prime)}}{\Lambda(X_i)(n+\omega m)}+\dfrac{L\overline{L}_1}{\underline{L}_1\underline{L}_3}l_0K_0(l_0)h^{-d+1},
    \end{eqnarray}
    where the second item in the last inequality is derived as follows due to Assumption \ref{ass2}:
    \begin{eqnarray}
        \notag & & \dfrac{\overset{}{\underset{X_j\notin B_{l_0}(X_i)}\sum}K(X_j,X_i;h)\|X_j-X_i\|+\omega\overset{}{\underset{X_j^\prime\notin B_{l_0}(X_i)}\sum} K(X_j^\prime,X_i;h)\dfrac{f_X(X_j^\prime)}{g_X(X_j^\prime)}\|X_j^\prime-X_i\|}{K_0(0)+\Lambda(X_i)(n+ m)}\\
        \notag & \leq & l_0K_0(l_0)h\dfrac{\overset{}{\underset{X_j\notin B_{l_0}(X_i)}\sum}1+\omega\overset{}{\underset{X_j^\prime\notin B_{l_0}(X_i)}\sum} f_X(X_j^\prime)/g_X(X_j^\prime)}{\Lambda(X_i)(n+\omega m)}\\
        \notag & \leq & l_0K_0(l_0)h\dfrac{n+\omega m\overline{L}_1/\underline{L}_1}{\Lambda(X_i)(n+\omega m)}\\
        \notag & \leq & \dfrac{\overline{L}_1}{\underline{L}_1\underline{L}_3}l_0K_0(l_0)h^{-d+1}\,.
    \end{eqnarray}
    
    For $\tau_5>0$, define
    \begin{equation*}
        A_i^{(5)}(\tau_5)=\left\{ \left|h^{-d}J_i^{(0,1)}+\omega h^{-d}J_i^{(0,8)}-\Lambda(X_i)(n+\omega m)h^{-d}\right|<\tau_5\Lambda(X_i)(n+\omega m)h^{-d} \right\}\,.
    \end{equation*}
    According to
    \begin{equation}
        \notag \left\{ \dfrac{f_X(X_j^\prime)}{g_X(X_j^\prime)} \right\}^2=\left[ E\left\{ r(X_{\ell}^\prime,S_{\ell}^\prime)\mid X_{\ell}^\prime=X_j^\prime \right\} \right]^2\leq E\left[ \left\{ r(X_{\ell}^\prime,S_{\ell}^\prime) \right\}^2\mid X_{\ell}^\prime=X_j^\prime \right]\leq V_r\,,
    \end{equation}
    we obtain that
    \begin{equation}
        \notag 1 - \mathrm{pr}\left( A_i^{(5)}(\tau_5) \right)\leq \exp\left[ -\dfrac{\underline{L}_3^2\tau_5^2(n+\omega m)^2h^d/2}{\overline{L}_3n+\overline{L}_3V_r\omega^2m+(1+ \underline{L}_1^{-1}\overline{L}_1\omega)\overline{L}_3K_0(0)\tau_5(n+\omega m)/3}\right].
    \end{equation}
    Take $l_0=K_0^{-1}(h^d)$, we arrive at 
    \begin{equation*}
        \left|D_i^{(2)}\right|\leq \left\{L(1+\tau_5)+\dfrac{L\overline{L}_1}{\underline{L}_1\underline{L}_3}\right\}K_0^{-1}(h^d)h
    \end{equation*}
    on $A_i^{(5)}(\tau_5)$.

    For $D_i^{(3)}$, it is obvious that $\left|D_i^{(3)}\right|\leq\tau_5$ on $A^{(5)}(\tau_5)$. 
    
    Finally, on the event
    \begin{equation*}
        \{\cap_{i=1}^{n+1}A_i^{(1)}\}\cap\{\cap_{i=1}^{n+1}A_i^{(2)}\}\cap\{\cap_{i=1}^{n+1}A_i^{(3)}\}\cap\{\cap_{i=1}^{n+1}A_i^{(4,1)}\}\cap\{\cap_{i=1}^{n+1}A_i^{(5)}\}\cap A(\gamma,r,k)\,,
    \end{equation*}
    we have
    \begin{eqnarray}
        \notag & & \underset{1\leq i\leq n+1}{\sup}\left|\hat{\beta}_{\omega,\hat{r}}(X_{i},S_{i})-F_{S\mid X}(S_i\mid X_i)\right| \\
        \notag & \leq &  \underset{1\leq i\leq n+1}{\sup}\left|J_i^{(1)}J_i^{(2)}\right| + \underset{1\leq i\leq n+1}{\sup}\left|J_i^{(1)}J_i^{(3)}D_i^{(1)}\right| + \underset{1\leq i\leq n+1}{\sup}\left|J_i^{(1)}J_i^{(3)}D_i^{(2)}\right| \\
        \notag & & + \underset{1\leq i\leq n+1}{\sup}\left|J_i^{(1)}J_i^{(3)}D_i^{(3)}\right| + \underset{1\leq i\leq n+1}{\sup}\left|J_i^{(1)}J_i^{(3)}F_{S\mid X}(S_i\mid X_i)-F_{S\mid X}(S_i\mid X_i)\right| \\
        \notag & \leq & \left|\left\{1+\dfrac{(\overline{L}_3\vee 1)\omega m}{(\underline{L}_3-\tau_1)(n+\omega m)}\dfrac{\epsilon_k(\gamma,r)}{h^{d/k}}+\dfrac{\tau_2}{\underline{L}_3-\tau_1}\right\}\left\{\dfrac{(\overline{L}_3\vee 1)\omega m}{(\underline{L}_3-\tau_1)(n+\omega m)}\dfrac{\epsilon_k(\gamma,r)}{h^{d/k}}+\dfrac{\tau_2}{\underline{L}_3-\tau_1}\right\}\right| \\
        \notag & & + \left|(1+\tau_3)\left\{\dfrac{(\overline{L}_3\vee 1)\omega m}{(\underline{L}_3-\tau_1)(n+\omega m)}\dfrac{\epsilon_k(\gamma,r)}{h^{d/k}}+\dfrac{\tau_2}{\underline{L}_3-\tau_1}\right\}\left\{\dfrac{K_0(0)}{\underline{L}_3(n+\omega m)h^{d}} + \dfrac{\tau_{4,1}}{\underline{L}_3}\right\}\right| \\
        \notag & & + \left|(1+\tau_3)\left\{\dfrac{(\overline{L}_3\vee 1)\omega m}{(\underline{L}_3-\tau_1)(n+\omega m)}\dfrac{\epsilon_k(\gamma,r)}{h^{d/k}}+\dfrac{\tau_2}{\underline{L}_3-\tau_1}\right\}\left[\left\{L(1+\tau_5)+\dfrac{L\overline{L}_1}{\underline{L}_1\underline{L}_3}\right\}K_0^{-1}(h^d)h\right]\right| \\
        \notag & & + \left|(1+\tau_3)\tau_5\left\{\dfrac{(\overline{L}_3\vee 1)\omega m}{(\underline{L}_3-\tau_1)(n+\omega m)}\dfrac{\epsilon_k(\gamma,r)}{h^{d/k}}+\dfrac{\tau_2}{\underline{L}_3-\tau_1}\right\}\right| \\
        \notag & &  + \left|(1+\tau_3)\left\{\dfrac{(\overline{L}_3\vee 1)\omega m}{(\underline{L}_3-\tau_1)(n+\omega m)}\dfrac{\epsilon_k(\gamma,r)}{h^{d/k}}+\dfrac{\tau_2}{\underline{L}_3-\tau_1}\right\}-1\right|\,,
    \end{eqnarray}
    where
    \begin{eqnarray}
        \notag & & \mathrm{pr}\left(\{\cap_{i=1}^{n+1}A_i^{(1)}\}\cap\{\cap_{i=1}^{n+1}A_i^{(2)}\}\cap\{\cap_{i=1}^{n+1}A_i^{(3)}\}\cap\{\cap_{i=1}^{n+1}A_i^{(4,1)}\}\cap\{\cap_{i=1}^{n+1}A_i^{(5)}\}\cap A(\gamma,r,k)\right) \\
        \notag & \geq & 1 - (n+1) \exp\left\{ -\dfrac{\tau_1^2(n+\omega m)^2h^d/2}{\overline{L}_3n+\overline{L}_3V_r\omega^2m+(1+\overline{L}_2\omega)K_0(0)\tau_1(n+\omega m)/3} \right\} \\
        \notag & & - (n+1)\exp\left[ -\dfrac{\tau_2^2(n+\omega m)^2h^d/2}{(\overline{L}_3\vee 1)\left[ h^{-2d/k}\{\epsilon_k(\gamma;r)\}^2\wedge\overline{L}_2^2 \right]\omega^2m+\overline{L}_2\omega K_0(0)\tau_2 (n+\omega m)/3} \right] \\
        \notag & & - (n+1)\exp\left\{ -\dfrac{\tau_3^2\underline{L}_3^2(n+\omega m)^2h^d/2}{\overline{L}_3n+\overline{L}_3V_r\omega^2m+(1+\overline{L}_2\omega)\overline{L}_3K_0(0)\tau_3(n+\omega m)/3} \right\} \\
        \notag & & - (n+1)\exp\left[ -\dfrac{\tau_{4,1}^2(n+\omega m)^2h^d/2}{(1+\tau_{4,2})\overline{L}_3\big(n+V_r\omega^2m\big)+(1+\overline{L}_2\omega)K_0(0)\tau_{4,1}(n+\omega m)/3} \right] \\
        \notag & & - (n+1)\exp\left[ -\dfrac{\tau_{4,2}^2\overline{L}_3\big(n+V_r\omega^2m\big)^2h^d/2}{n+V_r^2\omega^4m+\big(1+V_r\omega^2\big)\{K_0(0)\}^2\tau_{4,2}\big(n+\omega^2V_rm\big)/3} \right] \\
        \notag & & -(n+1)\exp\left[ -\dfrac{\underline{L}_3^2\tau_5^2(n+\omega m)^2h^d/2}{\overline{L}_3n+\overline{L}_3V_r\omega^2m+(1+ \underline{L}_1^{-1}\overline{L}_1\omega)\overline{L}_3K_0(0)\tau_5(n+\omega m)/3}\right]-\gamma\,.
    \end{eqnarray}
    
    With properly chosen $\tau_1,\tau_2,\tau_3,\tau_{4,1},\tau_{4,2}$ and $\tau_5$, there exists some $\tau>0$, which is bounded by a universal positive constant, such that
    \begin{eqnarray}
        \notag \underset{1\leq i\leq n+1}{\sup}\left|\hat{\beta}_{\hat{r}}(X_{i},S_{i})-F_{S\mid X}(S_i\mid X_i)\right| \leq C_0\left\{\dfrac{\omega m\epsilon_k(\gamma,r)}{(n+\omega m)h^{d/k}} + K_0^{-1}(h^{d})h + \tau\right\}\,.
    \end{eqnarray}
    for some positive constant $C_0$,
    with probability larger than
    \begin{eqnarray}
        \notag 
        1 - C_1n\exp\left\{ -C_2\tau^2(n+\omega m)h^d \right\}-\gamma\,,
    \end{eqnarray}
    for some positive constants $C_1$ and $C_2$.

    Solving the $\tau$ from $C_1n\exp\left\{ -C_2\tau^2(n+\omega m)h^d \right\}=\delta$, we get with probability at least $1-\delta-\gamma$,
    \begin{eqnarray}
        \notag \underset{1\leq i\leq n+1}{\sup}\Delta_{i} & \leq & C_0\Bigg\{\dfrac{\omega m\epsilon_k(\gamma,r)}{(n+\omega m)h^{d/k}} + K_0^{-1}(h^{d})h + \dfrac{1}{(n+\omega m)^{1/2}h^{d/2}}\log^{1/2}\left(\dfrac{n}{\delta}\right) \Bigg\}\,.
    \end{eqnarray}

    \textbf{Step II}: In the second part of the proof, we establish Lemma~\ref{lemma:delta_concentration2} under the setting where $\hat{r}(x,s)$ is estimated using $\mathcal{Z}_n \cup \{Z_{n+1}\}$ and $\mathcal{Z}_{m}^\prime$.
    Since only the assumption about the dependency between $\hat{r}(x,s)$ and auxiliary data $\mathcal{Z}_{m}^{\prime}$ is modified here, only events related to $J_i^{(0,3)}$ and $J_i^{(0,6)}$ are affected. 
    
    According to the proof in \textbf{Step I}, we have the following decomposition:
    \begin{gather*}
        \hat{\beta}_{\omega,\hat{r}}(X_{i},S_{i}) = J_i^{(1)}\left( J_i^{(2)}+J_i^{(3)}H_i^{(1)} \right)\,,
    \end{gather*}
    where only $J_i^{(1)}$ and $J_i^{(2)}$ involves $J_i^{(0,3)}$ and $J_i^{(0,6)}$. 
    Since $J_i^{(2)}$ can be treated in the same manner as $J_i^{(1)}$ following the derivation in \textbf{Step I}, it suffices to reconsider $J_i^{(1)}$ here.
    
    On the event $A_i^{(1)}(\tau_1)$, based on \eqref{eq:Ji1 upper1}, Lemma \ref{lemma:KernelFun} and Assumption \ref{ass3}, we have
    \begin{align}
        \notag &\left|J_i^{(1)}-1\right|\\
        \notag \leq&\dfrac{\omega h^{-d} \overset{m}{\underset{j=1}\sum}K(X_j^\prime,X_i;h)\left|\hat{r}(X_j^\prime,S_j^\prime)-r(X_j^\prime,S_j^\prime)\right|}{K_0(0)h^{-d}+h^{-d}J_{i}^{(0,1)}+\omega h^{-d}J_{i}^{(0,2)}}\\
        \notag \leq&\dfrac{\omega h^{-d} \overset{m}{\underset{j=1}\sum}K(X_j^\prime,X_i;h)\left|\hat{r}(X_j^\prime,S_j^\prime)-r(X_j^\prime,S_j^\prime)\right|}{(n+\omega m)(\underline{L}_3-\tau_1)}\\
        \notag \leq&\dfrac{\omega h^{-d} \overset{m}{\underset{j=1}\sum}K(X_j^\prime,X_i;h)\left\{ \left|\hat{r}(X_j^\prime,S_j^\prime)-\hat{r}_j^\prime(X_j^\prime,S_j^\prime)\right|+\left|\hat{r}_j^\prime(X_j^\prime,S_j^\prime)-r(X_j^\prime,S_j^\prime)\right| \right\}}{(n+\omega m)(\underline{L}_3-\tau_1)}\\
        \notag \leq&\dfrac{ C_r \omega h^{-d} \overset{m}{\underset{j=1}\sum}K(X_j^\prime,X_i;h)}{m(n+\omega m)(\underline{L}_3-\tau_1)}+\dfrac{\omega h^{-d} \overset{m}{\underset{j=1}\sum}K(X_j^\prime,X_i;h)\left|\hat{r}_j^\prime(X_j^\prime,S_j^\prime)-r(X_j^\prime,S_j^\prime)\right|}{(n+\omega m)(\underline{L}_3-\tau_1)}\,.
    \end{align}
    For $\tau_6>0$, consider the event
    \begin{gather*}
        A_i^{(6)}(\tau_6)=\left\{ \omega h^{-d} \overset{m}{\underset{j=1}\sum}K(X_j^\prime,X_i;h)-\omega m\overline{L}_3<\tau_6(n+\omega m) \right\}\,.
    \end{gather*}
    Due to Lemma \ref{lemma:KernelFun}, it can be shown that
    \begin{gather*}
       \left\{ \omega h^{-d} \overset{m}{\underset{j=1}\sum}K(X_j^\prime,X_i;h)-\omega h^{-d} \overset{m}{\underset{j=1}\sum}E\{K(X_j^\prime,X_i;h)\}<\tau_6(n+\omega m) \right\}\subset A_i^{(6)}(\tau_6)\,. 
    \end{gather*}
    Similar to the calculation of $\mathrm{pr}\big( A_i^{(1)}(\tau_1)\big)$ in \textbf{Step I}, we obtain
    \begin{gather*}
        1-\mathrm{pr}\left( A_i^{(6)}(\tau_6) \right)\leq\exp\left\{ -\dfrac{\tau_6^2(n+\omega m)^2h^d/2}{\overline{L}_3\omega^2m+K_0(0)\tau_6(n+\omega m)/3} \right\}\,.
    \end{gather*}
    Therefore, on the event $A_i^{(6)}(\tau_6)$, it holds that
    \begin{align*}
        \left|J_i^{(1)}-1\right|\leq \dfrac{ C_r (\overline{L}_3+\tau_6)\omega}{(\underline{L}_3-\tau_1)m}+\dfrac{\omega h^{-d} \overset{m}{\underset{j=1}\sum}K(X_j^\prime,X_i;h)\left|\hat{r}_j^\prime(X_j^\prime,S_j^\prime)-r(X_j^\prime,S_j^\prime)\right|}{(n+\omega m)(\underline{L}_3-\tau_1)}\,.
    \end{align*}
    
    We modify the definition of $A_i^{(2)}(\tau_2)$ in \textbf{Step I} as follows:
    \begin{align*}
        A_i^{(2)}(\tau_2) =&  \bigg\{ \omega h^{-d} \overset{m}{\underset{j=1}\sum}K(X_j^\prime,X_i;h)\left|\hat{r}_j^\prime(X_j^\prime,S_j^\prime)-r(X_j^\prime,S_j^\prime)\right| \\
        &~- \omega mh^{-d}E\left\{ K(X_j^\prime,X_i;h)\left|\hat{r}_j^\prime(X_j^\prime,S_j^\prime)-r(X_j^\prime,S_j^\prime)\right|\mid \mathcal{Z}_n\cup\{Z_{n+1}\} \right\} \\
        &~< \tau_2(n+\omega m) \bigg\}\,.
    \end{align*}
    Consider $(X^\prime,S^\prime)$ as an i.i.d.~copy of $(X_1^\prime,S_1^\prime)$. Regarding the fact that $\hat{r}_j^{\prime}(x,s)$ is independent of $Z_j^{\prime}$, we have that
    \begin{align*}
        &h^{-d}E\left\{ K(X_j^\prime,X_i;h)\left|\hat{r}_j^\prime(X_j^\prime,S_j^\prime)-r(X_j^\prime,S_j^\prime)\right|\mid \mathcal{Z}_n\cup\{Z_{n+1}\} \right\}\\
        =&h^{-d}E\left\{ K(X^\prime,X_i;h)\left|\hat{r}_j^\prime(X^\prime,S^\prime)-r(X^\prime,S^\prime)\right|\mid \mathcal{Z}_n\cup\{Z_{n+1}\} \right\}\,.
    \end{align*}
    Recall that $A(\gamma,r,k)=\{D_k(r,\hat{r})\leq\epsilon_k(\gamma;r)\}$. Since $1=\mathbb{1}(A(\gamma,r,k))+\mathbb{1}(A^c(\gamma,r,k))$, where $A^c(\gamma,r,k)$ denotes the complement of $A(\gamma,r,k)$, we can derive that
    \begin{align}
        \notag &h^{-d}E\left\{ K(X^\prime,X_i;h)\left|\hat{r}_j^\prime(X^\prime,S^\prime)-r(X^\prime,S^\prime)\right|\mid \mathcal{Z}_n\cup\{Z_{n+1}\} \right\}\\
        \notag =&h^{-d}E\left\{ K(X^\prime,X_i;h)\left|\hat{r}_j^\prime(X^\prime,S^\prime)-r(X^\prime,S^\prime)\right|\mathbb{1}(A(\gamma,r,k))\mid \mathcal{Z}_n\cup\{Z_{n+1}\} \right\}\\
        \notag &+h^{-d}E\left\{ K(X^\prime,X_i;h)\left|\hat{r}_j^\prime(X^\prime,S^\prime)-r(X^\prime,S^\prime)\right|\mathbb{1}(A^c(\gamma,r,k))\mid \mathcal{Z}_n\cup\{Z_{n+1}\} \right\}\\
        \leq&h^{-d}E\left\{ K(X^\prime,X_i;h)\left|\hat{r}_j^\prime(X^\prime,S^\prime)-r(X^\prime,S^\prime)\right|\mathbb{1}(A(\gamma,r,k))\mid \mathcal{Z}_n\cup\{Z_{n+1}\} \right\}\label{eq:Ai2 new exp 1}\\
        &+\overline{L}_2h^{-d}E\left\{ K(X^\prime,X_i;h)\mathbb{1}(A^c(\gamma,r,k))\mid \mathcal{Z}_n\cup\{Z_{n+1}\} \right\}\label{eq:Ai2 new exp 2}\,.
    \end{align}
    
    As $X^\prime$ is independent of both $\mathcal{Z}_n\cup\{Z_{n+1}\}$ and $\mathcal{Z}_m^\prime$, we can bound \eqref{eq:Ai2 new exp 2} as
    \begin{align*}
        &\overline{L}_2h^{-d}E\left\{ K(X^\prime,X_i;h)\mathbb{1}(A^c(\gamma,r,k))\mid \mathcal{Z}_n\cup\{Z_{n+1}\} \right\}\\
        =&\overline{L}_2h^{-d}E\left\{ K(X^\prime,X_i;h)\mid \mathcal{Z}_n\cup\{Z_{n+1}\} \right\}E\left\{ \mathbb{1}(A^c(\gamma,r,k))\mid \mathcal{Z}_n\cup\{Z_{n+1}\} \right\}\\
        \leq&\overline{L}_2\overline{L}_3E\left\{ \mathbb{1}(A^c(\gamma,r,k))\mid \mathcal{Z}_n\cup\{Z_{n+1}\} \right\}\,.
    \end{align*}
    By the definition of $A(\gamma,r,k)$, we know $\mathrm{pr}\{A(\gamma,r,k)\} \ge 1-\gamma$, and hence
    \begin{gather*}
        E\left[E\left\{ \mathbb{1}(A^c(\gamma,r,k))\mid \mathcal{Z}_n\cup\{Z_{n+1}\} \right\}\right]=\mathrm{pr}(A^c(\gamma,r,k))\leq\gamma\,.
    \end{gather*}
    Define the event $A_n(\gamma,r,k)=\{E\left\{ \mathbb{1}(A^c(\gamma,r,k))\mid \mathcal{Z}_n\cup\{Z_{n+1}\} \right\}\leq\gamma^{1/2}\}$. By Markov’s inequality,
    \begin{align*}
        \mathrm{pr}\big(A_n^c(\gamma,r,k)\big)
    = \mathrm{pr}\Big( E\big\{\mathbb{1}\big(A^c(\gamma,r,k)\big)\mid \mathcal{Z}_n\cup\{Z_{n+1}\}\big\} > \gamma^{1/2}\Big)
    \le \frac{\gamma}{\gamma^{1/2}} = \gamma^{1/2}.
    \end{align*}
    Therefore, 
    \begin{gather*}
        \mathrm{pr}(A_n(\gamma,r,k))\leq\gamma^{1/2}\,.
    \end{gather*}
    
    Using the iterated expectation theorem, we can analyze \eqref{eq:Ai2 new exp 1} as
    \begin{align*}
        &h^{-d}E\left\{ K(X^\prime,X_i;h)\left|\hat{r}_j^\prime(X^\prime,S^\prime)-r(X^\prime,S^\prime)\right|\mathbb{1}(A(\gamma,r,k))\mid \mathcal{Z}_n\cup\{Z_{n+1}\} \right\}\\
        =&h^{-d}E\left[ E\left\{ K(X^\prime,X_i;h)\left|\hat{r}_j^\prime(X^\prime,S^\prime)-r(X^\prime,S^\prime)\right|\mathbb{1}(A(\gamma,r,k))\mid \mathcal{Z}_n\cup\{Z_{n+1}\},\mathcal{Z}_m^\prime \right\}\mid\mathcal{Z}_n\cup\{Z_{n+1}\} \right]\,,
    \end{align*}
    where the inner expectation is taken over $(X^\prime,S^\prime)$. By Holder's inequality, Lemma \ref{lemma:KernelFun} and some simple algebra, we can derive
    \begin{align*}
        &h^{-d}E\left\{ K(X^\prime,X_i;h)\left|\hat{r}_j^\prime(X^\prime,S^\prime)-r(X^\prime,S^\prime)\right|\mathbb{1}(A(\gamma,r,k))\mid \mathcal{Z}_n\cup\{Z_{n+1}\},\mathcal{Z}_m^\prime \right\}\\
        \leq&h^{-d}E\left\{ K(X^\prime,X_i;h)\left|\hat{r}_j^\prime(X^\prime,S^\prime)-\hat{r}(X^\prime,S^\prime)\right|\mathbb{1}(A(\gamma,r,k))\mid \mathcal{Z}_n\cup\{Z_{n+1}\},\mathcal{Z}_m^\prime \right\}\\
        &+h^{-d}E\left\{ K(X^\prime,X_i;h)\left|\hat{r}(X^\prime,S^\prime)-r(X^\prime,S^\prime)\right|\mathbb{1}(A(\gamma,r,k))\mid \mathcal{Z}_n\cup\{Z_{n+1}\},\mathcal{Z}_m^\prime \right\}\\
        \leq& C_r h^{-d}E\left\{ K(X^\prime,X_i;h)\mid X_i \right\}m^{-1}\\
        &+\left( E\left[ \left\{ h^{-d}K(X^\prime,X_i;h) \right\}^{k/(k-1)}\mid X_i \right] \right)^{(k-1)/k}\\
        &~~~~~~\times\mathbb{1}(A(\gamma,r,k))\left[ E\left\{ \left|\hat{r}(X^\prime,S^\prime)-r(X^\prime,S^\prime)\right|^k\mid \mathcal{Z}_n\cup\{Z_{n+1}\},\mathcal{Z}_m^\prime \right\} \right]^{1/k}\\
        \leq& C_r \overline{L}_3m^{-1}+(\overline{L}_3\vee 1)h^{-d/k}D_k(r,\hat{r})\mathbb{1}(A(\gamma,r,k))\\
        \leq& C_r \overline{L}_3m^{-1}+(\overline{L}_3\vee 1)h^{-d/k}\epsilon_k(\gamma;r)\,.
    \end{align*}
    Thus, on the event $A_n(\gamma,r,k)$, combining the results on \eqref{eq:Ai2 new exp 1} and \eqref{eq:Ai2 new exp 2} yields
    \begin{align*}
        &h^{-d}E\left\{ K(X^\prime,X_i;h)\left|\hat{r}_j^\prime(X^\prime,S^\prime)-r(X^\prime,S^\prime)\right|\mid \mathcal{Z}_n\cup\{Z_{n+1}\} \right\}\\
        \leq&(\overline{L}_2+ C_r )(\overline{L}_3\vee 1)\left\{ \gamma^{1/2}+m^{-1}+h^{-d/k}\epsilon_k(\gamma;r) \right\}\,.
    \end{align*}
    Consequently, on the event $A_i^{(2)}(\tau_2)\cap A_i^{(6)}(\tau_6)\cap A(\gamma,r,k)\cap A_n(\gamma,r,k)$, we have
    \begin{align}
        \notag &\left|J_i^{(1)}-1\right|\\
        \notag \leq&\dfrac{\omega C_r (\overline{L}_3+\tau_6)}{(\underline{L}_3-\tau_1)m}+\dfrac{\omega m(\overline{L}_2+ C_r )(\overline{L}_3\vee 1)\left\{ \gamma^{1/2}+m^{-1}+h^{-d/k}\epsilon_k(\gamma;r) \right\}+\tau_2(n+\omega m)}{(n+\omega m)(\underline{L}_3-\tau_1)}\\
        \leq&\dfrac{(\overline{L}_2+ C_r )(\overline{L}_3\vee 1+\tau_6)}{(\underline{L}_3-\tau_1)}(m^{-1}+\tau_2+\gamma^{1/2})+\dfrac{(\overline{L}_2+ C_r )(\overline{L}_3\vee 1)}{(\underline{L}_3-\tau_1)} \dfrac{\omega m\epsilon_k(\gamma;r)}{(n+\omega m)h^{d/k}}\,.\label{eq:Ji1 Exp bound}
    \end{align}
    
    We proceed to derive a lower bound for the probability of $A_i^{(2)}(\tau_2)$. Let $\mathcal{F}_j^\prime$ denote the $\sigma$-algebra generated by $(X_1^\prime,S_1^\prime), \ldots, (X_j^\prime,S_j^\prime)$ and $\mathcal{Z}_n\cup\{Z_{n+1}\}$ for $j\in[m]$, so that $\mathcal{F}_0^\prime  \subset \cdots \subset \mathcal{F}_m^\prime$. Define $\xi_j=\omega h^{-d}K(X_j^\prime,X_i;h)|\hat{r}_j^\prime(X_j^\prime,S_j^\prime)-r(X_j^\prime,S_j^\prime)|$ for $j\in[m]$ and let
    \begin{align*}
        \widetilde{\xi}_j=&E\left( \sum_{\ell=1}^{m}\xi_{\ell}\mid \mathcal{F}_{j}^\prime \right)-E\left( \sum_{\ell=1}^{m}\xi_{\ell}\mid \mathcal{F}_{j-1}^\prime \right)\,,\,j\in[m],\\
        \widetilde{\xi}_0=&E\left( \sum_{\ell=1}^{m}\xi_{\ell}\mid \mathcal{F}_{0}^\prime \right)\,.
    \end{align*}
    be the martingale difference.
    Then,
    \begin{align*}
        \sum_{j=1}^{m}\widetilde{\xi}_j=&\omega h^{-d} \overset{m}{\underset{j=1}\sum}K(X_j^\prime,X_i;h)\left|\hat{r}_j^\prime(X_j^\prime,S_j^\prime)-r(X_j^\prime,S_j^\prime)\right| \\
        &~- \omega mh^{-d}E\left\{ K(X_j^\prime,X_i;h)\left|\hat{r}_j^\prime(X_j^\prime,S_j^\prime)-r(X_j^\prime,S_j^\prime)\right|\mid \mathcal{Z}_n\cup\{Z_{n+1}\} \right\}\\
        =&\omega h^{-d} \overset{m}{\underset{j=1}\sum}K(X_j^\prime,X_i;h)\left|\hat{r}_j^\prime(X_j^\prime,S_j^\prime)-r(X_j^\prime,S_j^\prime)\right|-\widetilde{\xi}_0\,
    \end{align*}
    and
    \begin{align*}
        \left|\widetilde{\xi}_j\right|\leq&\sum_{\ell=1}^{m}\left|E\left( \xi_{\ell}\mid\mathcal{F}_j^\prime \right)-E\left( \xi_{\ell}\mid\mathcal{F}_{j-1}^\prime \right)\right|\,.
    \end{align*}
    For $\ell\neq j$, define $\xi_\ell^{\setminus j}$ and $\hat{r}_\ell^{\setminus j}$ as the counterpart of $\xi_\ell$ and $\hat{r}_\ell^\prime$ by replacing $Z_j^\prime$ with $\widetilde{Z}_j^\prime$, where $\widetilde{Z}_j^\prime$ is an i.i.d.~copy of $Z_j^\prime$. Since both $\hat{r}_\ell^{\setminus j}$ and $\hat{r}_\ell^\prime$ utilize $m-1$ samples from the auxiliary dataset with only one differing sample, Assumption 3 yields $\sup_{x,s}\left|\hat{r}_\ell^{\setminus j}(x,s)-\hat{r}_\ell^\prime(x,s)\right|\leq 2 C_r (m-1)^{-1}\leq 4 C_r m^{-1}$. Under this circumstance, %$\left|\xi_l-\xi_l^{\setminus j}\right|$ is bounded:
    \begin{align*}
        \left|\xi_\ell-\xi_\ell^{\setminus j}\right| \leq&\omega h^{-d}K(X_\ell^\prime,X_i;h)\left|\hat{r}_\ell^{\setminus j}(X_\ell^\prime,S_\ell^\prime)-\hat{r}_\ell^{\prime}(X_\ell^\prime,S_\ell^\prime)\right|\\
        \leq&4\omega C_r h^{-d}K(X_\ell^\prime,X_i;h)m^{-1}\,.
    \end{align*}
    By definition, $\xi_\ell$ is independent of $\widetilde{Z}_j^\prime$, which indicates 
    \begin{gather*}
        E\left( \xi_\ell\mid\mathcal{F}_j^\prime \right)=E\left\{ \xi_\ell\mid\sigma(\mathcal{F}_{j}^\prime,\widetilde{Z}_j^\prime) \right\}\,,
    \end{gather*}
    where $\sigma(\mathcal{F}_{j}^\prime,\widetilde{Z}_j^\prime)$ is the $\sigma$-algebra generated by $(X_1^\prime,S_1^\prime), \ldots, (X_j^\prime,S_j^\prime)$, $\widetilde{Z}_j^\prime$ and $\mathcal{Z}_n\cup\{Z_{n+1}\}$.
    Moreover, since $\xi_\ell^{\setminus j}$ is obtained from $\xi_\ell$ by replacing $Z_j^\prime$ with $\widetilde{Z}_j^\prime$, and $Z_j^\prime$ and $\widetilde{Z}_j^\prime$ are i.i.d.~while being independent of other variables, it follows that $E\left( \xi_\ell\mid\mathcal{F}_{j-1}^\prime \right)=E\big( \xi_\ell^{\setminus j}\mid\mathcal{F}_{j-1}^\prime \big)$. As $\xi_\ell^{\setminus j}$ and $\widetilde{Z}_j^\prime$ are independent, we have $E\big( \xi_\ell^{\setminus j}\mid\mathcal{F}_{j-1}^\prime \big)=E\big( \xi_\ell^{\setminus j}\mid\mathcal{F}_j^\prime \big)$. Applying the iterated expectation theorem yields
    \begin{gather*}
        E\left( \xi_\ell\mid\mathcal{F}_{j-1}^\prime \right)=E\left[ E\left\{\xi_\ell^{\setminus j}\mid\sigma(\mathcal{F}_{j}^\prime,\widetilde{Z}_j^\prime) \right\}\mid\mathcal{F}_{j}^\prime \right]\,.
    \end{gather*}
    Furthermore, the martingale difference can be bounded by
    \begin{align*}
        \left|\widetilde{\xi}_j\right|\leq&\sum_{\ell=1}^{m}\left|E\left( \xi_\ell\mid\mathcal{F}_j^\prime \right)-E\left( \xi_\ell\mid\mathcal{F}_{j-1}^\prime \right)\right|\\
        =&\sum_{\ell=1}^{m}\left|E\left\{ \xi_\ell\mid\sigma(\mathcal{F}_{j}^\prime,\widetilde{Z}_j^\prime) \right\}-E\left[ E\left\{\xi_\ell^{\setminus j}\mid\sigma(\mathcal{F}_{j}^\prime,\widetilde{Z}_j^\prime) \right\}\mid\mathcal{F}_{j}^\prime \right]\right|\\
        \leq&\sum_{\ell=1}^{m}E\left[ E\left\{ \left|\xi_\ell-\xi_\ell^{\setminus j}\right|\mid\sigma(\mathcal{F}_{j}^\prime,\widetilde{Z}_j^\prime) \right\}\mid\mathcal{F}_{j}^\prime \right]\\
        \leq&4 C_r \omega h^{-d}\sum_{\ell=1}^{m}E\left\{ K(X_\ell^\prime,X_i;h)\mid\mathcal{F}_j^\prime\right\}m^{-1}\\
        \leq&4 C_r \omega\left\{ (m-j)\overline{L}_3+h^{-d}\sum_{\ell=1}^{j}K(X_\ell^\prime,X_i;h) \right\}m^{-1}\\
        \leq&4 C_r \omega\left\{ \overline{L}_3+m^{-1}h^{-d}\sum_{\ell=1}^{j-1}K(X_\ell^\prime,X_i;h)+m^{-1}h^{-d}K(X_j^\prime,X_i;h) \right\}\,.
    \end{align*}
    It follows from Cauchy inequality and Lemma \ref{lemma:KernelFun} that
    \begin{align*}
        &E\left( \widetilde{\xi}_j^2\mid\mathcal{F}_{j-1}^\prime \right)\\
        \leq&(4 C_r \omega)^2E\left[ 3\overline{L}_3^2+3\left\{ m^{-1}h^{-d}\sum_{\ell=1}^{j-1}K(X_\ell^\prime,X_i;h) \right\}^2+3\left\{ m^{-1}h^{-d}K(X_j^\prime,X_i;h) \right\}^2\mid\mathcal{F}_{j-1}^\prime  \right]\\
        %\leq&48(\omega C_r )^2\left[ \overline{L}_3^2+\left\{ m^{-1}h^{-d}\sum_{l=1}^{j-1}K(X_l^\prime,X_i;h) \right\}^2+\overline{L}_3m^{-2}h^{-d} \right]\\
        \leq&48( C_r \omega)^2\left[ \overline{L}_3^2+\left\{ m^{-1}h^{-d}\sum_{\ell=1}^{m}K(X_\ell^\prime,X_i;h) \right\}^2+\overline{L}_3m^{-2}h^{-d} \right]\,.
    \end{align*}
    Finally, the accumulated variance is bounded by 
    \begin{gather*}
        \sum_{j=1}^{m}E\left( \widetilde{\xi}_j^2\mid\mathcal{F}_{j-1}^\prime \right)\leq48( C_r \omega)^2\left[ \overline{L}_3^2+\left\{ m^{-1}h^{-d}\sum_{\ell=1}^{m}K(X_\ell^\prime,X_i;h) \right\}^2+\overline{L}_3m^{-2}h^{-d} \right]m\,.
    \end{gather*}
    Under the event $A_i^{(6)}(\tau_6)$, the bound can be further refined as
    \begin{align*}
        \sum_{j=1}^{m}E\left( \widetilde{\xi}_j^2\mid\mathcal{F}_{j-1}^\prime \right)\leq&48 C_r ^2\left[ \overline{L}_3^2\omega^2+\left\{ \omega m^{-1}h^{-d}\sum_{\ell=1}^{m}K(X_\ell^\prime,X_i;h) \right\}^2+\overline{L}_3\omega^2m^{-2}h^{-d} \right]m\\
        \leq&48 C_r ^2\left[ \overline{L}_3^2\omega^2+\left\{ \overline{L}_3\omega+\tau_6(n/m+\omega) \right\}^2+\omega^2\overline{L}_3m^{-2}h^{-d} \right]m\,.
    \end{align*}
    Denote $v_{n,m}^{\omega, h}=48 C_r ^2\left[ \overline{L}_3^2\omega^2+\left\{ \overline{L}_3\omega+\tau_6(n/m+\omega) \right\}^2+\overline{L}_3\omega^2m^{-2}h^{-d} \right]$.
    For any $j\in[m]$, the following bound always holds:
    \begin{gather*}
        \left|\widetilde{\xi}_j\right|\leq 4 C_r \omega K_0(0)h^{-d}\,.
    \end{gather*}
    Apply Freedman's inequality, we get
    \begin{align*}
        &\mathrm{pr}\left( \sum_{j=1}^{m}\widetilde{\xi}_j\geq\tau_2(n+\omega m) \right)\\
        \leq&\mathrm{pr}\left( \sum_{j=1}^{m}\widetilde{\xi}_j\geq\tau_2(n+\omega m),A_i^{(6)}(\tau_6) \right)+\left\{ 1-\mathrm{pr}\left( A_i^{(6)}(\tau_6) \right) \right\}\\
        \leq&\mathrm{pr}\left( \sum_{j=1}^{m}\widetilde{\xi}_j\geq\tau_2(n+\omega m)\,,\,\sum_{j=1}^{m}E\left( \widetilde{\xi}_j^2\mid\mathcal{F}_{j-1}^\prime \right)\leq v_{n,m}^{\omega, h}m \right)+\left\{ 1-\mathrm{pr}\left( A_i^{(6)}(\tau_6) \right) \right\}\\
        \leq&\exp\left( -\dfrac{\tau_2^2(n+\omega m)^2h^{d}}{2\left\{ v_{n,m}^{\omega, h}h^{d}m+4 C_r \omega K_0(0)\tau_2(n+\omega m)/3 \right\}} \right)+\left\{ 1-\mathrm{pr}\left( A_i^{(6)}(\tau_6) \right) \right\}\,.
    \end{align*}
    Consequently,
    \begin{align*}
        & 1 - \mathrm{pr}\left(A_i^{(2)}(\tau_2)\right)=\mathrm{pr}\left( \sum_{j=1}^{m}\widetilde{\xi}_j\geq\tau_2(n+\omega m) \right)\\
        \leq &\exp\left( -\dfrac{\tau_2^2(n+\omega m)^2h^{d}}{2\left\{ v_{n,m}^{\omega, h}h^{d}m+4 C_r \omega K_0(0)\tau_2(n+\omega m)/3 \right\}} \right)\\
        &~~~~~~~~~~~~~~~~~~~~~~~~~~~~~~~~~~~~~~~~+\exp\left\{ -\dfrac{\tau_6^2(n+\omega m)^2h^d/2}{\overline{L}_3\omega^2m+K_0(0)\tau_6(n+\omega m)/3} \right\}\,.
    \end{align*}
    
    Similar to the proof in \textbf{Step I}, with properly chosen $\tau_1,\tau_2,\tau_3,\tau_{4,1},\tau_{4,2},\tau_5$ and $\tau_6$, there exists some $\tau>0$, which is bounded by a universal positive constant, such that
    \begin{eqnarray}
        \notag \underset{1\leq i\leq n+1}{\sup}\left|\hat{\beta}_{\hat{r}}(X_{i},S_{i})-F_{S\mid X}(S_i\mid X_i)\right| \leq C_0\left\{\dfrac{\omega m\epsilon_k(\gamma,r)}{(n+\omega m)h^{d/k}} + K_0^{-1}(h^{d})h + \tau + \gamma^{1/2} + m^{-1}\right\}\,.
    \end{eqnarray}
    for some positive constant $C_0$,
    with probability larger than
    \begin{eqnarray}
        \notag 
        1 - C_1n\exp\left\{ -C_2\tau^2(n+\omega m)h^d \right\}-\gamma^{1/2}\,,
    \end{eqnarray}
    for some positive constants $C_1$ and $C_2$. Solve the $\tau$ from $C_1n\exp\left\{ -C_2\tau^2(n+\omega m)h^d \right\}=\delta$, we get with probability at least $1-\delta-\gamma^{1/2}$,
    \begin{eqnarray}
        \notag \underset{1\leq i\leq n+1}{\sup}\Delta_{i} & \leq & C_0\Bigg\{\dfrac{\omega m\epsilon_k(\gamma,r)}{(n+\omega m)h^{d/k}} + K_0^{-1}(h^{d})h + \dfrac{1}{(n+\omega m)^{1/2}h^{d/2}}\log^{1/2}\left(\dfrac{n}{\delta}\right)+\gamma^{1/2} \Bigg\}\,,
    \end{eqnarray}
    and thus complete the proof of this lemma.
\end{proof}

\subsection{Proof of Theorem \ref{theo:conditional_coverage_error_bound}}\label{Proof of theo:conditional_coverage_error_bound}

\begin{proof}
    Lemma \ref{lemma:conditional_coverage_bound} and Lemma \ref{lemma:delta_concentration2} imply that
    \begin{eqnarray}
        \notag & & \left|\mathrm{pr}\left( Y_{n+1}\in\widehat{C}_{\alpha}^{\mathrm{ELCP}}(X_{n+1})\mid X_{n+1}=x_0 \right)-(1-\alpha)\right| - (n+1)^{-1}  \\
        \notag & \leq & \delta + \gamma^{1/2} + 2C_0\Bigg\{\dfrac{\omega m\epsilon_k(\gamma,r)}{(n+\omega m)h^{d/k}} + K_0^{-1}(h^{d})h + \dfrac{1}{(n+\omega m)^{1/2}h^{d/2}}\log^{1/2}\left(\dfrac{n}{\delta}\right)+\gamma^{1/2} \Bigg\}\,.
    \end{eqnarray}
    Define $M_{n,m,\omega}=(n+\omega m)^{-1/2}h^{-d/2}$, and consider the function
    \begin{equation*}
        \varphi_0(\delta)=\delta+2C_0M_{n,m,\omega}^{1/2}\log^{1/2}\left(n/\delta\right)\,.
    \end{equation*}
    Let the derivative of $\varphi_0(\delta)$ equal to $0$, we obtain
    $\delta\log^{1/2}\left( n /\delta\right)=C_0M_{n,m,\omega}^{1/2}$.
    Let $\delta^*$ be the minimizer of $\varphi_0(\delta)$, then we know that 
    $\delta^*<C_0M_{n,m,\omega}^{1/2}/\log^{1/2}(n)$.
    Since $1<\log^{1/2}(n)<\log(n)<n$ for $n>3$, we have
    \begin{eqnarray}
        \notag \varphi_0(\delta^*) & < & \dfrac{C_0M_{n,m,\omega}^{1/2}}{\log^{1/2}(n)} + 2C_0M_{n,m,\omega}^{1/2}\log^{1/2}\left(\dfrac{n\log^{1/2}(n)}{C_0M_{n,m,\omega}^{1/2}}\right)\\
        \notag &\leq& \dfrac{C_0M_{n,m,\omega}^{1/2}}{\log^{1/2}(n)} + 2C_0M_{n,m,\omega}^{1/2}\left\{ \sqrt{2}\log^{1/2}\left( n \right)+\log^{1/2}\left( 1/M_{n,m,\omega} \right) \right\}\\
        \notag &\leq & 3C_0M_{n,m,\omega}^{1/2}\left\{ \log^{1/2}\left( n \right)+\log^{1/2}\left( 1/M_{n,m,\omega} \right)\right\}\,.
    \end{eqnarray}
    Therefore, we can derive that
    \begin{eqnarray}
        \notag & & \left|\mathrm{pr}\left\{ Y_{n+1}\in\widehat{C}_{\alpha}^{\mathrm{ELCP}}(X_{n+1})\mid X_{n+1}=x_0 \right\}-(1-\alpha)\right| - (n+1)^{-1} \\
        \notag & \leq & \widetilde{C}_0\left[\dfrac{\omega m\epsilon_k(\gamma,r)}{(n+\omega m)h^{d/k}} + K_0^{-1}(h^{d})h + M_{n,m,\omega}^{1/2}\left\{\log^{1/2}(n)+\log^{1/2}(1/M_{n,m,\omega})\right\} + \gamma^{1/2}\right]\,
    \end{eqnarray}
    for some positive constant $\widetilde{C}_0$.
\end{proof}

\subsection{Proof of Theorem \ref{theo:weak_test}}\label{Proof of lemma:test_conditional_B}

Theorem \ref{theo:weak_test} is a direct consequence of Lemma \ref{lemma:delta_concentration2} and the following lemma.

\begin{lemma}\label{lemma:test_conditional_B}
    Assume $\sup_{1\leq i\leq n+1}\Delta_{i}\leq\varepsilon$ with probability at least $1-\delta$. Then for any fixed set $\mathcal{B}\subset \mathcal{X}$ with $\mathrm{pr}(X_{n+1} \in \mathcal{B})=p_0$,
    \begin{align*}
        & \left|\mathrm{pr}\left( Y_{n+1}\in\widehat{C}_{\alpha}^{\mathrm{ELCP}}(X_{n+1})\mid X_{n+1}\in \mathcal{B} \right) - (1-\alpha)\right| \\
        \leq & (n+1)^{-1}+p_0^{-1}(1-p_0)\left\{ 1-(1-p_0)^n \right\}(\delta+2\varepsilon)\,.
    \end{align*}
\end{lemma}
\begin{proof}
    Define $\mathcal{A}_0=\big\{{\sup}_{1\leq i\leq n+1}\left|\Delta_i\right|\leq\varepsilon \big\}$ such that $\mathrm{pr}\left( \mathcal{A}_0 \right)\geq 1-\delta$, as assumed. 
    For $z=(z_1,\ldots,z_{n+1})$ with $z_i=(x_i,y_i)$ for $i\in[n+1]$, define $\mathcal{A}_z=\left\{\{Z_1,\ldots,Z_{n+1}\}=\{z_1,\ldots,z_{n+1}\}\right\}$, $n_{\mathcal{B}}(z)=\sum_{i=1}^{n+1}\mathbb{1}\left(x_i\in \mathcal{B}\right)$, $\mathcal{I}_z=\{i:\ x_i\in\mathcal{B},i\in[n+1]\}$ and $\mathcal{I}_z^c=\{i:\ x_i\notin\mathcal{B},i\in[n+1]\}$. The observed values of $\hat{\beta}_{\omega,\hat{r}}(X_1,S_1),\ldots, \hat{\beta}_{\omega,\hat{r}}(X_{n+1},S_{n+1})$ when $\{Z_1=z_1,\ldots,Z_{n+1}=z_{n+1}\}$ are $\beta_1,\ldots,\beta_{n+1}$. 
    Let $\mathcal{N}_k=\left\{ n_{\mathcal{B}}(Z)=k+1 \right\}$, then $\mathrm{pr}\left( \mathcal{N}_k\mid X_{n+1}\in \mathcal{B} \right)={n\choose k}p_0^k(1-p_0)^{n-k}$. Thus, for any $z\in\mathcal{N}_k$, we have $n_{\mathcal{B}}(z)=k+1$. It follows that
    \begin{align}
        &\mathrm{pr}\left( \hat{\beta}_{\omega,\hat{r}}(X_{n+1},S_{n+1})\leq Q\left( 1-\alpha;\dfrac{1}{n+1}\overset{n+1}{\underset{i=1}\sum}\delta_{\hat{\beta}_{\omega,\hat{r}}(X_i,S_i)} \right)\mid\mathcal{A}_z\cap\{X_{n+1}\in \mathcal{B}\} \right)\notag\\
        =&E\left\{ \mathbb{1}\left(\hat{\beta}_{\omega,\hat{r}}(X_{n+1},S_{n+1})\leq Q\left( 1-\alpha;\dfrac{1}{n+1}\overset{n+1}{\underset{i=1}\sum}\delta_{\hat{\beta}_{\omega,\hat{r}}(X_i,S_i)} \right)\right) \mid \mathcal{A}_z\cap\{X_{n+1}\in \mathcal{B}\}\right\}\notag\\
        =&\overset{}{\underset{i\in\mathcal{I}_z}{\sum}}\dfrac{1}{n_\mathcal{B}(z)}\mathbb{1}\left( \beta_i\leq Q\left( 1-\alpha;\dfrac{1}{n+1}\overset{n+1}{\underset{j=1}\sum}\delta_{\beta_j} \right) \right)\notag\\
        \notag =&\dfrac{\lceil (n+1)(1-\alpha)\rceil}{k+1}-\dfrac{1}{k+1}\overset{}{\underset{i\in\mathcal{I}_z^c}{\sum}}\mathbb{1}\left( \beta_i\leq Q\left( 1-\alpha;\dfrac{1}{n+1}\overset{n+1}{\underset{j=1}\sum}\delta_{\beta_j} \right) \right)\,.%\label{Proof of lemma:test_conditional_B form1}
    \end{align}
    As $Z_1,\ldots,Z_{n+1}$ are exchangeable, on the event $\mathcal{N}_k$,
    \begin{align*}
        &\mathrm{pr}\left( \hat{\beta}_{\omega,\hat{r}}(X_{n+1},S_{n+1})\leq Q\left( 1-\alpha;\dfrac{1}{n+1}\overset{n+1}{\underset{i=1}\sum}\delta_{\hat{\beta}_{\omega,\hat{r}}(X_i,S_i)} \right)\mid \mathcal{N}_k\cap\left\{ X_{n+1}\in \mathcal{B} \right\} \right)\\
        =&E\left\{ \mathrm{pr}\left( \hat{\beta}_{\omega,\hat{r}}(X_{n+1},S_{n+1})\leq Q\left( 1-\alpha;\dfrac{1}{n+1}\overset{n+1}{\underset{i=1}\sum}\delta_{\hat{\beta}_{\omega,\hat{r}}(X_i,S_i)} \right)\mid \mathcal{A}_{Z} \right)\mid \mathcal{N}_k\cap\left\{ X_{n+1}\in \mathcal{B} \right\} \right\}\\
        =&E\left\{ \overset{}{\underset{i\in\mathcal{I}_Z}{\sum}}\dfrac{1}{n_\mathcal{B}(Z)}\mathbb{1}\left( \beta_i\leq Q\left( 1-\alpha;\dfrac{1}{n+1}\overset{n+1}{\underset{j=1}\sum}\delta_{\beta_j} \right) \right)\mid \mathcal{N}_k\cap\left\{ X_{n+1}\in \mathcal{B} \right\} \right\}\\
        \notag =&\dfrac{\lceil (n+1)(1-\alpha)\rceil}{k+1} -\dfrac{1}{k+1}E\left\{ \overset{}{\underset{i\in\mathcal{I}_Z^c}{\sum}}\mathbb{1}\left( \beta_i\leq Q\left( 1-\alpha;\dfrac{1}{n+1}\overset{n+1}{\underset{j=1}\sum}\delta_{\beta_j} \right) \right) \mid \mathcal{N}_k\cap\left\{ X_{n+1}\in \mathcal{B} \right\}\right\}\\
        \notag =&\dfrac{\lceil (n+1)(1-\alpha)\rceil}{k+1} \\
        \notag & -\dfrac{n-k}{k+1}E\left\{ \mathbb{1}\left(\beta_i\leq Q\left( 1-\alpha;\dfrac{1}{n+1}\overset{n+1}{\underset{j=1}\sum}\delta_{\beta_j} \right) \right)\mid \{i\in\mathcal{I}_Z^c\}\cap\mathcal{N}_k\cap\left\{ X_{n+1}\in \mathcal{B}\right\} \right\}\\
        =&\dfrac{\lceil (n+1)(1-\alpha)\rceil}{k+1} \notag \\
        \notag & -\dfrac{n-k}{k+1}\mathrm{pr}\left(\beta_i\leq Q\left( 1-\alpha;\dfrac{1}{n+1}\overset{n+1}{\underset{j=1}\sum}\delta_{\beta_j} \right)\mid \{i\in\mathcal{I}_Z^c\}\cap\mathcal{N}_k\cap\left\{ X_{n+1}\in \mathcal{B}\right\} \right)\,.
    \end{align*} 
    Denote $U_i=F_{S\mid X}(S_i\mid X_i)$ for $i\in[n+1]$, the probability term in the last equation can derived as
    \begin{align*}
        &\mathrm{pr}\left(\beta_i\leq Q\left( 1-\alpha;\dfrac{1}{n+1}\overset{n+1}{\underset{j=1}\sum}\delta_{\beta_j} \right)\mid \{i\in\mathcal{I}_Z^c\}\cap\mathcal{N}_k\cap\left\{ X_{n+1}\in \mathcal{B}\right\} \right)\\
        \leq&\mathrm{pr}\left(\beta_i\leq Q\left( 1-\alpha;\dfrac{1}{n+1}\overset{n+1}{\underset{j=1}\sum}\delta_{\beta_j} \right),\mathcal{A}_0\mid \{i\in\mathcal{I}_Z^c\}\cap\mathcal{N}_k\cap\left\{ X_{n+1}\in \mathcal{B}\right\} \right)+\mathrm{pr}\left( \mathcal{A}_0^c \right)\\
        \leq&\mathrm{pr}\left(U_i\leq Q\left( 1-\alpha;\dfrac{1}{n+1}\overset{n+1}{\underset{j=1}\sum}\delta_{U_j} \right)+2\varepsilon\mid \{i\in\mathcal{I}_Z^c\}\cap\mathcal{N}_k\cap\left\{ X_{n+1}\in \mathcal{B}\right\} \right)+\delta\\
        =&\dfrac{\lceil (n+1)(1-\alpha)\rceil}{n+1}+2\varepsilon+\delta.
    \end{align*}
    The last equation follows a similar derivation of the proof of Lemma \ref{lemma:conditional_coverage_bound}. Thus
    \begin{align*}
        &\mathrm{pr}\left( \hat{\beta}_{\omega,\hat{r}}(X_{n+1},S_{n+1})\leq Q\left( 1-\alpha;\dfrac{1}{n+1}\overset{n+1}{\underset{i=1}\sum}\delta_{\hat{\beta}_{\omega,\hat{r}}(X_i,S_i)} \right)\mid X_{n+1}\in \mathcal{B} \right)\\
        =&\overset{n}{\underset{k=0}\sum}\mathrm{pr}\left( \hat{\beta}_{\omega,\hat{r}}(X_{n+1},S_{n+1})\leq Q\left( 1-\alpha;\dfrac{1}{n+1}\overset{n+1}{\underset{i=1}\sum}\delta_{\hat{\beta}_{\omega,\hat{r}}(X_i,S_i)} \right)\mid \mathcal{N}_k\cap\left\{ X_{n+1}\in \mathcal{B} \right\} \right) \\
        \notag & \times\mathrm{pr}\left( \mathcal{N}_k\mid X_{n+1}\in\mathcal{B} \right)\\
        \geq&\overset{n}{\underset{k=0}\sum}\mathrm{pr}\left( \mathcal{N}_k\mid X_{n+1}\in\mathcal{B} \right)\left[ \dfrac{\lceil (n+1)(1-\alpha)\rceil}{k+1}-\dfrac{n-k}{k+1}\left\{ \dfrac{\lceil (n+1)(1-\alpha)\rceil}{n+1}+2\varepsilon+\delta \right\} \right].
    \end{align*}
    By the properties of combination numbers,
    \begin{align*}
        \overset{n}{\underset{k=0}\sum}\dfrac{\mathrm{pr}\left( \mathcal{N}_k\mid X_{n+1}\in\mathcal{B} \right)}{k+1}& =\overset{n}{\underset{k=0}\sum}\dfrac{1}{p_0(n+1)}{n+1 \choose k+1}p_0^{k+1}(1-p_0)^{n-k}=\dfrac{1-(1-p_0)^{n+1}}{p_0(n+1)},\\
        \overset{n}{\underset{k=0}\sum}\dfrac{\mathrm{pr}\left( \mathcal{N}_k\mid X_{n+1}\in\mathcal{B} \right)(n-k)}{k+1}& =\overset{n-1}{\underset{k=0}\sum}\dfrac{1-p_0}{p_0}{n \choose k+1}p_0^{k+1}(1-p_0)^{n-k-1} \\
        & =\dfrac{1-p_0}{p_0}\left\{ 1-(1-p_0)^n \right\}.
    \end{align*}
    It follows that
    \begin{align*}
        &\mathrm{pr}\left( \hat{\beta}_{\omega,\hat{r}}(X_{n+1},S_{n+1})\leq Q\left( 1-\alpha;\dfrac{1}{n+1}\overset{n+1}{\underset{i=1}\sum}\delta_{\hat{\beta}_{\omega,\hat{r}}(X_i,S_i)} \right)\mid X_{n+1}\in \mathcal{B} \right)\\
        \geq&\dfrac{\lceil (n+1)(1-\alpha)\rceil}{n+1}-\dfrac{1-p_0}{p_0}\left\{ 1-(1-p_0)^n \right\}(\delta+2\varepsilon)\\
        \geq&1-\alpha-\dfrac{1-p_0}{p_0}\left\{ 1-(1-p_0)^n \right\}(\delta+2\varepsilon)\,.
    \end{align*}
    On the other hand,
    \begin{align*}
        &\mathrm{pr}\left(\beta_i\leq Q\left( 1-\alpha;\dfrac{1}{n+1}\overset{n+1}{\underset{j=1}\sum}\delta_{\beta_j} \right)\mid \{i\in\mathcal{I}_Z^c\}\cap\mathcal{N}_k\cap\left\{ X_{n+1}\in \mathcal{B}\right\} \right)\\
        \geq&\mathrm{pr}\left(\beta_i\leq Q\left( 1-\alpha;\dfrac{1}{n+1}\overset{n+1}{\underset{j=1}\sum}\delta_{\beta_j} \right),\mathcal{A}_0\mid \{i\in\mathcal{I}_Z^c\}\cap\mathcal{N}_k\cap\left\{ X_{n+1}\in \mathcal{B}\right\} \right)\\
        \geq&\mathrm{pr}\left(U_i\leq Q\left( 1-\alpha;\dfrac{1}{n+1}\overset{n+1}{\underset{j=1}\sum}\delta_{U_j} \right)-2\varepsilon\mid \{i\in\mathcal{I}_Z^c\}\cap\mathcal{N}_k\cap\left\{ X_{n+1}\in \mathcal{B}\right\} \right)-\mathrm{pr}\left( \mathcal{A}_0^c \right)\\
        \notag =&\dfrac{\lceil (n+1)(1-\alpha) \rceil-1}{n+1} - 2\varepsilon-\delta\,.
    \end{align*}
    Therefore
    \begin{align*}
        &\mathrm{pr}\left( \hat{\beta}_{\omega,\hat{r}}(X_{n+1},S_{n+1})\leq Q\left( 1-\alpha;\dfrac{1}{n+1}\overset{n+1}{\underset{i=1}\sum}\delta_{\hat{\beta}_{\omega,\hat{r}}(X_i,S_i)} \right)\mid X_{n+1}\in \mathcal{B} \right)\\
        \leq&\dfrac{\lceil (n+1)(1-\alpha) \rceil-1}{n+1}-\dfrac{1-p_0}{p_0}\left\{ 1-(1-p_0)^n \right\}(\delta+2\varepsilon)\,.
    \end{align*}
    In conclusion, 
    \begin{eqnarray}
        \notag & & \left|\mathrm{pr}\left( Y_{n+1}\in\widehat{C}_{\alpha}^{\mathrm{ELCP}}(X_{n+1})\mid X_{n+1}\in \mathcal{B} \right) - (1-\alpha)\right| \\
        \notag & \leq & \dfrac{1}{n+1}+\dfrac{1-p_0}{p_0}\left\{ 1-(1-p_0)^n \right\}(\delta+2\varepsilon)\,,
    \end{eqnarray}
    which completes the proof of this lemma.
\end{proof}

\subsection{Proof of Theorem \ref{theo:modsel elcp}}

\begin{proof}
    To highlight the impact of $\omega$ and $h$ in $\hat{\beta}_{\omega,\hat{r}}(X_i,S_i)$, we rewrite it as $\hat{\beta}_{\omega,h}(X_i,S_i)$.
    By the definition of $\widehat{C}_\alpha^{\mathrm{ELCP-PS}}(X_{n+1})$, the event $\left\{Y_{n+1}\in\widehat{C}_\alpha^{\mathrm{ELCP-PS}}(X_{n+1})\right\}$ is equivalent to the following:
    \begin{gather}
        \left\{\hat{\beta}_{\hat{\omega}_{n+1},\hat{h}_{n+1}}(X_{n+1},S_{n+1})\leq Q\left(1-\alpha;(n+1)^{-1}\sum_{i=1}^{n+1}\delta_{\hat{\beta}_{\hat{\omega}_{n+1},\hat{h}_{n+1}}(X_i,S_i)}\right)\right\}\,,\label{eq:modsel event}
    \end{gather}
    where $\hat{\omega}_{n+1}=\hat{\omega}(Y_{n+1})$ and $\hat{h}_{n+1}=\hat{h}(Y_{n+1})$.
    By the definitions of $\hat{\omega}_{n+1}$ and $\hat{h}_{n+1}$, they are invariant under any permutation of the $n+1$ data points $\mathcal{Z}_n\cup\{(X_{n+1},Y_{n+1})\}$. 
    Consequently, the probability of the event described in \eqref{eq:modsel event} is equal to the probability of the following event:
    \begin{gather*}
        \left\{\hat{\beta}_{\hat{\omega}_{n+1},\hat{h}_{n+1}}(X_j,S_j)\leq Q\left(1-\alpha;(n+1)^{-1}\sum_{i=1}^{n+1}\delta_{\hat{\beta}_{\hat{\omega}_{n+1},\hat{h}_{n+1}}(X_i,S_i)}\right)\right\}\,,\forall ~j\in[n+1]\,.
    \end{gather*}
    Thus,
    \begin{align*}
        &\mathrm{pr}\left(Y_{n+1}\in\widehat{C}_\alpha^{\mathrm{ELCP-PS}}(X_{n+1})\right)\\
        =&\mathrm{pr}\left(\hat{\beta}_{\hat{\omega}_{n+1},\hat{h}_{n+1}}(X_{n+1},S_{n+1})\leq Q\left(1-\alpha;(n+1)^{-1}\sum_{i=1}^{n+1}\delta_{\hat{\beta}_{\hat{\omega}_{n+1},\hat{h}_{n+1}}(X_i,S_i)}\right)\right)\\
        =&\dfrac{1}{n+1}\sum_{j=1}^{n+1}\mathrm{pr}\left(\hat{\beta}_{\hat{\omega}_{n+1},\hat{h}_{n+1}}(X_j,S_j)\leq Q\left(1-\alpha;(n+1)^{-1}\sum_{i=1}^{n+1}\delta_{\hat{\beta}_{\hat{\omega}_{n+1},\hat{h}_{n+1}}(X_i,S_i)}\right)\right)\\
        =&E\left\{\dfrac{1}{n+1}\sum_{j=1}^{n+1}\mathbb{1}\left(\hat{\beta}_{\hat{\omega}_{n+1},\hat{h}_{n+1}}(X_j,S_j)\leq Q\left(1-\alpha;(n+1)^{-1}\sum_{i=1}^{n+1}\delta_{\hat{\beta}_{\hat{\omega}_{n+1},\hat{h}_{n+1}}(X_i,S_i)}\right)\right)\right\}\,,
    \end{align*}
    and the last formula is in $\left[1-\alpha,1-\alpha+\dfrac{1}{n+1}\right)$ by the definition of the quantile.
\end{proof}

\subsection{Proof of Theorem \ref{theo:parameter_selection_consistent}}\label{proof of theo:parameter_selection_consistent}

The proof of Theorem \ref{theo:parameter_selection_consistent} builds on the following lemma, which establishes a concentration inequality for the loss function $\mathcal{L}_2$.

\begin{lemma}\label{lemma: para_select_conv}
    Suppose Assumptions \ref{assump:permutable}--\ref{ass3} hold.
    Assume that the kernel functions $K_1(\cdot,\cdot)$ and $K_2(\cdot,\cdot)$ are uniformly bounded by $D_{K,0}>0$ and that their partial derivatives are bounded in absolute value by $D_{K,1}>0$.
    Then for $h$ satisfies $(n+\omega m)h^d\log^{-1}(n)\rightarrow\infty$ and $\tau>0$, it holds that
    \begin{align*}
        & \mathrm{pr}\left( \left|\mathcal{L}_{2}\left(\omega,h;\mathcal{Z}_n\cup\{Z_{n+1}\},\mathcal{Z}_m^\prime\right)-\mathcal{R}_{\omega,h}^{(n,m)}\right|>\tau \right) \\
        < & n\exp\left\{ -\overline{C}_0(n+\omega m)h^d \right\}+\exp\left( -\overline{C}_1\tau^2n \right)\,,
    \end{align*}
    for some positive constants $\overline{C}_0$ and $\overline{C}_1$, where
    \begin{gather*}
        \mathcal{R}_{\omega,h}^{(n,m)}=E\left[ \left\{ K_1(\hat{\beta}_{\omega,h}(X_1,S_1),\hat{\beta}_{\omega,h}(X_2,S_2))-2\phi(\hat{\beta}_{\omega,h}(X_1,S_1)) \right\}K_2(X_1,X_2) \right]\,.
    \end{gather*}    
\end{lemma}

\begin{proof}
    Let $\widetilde{\mathcal{Z}}_{n+1}=\{\widetilde{Z}_i\}_{i\in[n+1]}$ be an i.i.d.~copy of $\mathcal{Z}_n\cup\{Z_{n+1}\}$, and the corresponding scores are $\{\widetilde{S}_i\}_{i\in[n+1]}$. 
    Note that $\hat{\beta}_{\omega,h}(x,s)$ and $\hat{r}(x,s)$ are computed using $\mathcal{Z}_n\cup\{Z_{n+1}\}$ and $\mathcal{Z}_m^\prime$. 
    For any subset of indices $\{i_1,\ldots,i_k\}\subset[n+1]$, define $\hat{\beta}_{\omega,h}^{\setminus i_1,\ldots,i_k}(x,s)$ and $\hat{r}^{\setminus i_1,\ldots,i_k}(x,s)$ as those obtained using $\mathcal{Z}_n\cup\{Z_{n+1}\}\setminus\{Z_{i_\ell}\}_{\ell=1}^k\cup\{\widetilde{Z}_{i_\ell}\}_{\ell=1}^k$ and $\mathcal{Z}_m^\prime$.
    
    By the exchangeability of the pairs $(X_1,S_1),\ldots,(X_n,S_n)$ and $(X_{n+1},S_{n+1})$, the transformed sequence $\hat{\beta}_{\omega,h}(X_1,S_1),\ldots,\hat{\beta}_{\omega,h}(X_n,S_n), \hat{\beta}_{\omega,h}(X_{n+1},S_{n+1})$ inherits this exchangeability property. Therefore, $E\left\{ \mathcal{L}_{2}\left(\omega,h;\mathcal{Z}_n\cup\{Z_{n+1}\},\mathcal{Z}_m^\prime\right) \right\}=\mathcal{R}_{\omega,h}^{(n,m)}$.
    
    \noindent\textbf{Step 1: } We first quantify the difference between $\hat{\beta}_{\omega,h}^{\setminus j}(X_i,S_i)$ and $\hat{\beta}_{\omega,h}(X_i,S_i)$ for $j\neq i$.

    For notation simplicity, we define $X_\ell^{\setminus i_1,\ldots,i_k}=X_\ell$ if $\ell\neq i_p$ for all $p\in[k]$, otherwise $X_\ell^{\setminus i_1,\ldots,i_k}=\widetilde{X}_{i_p}$ if $\ell=i_p$ for some $p\in[k]$. Similarly, we define $S_\ell^{\setminus i_1,\ldots,i_k}=S_\ell$ if $\ell\neq i_p$ for all $p\in[k]$, otherwise $S_\ell^{\setminus i_1,\ldots,i_k}=\widetilde{S}_{i_p}$ if $\ell=i_p$ for some $p\in[k]$. Then, 
    \begin{align}
        &|\hat{\beta}_{\omega,h}^{\setminus j}(X_i,S_i)-\hat{\beta}_{\omega,h}(X_i,S_i)|\notag\\
        =&\Big|\dfrac{\sum_{\ell=1}^{n+1}K(X_i,X_\ell^{\setminus j};h)\mathbb{1}(S_\ell^{\setminus j}\leq S_i)+\omega\sum_{\ell=1}^{m}K(X_i,X_\ell^\prime;h)\hat{r}^{\setminus j}(X_\ell^\prime,S_\ell^\prime)\mathbb{1}(S_\ell^\prime\leq S_i)}{\sum_{\ell=1}^{n+1}K(X_i,X_\ell^{\setminus j};h)+\omega\sum_{\ell=1}^{m}K(X_i,X_\ell^\prime;h)\hat{r}^{\setminus j}(X_\ell^\prime,S_\ell^\prime)}\notag\\
        &-\dfrac{\sum_{\ell=1}^{n+1}K(X_i,X_\ell;h)\mathbb{1}(S_\ell\leq S_i)+\omega\sum_{\ell=1}^{m}K(X_i,X_\ell^\prime;h)\hat{r}(X_\ell^\prime,S_l^\prime)\mathbb{1}(S_\ell^\prime\leq S_i)}{\sum_{\ell=1}^{n+1}K(X_i,X_\ell;h)+\omega\sum_{\ell=1}^{m}K(X_i,X_\ell^\prime;h)\hat{r}(X_\ell^\prime,S_\ell^\prime)}\Big|\notag\\
        \notag \leq & \frac{|\Pi_{1,1}\Pi_{4,1}-\Pi_{3,1}\Pi_{2,1}| + |\Pi_{1,1}\Pi_{4,2}-\Pi_{3,1}\Pi_{2,2}| + |\Pi_{1,2}\Pi_{4,1}-\Pi_{3,2}\Pi_{2,1}| + |\Pi_{1,2}\Pi_{4,2}-\Pi_{3,2}\Pi_{2,2}|}{(\Pi_{2,1}+\Pi_{2,2})(\Pi_{4,1}+\Pi_{4,2})}\,,
    \end{align}
    where $\Pi_{1,1}=\sum_{\ell=1}^{n+1}K(X_i,X_\ell^{\setminus j};h)\mathbb{1}(S_\ell^{\setminus j}\leq S_i)$, $\Pi_{1,2}=\omega\sum_{\ell=1}^{m}K(X_i,X_\ell^\prime;h)\hat{r}^{\setminus j}(X_\ell^\prime,S_\ell^\prime)\mathbb{1}(S_\ell^\prime\leq S_i)$,
    $\Pi_{2,1}=\sum_{\ell=1}^{n+1}K(X_i,X_\ell^{\setminus j};h)$, $\Pi_{2,2}=\omega\sum_{\ell=1}^{m}K(X_i,X_\ell^\prime;h)\hat{r}^{\setminus j}(X_\ell^\prime,S_\ell^\prime)$, \\
    $\Pi_{3,1}=\sum_{\ell=1}^{n+1}K(X_i,X_\ell;h)\mathbb{1}(S_\ell\leq S_i)$, $\Pi_{3,2}=\omega\sum_{\ell=1}^{m}K(X_i,X_\ell^\prime;h)\hat{r}(X_\ell^\prime,S_\ell^\prime)\mathbb{1}(S_\ell^\prime\leq S_i)$, $\Pi_{4,1}=\sum_{\ell=1}^{n+1}K(X_i,X_\ell;h)$ and $\Pi_{4,2}=\omega\sum_{\ell=1}^{m}K(X_i,X_\ell^\prime;h)\hat{r}(X_\ell^\prime,S_\ell^\prime)$.

    By the fact that $\sup_{x,s}|\hat{r}^{\setminus j}(x,s)-\hat{r}(x,s)|\leq  C_r  n^{-1}$, simple algebra yields that
    \begin{align*}
        |\Pi_{1,1}\Pi_{4,1}-\Pi_{3,1}\Pi_{2,1}|&\leq 2\left\{ \sum_{\ell=1}^{n+1}K(X_i,X_\ell;h) \right\}\left|K(X_i,X_j;h)-K(X_i,\widetilde{X}_j;h)\right|\,,\\
        |\Pi_{1,1}\Pi_{4,2}-\Pi_{3,1}\Pi_{2,2}| &\leq \left\{ \omega\sum_{\ell=1}^{m}K(X_i,X_\ell^\prime;h)\hat{r}(X_\ell^\prime,S_\ell^\prime) \right\}\left|K(X_i,X_j;h)-K(X_i,\widetilde{X}_j;h)\right|\\
        &~~~~~~~~~~~~~~~~~~+\dfrac{ C_r }{n}\left\{ \sum_{\ell=1}^{n+1}K(X_i,X_\ell;h) \right\}\left\{ \omega\sum_{\ell=1}^{m}K(X_i,X_\ell^\prime;h) \right\}\,,\\
        |\Pi_{1,2}\Pi_{4,1}-\Pi_{3,2}\Pi_{2,1}| &\leq\dfrac{ C_r }{n}\left\{ \sum_{\ell=1}^{n+1}K(X_i,X_\ell;h) \right\}\left\{ \omega\sum_{\ell=1}^{m}K(X_i,X_\ell^\prime;h) \right\}\\
        &~+\left\{ \omega\sum_{\ell=1}^{m}K(X_i,X_\ell^\prime;h)\hat{r}(X_\ell^\prime,S_\ell^\prime) \right\}\left|K(X_i,X_j;h)-K(X_i,\widetilde{X}_j;h)\right|\,,\\
        |\Pi_{1,2}\Pi_{4,2}-\Pi_{3,2}\Pi_{2,2}| &\leq\dfrac{2 C_r }{n}\left\{ \omega\sum_{\ell=1}^{m}K(X_i,X_\ell^\prime;h) \right\}\left\{ \omega\sum_{\ell=1}^{m}K(X_i,X_\ell^\prime;h)\hat{r}(X_\ell^\prime,S_\ell^\prime) \right\}\,.
    \end{align*}
    %Similar to the analysis of $A_i^{(2)}(\tau_2)$ in the proof of \eqref{Proof of lemma:delta_concentration2} when $\hat{r}$ is dependent on $\mathcal{Z}_m^\prime$, 
    Our preceding derivation established that the upper bound of $|\hat{\beta}_{\omega,h}^{\setminus j}(X_i,S_i)-\hat{\beta}_{\omega,h}(X_i,S_i)|$ depends on $\Pi_{2,1}+\Pi_{2,2}$ and $\Pi_{4,1}+\Pi_{4,2}$ in the denominator, motivating the need to control their lower bounds. For $i\in[n+1]$, define the following event: 
    \begin{align*}
        A_i=&\left\{\Pi_{4,1}+\Pi_{4,2}>\underline{L}_3(n+\omega m\underline{L}_2)h^{d}/2\right\}\cap\\
        &\left[ \underset{j\neq i}{\cap}\left\{ \sum_{\ell=1}^{n+1}K(X_i,X_\ell^{\setminus j};h)+\omega\sum_{\ell=1}^{m}K(X_i,X_\ell^\prime;h)\hat{r}^{\setminus j}(X_\ell^\prime,S_\ell^\prime)>\underline{L}_3(n+\omega m\underline{L}_2)h^{d}/2 \right\} \right]\,.
    \end{align*}
    Since $\sum_{\ell\neq i}K(X_i,X_\ell;h)=\sum_{\ell=1}^{n+1}K(X_i,X_\ell;h)-K_0(0)\leq\sum_{\ell=1}^{n+1}K(X_i,X_\ell^j;h)$ and the density ratio is lower bounded by $\hat{r}^{\setminus j}(X_\ell^\prime,S_\ell^\prime)\geq\underline{L}_2$, we obtain
    \begin{gather*}
        A_i\supset\left\{ \sum_{\ell\neq i}K(X_i,X_\ell;h)+\underline{L}_2\omega \sum_{\ell=1}^{m}K(X_i,X_\ell^\prime;h)>\underline{L}_3(n+\omega m\underline{L}_2)h^{d}/2 \right\}\,.
    \end{gather*}
    According to Lemma \ref{lemma:KernelFun}, $E\left\{ K(X_i,X_\ell;h) \right\}>\underline{L}_3h^d$ and $E\left\{ K(X_i,X_\ell^\prime;h) \right\}>\underline{L}_3h^d$. Similar to the proof of Lemma \ref{lemma:delta_concentration2}, the probability of $A_i$ satisfies
    \begin{gather*}
        \mathrm{pr}\left( A_i \right)\geq 1-\exp\left\{ -\overline{C}_0(n+\omega m)h^d \right\}\,,
    \end{gather*}
    where $\overline{C}_0$ is a positive constant independent of $i$. Denote $A_0=\cup_{i=1}^{n+1}A_i$, then 
    \begin{gather*}
        \mathrm{pr}\left( A_0 \right)\geq 1-(n+1)\exp\left( -\overline{C}_0(n+\omega m)h^d \right)\,.
    \end{gather*}
    
    Summarizing the preceding analysis, we establish the following bound under event $A_0$:
    \begin{gather*}
        |\hat{\beta}_{\omega,h}^{\setminus j}(X_i,S_i)-\hat{\beta}_{\omega,h}(X_i,S_i)|\leq\dfrac{4\underline{L}_2\left|K(X_i,X_j;h)-K(X_i,\widetilde{X}_j;h)\right|}{\underline{L}_3(n+\omega m)h^d}+\dfrac{4\underline{L}_2}{n}\overset{\rm def.}{\equiv}\delta_{\beta,i,j}\,.
    \end{gather*}

    \noindent\textbf{Step 2: } We now analyze the concentration of $\mathcal{L}_{2}\left(\omega,h;\mathcal{Z}_n\cup\{Z_{n+1}\},\mathcal{Z}_m^\prime\right)$ around $\mathcal{R}_{\omega,h}^{(n,m)}$.
    
    Define filters $\mathcal{F}_i=\sigma(Z_1,\ldots,Z_i)$ for $i\in[n+1]$ and $\mathcal{F}_0$ as the trivial $\sigma$-algebra. Let
    \begin{gather*}
        L_{\omega,h,i,j}=\dfrac{1}{n(n+1)}\left\{ K_1(\hat{\beta}_{\omega,h}(X_i,S_i),\hat{\beta}_{\omega,h}(X_j,S_j))-2\phi(\hat{\beta}_{\omega,h}(X_i,S_i)) \right\}K_2(X_i,X_j).
    \end{gather*}
    Then, $\mathcal{L}_{2}\left(\omega,h;\mathcal{Z}_n\cup\{Z_{n+1}\},\mathcal{Z}_m^\prime\right)=\sum_{1\leq i\neq j\leq n+1}L_{\omega,h,i,j}$. Denote
    \begin{gather*}
        \xi^{\omega,h,\ell}=E\left\{ \mathcal{L}_{2}\left(\omega,h;\mathcal{Z}_n\cup\{Z_{n+1}\},\mathcal{Z}_m^\prime\right)\mid\mathcal{F}_\ell \right\}-E\left\{ \mathcal{L}_{2}\left(\omega,h;\mathcal{Z}_n\cup\{Z_{n+1}\},\mathcal{Z}_m^\prime\right)\mid\mathcal{F}_{\ell-1} \right\}\,,\\
        \xi_{i,j}^{\omega,h,\ell}=E\left\{ L_{\omega,h,i,j}\mid\mathcal{F}_\ell \right\}-E\left\{ L_{\omega,h,i,j}\mid\mathcal{F}_{\ell-1} \right\}\,.
    \end{gather*}
    Therefore,
    \begin{gather*}
        \mathcal{L}_{2}\left(\omega,h;\mathcal{Z}_n\cup\{Z_{n+1}\},\mathcal{Z}_m^\prime\right)-\mathcal{R}_{\omega,h}^{(n,m)}=\sum_{\ell=1}^{n+1}\xi^{\omega,h,\ell}=\sum_{\ell=1}^{n+1}\sum_{1\leq i\neq j\leq n+1}\xi_{i,j}^{\omega,h,\ell}
    \end{gather*}
    
    In the analysis of Step 1, when $i\neq \ell$, the upper bound of $|\hat{\beta}_{\omega,h}^{\setminus \ell}(X_i,S_i)-\hat{\beta}_{\omega,h}(X_i,S_i)|$ is controlled by $\delta_{\beta,i,\ell}$ under event $A_0$. Since $\xi_{i,j}^{\omega,h,\ell}$ involves both $|\hat{\beta}_{\omega,h}^{\setminus \ell}(X_i,S_i)-\hat{\beta}_{\omega,h}(X_i,S_i)|$ and $|\hat{\beta}_{\omega,h}^{\setminus \ell}(X_j,S_j)-\hat{\beta}_{\omega,h}(X_j,S_j)|$, we consider Case 1 where at least one of $i$ and $j$ is equal to $\ell$ or $\ell-1$, and Case 2 where both $i$ and $j$ are not equal to $\ell$ or $\ell-1$.
    
    \noindent\textbf{Case 1:} at least one of $i$ and $j$ is equal to $\ell$ or $\ell-1$: 
    
    \noindent Since each term in $\xi_{i,j}^{\omega,h,\ell}$ is bounded, taking their suprema yields
    \begin{gather*}
        \left|\xi_{i,j}^{\omega,h,\ell}\right|\leq\dfrac{D_{K,0}^2}{n(n+1)}\,.
    \end{gather*}

    \noindent\textbf{Case 2:} both $i$ and $j$ are not equal to $\ell$ or $\ell-1$: 

    \noindent Based on the definition, $L_{\omega,h,i,j}$ is independent of $\widetilde{Z}_\ell$. Replace the $Z_\ell$ in $L_{\omega,h,i,j}$ with $\widetilde{Z}_\ell$, and denote it as $L_{\omega,h,i,j}^{\setminus \ell}$. Therefore,
    \begin{gather}
        E\left\{ L_{\omega,h,i,j}\mid\mathcal{F}_{\ell-1} \right\}=E\left\{ L_{\omega,h,i,j}^{\setminus \ell}\mid\mathcal{F}_{\ell-1} \right\}=E\left[ E\left\{ L_{\omega,h,i,j}^{\setminus \ell}\mid\sigma(\mathcal{F}_{\ell-1},\widetilde{Z}_\ell) \right\}\mid \mathcal{F}_{\ell-1} \right]\,,\label{eq:towerLwhij}
    \end{gather}
    where the outer expectation of the last formula is taken with respect to $\widetilde{Z}_\ell$. When both $i$ and $j$ are not equal to $\ell$ or $\ell-1$, on the event $A_0$ we can derive that:
    \begin{align*}
        &n(n+1)\left|L_{\omega,h,i,j}-L_{\omega,h,i,j}^{\setminus \ell}\right|\\
        \leq&K_2(X_i,X_j)\left|K_1(\hat{\beta}_{\omega,h}(X_i,S_i),\hat{\beta}_{\omega,h}(X_j,S_j))-K_1(\hat{\beta}_{\omega,h}^{\setminus \ell}(X_i,S_i),\hat{\beta}_{\omega,h}^{\setminus \ell}(X_j,S_j))\right|\\
        &+2K_2(X_i,X_j)\left|\phi(\hat{\beta}_{\omega,h}(X_i,S_i))-\phi(\hat{\beta}_{\omega,h}^{\setminus \ell}(X_i,S_i))\right|\\
        \leq&D_{K,0}D_{K,1}\left\{ 3\left|\hat{\beta}_{\omega,h}(X_i,S_i)-\hat{\beta}_{\omega,h}^{\setminus \ell}(X_i,S_i)\right|+\left|\hat{\beta}_{\omega,h}(X_j,S_j)-\hat{\beta}_{\omega,h}^{\setminus \ell}(X_j,S_j)\right| \right\}\\
        =&D_{K,0}D_{K,1}\left( 3\delta_{\beta,i,\ell}+\delta_{\beta,j,\ell} \right)\,.
    \end{align*}
    Moreover, $L_{\omega,h,i,j}$ and $L_{\omega,h,i,j}^{\setminus \ell}$ are always bounded by $D_{K,0}^2/\{n(n+1)\}$.
    Based on \eqref{eq:towerLwhij}, we obtain
    \begin{align*}
        &n(n+1)\left|\xi_{i,j}^{\omega,h,\ell}\right|\\
        =&\left|E\left[ E\left\{ L_{\omega,h,i,j}-L_{\omega,h,i,j}^{\setminus \ell}\mid\sigma(\mathcal{F}_\ell,\widetilde{Z}_\ell) \right\}\mid\mathcal{F}_{\ell-1} \right]\right|\\
        \leq&E\left[ E\left\{ \left|L_{\omega,h,i,j}-L_{\omega,h,i,j}^{\setminus \ell}\right|\mid\sigma(\mathcal{F}_\ell,\widetilde{Z}_\ell) \right\}\mid\mathcal{F}_{\ell-1} \right]\\
        =&E\left[ E\left\{ \left|L_{\omega,h,i,j}-L_{\omega,h,i,j}^{\setminus \ell}\right|\mathbb{1}\left(A_0\right)\mid\sigma(\mathcal{F}_\ell,\widetilde{Z}_\ell) \right\}\mid\mathcal{F}_{\ell-1} \right]\\
        &~~~~~~+E\left[ E\left\{ \left|L_{\omega,h,i,j}-L_{\omega,h,i,j}^{\setminus \ell}\right|\mathbb{1}\left(A_0^c\right)\mid\sigma(\mathcal{F}_\ell,\widetilde{Z}_\ell) \right\}\mid\mathcal{F}_{\ell-1} \right]\\
        \leq&D_{K,0}D_{K,1}E\left[ E\left\{ 3\delta_{\beta,i,\ell}+\delta_{\beta,j,\ell}\mid\sigma(\mathcal{F}_\ell,\widetilde{Z}_\ell) \right\}\mid\mathcal{F}_{\ell-1} \right]+2D_{K,0}^2E\left\{\mathbb{1}\left(A_0^c\right)\mid\mathcal{F}_{\ell-1}\right\}\,.
    \end{align*}
    We analyze $E\left[ E\left\{ \delta_{\beta,i,\ell}\mid\sigma(\mathcal{F}_\ell,\widetilde{Z}_\ell) \right\}\mid\mathcal{F}_{\ell-1} \right]$ as follows:
    \begin{align}
        &E\left[ E\left\{ \delta_{\beta,i,\ell}\mid\sigma(\mathcal{F}_\ell,\widetilde{Z}_\ell) \right\}\mid\mathcal{F}_{\ell-1} \right]-\dfrac{4\underline{L}_2}{n}\notag\\
        \leq&\dfrac{4\underline{L}_2}{\underline{L}_3(n+\omega m)h^d}E\left[ E\left\{ \left|K(X_i,X_\ell;h)-K(X_i,\widetilde{X}_\ell;h)\right|\mid\sigma(\mathcal{F}_\ell,\widetilde{Z}_\ell) \right\}\mid\mathcal{F}_{\ell-1} \right]\notag\\
        \leq&\dfrac{4\underline{L}_2}{\underline{L}_3(n+\omega m)h^d}E\left[ E\left\{ K(X_i,X_\ell;h)+K(X_i,\widetilde{X}_\ell;h)\mid\sigma(\mathcal{F}_{\ell-1},X_i) \right\}\mid\mathcal{F}_{\ell-1} \right].\label{eq:conce loss eq1}
    \end{align}
    From Lemma \ref{lemma:KernelFun}, we can get that
    \begin{align*}
        n(n+1)\left|\xi_{i,j}^{\omega,h,\ell}\right|&\leq D_{K,0}D_{K,1}\left\{ \dfrac{16\underline{L}_2}{n}+\dfrac{32\underline{L}_2\overline{L}_3}{\underline{L}_3(n+\omega m)} \right\}+2D_{K,0}^2E\left\{\mathbb{1}\left(A_0^c\right)\mid\mathcal{F}_{\ell-1}\right\}\,.
    \end{align*}
    The second term in the preceding equation constitutes a non-negative martingale with expectation $E\left\{\mathbb{1}\left(A_0^c\right)\right\}=1-E\left\{\mathbb{1}\left(A_0\right)\right\}\leq(n+1)\exp\{-\overline{C}_0(n+\omega m)h^d\}$. Consequently, applying Doob's maximal inequality yields:
    \begin{align*}
        &\mathrm{pr}\left(\sup_{1\leq \ell\leq n+1}E\left\{\mathbb{1}\left(A_0^c\right)\mid\mathcal{F}_{\ell-1}\right\}\geq(n+1)^{1/2}\exp\{-\overline{C}_0(n+\omega m)h^d/2\}\right)\\
        \leq&(n+1)^{1/2}\exp\{-\overline{C}_0(n+\omega m)h^d/2\}\,.
    \end{align*}
    Consider the following transformation:
    \begin{align*}
        (n+1)^{3/2}\exp\{-\overline{C}_0(n+\omega m)h^d/2\}\leq&\exp\{-\overline{C}_0(n+\omega m)h^d/2+3\log(n)\}\\
        =&\exp[-\log(n)\{-\overline{C}_0(n+\omega m)h^d\log^{-1}(n)/2+3\}]\,.
    \end{align*}
    As $h$ satisfies $(n+\omega m)h^d\log^{-1}(n)\to\infty$, the above expression admits a uniform upper bound, which we denote by $C_h$. Therefore $(n+1)^{1/2}\exp\{-\overline{C}_0(n+\omega m)h^d/2\}\leq C_h n^{-1}$. Denote the event
    \begin{gather*}
        \widetilde{A}=\left\{\sup_{1\leq \ell\leq n+1}E\left\{\mathbb{1}\left(A_0^c\right)\mid\mathcal{F}_{\ell-1}\right\}<(n+1)^{1/2}\exp\{-\overline{C}_0(n+\omega m)h^d/2\}\right\}\,.
    \end{gather*}
    Previous analysis gives $\mathrm{pr}(\widetilde{A})> 1-(n+1)^{1/2}\exp\{-\overline{C}_0(n+\omega m)h^d/2\}$. Therefore, on the event $\widetilde{A}$,
    \begin{gather*}
        n(n+1)\left|\xi_{i,j}^{\omega,h,\ell}\right|\leq D_{K,0}D_{K,1}\left\{ \dfrac{16\underline{L}_2}{n}+\dfrac{32\underline{L}_2\overline{L}_3}{\underline{L}_3(n+\omega m)} \right\}+\dfrac{2D_{K,0}^2C_h}{n}\,.
    \end{gather*}
    Based on \textbf{Case 1} and \textbf{Case 2}, on the event $\widetilde{A}$,
    \begin{align*}
        \left|\xi^{\omega,h,\ell}\right|&\leq\dfrac{4D_{K,0}^2}{n}+D_{K,0}D_{K,1}\left\{ \dfrac{16\underline{L}_2}{n}+\dfrac{32\underline{L}_2\overline{L}_3}{\underline{L}_3(n+\omega m)} \right\}+\dfrac{2D_{K,0}^2C_h}{n}\\
        &\leq\left( 4D_{K,0}^2+16D_{K,0}D_{K,1}\underline{L}_2+\dfrac{32D_{K,0}D_{K,1}\underline{L}_2\underline{L}_3}{\underline{L}_3}+2D_{K,0}^2C_h \right)\dfrac{1}{n}\,.
    \end{align*}
    Define the constant
    \begin{gather*}
        D_\xi\ = 4D_{K,0}^2+16D_{K,0}D_{K,1}\underline{L}_2+\dfrac{32D_{K,0}D_{K,1}\underline{L}_2\underline{L}_3}{\underline{L}_3}+2D_{K,0}^2C_h\,,
    \end{gather*}
    and we have $\sum_{\ell=1}^{n+1}E\left( \left|\xi^{\omega,h,\ell}\right|^2\mid\mathcal{F}_{\ell-1} \right)\leq D_\xi^2n^{-1}$ on the event $A_0$ and $\widetilde{A}$. 
    As $\xi^{\omega,h,\ell}$ is martingale with respect to filter $\mathcal{F}_\ell$, Freedman's inequality gives
    \begin{align*}
        &\mathrm{pr}\left( \left|\sum_{\ell=1}^{n+1}\xi^{\omega,h,\ell}\right|>\tau,\widetilde{A} \right)\\
        \leq&\mathrm{pr}\left( \left|\sum_{\ell=1}^{n+1}\xi^{\omega,h,\ell}\right|>\tau,\sum_{\ell=1}^{n+1}E\left( \left|\xi^{\omega,h,\ell}\right|^2\mid\mathcal{F}_{\ell-1} \right)\leq D_\xi^2n^{-1}, \left|\xi^{\omega,h,\ell}\right|\leq D_\xi n^{-1},~\forall ~1\leq \ell\leq n+1 \right)\\
        \leq&\exp\left( -\dfrac{\tau^2n}{D_\xi^2+2D_\xi\tau/3} \right)\,.
    \end{align*}
    Therefore,
    \begin{align*}
        \mathrm{pr}\left( \left|\sum_{\ell=1}^{n+1}\xi^{\omega,h,\ell}\right|>\tau \right)&\leq\mathrm{pr}\left( \left|\sum_{\ell=1}^{n+1}\xi^{\omega,h,\ell}\right|>\tau,\widetilde{A} \right)+\left\{ 1-\mathrm{pr}\left( \widetilde{A} \right) \right\}\\
        &\leq\exp\left( -\dfrac{\tau^2n}{D_\xi^2+2D_\xi\tau/3} \right)+(n+1)^{1/2}\exp\{-\overline{C}_0(n+\omega m)h^d/2\}\,.
    \end{align*}
    For $\tau<D_{K,0}^2$, let $\overline{C}_1=(D_\xi^2+2D_\xi D_{K,0}^2/3)^{-1}$. For all $n>1$, $n>(n+1)^{1/2}$ holds, and replacing the constant $\overline{C}_0/2$ with $\overline{C}_0$ in the preceding inequality yields the following probabilistic bound:
    \begin{align*}
        & \mathrm{pr}\left( \left|\mathcal{L}_{2}\left(\omega,h;\mathcal{Z}_n\cup\{Z_{n+1}\},\mathcal{Z}_m^\prime\right)-\mathcal{R}_{\omega,h}^{(n,m)}\right|>\tau \right) \\
        = & \mathrm{pr}\left( \left|\sum_{\ell=1}^{n+1}\xi^{\omega,h,\ell}\right|>\tau \right)
        \leq n\exp\left\{ -\overline{C}_0(n+\omega m)h^d \right\}+\exp\left( -\overline{C}_1\tau^2n \right)\,.
    \end{align*}
    We finish the proof of this lemma.
\end{proof}    

\begin{proof}[Proof of Theorem \ref{theo:parameter_selection_consistent}]
    Define the counterpart of $\widetilde{A}$ with bandwidth $h$ as $\widetilde{A}_h$. 
    By Lemma~\ref{lemma: para_select_conv}, for any $(\omega,h)\in\mathcal{G}$, 
    \begin{align*}
        \mathrm{pr}\left( \left|\mathcal{L}_{2}\left(\omega,h;\mathcal{Z}_n\cup\{Z_{n+1}\},\mathcal{Z}_m^\prime\right)-\mathcal{R}_{\omega,h}^{(n,m)}\right|>\tau,\widetilde{A}_h \right)\leq\exp\left( -\overline{C}_1\tau^2n \right)\,.
    \end{align*}
    For $\zeta>0$, define the event 
    \begin{gather*}
        A_{\omega,h}(\zeta)=\left\{ \left|\mathcal{L}_{2}\left(\omega,h;\mathcal{Z}_n\cup\{Z_{n+1}\},\mathcal{Z}_m^\prime\right)-\mathcal{R}_{\omega,h}^{(n,m)}\right|\leq\zeta n^{-1/2}\mathrm{log}^{1/2} n \right\}\,.
    \end{gather*}
    Its complement is denoted as $A_{\omega,h}^c(\zeta)$, and we have  
    \begin{align*}
        \mathrm{pr}\left( A_{\omega,h}^c(\zeta)\cap\widetilde{A}_h \right)\leq n^{-\overline{C}_1\zeta^2}\,,
    \end{align*}
    
    It is assumed that
    \begin{gather*}
        \inf_{(\omega,h)\in\mathcal{G}\setminus\{(\omega^*,h^*)\}}\mathcal{R}_{\omega,h}^{(n,m)}-\mathcal{R}_{\omega^*,h^*}^{(n,m)}>\zeta n^{-1/2}\mathrm{log}^{1/2}n\,.
    \end{gather*}
    On the event $\bigcap_{(\omega,h)\in \mathcal{G}} A_{\omega,h}(\zeta)$, it then follows that  
    \begin{gather*}
        \mathcal{L}_{2}(\omega^*,h^*,\mathcal{Z}_n\cup\{Z_{n+1}\},\mathcal{Z}_m^\prime)<\inf_{(\omega,h)\in\mathcal{G}\setminus\{(\omega^*,h^*)\}}\mathcal{L}_{2}(\omega,h,\mathcal{Z}_n\cup\{Z_{n+1}\},\mathcal{Z}_m^\prime)\,,
    \end{gather*}
    which implies $(\hat{\omega},\hat{h})=(\omega^*,h^*)$. The probability of this event can be bounded as follows:
    \begin{align*}
        \mathrm{pr}\left( (\hat{\omega},\hat{h})=(\omega^*,h^*) \right)&\geq\mathrm{pr}\left( \underset{(\omega,h)\in\mathcal{G}}{\cap}A_{\omega,h}(\zeta) \right)\\
        &\geq 1-\sum_{\ell=1}^{L}\mathrm{pr}\left( \underset{j\in[L]}{\cap}A_{\omega_j,h_\ell}(\zeta) \right)\\
        &\geq 1-\sum_{\ell=1}^{L}\mathrm{pr}\left( \underset{j\in[L]}{\cap}A_{\omega_j,h_\ell}(\zeta)\cap\widetilde{A}_{h_\ell} \right)-\sum_{\ell=1}^{L}\mathrm{pr}\left( \widetilde{A}_{h_\ell} \right)\\
        &\geq 1-\sum_{\ell,j\in[L]}\mathrm{pr}\left( A_{\omega_j,h_\ell}(\zeta)\cap\widetilde{A}_{h_\ell} \right)-\sum_{\ell=1}^{L}\mathrm{pr}\left( \widetilde{A}_{h_\ell} \right)\\
        &=1-L^2n^{-\overline{C}_1\zeta^2}-\sum_{\ell=1}^{L}n\exp\left\{ -\overline{C}_0(n+\omega m)h_i^d \right\}\\
        &\geq 1-L^2n^{-\overline{C}_1\zeta^2}-Ln\exp\left\{ -\overline{C}_0(n+\omega m)[\inf_{\ell\in[L]}h_\ell]^d \right\}
    \end{align*}
    When $(n+\omega m)[\inf_{\ell\in[L]}h_\ell]^d\log^{-1} (Ln)\rightarrow\infty$ and $L=o(n^{\overline{C}_1\zeta/2})$, we get
    \begin{gather*}
        \mathrm{pr}\left( (\hat{\omega},\hat{h})=(\omega^*,h^*) \right)\rightarrow 1\,,
    \end{gather*}
    which completes the proof of this theorem.
\end{proof}

\subsection{Proof of Theorem \ref{theo:marginal_coverage_r}}\label{sec:proof_sm_computing}

\begin{proof}
    Denote $\hat{q}_\alpha^y=Q\left( 1-\alpha;(n+1)^{-1}\left\{ \sum_{i=1}^{n+1}\delta_{\hat{\beta}_{\omega,\hat{r}}^y(X_i,S_i^y)} \right\} \right)$.
    Define the prediction set constructed using $\hat{\beta}_{\omega,\hat{r}}^y(X_{n+1},S_{n+1}^y)$ and $\tilde{q}_\alpha^y$ as
    \begin{gather*}
        \overline{C}_\alpha^{\rm ELCP}=\left\{ y:\hat{\beta}_{\omega,\hat{r}}^y(X_{n+1},S_{n+1}^y)\leq \tilde{q}_\alpha^y \right\}\,.
    \end{gather*}
    We decompose $\widehat{C}_\alpha^{\rm ELCP}\setminus\widetilde{C}_\alpha^{\rm ELCP}$ by $\widehat{C}_\alpha^{\rm ELCP}\setminus\widetilde{C}_\alpha^{\rm ELCP}\subset\left( \widehat{C}_\alpha^{\rm ELCP}\setminus\overline{C}_\alpha^{\rm ELCP} \right)\cup\left( \overline{C}_\alpha^{\rm ELCP}\setminus\widetilde{C}_\alpha^{\rm ELCP} \right)$.
    As $\widehat{C}_\alpha^{\rm ELCP}=\left\{ y:\hat{\beta}_{\omega,\hat{r}}^y(X_{n+1},S_{n+1}^y)\leq \hat{q}_\alpha^y \right\}$, we have
    \begin{gather*}
        \widehat{C}_\alpha^{\rm ELCP}\setminus\overline{C}_\alpha^{\rm ELCP}\subset\left\{ y:\tilde{q}_\alpha^y< \hat{\beta}_{\omega,\hat{r}}^y(X_{n+1},S_{n+1}^y)\leq \hat{q}_\alpha^y \right\}\,.
    \end{gather*}
    Based on the definition of $\tilde{\delta}^y=\underset{1\leq i\leq n+1}{\sup}|\hat{\beta}_{\omega,\hat{r}}^y(X_i,S_i^y)-\hat{\beta}_{\omega,\tilde{r}}^y(X_i,S_i^y)|$, we get that $|\hat{q}_\alpha^y-\tilde{q}_\alpha^y|\leq\tilde{\delta}^y$. Therefore,
    \begin{align*}
        \widehat{C}_\alpha^{\rm ELCP}\setminus\overline{C}_\alpha^{\rm ELCP}&\subset\left\{ y:\tilde{q}_\alpha^y< \hat{\beta}_{\omega,\hat{r}}^y(X_{n+1},S_{n+1}^y)\leq \tilde{q}_\alpha^y+\tilde{\delta}^y \right\}\\
        &\subset\left\{ y:\tilde{q}_\alpha^y-\tilde{\delta}^y< \hat{\beta}_{\omega,\tilde{r}}^y(X_{n+1},S_{n+1}^y)\leq \tilde{q}_\alpha^y+2\tilde{\delta}^y \right\}\,.
    \end{align*}
    On the other hand,
    \begin{align*}
        \overline{C}_\alpha^{\rm ELCP}\setminus\widetilde{C}_\alpha^{\rm ELCP}&\subset\left\{ y:\hat{\beta}_{\omega,\hat{r}}^y(X_{n+1},S_{n+1}^y)\leq \tilde{q}_\alpha^y<\hat{\beta}_{\omega,\tilde{r}}^y(X_{n+1},S_{n+1}^y) \right\}\\
        &\subset\left\{ y:\tilde{q}_\alpha^y< \hat{\beta}_{\omega,\tilde{r}}^y(X_{n+1},S_{n+1}^y)\leq \tilde{q}_\alpha^y+\tilde{\delta}^y \right\}\,.
    \end{align*}
    It follows that
    \begin{gather*}
        \widehat{C}_\alpha^{\rm ELCP}\setminus\widetilde{C}_\alpha^{\rm ELCP}\subset\left\{ y:\tilde{q}_\alpha^y-\tilde{\delta}^y< \hat{\beta}_{\omega,\tilde{r}}^y(X_{n+1},S_{n+1}^y)\leq \tilde{q}_\alpha^y+2\tilde{\delta}^y \right\}\,.
    \end{gather*}
    Similarly, we can show that
    \begin{gather*}
        \widetilde{C}_\alpha^{\rm ELCP}\setminus\widehat{C}_\alpha^{\rm ELCP}\subset\left\{ y:\tilde{q}_\alpha^y-2\tilde{\delta}^y< \hat{\beta}_{\omega,\tilde{r}}^y(X_{n+1},S_{n+1}^y)\leq \tilde{q}_\alpha^y+\tilde{\delta}^y \right\}\,.
    \end{gather*}
    
    Therefore, the difference between the prediction sets $\widetilde{C}_\alpha^{\rm ELCP}(X_{n+1})$ and $\widehat{C}_\alpha^{\rm ELCP}(X_{n+1})$, denoted as $\widetilde{C}_\alpha^{\rm ELCP}(X_{n+1})\Delta\widehat{C}_\alpha^{\rm ELCP}(X_{n+1})$, is a subset of $$D_{\hat{r},\tilde{r}}(X_{n+1})=\left\{ y:\tilde{q}_\alpha^y-2\tilde{\delta}^y< \hat{\beta}_{\omega,\tilde{r}}^y(X_{n+1},S_{n+1}^y)\leq \tilde{q}_\alpha^y+\tilde{\delta}^y \right\}.$$ 
    
    Under the exchangeability of $\mathcal{Z}_n\cup\{Z_{n+1}\}$ and Assumption \ref{assump:permutable}, we can show immediately that
    \begin{align*}
        &\mathrm{pr}\left( Y_{n+1}\in\widetilde{C}_\alpha^{\rm ELCP}(X_{n+1}) \right)\\
        \leq&\mathrm{pr}\left( Y_{n+1}\in\widehat{C}_\alpha^{\rm ELCP}(X_{n+1}) \right)+\mathrm{pr}\left( Y_{n+1}\in\widetilde{C}_\alpha^{\rm ELCP}(X_{n+1})\Delta\widehat{C}_\alpha^{\rm ELCP}(X_{n+1}) \right)\\
        \leq& 1-\alpha+(n+1)^{-1}+\mathrm{pr}(Y_{n+1}\in D_{\hat{r},\tilde{r}}(X_{n+1}))\,,
    \end{align*}
    and
    \begin{align*}
        &\mathrm{pr}\left( Y_{n+1}\in\widetilde{C}_\alpha^{\rm ELCP}(X_{n+1}) \right)\\
        \geq&\mathrm{pr}\left( Y_{n+1}\in\widehat{C}_\alpha^{\rm ELCP}(X_{n+1}) \right)-\mathrm{pr}\left( Y_{n+1}\in\widetilde{C}_\alpha^{\rm ELCP}(X_{n+1})\Delta\widehat{C}_\alpha^{\rm ELCP}(X_{n+1}) \right)\\
        \geq& 1-\alpha-\mathrm{pr}(Y_{n+1}\in D_{\hat{r},\tilde{r}}(X_{n+1}))\,.
    \end{align*}

    Furthermore, under  Assumption \ref{ass2} and the condition that  $\underset{x,s}{\sup}|\hat{r}(x,s)-\tilde{r}(x,s)|\leq C_r n^{-1}$ for some positive constant $ C_r $, we can derive
    \begin{align}
        &|\hat{\beta}_{\omega,\hat{r}}^y(X_i,S_i^y)-\hat{\beta}_{\omega,\tilde{r}}^y(X_i,S_i^y)|\notag\\
        =&\left|\dfrac{\sum_{j=1}^{n+1}K(X_i,X_j;h)\mathbb{1}(S_j^y\leq S_i^y)+\omega\sum_{j=1}^{m}K(X_i,X_j^\prime;h)\tilde{r}(X_j^\prime,S_j^\prime)\mathbb{1}(S_j^\prime\leq S_i^y)}{\sum_{j=1}^{n+1}K(X_i,X_j;h)+\omega\sum_{j=1}^{m}K(X_i,X_j^\prime;h)\tilde{r}(X_j^\prime,S_j^\prime)}\right.\notag\\
        &-\left.\dfrac{\sum_{j=1}^{n+1}K(X_i,X_j;h)\mathbb{1}(S_j^y\leq S_i^y)+\omega\sum_{j=1}^{m}K(X_i,X_j^\prime;h)\hat{r}(X_j^\prime,S_j^\prime)\mathbb{1}(S_j^\prime\leq S_i^y)}{\sum_{j=1}^{n+1}K(X_i,X_j;h)+\omega\sum_{j=1}^{m}K(X_i,X_j^\prime;h)\hat{r}(X_j^\prime,S_j^\prime)}\right|\notag\\
        \leq&\left|\dfrac{\sum_{j=1}^{n+1}K(X_i,X_j;h)\mathbb{1}(S_j^y\leq S_i^y)}{\sum_{j=1}^{n+1}K(X_i,X_j;h)+\omega\sum_{j=1}^{m}K(X_i,X_j^\prime;h)\tilde{r}(X_j^\prime,S_j^\prime)}\right.\notag\\
        &~~~~~~~~~~~~~~~~~~~~~~~~~~-\left.\dfrac{\sum_{j=1}^{n+1}K(X_i,X_j;h)\mathbb{1}(S_j^y\leq S_i^y)}{\sum_{j=1}^{n+1}K(X_i,X_j;h)+\omega\sum_{j=1}^{m}K(X_i,X_j^\prime;h)\hat{r}(X_j^\prime,S_j^\prime)}\right|\label{eq:delta_decomp1}\\
        &+\left|\dfrac{\omega\sum_{j=1}^{m}K(X_i,X_j^\prime;h)\tilde{r}(X_j^\prime,S_j^\prime)\mathbb{1}(S_j^\prime\leq S_i^y)}{\sum_{j=1}^{n+1}K(X_i,X_j;h)+\omega\sum_{j=1}^{m}K(X_i,X_j^\prime;h)\tilde{r}(X_j^\prime,S_j^\prime)}\right.\notag\\
        &~~~~~~~~~~~~~~~~~~~~~~~~~~-\left.\dfrac{\omega\sum_{j=1}^{m}K(X_i,X_j^\prime;h)\hat{r}(X_j^\prime,S_j^\prime)\mathbb{1}(S_j^\prime\leq S_i^y)}{\sum_{j=1}^{n+1}K(X_i,X_j;h)+\omega\sum_{j=1}^{m}K(X_i,X_j^\prime;h)\hat{r}(X_j^\prime,S_j^\prime)}\right|\label{eq:delta_decomp2}\,,
    \end{align}
    where
    \begin{align*}
        \text{\eqref{eq:delta_decomp1}}&\leq{\dfrac{\scriptstyle \sum_{j=1}^{n+1}K(X_i,X_j;h)\mathbb{1}(S_j^y\leq S_i^y)\left\{ \omega\sum_{j=1}^{m}K(X_i,X_j^\prime;h)|\hat{r}(X_j^\prime,S_j^\prime)-\tilde{r}(X_j^\prime,S_j^\prime)| \right\}}{\scriptstyle \left\{ \sum_{j=1}^{n+1}K(X_i,X_j;h)+\omega\sum_{j=1}^{m}K(X_i,X_j^\prime;h)\hat{r}(X_j^\prime,S_j^\prime) \right\}\left\{ \sum_{j=1}^{n+1}K(X_i,X_j;h)+\omega\sum_{j=1}^{m}K(X_i,X_j^\prime;h)\tilde{r}(X_j^\prime,S_j^\prime) \right\}}}\\
        &\leq\dfrac{\left\{ \sum_{j=1}^{n+1}K(X_i,X_j;h) \right\}\left\{ \omega\sum_{j=1}^{m}K(X_i,X_j^\prime;h) \right\}\underset{x,s}{\sup}|\hat{r}(x,s)-\tilde{r}(x,s)|}{\left\{ \sum_{j=1}^{n+1}K(X_i,X_j;h) \right\}\left\{ \omega\sum_{j=1}^{m}K(X_i,X_j^\prime;h)\tilde{r}(X_j^\prime,S_j^\prime) \right\}}\\
        &=\underset{x,s}{\sup}|\hat{r}(x,s)-\tilde{r}(x,s)|/\underset{x,s}{\inf}\ \tilde{r}(x,s)\leq\underline{L}_2^{-1} C_r n^{-1}\,.
    \end{align*}
    and
    \begin{align}
        \notag & \text{\eqref{eq:delta_decomp2}} \\
        \notag \leq&\dfrac{\scriptstyle \sum_{j,l}K(X_i,X_j^\prime;h)K(X_i,X_l^\prime;h)\mathbb{1}(S_j^\prime\leq S_i^y)\left| \hat{r}(X_j^\prime,S_j^\prime)\tilde{r}(X_l^\prime,S_l^\prime)-\hat{r}(X_l^\prime,S_l^\prime)\tilde{r}(X_j^\prime,S_j^\prime) \right|}{\scriptstyle \left\{ \sum_{j=1}^{n+1}K(X_i,X_j;h)+\omega\sum_{j=1}^{m}K(X_i,X_j^\prime;h)\hat{r}(X_j^\prime,S_j^\prime) \right\}\left\{ \sum_{j=1}^{n+1}K(X_i,X_j;h)+\omega\sum_{j=1}^{m}K(X_i,X_j^\prime;h)\tilde{r}(X_j^\prime,S_j^\prime) \right\}} \\
        \notag \leq& \dfrac{\scriptstyle \sum_{j,l}K(X_i,X_j^\prime;h)K(X_i,X_l^\prime;h)\mathbb{1}(S_j^\prime\leq S_i^y)\left|\hat{r}(X_j^\prime,S_j^\prime)\tilde{r}(X_l^\prime,S_l^\prime)-\tilde{r}(X_j^\prime,S_j^\prime)\tilde{r}(X_l^\prime,S_l^\prime)+\tilde{r}(X_j^\prime,S_j^\prime)\tilde{r}(X_l^\prime,S_l^\prime)-\hat{r}(X_l^\prime,S_l^\prime)\tilde{r}(X_j^\prime,S_j^\prime)\right|}{\scriptstyle \left\{ \sum_{j=1}^{n+1}K(X_i,X_j;h)+\omega\sum_{j=1}^{m}K(X_i,X_j^\prime;h)\hat{r}(X_j^\prime,S_j^\prime) \right\}\left\{ \sum_{j=1}^{n+1}K(X_i,X_j;h)+\omega\sum_{j=1}^{m}K(X_i,X_j^\prime;h)\tilde{r}(X_j^\prime,S_j^\prime) \right\}} \\
        \notag \leq&\dfrac{\scriptstyle  C_r n^{-1}\left[ \omega\sum_{j=1}^{m}K(X_i,X_j^\prime;h)\left\{\tilde{r}(X_j^\prime,S_j^\prime)+\hat{r}(X_j^\prime,S_j^\prime)\right\}\mathbb{1}(S_j^\prime\leq S_i^y) \right]\left\{ \sum_{l=1}^{m}K(X_i,X_l^\prime;h) \right\}}{\scriptstyle \left\{ \sum_{j=1}^{n+1}K(X_i,X_j;h)+\omega\sum_{j=1}^{m}K(X_i,X_j^\prime;h)\hat{r}(X_j^\prime,S_j^\prime) \right\}\left\{ \sum_{j=1}^{n+1}K(X_i,X_j;h)+\omega\sum_{j=1}^{m}K(X_i,X_j^\prime;h)\tilde{r}(X_j^\prime,S_j^\prime) \right\}}\\        \notag \leq &2 C_r n^{-1}\,.
    \end{align}
    Finally, we get that $|\hat{\beta}_{\omega,\hat{r}}^y(X_i,S_i^y)-\hat{\beta}_{\omega,\tilde{r}}^y(X_i,S_i^y)|\leq C_r (2+\underline{L}_2^{-1})n^{-1}$ and thus $\tilde{\delta}^y=O(n^{-1})$.
\end{proof}

%\appendixtwo
\section{Technical details}\label{technical details}

\subsection{Detailed definition of $\alpha(y)$ in $\widehat{C}_{\alpha}^{\mathrm{LCP}}(X_{n+1})$}\label{sec:LCP_detail}

    Define $\alpha_1(y)$ as the largest value in $\Gamma(\omega)=\{\sum_{j\in I_i}\omega_{i,j}-\epsilon: i\in[n+1], I_i\subseteq [n+1]\}$, where $0<\epsilon<\underset{\{i,j\in[n+1]:\omega_{i,j}>0\}}{\min}\omega_{i,j}$, such that
    \begin{eqnarray}\label{eq:LCP_alpha1}
        (n+1)^{-1}\sum_{i=1}^{n+1}\mathbb{1}\left( S_i^y<Q\left(1-\alpha_1(y);\hat{F}_i^y\right) \right) \geq 1-\alpha\,.
    \end{eqnarray}
    
    In $\widehat{C}_{\alpha}^{\mathrm{LCP}}(X_{n+1})$ as defined by (4), the value of $\alpha(y)$ is chosen as the largest value in the range of $\hat{F}_{n+1}^y$ that is smaller than $1 - \alpha_1(y)$, minus $\epsilon$.
    By the definition of $\alpha(y)$, there is no value of $\hat{F}_i^y(s)$ or $\hat{\beta}_i^{\mathrm{LCP}}(y)$, $i\in[n+1]$, equal to $1-\alpha(y)$. Therefore
    \begin{align}
        \notag \widehat{C}_\alpha^{\mathrm{LCP}}(X_{n+1})= & \left\{y:S(X_{n+1},y)\leq Q\left(1-\alpha(y); \hat{F}_{n+1}^{y}\right)\right\}\\
        \notag = & \left\{y:S(X_{n+1},y)< Q\left(1-\alpha_1(y); \hat{F}_{n+1}^{y}\right)\right\} \\
        \notag = & \{y:\hat{\beta}_{n+1}^{\mathrm{LCP}}(y)\leq 1-\alpha(y)\}\,,
    \end{align}
    and $\mathbb{1}\left( S_i^y<Q\left(1-\alpha(y);\hat{F}_i^y\right) \right)=\mathbb{1}\left( \hat{\beta}_i^{\mathrm{LCP}}(y)\leq 1-\alpha(y) \right)$ for $i\in[n+1]$.
    
    By the definition of $Q(\cdot;\cdot)$, $z^*=Q\left(1-\alpha;(n+1)^{-1}\sum_{j=1}^{n+1}\delta_{\hat{\beta}_j^{\mathrm{LCP}}(y)}\right)$ is the smallest $z$ such that
    \begin{gather*}
        (n+1)^{-1}\sum_{i=1}^{n+1}\mathbb{1}\left( \hat{\beta}_i^{\mathrm{LCP}}(y)\leq z \right)\geq 1-\alpha\,.
    \end{gather*}
    Since $(n+1)^{-1}\sum_{i=1}^{n+1}\mathbb{1}\left( \hat{\beta}_i^{\mathrm{LCP}}(y)\leq z \right)$ is monotonically increasing with respect to $z$, we have $z^*+\epsilon\geq 1-\alpha_1(y)>z^*$.
    Therefore,
    \begin{gather*}
        \left\{y:\hat{\beta}_{n+1}^{\mathrm{LCP}}(y)\leq z^*\right\}\subset\widehat{C}_\alpha^{\mathrm{LCP}}(X_{n+1})\subset\left\{y:\hat{\beta}_{n+1}^{\mathrm{LCP}}(y)\leq z^*+\epsilon\right\}\,.
    \end{gather*}
    
    For a given weight sequence $\omega_{i,j}$, define $\widetilde{\Gamma}(\omega)=\{\sum_{j\in I_i}\omega_{i,j}: i\in[n+1], I_i\subseteq [n+1]\}$.
    Since $\hat{\beta}_i^{\mathrm{LCP}}(y), i\in[n+1]$ and $z^*$ belong to $\widetilde{\Gamma}(\omega)$, and by the definition of $\epsilon$, the interval $(z^*,z^*+\epsilon]$ contains no elements of $\widetilde{\Gamma}(\omega)$, that is, $(z^*,z^*+\epsilon]\cap\widetilde{\Gamma}(\omega)=\emptyset$. 
    As $\hat{\beta}_{n+1}^{\mathrm{LCP}}(y)$ also takes value only in $\widetilde{\Gamma}(\omega)$, we have $\left\{y:z^*<\hat{\beta}_{n+1}^{\mathrm{LCP}}(y)\leq z^*+\epsilon\right\}=\emptyset$ and consequently $\left\{y:\hat{\beta}_{n+1}^{\mathrm{LCP}}(y)\leq z^*\right\}=\left\{y:\hat{\beta}_{n+1}^{\mathrm{LCP}}(y)\leq z^*+\epsilon\right\}$. 
    This implies that
    \begin{gather*}
        \widehat{C}_\alpha^{\mathrm{LCP}}(X_{n+1})=\left\{y:\hat{\beta}_{n+1}^{\mathrm{LCP}}(y)\leq z^*\right\}=\left\{y:\hat{\beta}_{n+1}^{\mathrm{LCP}}(y)\leq Q\left(1-\alpha;(n+1)^{-1}\sum_{j=1}^{n+1}\delta_{\hat{\beta}_j^{\mathrm{LCP}}(y)}\right)\right\}\,.
    \end{gather*}

\subsection{Comparison of $\widehat{C}_{\alpha}^{\mathrm{LCP}}(X_{n+1})$ and the LCP set in \citet{guan2023localized}}\label{sec:LCP_compare}

    Note that the LCP set $\widehat{C}_{\alpha}^{\mathrm{LCP}}(X_{n+1})$ in (6) is slightly different from that proposed by \citet{guan2023localized}, which is
    \begin{gather*}
        \widetilde{C}_\alpha^{\mathrm{LCP}}(X_{n+1})=\left\{y:S(X_{n+1},y)\leq Q\left(1-\widetilde{\alpha}(y); \hat{F}_{n+1}^{y}\right)\right\}\,,
    \end{gather*}
    where $\widetilde{\alpha}(y)$ is the largest value in $\widetilde{\Gamma}(\omega)=\{\sum_{j\in I_i}\omega_{i,j}: i\in[n+1], I_i\subseteq [n+1]\}$ such that
    \begin{gather*}
        (n+1)^{-1}\sum_{i=1}^{n+1}\mathbb{1}\left(S_i^y\leq Q\left(1-\widetilde{\alpha}(y);\hat{F}_i^y\right)\right)\geq 1-\alpha\,.
    \end{gather*}
    Define \begin{gather*}
        \tilde{\beta}_i^{\mathrm{LCP}}(y)=\sum_{j=1}^{n+1}\omega_{i,j}\mathbb{1}(S_j^y<S_i^y)\,
    \end{gather*}
    for $i\in[n+1]$. Then $\widetilde{C}_\alpha^{\mathrm{LCP}}(X_{n+1})$ is equivalent to
    \begin{eqnarray}
        \notag \widetilde{C}_\alpha^{\mathrm{LCP}}(X_{n+1}) = \left\{y:  \tilde{\beta}_{n+1}^{\mathrm{LCP}}(y) \leq Q\left( 1-\alpha; (n+1)^{-1}\overset{n+1}{\underset{i=1}\sum}\delta_{\tilde{\beta}_i^{\mathrm{LCP}}(y)} \right) \right\}\,.
    \end{eqnarray}

    In contrast, $\widehat{C}_\alpha^{\mathrm{LCP}}(X_{n+1})$ is defined as 
    \begin{eqnarray}
        \notag \widehat{C}_\alpha^{\mathrm{LCP}}(X_{n+1}) = \left\{y:  \hat{\beta}_{n+1}^{\mathrm{LCP}}(y) \leq Q\left( 1-\alpha; (n+1)^{-1}\overset{n+1}{\underset{i=1}\sum}\delta_{\hat{\beta}_i^{\mathrm{LCP}}(y)} \right) \right\}\,,
    \end{eqnarray}
    with
    \begin{equation}
        \notag \hat{\beta}_i^{\mathrm{LCP}}(y) = \sum_{j=1}^{n+1}\omega_{i,j}\mathbb{1}(S_j^{y}\leq S_i^{y}), \ i\in[n+1].
    \end{equation}
    The only difference between $\widetilde{C}_\alpha^{\mathrm{LCP}}(X_{n+1})$ and $\widehat{C}_\alpha^{\mathrm{LCP}}(X_{n+1})$ lies in the definition of $\widetilde{\alpha}(y)$ and $\alpha(y)$, or, $\tilde{\beta}_{i}^{\mathrm{LCP}}(y)$ and $\hat{\beta}_{i}^{\mathrm{LCP}}(y)$.

    If there are no ties among the values $S_{1}^{y},\ldots,S_{n+1}^{y}$ for $y\in\mathcal{Y}$, we have $\sup_{y\in\mathcal{Y}}\tilde{\beta}_i^{\mathrm{LCP}}(y)-\inf_{y\in\mathcal{Y}}\tilde{\beta}_i^{\mathrm{LCP}}(y)\leq\omega_{i,n+1}$ for $i\in[n]$ and $\sup_{y\in\mathcal{Y}}\tilde{\beta}_{n+1}^{\mathrm{LCP}}(y)-\inf_{y\in\mathcal{Y}}\tilde{\beta}_{n+1}^{\mathrm{LCP}}(y)\leq 1-\omega_{n+1,n+1}$.
    Therefore, the variation in the $(1-\alpha)$-th quantile of $\tilde{\beta}_1^{\mathrm{LCP}}(y),\ldots,\tilde{\beta}_{n+1}^{\mathrm{LCP}}(y)$ over different $y\in\mathcal{Y}$ is bounded by $\max\{\omega_{1,n+1},\ldots,\omega_{n,n+1},1-\omega_{n+1,n+1}\}$.
    By the definition of $\widetilde{\alpha}(y)$,  $\widetilde{\alpha}(y)$ is the largest value in $\Gamma(\omega)$ satisfying $1-\widetilde{\alpha}(y)\geq Q\left( 1-\alpha; (n+1)^{-1}\overset{n+1}{\underset{i=1}\sum}\delta_{\tilde{\beta}_i^{\mathrm{LCP}}(y)} \right)$.
    This implies that $\sup_{y\in\mathcal{Y}}\widetilde{\alpha}(y)-\inf_{y\in\mathcal{Y}}\widetilde{\alpha}(y)\leq \max\{\omega_{1,n+1},\ldots,\omega_{n,n+1},1-\omega_{n+1,n+1}\}$.
    Consider the situation where $X_{n+1}$ is located far from all of  $X_1,\ldots,X_n$ such that the weight $\omega_{n+1,n+1}$ is sufficiently close to 1 and $\{\omega_{i,n+1}\}_{i=1}^{n}$ are sufficiently small.
    In this case, it follows that $\sup_{y\in\mathcal{Y}}\widetilde{\alpha}(y) < \omega_{n+1,n+1}$.
    This further implies $\widetilde{\beta}_{n+1}^{\mathrm{LCP}}(y)\leq 1-\omega_{n+1,n+1}<1 - \widetilde{\alpha}(y)$ for any $y \in \mathcal{Y}$.
    As a result, the localized conformal prediction set $\widetilde{C}_\alpha^{\mathrm{LCP}}(X_{n+1}) = \mathcal{Y}$, which is uninformative.
    Therefore, the original LCP method may produce an unnecessarily large and uninformative prediction set for some values of $X_{n+1}=x\in\mathcal{X}$.

    Now consider $\widehat{C}_{\alpha}^{\mathrm{LCP}}(X_{n+1})$ under the situation where $X_{n+1}$ is located far from all of  $X_1,\ldots,X_n$.
    Assume there exists $y^*\in\mathcal{Y}$ such that $S(X_{n+1},y^*)\geq S_i$ for $i\in[n]$.
    Then, for any $y\in\mathcal{Y}$ that satisfies $S(X_{n+1},y)\geq S(X_{n+1},y^*)$, we have $\hat{\beta}_{n+1}^{\mathrm{LCP}}(y)=1$.
    In this case, if $\alpha>2/(n+1)$ and there are two distinct values among $\hat{\beta}_{1}^{\mathrm{LCP}}(y),\ldots,\hat{\beta}_{n+1}^{\mathrm{LCP}}(y)$, then $Q\left( 1-\alpha; (n+1)^{-1}\overset{n+1}{\underset{i=1}\sum}\delta_{\hat{\beta}_i^{\mathrm{LCP}}(y)} \right)<1$.
    This ensures that the modified prediction set $\widehat{C}_\alpha^{\mathrm{LCP}}(X_{n+1})$ is non-trivial as $\mathcal{Y}$.

    Figure \ref{fig:motivation_example1} illustrates a comparison between the original LCP (LCP-O) and the modified LCP (LCP-M).
    \begin{figure}[htbp]
        \centering
        \begin{minipage}{0.29\linewidth}
            \centering
            \includegraphics[scale=0.3]{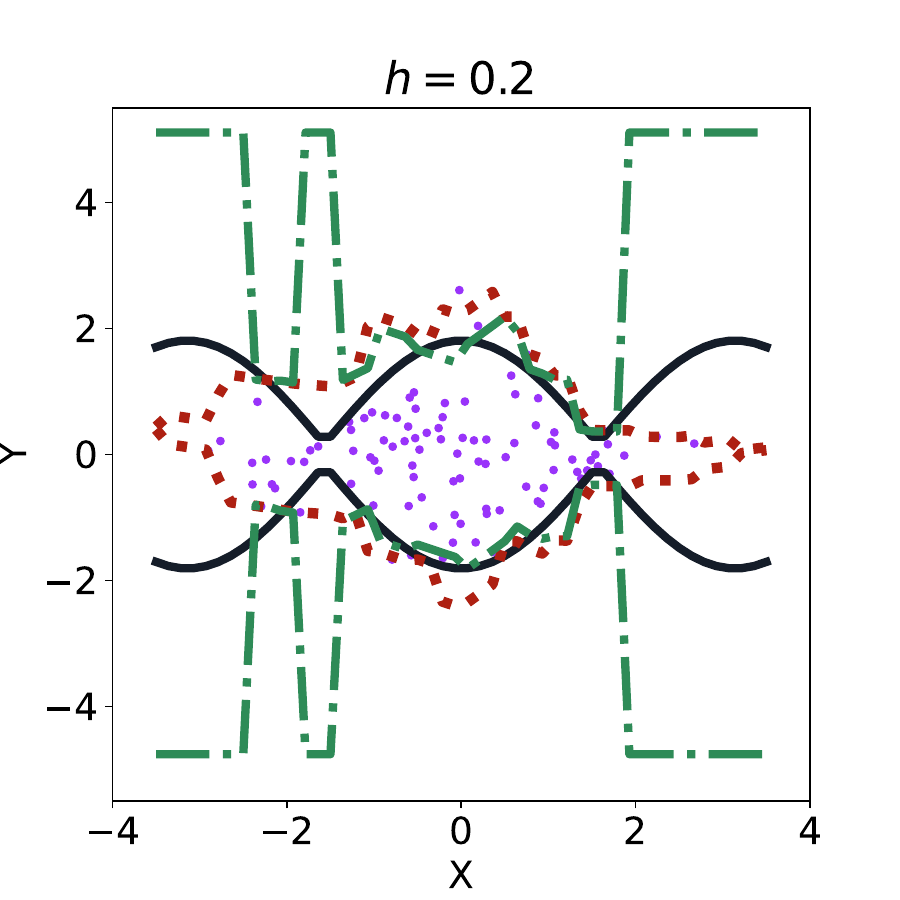}
        \end{minipage}
        \begin{minipage}{0.29\linewidth}
            \centering
            \includegraphics[scale=0.3]{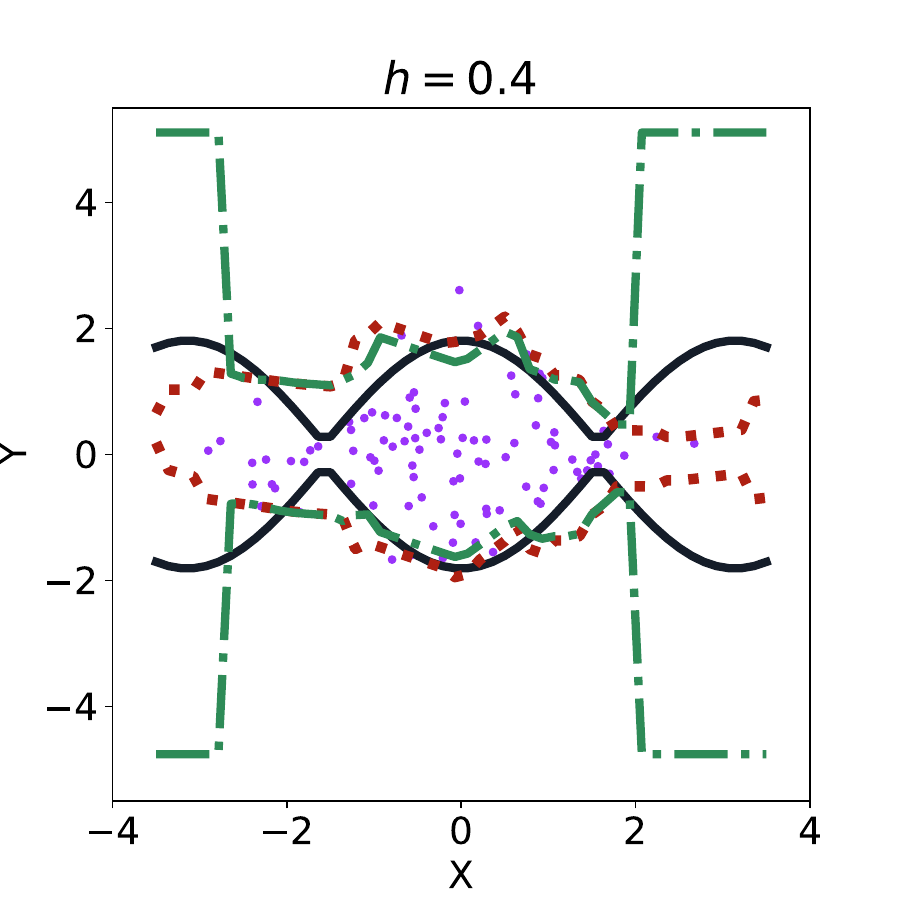}
        \end{minipage}
        \begin{minipage}{0.39\linewidth}
            \centering
            \includegraphics[scale=0.3]{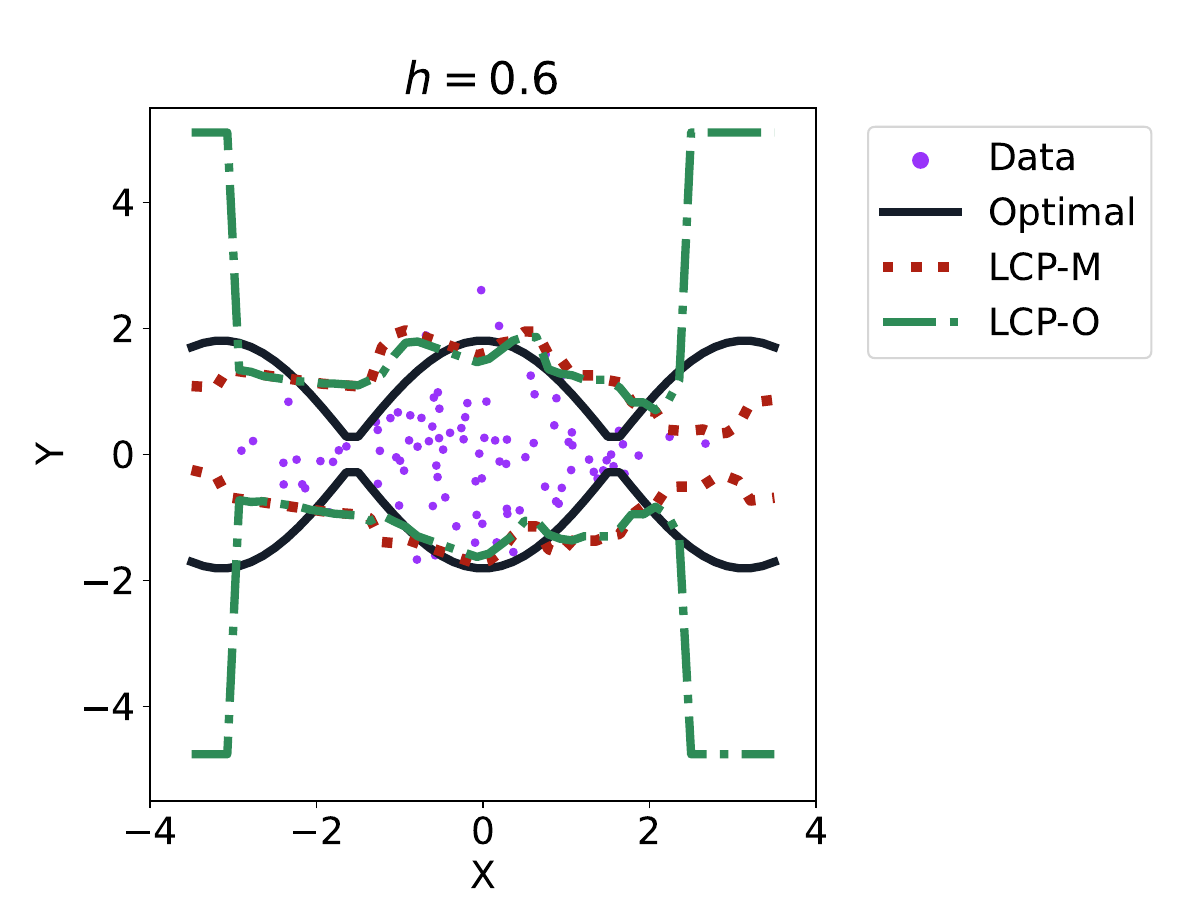}
        \end{minipage}
        \caption{Comparison of prediction bands by original LCP $\widetilde{C}_\alpha^{\mathrm{LCP}}(X_{n+1})$ (dash-dotted), modified LCP defined by $\widehat{C}_\alpha^{\mathrm{LCP}}(X_{n+1})$ (dotted) and optimal band (solid) with bandwidth $h=0.2,0.4$ and $0.6$. The calibration data of size $n=100$ is simulated from $X_i\sim N(0,1.5)$ and $Y_i=\{|\cos(X_i)|+0.1\}\varepsilon_i$, where $\varepsilon_i\sim N(0,1)$.
        }\label{fig:motivation_example1}
    \end{figure}

\subsection{Discussion of justification for Assumption \ref{ass3}}\label{sec:verify_ass3}
    The density ratio estimator $\hat{r}(x,s)$ can be viewed as a functional $T_{x,s}(F_{n+1},F_m^\prime)$ \citep{shao2012jackknife}, where $F_{n+1}$ and $F_m^\prime$ are the empirical distribution of $\mathcal{Z}_n\cup\{Z_{n+1}\}$ and $\mathcal{Z}_m^\prime$, respectively. Let $F_{m,-j}^\prime$ denote the empirical distribution of $\mathcal{Z}_m^\prime\setminus\{Z_j^\prime\}$ and define $\delta_j^\prime=m(F_{m}^\prime-F_{m,-j}^\prime)$, which satisfies $\int d\delta_j^\prime(z)=0$. 
    
    Applying the Von Mises expansion \citep{fernholz2012mises} gives
    \begin{gather*}
        T_{x,s}(F_{n+1},F_m^\prime)-T_{x,s}(F_{n+1},F_{m,-j}^\prime)=m^{-1}\dfrac{d}{dt}T_{x,s}(F_{n+1},F_{m}^\prime+t\delta_j^\prime)\Big|_{t=0}+R(m^{-1},\delta_j^\prime)\,,
    \end{gather*}
    where the first-order term involves the integral of the influence function. If the influence function is uniformly bounded and the domain of $(x,s)$ is bounded, the first-order term is $O(m^{-1})$. 
    Furthermore, if $T_{x,s}(\cdot,\cdot)$ is sufficiently smooth with uniformly bounded higher-order derivatives, the remainder term $R(m^{-1},\delta_j^\prime)$ is $O(m^{-2})$. Combining these results yields an overall expansion of order $O(m^{-1})$, with a constant independent of $(x,s)$, thereby verifying Assumption~\ref{ass3}.

\subsection{Extended results of \citet{filipovic2025kernel} for $D_k(r,\hat{r})$}\label{sec:ext_dk}

    We follow the formulation of \citet{filipovic2025kernel}. 
    Let $\mathcal{Z}$ be a countably generated measurable space with probability measures $P$ and $Q$ where $Q \ll P$, and let $$g_\star := \frac{dQ}{dP} \in L_P^k$$ be the density ratio, where $k\geq 2$ is given. Define $L_M^k$ norm of $h\in L_M^k$ by
    \begin{gather*}
        \|h\|_{L_M^k}=\left\{ \int |h(x)|^kdM(x) \right\}^{1/k}\,.
    \end{gather*}
    We consider a separable reproducing kernel Hilbert space (RKHS) $\mathcal{H}\subset L_M^k\subset L_M^2$ with bounded measurable kernel $\kappa(\cdot,\cdot)$ satisfying $\kappa_\infty := \sup_{z\in\mathcal{Z}} \kappa(z,z) < \infty$, where the canonical embeddings $J_M^k : \mathcal{H} \to L_M^k$, are Hilbert-Schmidt operators, and with adjoints $$J_M^* f = \int_\mathcal{Z} \kappa(\cdot,z)f(z)M(dz)\,,$$ for $M \in \{P,Q\}$. %For simplicity we assume $J_Mh=h$, which means $\mathcal{Z}$ is $L_M^k$ measurable. %When the support of $P,Q$ is bounded and tight
    The hypothesis density ratio is modeled as $g = p_\star + J_P^k h$ for $h \in \mathcal{H}$, where $p_\star : \mathcal{Z} \to \mathbb{R}$ is an exogenous prior function bounded by $\pi_\infty := \sup_{z\in\mathcal{Z}} |p_\star(z)| < \infty$. 
    The error functional $\mathcal{E}_k(h) := \|g_\star - p_\star - J_P^k h\|_{L_P^k}$ measures the worst-case expectation error, leading to the regularized convex problem 
    $$
    \min_{h\in\mathcal{H}} \{\mathcal{E}_2(h)^2 + \lambda\|h\|_\mathcal{H}^2\}\,,
    $$ 
    with unique solution $h_\lambda = (J_P^* J_P^k + \lambda)^{-1} J_P^* (g_\star - p_\star)$. Using Lemma 2.2 of \citet{filipovic2025kernel}, this solution can be equivalently expressed as 
    $$
    h_\lambda = (J_P^* J_P^k + \lambda)^{-1}(J_Q^* 1 - J_P^* p_\star).
    $$ 
    Note that for any $\hat{h}\in\mathcal{H}$ the corresponding density ratio estimator is $\hat{g}=p_\star+J_P^k\hat{h}$. Thus $\mathcal{E}_k(\hat{h})$ is the desired $D_k(g_\star,\hat{g})$. 

    % Using Mercer's theorem, any $h\in\mathcal{H}$ can be decomposed as $h=\sum_{i\in I}\langle h,u_i\rangle u_i$, where $\{u_i\}_{i\in I}$ is an standard orthogonal system in $\mathcal{H}$. Denote the subspace of $\mathcal{H}$ by $\mathcal{H}_1=\{\sum_{i\in I}\langle h,u_i\rangle_\mathcal{H} u_i:\sum_{i\in I}|\langle h,u_i\rangle_\mathcal{H}|<\infty\}$.

    Given i.i.d.~samples $\{z_{P,i}\}_{i=1}^n \sim P$ and $\{z_{Q,j}\}_{j=1}^n \sim Q$, we define the empirical operators $S_M : \mathcal{H} \to \mathbb{R}^n$ for $M \in \{P,Q\}$ as $(S_M h)_i := h(z_{M,i})$, with adjoints $S_M^* v := \sum_{i=1}^n \kappa(\cdot,z_{M,i})v_i$. The sample analogue of $J_M^*$ is $n^{-1}S_M^*$, yielding the empirical convex problem:
    \begin{gather*}
        \min_{h\in\mathcal{H}} \left\{ -2\langle S_Q^* \mathbf{1} - S_P^* \mathbf{p}_\star, h \rangle_\mathcal{H} + \langle (S_P^* S_P + n\lambda) h, h \rangle_\mathcal{H} \right\}\,,
    \end{gather*}
    where $\mathbf{1} := (1,...,1)^\top$ and $\mathbf{p}_\star := (p_\star(z_{P,1}),...,p_\star(z_{P,n}))^\top$. The unique empirical solution is:
    \begin{gather*}
        \hat{h}_\lambda := (S_P^* S_P + n\lambda)^{-1}(S_Q^* \mathbf{1} - S_P^* \mathbf{p}_\star)\,.
    \end{gather*}
    This estimator $\hat{h}_\lambda$ preserves the structure of the population solution $h_\lambda$ but replaces $J_M^*$ with $S_M^*$ and scales the regularization by $n$. Theorem 3.1 of \citet{filipovic2025kernel} establishes that $\hat{h}_\lambda \to h_\lambda$ almost surely as $n \to \infty$, with a functional CLT for $n^{1/2}(\hat{h}_\lambda - h_\lambda)$.

    Consider $\hat{g}_\lambda=p_\star+J_P^k\hat{h}_\lambda$ as the density ratio estimator. We bound $\mathcal{E}_k(\hat{h}_\lambda)$ by
    \begin{gather*}
        \mathcal{E}_k(\hat{h}_\lambda)\leq\mathcal{E}_k(h_\lambda)+\|J_P^k(\hat{h}_\lambda-h_\lambda)\|_{L_P^k}\,.
    \end{gather*}

    \begin{lemma}\label{lemma:dmem}
        Assume conditions of Lemma 2.3 (iii) in \citet{filipovic2025kernel} hold. For an orthogonal system $\{u_i\}_{i\in I}$ in $\mathcal{H}$, with eigenfunctions of $J_{P}^*J_{P}^k$ and eigenvalues $\mu_i>0$, assume $h_0=\sum_{i\in I}\langle h_0,u_i\rangle_\mathcal{H}u_i$ and there exists $\varepsilon\in[0,1]$ such that
        $$
        L_\varepsilon(h_0)\overset{\rm def.}{\equiv}\sum_{i\in I}\langle h_0,u_i\rangle_\mathcal{H}^2\mu_i^{2\varepsilon-2}<\infty\,.
        $$
        If the support of $P$ is a subset of $[0,1]^d$ and $g_\star\in L_P^k$, then
        \begin{equation*}
            \mathcal{E}_k(h_\lambda)\leq\kappa_\infty^{1/2}c_\varepsilon^{-1}L_{\varepsilon}(h_0)\lambda^{1-\varepsilon}\,.
        \end{equation*}
    \end{lemma}
    \begin{proof}
        We continue with the proof in Lemma 2.3 (iii) in \citet{filipovic2025kernel} that $h_\lambda=(J_{P}^*J_{P}^k+\lambda)^{-1}J_{P}^*J_{P}^kh_0=\sum_{i\in I}\mu_i(\mu_i+\lambda)^{-1}\langle h_0,u_i\rangle_\mathcal{H}u_i$, where $\{u_i\}_{i\in I}$ is an orthogonal system in $\mathcal{H}$ of eigenfunctions of $J_{P}^*J_{P}^k$ with eigenvalues $\mu_i>0$. %As $J_\mathbb{P}h=h$, the eigenvalues becomes $\mu_i=1$ for any $i$. % and $J_{\mathbb{P}}h_\lambda=\sum_{i\in I}\mu_i^{3/2}(\mu_i+\lambda)^{-1}\langle h_0,u_i\rangle_\mathcal{H}v_i$ with $v_i=\mu_i^{-1/2}J_{\mathbb{P}}u_i$. We have
        % \begin{equation*}
        %     J_{\mathbb{P}}h_0-J_{\mathbb{P}}h_\lambda=\sum_{i\in I}\dfrac{\mu_i^{1/2}\lambda}{\mu_i+\lambda}\langle h_0,u_i\rangle_\mathcal{H}u_i\,.
        % \end{equation*}
        For any $h\in\mathcal{H}$, as $h(x)=\langle h,\kappa(x,\cdot) \rangle_\mathcal{H}$, with Cauchy-Schwarz inequality we have $|h(x)|\leq\|h\|_\mathcal{H}\|\kappa(x,\cdot)\|_\mathcal{H}\leq\kappa_\infty^{1/2}\|h\|_\mathcal{H}$. 
        The operator norm of $J_P$ can be defined as
        \begin{gather*}
            \|J_{P}^k\|=\underset{\|h\|_\mathcal{H}\leq 1}{\sup}\|J_{P}^kh\|_{L_P^k}=\underset{\|h\|_\mathcal{H}\leq 1}{\sup}\left\{ \int|h(x)|^kdP(x) \right\}^{1/k}\leq\kappa_\infty^{1/2}\,.
        \end{gather*}
        Therefore we can derive:
        \begin{align*}
            \mathcal{E}_k(h_\lambda)&=\|J_{P}^k(h_0-h_\lambda)\|_{L_P^k}
            % &=\|\sum_{i\in I}\mu_i^{1/2}\lambda(\mu_i+\lambda)^{-1}\langle h_0,u_i\rangle_\mathcal{H}v_i\|_{L_P^k}\\
            % &=\|b_iv_i\|_{L_P^k}\leq M\sum_{i\in I}|b_i|
            \leq\|J_{P}^k\|\|h_0-h_\lambda\|_\mathcal{H} \leq \kappa_\infty^{1/2}\left\{ \sum_{i\in I}\dfrac{\lambda^2}{(\mu_i+\lambda)^2}\langle h_0,u_i\rangle_\mathcal{H}^2 \right\}^{1/2}\,.
            % &\leq \left( \dfrac{\kappa_\infty\lambda}{2} \right)^{1/2}\|h_0\|_\mathcal{H}.
        \end{align*}
        For $\lambda>0$ and $\mu_i>0$, there exists constant $c_\varepsilon>0$ such that $\lambda+\mu_i\geq c_\varepsilon\lambda^\varepsilon\mu_i^{1-\varepsilon}$ always holds. Therefore:
        \begin{align*}
            \mathcal{E}_k(h_\lambda)&\leq\kappa_\infty^{1/2}\left\{ \sum_{i\in I}\dfrac{\lambda^{2-2\varepsilon}}{c_\varepsilon^2\mu_i^{2-2\varepsilon}}\langle h_0,u_i\rangle_\mathcal{H}^2 \right\}^{1/2} =\kappa_\infty^{1/2}c_\varepsilon^{-1}L_{\varepsilon}(h_0)\lambda^{1-\varepsilon}\,,
        \end{align*}
        which completes the proof of this lemma.
    \end{proof}

    Combine Lemma \ref{lemma:dmem} with Theorem 3.1 (iii) of \citet{filipovic2025kernel}, we have a general error bound for $\mathcal{E}_k(\hat{h}_\lambda)$.
    \begin{theorem}\label{theo:dmem}
        Assume conditions of Lemma 2.3 (iii) and Theorem 3.1 (iii) in \citet{filipovic2025kernel} hold. If the support of $P$ is a subset of $[0,1]^d$ and $g_\star\in L_P^k$, and $L_\varepsilon(h_0)<\infty$ for $\varepsilon\in[0,1]$, then with probability over $1-\gamma$,
        \begin{equation*}
            \mathcal{E}_k(\hat{h}_\lambda)\leq C_3\{n^{-1}\log(2/\gamma)\}^{(1-\varepsilon)/(4-2\varepsilon)}\,
        \end{equation*}
        for some positive constant $C_3$.
    \end{theorem}
    \begin{proof}
        Use Lemma \ref{lemma:dmem}, and we bound $\mathcal{E}_m(\hat{h}_\lambda)$ by
        \begin{gather}
            \mathcal{E}_k(\hat{h}_\lambda)\leq\mathcal{E}_k(h_\lambda)+\|J_P^k(\hat{h}_\lambda-h_\lambda)\|_{L_P^k}\leq\kappa_\infty^{1/2}c_\varepsilon^{-1}L_{\varepsilon}(h_0)\lambda^{1-\varepsilon}+\kappa_\infty^{1/2}\|\hat{h}_\lambda-h_\lambda\|_\mathcal{H}\,.\label{eq:emupper}
        \end{gather}
        Theorem 3.1 (iii) in \citet{filipovic2025kernel} implies
        \begin{gather*}
            \|\hat{h}_\lambda-h_\lambda\|_\mathcal{H}\leq 2\{2\log(2/\gamma)\kappa_\infty\}^{1/2}(1+\pi_\infty+\|h_\lambda\|_\mathcal{H}\kappa_\infty^{1/2})\lambda^{-1}n^{-1/2}\,
        \end{gather*}
        with probability over $1-\gamma$. 
        In the proof of Lemma \ref{lemma:dmem}, we have 
        \begin{gather*}
            \|h_\lambda\|_\mathcal{H}\leq \|h_0\|_\mathcal{H}+\|h_0-h_\lambda\|_\mathcal{H}\leq \|h_0\|_\mathcal{H}+C_1L_\varepsilon(h_0)\,,
        \end{gather*}
        where $C_1$ is an upper bound for $\kappa_\infty^{1/2}c_\varepsilon^{-1}L_{\varepsilon}(h_0)$. Denote $C_2=2\sqrt{2\kappa_\infty}[1+\pi_\infty+\{\|h_0\|_\mathcal{H}+C_1L_\varepsilon(h_0)\}\sqrt{\kappa_\infty}]$ and 
        $
        C_3=\{(1-\varepsilon)^{1/(2-\varepsilon)}+(1-\varepsilon)^{(\varepsilon-1)/(2-\varepsilon)}\}C_1^{1/(2-\varepsilon)}C_2^{(1-\varepsilon)/(2-\varepsilon)}
        $.
        By simple algebra, the minimum of the right side of \eqref{eq:emupper} is
        $
        C_3\{n^{-1}\log(2/\gamma)\}^{(1-\varepsilon)/(4-2\varepsilon)}\,.
        $
    \end{proof}
    Theorem \ref{theo:dmem} indicates that if $g_\star\in L_P^k$ for some $k\geq 2$, we can reach $
    \epsilon_k(\gamma;r)=C_3\{n^{-1}\log(2/\gamma)\}^{(1-\varepsilon)/(4-2\varepsilon)}.
    $ 
    When the number of nonzero components $\langle h_0,u_i\rangle$ is finite, we can take $\varepsilon=0$, in which case the convergence rate of $\epsilon_k(\gamma;r)$ is $O(n^{-1/4})$. The $\varepsilon$ quantifies the information of $h_0$ along the rapidly contracting components ($u_i$ with smaller $\mu_i$). Small value of $\varepsilon$ indicates reduced information about $h_0$ in these fast decaying directions, making the estimation problem more tractable; therefore, results in a faster rate of $\epsilon_k(\gamma;r)$.
    With $\varepsilon\leq 1/2$, $$\epsilon_k(\gamma;r)=O(\log^{1/6}(2/\gamma)n^{-1/6}).$$

\subsection{Detailed construction of loss function $\mathcal{L}_2$}\label{sec:SM_detailed L2}

    For i.i.d.~data pairs $(X_1,W_1),\ldots,(X_n,W_n)$, let $F_{W\mid X}$ be the conditional distribution of $W_i$ given $X_i$. Then the conditional maximum mean discrepancy \citep[CMMD]{yan2022distance} between $F_{W\mid X}$ and $\mathrm{Uniform}[0,1]$ at $X_i=x$, using kernel $K_1(\cdot,\cdot)$ over the support of $W_i$, is defined as
    \begin{gather}
        \notag E\left\{ K_1(W_1,W_2)\mid X_1=x, X_2=x \right\}+E\left\{ K_1(U_1,U_2) \right\}-2E\left\{ K_1(W_1,U_1)\mid X_1=x \right\}\,,
    \end{gather}
    where $U_1,U_2\sim\mathrm{Uniform}[0,1]$. As $E\left\{ K_1(U_1,U_2) \right\}$ is a constant, we ignore this term in the following formulation.
    Then the integrated conditional maximum mean discrepancy (ICMMD) between $F_{W\mid X}$ and $\mathrm{Uniform}[0,1]$ with weighted function $\{f_{X}(x)\}^2$ is
    \begin{align*}
        & \int \left[E\left\{ K_1(W_1,W_2)\mid X_1=x, X_2=x \right\}-2E\left\{ K_1(W_1,U_1)\mid X_1=x \right\}\right]\{f_{X}(x)\}^2dx \\
        = & \int \left[E\left\{ K_1(W_1,W_2)\mid X_1=x, X_2=x \right\}-2E\left\{ \phi(W_1)\mid X_1=x \right\}\right]\{f_{X}(x)\}^2dx\,,
    \end{align*}
    where $\phi(u)=E\left\{ K_1(u,U_1) \right\}$.
    
    According to \citet{yan2022distance}, an estimator of the ICMMD is given by
    \begin{align*}
        &\dfrac{1}{n(n-1)}\sum_{1\leq i\neq j\leq n}\left\{K_1(W_i,W_j)-2\phi(W_i)\right\}K_2(X_i,X_j)\,,
    \end{align*}
    where $K_2(\cdot,\cdot)$ is a kernel function supported on $\mathcal{X}$.
    
    The loss function $\mathcal{L}_2$ is proposed to measure the discrepancy between the conditional distribution $\hat{\beta}_{\omega,h}(X_i,S_i)$ given $X_i$ and $\mathrm{Uniform}[0,1]$ based on the ICMMD, and the available data pairs are $(X_1,\hat{\beta}_{\omega,h}(X_1,S_1^y)),\ldots, (X_n,\hat{\beta}_{\omega,h}(X_n,S_n^y))$ and $(X_{n+1},\hat{\beta}_{\omega,h}(X_{n+1},S_{n+1}^y))$. Based on the discussion in the previous section, we defined the loss function as
    \begin{align*}
        & \mathcal{L}_{2}\left(\omega,h;\mathcal{Z}_n\cup\{(X_{n+1},y)\},\mathcal{Z}_m^\prime\right) \\
        = &
        \dfrac{1}{n(n+1)}\sum_{1\leq i\neq j\leq n+1}\Big\{K_1(\hat{\beta}_{\omega,h}(X_i,S_i^y),\hat{\beta}_{\omega,h}(X_j,S_j^y)) -2\phi(\hat{\beta}_{\omega,h}(X_i,S_i^y))\Big\}K_2(X_i,X_j)\,.
    \end{align*}

    Algorithm \ref{twoagentsProcedure-parasel} below presents an end-to-end algorithm that implements ELCP with parameter selection.
    \begin{algorithm}[H]
        \caption{Enhanced Localized Conformal Prediction with Parameter Selection}\label{twoagentsProcedure-parasel}
        {\bf Input:} Calibration and auxiliary data $\mathcal{Z}_n$, $\mathcal{Z}_m^\prime$, test point $X_{n+1}$, training data $\mathcal{D}_{\mathrm{tr}}$ and $\mathcal{D}_{\mathrm{tr}}^{\prime}$, score functions $S(\cdot,\cdot)$ and $S^{\prime}(\cdot,\cdot)$, function $K(\cdot,\cdot;\cdot)$, loss function $\mathcal{L}(\cdot,\cdot;\cdot,\cdot)$, parameter candidate set $\mathcal{G}$, level $1-\alpha$
        \begin{algorithmic}[1]
        \State Pretrain model $\hat{\mu}(\cdot)$ from $\mathcal{D}_{\mathrm{tr}}$ and $\hat{\mu}^{\prime}(\cdot)$ from $\mathcal{D}_{\mathrm{tr}}^{\prime}$
        \State Calculate $S_i^{y}=S(X_i,Y_i)$, $i\in[n]$ using $\hat{\mu}(\cdot)$, and $S_j^\prime=S^\prime(X_j^\prime,Y_j^\prime)$, $j\in[m]$ using $\hat{\mu}^{\prime}(\cdot)$
        \For{$y\in\mathcal{Y}$}
            \State\qquad Calculate $S_{n+1}^y=S(X_{n+1},y)$;
            \State\qquad Obtain density ratio estimator $\hat{r}(\cdot,\cdot)$ using $\mathcal{Z}_n\cup\{(X_{n+1}, y)\}$ and $\mathcal{Z}^\prime_m$;
            \State\qquad Solve $(\hat{\omega}(y),\hat{h}(y)) =\underset{(\omega,h)\in\mathcal{G}}{\arg\min}~\mathcal{L}\left(\omega,h;\mathcal{Z}_n\cup\{(X_{n+1},y)\},\mathcal{Z}_m^\prime\right)$
            \State\qquad Calculate $\hat{\beta}_{\hat{\omega}(y),\hat{h}(y)}^{y}(X_i,S_i^{y})$ for $i\in[n+1]$;
            \State\qquad Calculate $\hat{q}=Q\Big(1-\alpha; (n+1)^{-1}\sum_{i=1}^{n+1}\delta_{\hat{\beta}_{\hat{\omega}(y),\hat{h}(y)}^{y}(X_i,S_i^{y})}\Big)$;
            \State\qquad $y$ is included in set $\widehat{C}_\alpha^{\mathrm{ELCP}}(X_{n+1})$ as long as $\hat{\beta}_{\hat{\omega}(y),\hat{h}(y)}^{y}(X_{n+1},S_{n+1}^{y})\leq\hat{q}$.
        \EndFor
        \State \Return $\widehat{C}_\alpha^{\mathrm{ELCP}}(X_{n+1})$
        \end{algorithmic}
    \end{algorithm}

\subsection{Detailed analysis of computational efficient implementation of ELCP}\label{sec:sup computation}

    First, we present an end-to-end computationally efficient implementation of ELCP in Algorithm \ref{algo:computeeffimple ELCP}. 
    The algorithm is generic in the sense that it accommodates both the case of fixed parameters $(\omega,h)$ and the case of parameter selection.
    When $\omega$ and $h$ are fixed and given, then the candidate set $\mathcal{G}$ reduces to the singleton $(\omega,h)$.

    \begin{algorithm}[H]
        \caption{Computationally Efficient Implementation of ELCP}\label{algo:computeeffimple ELCP}
        {\bf Input:} Calibration and auxiliary data $\mathcal{Z}_n$, $\mathcal{Z}_m^\prime$, test point $X_{n+1}$, training data $\mathcal{D}_{\mathrm{tr}}$ and $\mathcal{D}_{\mathrm{tr}}^{\prime}$, score functions $S(\cdot,\cdot)$ and $S^{\prime}(\cdot,\cdot)$, function $K(\cdot,\cdot;\cdot)$, loss function $\mathcal{L}(\cdot,\cdot;\cdot,\cdot)$, parameter candidate set $\mathcal{G}$, level $1-\alpha$
        \begin{algorithmic}[1]
        \State Pretrain model $\hat{\mu}(\cdot)$ from $\mathcal{D}_{\mathrm{tr}}$ and $\hat{\mu}^{\prime}(\cdot)$ from $\mathcal{D}_{\mathrm{tr}}^{\prime}$
        \State Calculate $S_i^{y}=S(X_i,Y_i)$, $i\in[n]$ using $\hat{\mu}(\cdot)$, and $S_j^\prime=S^\prime(X_j^\prime,Y_j^\prime)$, $j\in[m]$ using $\hat{\mu}^{\prime}(\cdot)$
        \State Obtain density ratio estimator $\hat{r}(\cdot,\cdot)$ using $\mathcal{Z}_n$ and $\mathcal{Z}^\prime_m$
        \State Solve $(\hat{\omega},\hat{h}) =\underset{(\omega,h)\in\mathcal{G}}{\arg\min}~\mathcal{L}\left(\omega,h;\mathcal{Z}_n,\mathcal{Z}_m^\prime\right)$
        \For{$y\in\mathcal{Y}$}
            \State\qquad Calculate $S_{n+1}^y=S(X_{n+1},y)$;
            \State\qquad Calculate $\hat{\beta}_{\hat{\omega},\hat{h}}^{y}(X_i,S_i^{y})$ for $i\in[n+1]$;
            \State\qquad Calculate $\hat{q}=Q\Big(1-\alpha; (n+1)^{-1}\sum_{i=1}^{n+1}\delta_{\hat{\beta}_{\hat{\omega},\hat{h}}^{y}(X_i,S_i^{y})}\Big)$;
            \State\qquad $y$ is included in set $\widehat{C}_\alpha^{\mathrm{ELCP}}(X_{n+1})$ as long as $\hat{\beta}_{\hat{\omega},\hat{h}}^{y}(X_{n+1},S_{n+1}^{y})\leq\hat{q}$.
        \EndFor
        \State \Return $\widehat{C}_\alpha^{\mathrm{ELCP}}(X_{n+1})$
        \end{algorithmic}
    \end{algorithm}
    
    To highlight the impacts of the bandwidth $h$, the weighting parameter $\omega$, and the density ratio estimator $\hat{r}(\cdot,\cdot)$, we denote $\hat{\beta}_{\omega,h,\hat{r}}^y(x,s)$ as the counterpart of $\hat{\beta}_{\omega,\hat{r}}^y(x,s)$.
    Let $\tilde{h}$ and $\tilde{\omega}$ denote the parameters selected without using $(X_{n+1},y)$, i.e.,
    \begin{align*}
        (\tilde{h},\tilde{\omega})=\underset{(\omega,h)\in\mathcal{G}}{\arg\min}~\mathcal{L}\left(\omega,h;\mathcal{Z}_n,\mathcal{Z}_m^\prime\right)\,.
    \end{align*}
    
    Replacing $\hat{h}(y)$, $\hat{\omega}(y)$, and $\hat{r}(x,s)$ in $\{\hat{\beta}_{\hat{h}(y),\hat{\omega}(y),\hat{r}}^{y}(X_i,S_i^{y})\}_{i\in[n+1]}$ with $\tilde{h}$, $\tilde{\omega}$, and $\tilde{r}(x,s)$, respectively, yields $\{\hat{\beta}_{\tilde{h}(y),\tilde{\omega}(y),\tilde{r}}^{y}(X_i,S_i^{y})\}_{i\in[n+1]}$, and consequently, 
    \begin{align}
        \notag & \widetilde{C}_\alpha^{\mathrm{ELCP-PS}}(X_{n+1}) = \left\{y:\tilde{\beta}_{\tilde{\omega}(y),\tilde{h}(y)}^y(X_{n+1},S_{n+1}^y)\leq Q\left(1-\alpha;\dfrac{1}{n+1}\sum_{i=1}^{n+1}\delta_{\tilde{\beta}_{\tilde{\omega}(y),\tilde{h}(y)}^y(X_i,S_i^y)}\right)\right\}\,,
    \end{align}
    which is the computationally efficient version of ELCP with parameter selection.

%\appendixthree
\section{Additional numerical results}\label{sec:SM_numerical}

    Section \ref{sec:supp_simu_add_fixed_parameter} presents additional results for synthesized data under fixed $\omega$ and $h$, including: marginal coverage results under fixed $\omega$ and $h$ (Section \ref{sec:supp_marginal}); the effect of $\omega$ and $h$ on the test-conditional miscoverage error of ELCP (Section \ref{sec:supp_simu_parameter1}); results for different density ratio estimators (Section \ref{sec:supp_simu_dre}); the impact of auxiliary data size (Section \ref{sec:supp_simu_auxiliary_size}); experiments on different score functions (Section \ref{sec:supp_simu_score}); experiments with varying nominal coverage level $1-\alpha$ (Section \ref{sec:supp_simu_varying_alpha}); and experiments with extremely limited calibration size $n$ (Section \ref{sec:supp_simu_limited_n}).
    
    Additional results for synthesized data with data-driven selection of $\omega$ and $h$ are given in Section \ref{sec:supp_simu_add_selected_parameter}.
    Section \ref{sec:simu_semi} covers simulations under semi-supervised settings. Section \ref{sec:supp_housing} provides detailed implementations of the Moscow housing price prediction analysis. 
    Another real data analysis on medical insurance cost prediction is presented in Section \ref{sec:real_data}.

\FloatBarrier
\subsection{Additional results for synthesized data under fixed $\omega$ and $h$}\label{sec:supp_simu_add_fixed_parameter}

\subsubsection{Marginal coverage under fixed $\omega$ and $h$}\label{sec:supp_marginal}
    
    Tables~\ref{table:mardgp1}--\ref{table:mardgp3} report the marginal coverage of ELCP with $\omega=1$, LCP, RLCP, LCP-C and RLCP-C for DGP1--DGP3, across varying $h$, $n$ and $d$. 
    Results for ELCP with other values of $\omega$ are similar and therefore are omitted.
    Overall, ELCP, LCP, and RLCP achieve valid marginal coverage, whereas both LCP-C and RLCP-C exhibit substantially overcoverage.
    This suggests that when the auxiliary information is imperfect, directly combining calibration and auxiliary data can lead to unreliable prediction sets.
    \begin{table*}[htbp]
        \centering
        \caption{Marginal coverage rates for DGP1 with $\textbf{--}$ indicates cases with more than $30\%$ infinite prediction intervals.\label{table:mardgp1}}
        {\fontsize{8}{4.5}\selectfont{\setlength{\tabcolsep}{3.5pt}
        \begin{tabular}{cccccccccccccccc}
        \toprule
        $n$ & $d$ & $h$ & $0.4$ & $0.6$ & $0.8$ & $1.0$ & $1.2$ & $1.4$ & $1.6$ & $1.8$ & $2.0$ & $2.5$ & $3.0$ & $3.5$ & $4.0$ \vspace{2pt} \\ \midrule
        $100$ & $5$ & ELCP & $0.894$ & $0.895$ & $0.895$ & $0.896$ & $0.898$ & $0.899$ & $0.898$ & $0.898$ & $0.898$ & $0.898$ & $0.898$ & $0.898$ & $0.898$\\
        &   & LCP & $0.900$ & $0.900$ & $0.899$ & $0.900$ & $0.900$ & $0.900$ & $0.900$ & $0.900$ & $0.901$ & $0.901$ & $0.901$ & $0.901$ & $0.902$\\
        &   & RLCP & $\textbf{--}$ & $\textbf{--}$ & $\textbf{--}$ & $0.905$ & $0.905$ & $0.906$ & $0.906$ & $0.906$ & $0.906$ & $0.907$ & $0.907$ & $0.907$ & $0.907$\\
        &   & LCP-C & $\textbf{0.933}$ & $\textbf{0.937}$ & $\textbf{0.940}$ & $\textbf{0.943}$ & $\textbf{0.943}$ & $\textbf{0.943}$ & $\textbf{0.943}$ & $\textbf{0.943}$ & $\textbf{0.943}$ & $\textbf{0.943}$ & $\textbf{0.943}$ & $\textbf{0.943}$ & $\textbf{0.943}$\\
        &   & RLCP-C & $\textbf{--}$ & $\textbf{0.935}$ & $\textbf{0.940}$ & $\textbf{0.942}$ & $\textbf{0.943}$ & $\textbf{0.943}$ & $\textbf{0.944}$ & $\textbf{0.944}$ & $\textbf{0.944}$ & $\textbf{0.944}$ & $\textbf{0.944}$ & $\textbf{0.944}$ & $\textbf{0.944}$\\
        & $10$ & ELCP & $\textbf{--}$ & $0.900$ & $0.902$ & $0.901$ & $0.904$ & $0.906$ & $0.905$ & $0.905$ & $0.905$ & $0.904$ & $0.904$ & $0.903$ & $0.903$\\
        &   & LCP & $\textbf{--}$ & $\textbf{--}$ & $0.907$ & $0.906$ & $0.905$ & $0.905$ & $0.906$ & $0.905$ & $0.904$ & $0.905$ & $0.905$ & $0.905$ & $0.906$\\
        &   & RLCP & $\textbf{--}$ & $\textbf{--}$ & $\textbf{--}$ & $\textbf{--}$ & $\textbf{--}$ & $\textbf{--}$ & $0.909$ & $0.909$ & $0.910$ & $0.911$ & $0.912$ & $0.912$ & $0.912$\\
        &   & LCP-C & $\textbf{--}$ & $\textbf{0.929}$ & $\textbf{0.931}$ & $\textbf{0.933}$ & $\textbf{0.935}$ & $\textbf{0.936}$ & $\textbf{0.937}$ & $\textbf{0.937}$ & $\textbf{0.937}$ & $\textbf{0.937}$ & $\textbf{0.937}$ & $\textbf{0.937}$ & $\textbf{0.937}$\\
        &   & RLCP-C & $\textbf{--}$ & $\textbf{--}$ & $\textbf{--}$ & $\textbf{--}$ & $\textbf{0.932}$ & $\textbf{0.934}$ & $\textbf{0.936}$ & $\textbf{0.937}$ & $\textbf{0.937}$ & $\textbf{0.938}$ & $\textbf{0.938}$ & $\textbf{0.938}$ & $\textbf{0.938}$\\
        & $15$ & ELCP & $\textbf{--}$ & $\textbf{--}$ & $0.904$ & $0.905$ & $0.907$ & $0.907$ & $0.908$ & $0.909$ & $0.910$ & $0.909$ & $0.909$ & $0.908$ & $0.908$\\
        &   & LCP & $\textbf{--}$ & $\textbf{--}$ & $\textbf{--}$ & $0.909$ & $0.909$ & $0.908$ & $0.908$ & $0.908$ & $0.909$ & $0.909$ & $0.910$ & $0.911$ & $0.911$\\
        &   & RLCP & $\textbf{--}$ & $\textbf{--}$ & $\textbf{--}$ & $\textbf{--}$ & $\textbf{--}$ & $\textbf{--}$ & $\textbf{--}$ & $\textbf{--}$ & $0.912$ & $0.914$ & $0.915$ & $0.916$ & $0.916$\\
        &   & LCP-C & $\textbf{--}$ & $\textbf{--}$ & $\textbf{0.927}$ & $\textbf{0.929}$ & $\textbf{0.931}$ & $\textbf{0.933}$ & $\textbf{0.934}$ & $\textbf{0.934}$ & $\textbf{0.935}$ & $\textbf{0.935}$ & $\textbf{0.934}$ & $\textbf{0.934}$ & $\textbf{0.934}$\\
        &   & RLCP-C & $\textbf{--}$ & $\textbf{--}$ & $\textbf{--}$ & $\textbf{--}$ & $\textbf{--}$ & $\textbf{--}$ & $\textbf{0.930}$ & $\textbf{0.932}$ & $\textbf{0.933}$ & $\textbf{0.934}$ & $\textbf{0.935}$ & $\textbf{0.935}$ & $\textbf{0.935}$\\
        & $20$ & ELCP & $\textbf{--}$ & $\textbf{--}$ & $\textbf{--}$ & $0.899$ & $0.896$ & $0.896$ & $0.895$ & $0.897$ & $0.897$ & $0.897$ & $0.896$ & $0.896$ & $0.896$\\
        &   & LCP & $\textbf{--}$ & $\textbf{--}$ & $\textbf{--}$ & $\textbf{--}$ & $0.896$ & $0.895$ & $0.896$ & $0.897$ & $0.897$ & $0.897$ & $0.898$ & $0.898$ & $0.899$\\
        &   & RLCP & $\textbf{--}$ & $\textbf{--}$ & $\textbf{--}$ & $\textbf{--}$ & $\textbf{--}$ & $\textbf{--}$ & $\textbf{--}$ & $\textbf{--}$ & $\textbf{--}$ & $0.904$ & $0.905$ & $0.905$ & $0.905$\\
        &   & LCP-C & $\textbf{--}$ & $\textbf{--}$ & $\textbf{--}$ & $\textbf{0.920}$ & $\textbf{0.922}$ & $\textbf{0.924}$ & $\textbf{0.925}$ & $\textbf{0.926}$ & $\textbf{0.926}$ & $\textbf{0.926}$ & $\textbf{0.926}$ & $\textbf{0.927}$ & $\textbf{0.927}$\\
        &   & RLCP-C & $\textbf{--}$ & $\textbf{--}$ & $\textbf{--}$ & $\textbf{--}$ & $\textbf{--}$ & $\textbf{--}$ & $\textbf{--}$ & $\textbf{0.922}$ & $\textbf{0.924}$ & $\textbf{0.926}$ & $\textbf{0.927}$ & $\textbf{0.927}$ & $\textbf{0.927}$\\
        $150$ & $5$ & ELCP & $0.899$ & $0.898$ & $0.899$ & $0.900$ & $0.900$ & $0.901$ & $0.901$ & $0.901$ & $0.901$ & $0.901$ & $0.901$ & $0.901$ & $0.901$\\
        &   & LCP & $0.903$ & $0.904$ & $0.902$ & $0.903$ & $0.903$ & $0.902$ & $0.902$ & $0.902$ & $0.902$ & $0.902$ & $0.903$ & $0.903$ & $0.904$\\
        &   & RLCP & $\textbf{--}$ & $\textbf{--}$ & $\textbf{--}$ & $0.907$ & $0.907$ & $0.908$ & $0.908$ & $0.908$ & $0.908$ & $0.908$ & $0.908$ & $0.908$ & $0.908$\\
        &   & LCP-C & $\textbf{0.937}$ & $\textbf{0.941}$ & $\textbf{0.944}$ & $\textbf{0.945}$ & $\textbf{0.946}$ & $\textbf{0.946}$ & $\textbf{0.946}$ & $\textbf{0.946}$ & $\textbf{0.946}$ & $\textbf{0.946}$ & $\textbf{0.946}$ & $\textbf{0.946}$ & $\textbf{0.946}$\\
        &   & RLCP-C & $\textbf{--}$ & $\textbf{0.939}$ & $\textbf{0.943}$ & $\textbf{0.945}$ & $\textbf{0.946}$ & $\textbf{0.946}$ & $\textbf{0.946}$ & $\textbf{0.947}$ & $\textbf{0.947}$ & $\textbf{0.947}$ & $\textbf{0.947}$ & $\textbf{0.947}$ & $\textbf{0.947}$\\
        & $10$ & ELCP & $\textbf{--}$ & $0.899$ & $0.899$ & $0.896$ & $0.896$ & $0.898$ & $0.897$ & $0.897$ & $0.898$ & $0.899$ & $0.899$ & $0.899$ & $0.899$\\
        &   & LCP & $\textbf{--}$ & $0.898$ & $0.897$ & $0.899$ & $0.899$ & $0.899$ & $0.898$ & $0.900$ & $0.900$ & $0.900$ & $0.900$ & $0.900$ & $0.900$\\
        &   & RLCP & $\textbf{--}$ & $\textbf{--}$ & $\textbf{--}$ & $\textbf{--}$ & $\textbf{--}$ & $0.904$ & $0.904$ & $0.905$ & $0.905$ & $0.905$ & $0.905$ & $0.906$ & $0.906$\\
        &   & LCP-C & $\textbf{--}$ & $\textbf{0.935}$ & $\textbf{0.937}$ & $\textbf{0.939}$ & $\textbf{0.941}$ & $\textbf{0.942}$ & $\textbf{0.942}$ & $\textbf{0.942}$ & $\textbf{0.942}$ & $\textbf{0.941}$ & $\textbf{0.941}$ & $\textbf{0.941}$ & $\textbf{0.941}$\\
        &   & RLCP-C & $\textbf{--}$ & $\textbf{--}$ & $\textbf{--}$ & $\textbf{0.931}$ & $\textbf{0.936}$ & $\textbf{0.939}$ & $\textbf{0.940}$ & $\textbf{0.941}$ & $\textbf{0.941}$ & $\textbf{0.941}$ & $\textbf{0.942}$ & $\textbf{0.942}$ & $\textbf{0.942}$\\
        & $15$ & ELCP & $\textbf{--}$ & $\textbf{--}$ & $0.903$ & $0.902$ & $0.901$ & $0.901$ & $0.902$ & $0.901$ & $0.901$ & $0.900$ & $0.900$ & $0.900$ & $0.900$\\
        &   & LCP & $\textbf{--}$ & $\textbf{--}$ & $\textbf{--}$ & $0.901$ & $0.901$ & $0.901$ & $0.901$ & $0.901$ & $0.901$ & $0.901$ & $0.901$ & $0.902$ & $0.902$\\
        &   & RLCP & $\textbf{--}$ & $\textbf{--}$ & $\textbf{--}$ & $\textbf{--}$ & $\textbf{--}$ & $\textbf{--}$ & $\textbf{--}$ & $0.903$ & $0.904$ & $0.905$ & $0.906$ & $0.906$ & $0.906$\\
        &   & LCP-C & $\textbf{--}$ & $\textbf{--}$ & $\textbf{0.932}$ & $\textbf{0.934}$ & $\textbf{0.937}$ & $\textbf{0.940}$ & $\textbf{0.940}$ & $\textbf{0.941}$ & $\textbf{0.941}$ & $\textbf{0.941}$ & $\textbf{0.941}$ & $\textbf{0.941}$ & $\textbf{0.941}$\\
        &   & RLCP-C & $\textbf{--}$ & $\textbf{--}$ & $\textbf{--}$ & $\textbf{--}$ & $\textbf{--}$ & $\textbf{0.932}$ & $\textbf{0.936}$ & $\textbf{0.938}$ & $\textbf{0.939}$ & $\textbf{0.941}$ & $\textbf{0.941}$ & $\textbf{0.942}$ & $\textbf{0.942}$\\
        & $20$ & ELCP & $\textbf{--}$ & $\textbf{--}$ & $\textbf{--}$ & $0.896$ & $0.896$ & $0.897$ & $0.898$ & $0.898$ & $0.899$ & $0.898$ & $0.898$ & $0.897$ & $0.897$\\
        &   & LCP & $\textbf{--}$ & $\textbf{--}$ & $\textbf{--}$ & $\textbf{--}$ & $0.900$ & $0.900$ & $0.899$ & $0.898$ & $0.898$ & $0.898$ & $0.898$ & $0.899$ & $0.899$\\
        &   & RLCP & $\textbf{--}$ & $\textbf{--}$ & $\textbf{--}$ & $\textbf{--}$ & $\textbf{--}$ & $\textbf{--}$ & $\textbf{--}$ & $\textbf{--}$ & $\textbf{--}$ & $0.904$ & $0.904$ & $0.904$ & $0.904$\\
        &   & LCP-C & $\textbf{--}$ & $\textbf{--}$ & $\textbf{--}$ & $\textbf{0.928}$ & $\textbf{0.930}$ & $\textbf{0.933}$ & $\textbf{0.934}$ & $\textbf{0.935}$ & $\textbf{0.935}$ & $\textbf{0.935}$ & $\textbf{0.935}$ & $\textbf{0.935}$ & $\textbf{0.935}$\\
        &   & RLCP-C & $\textbf{--}$ & $\textbf{--}$ & $\textbf{--}$ & $\textbf{--}$ & $\textbf{--}$ & $\textbf{--}$ & $\textbf{0.926}$ & $\textbf{0.929}$ & $\textbf{0.932}$ & $\textbf{0.934}$ & $\textbf{0.935}$ & $\textbf{0.936}$ & $\textbf{0.936}$\\
        $200$ & $5$ & ELCP & $0.899$ & $0.900$ & $0.900$ & $0.901$ & $0.900$ & $0.901$ & $0.901$ & $0.901$ & $0.901$ & $0.901$ & $0.901$ & $0.901$ & $0.901$\\
        &   & LCP & $0.901$ & $0.901$ & $0.901$ & $0.900$ & $0.901$ & $0.902$ & $0.902$ & $0.903$ & $0.903$ & $0.903$ & $0.903$ & $0.903$ & $0.903$\\
        &   & RLCP & $\textbf{--}$ & $\textbf{--}$ & $0.905$ & $0.905$ & $0.906$ & $0.906$ & $0.906$ & $0.906$ & $0.906$ & $0.906$ & $0.906$ & $0.906$ & $0.906$\\
        &   & LCP-C & $\textbf{0.937}$ & $\textbf{0.942}$ & $\textbf{0.944}$ & $\textbf{0.946}$ & $\textbf{0.946}$ & $\textbf{0.946}$ & $\textbf{0.947}$ & $\textbf{0.946}$ & $\textbf{0.947}$ & $\textbf{0.947}$ & $\textbf{0.946}$ & $\textbf{0.947}$ & $\textbf{0.947}$\\
        &   & RLCP-C & $\textbf{--}$ & $\textbf{0.940}$ & $\textbf{0.944}$ & $\textbf{0.945}$ & $\textbf{0.946}$ & $\textbf{0.946}$ & $\textbf{0.946}$ & $\textbf{0.947}$ & $\textbf{0.947}$ & $\textbf{0.947}$ & $\textbf{0.947}$ & $\textbf{0.947}$ & $\textbf{0.947}$\\
        & $10$ & ELCP & $\textbf{--}$ & $0.898$ & $0.896$ & $0.895$ & $0.896$ & $0.897$ & $0.897$ & $0.896$ & $0.896$ & $0.896$ & $0.896$ & $0.896$ & $0.896$\\
        &   & LCP & $\textbf{--}$ & $0.896$ & $0.895$ & $0.895$ & $0.895$ & $0.895$ & $0.895$ & $0.896$ & $0.896$ & $0.898$ & $0.898$ & $0.898$ & $0.898$\\
        &   & RLCP & $\textbf{--}$ & $\textbf{--}$ & $\textbf{--}$ & $\textbf{--}$ & $\textbf{--}$ & $0.900$ & $0.900$ & $0.901$ & $0.901$ & $0.901$ & $0.901$ & $0.901$ & $0.901$\\
        &   & LCP-C & $\textbf{--}$ & $\textbf{0.936}$ & $\textbf{0.938}$ & $\textbf{0.941}$ & $\textbf{0.943}$ & $\textbf{0.944}$ & $\textbf{0.943}$ & $\textbf{0.944}$ & $\textbf{0.943}$ & $\textbf{0.943}$ & $\textbf{0.943}$ & $\textbf{0.943}$ & $\textbf{0.943}$\\
        &   & RLCP-C & $\textbf{--}$ & $\textbf{--}$ & $\textbf{--}$ & $\textbf{0.934}$ & $\textbf{0.939}$ & $\textbf{0.941}$ & $\textbf{0.942}$ & $\textbf{0.943}$ & $\textbf{0.943}$ & $\textbf{0.943}$ & $\textbf{0.943}$ & $\textbf{0.943}$ & $\textbf{0.943}$\\
        & $15$ & ELCP & $\textbf{--}$ & $\textbf{--}$ & $0.900$ & $0.899$ & $0.900$ & $0.900$ & $0.900$ & $0.901$ & $0.901$ & $0.901$ & $0.901$ & $0.901$ & $0.901$\\
        &   & LCP & $\textbf{--}$ & $\textbf{--}$ & $0.900$ & $0.899$ & $0.898$ & $0.900$ & $0.901$ & $0.902$ & $0.902$ & $0.902$ & $0.902$ & $0.902$ & $0.903$\\
        &   & RLCP & $\textbf{--}$ & $\textbf{--}$ & $\textbf{--}$ & $\textbf{--}$ & $\textbf{--}$ & $\textbf{--}$ & $\textbf{--}$ & $0.904$ & $0.904$ & $0.905$ & $0.905$ & $0.905$ & $0.905$\\
        &   & LCP-C & $\textbf{--}$ & $\textbf{--}$ & $\textbf{0.934}$ & $\textbf{0.937}$ & $\textbf{0.939}$ & $\textbf{0.942}$ & $\textbf{0.942}$ & $\textbf{0.943}$ & $\textbf{0.943}$ & $\textbf{0.943}$ & $\textbf{0.943}$ & $\textbf{0.943}$ & $\textbf{0.943}$\\
        &   & RLCP-C & $\textbf{--}$ & $\textbf{--}$ & $\textbf{--}$ & $\textbf{--}$ & $\textbf{--}$ & $\textbf{0.935}$ & $\textbf{0.938}$ & $\textbf{0.940}$ & $\textbf{0.941}$ & $\textbf{0.943}$ & $\textbf{0.943}$ & $\textbf{0.943}$ & $\textbf{0.943}$\\
        & $20$ & ELCP & $\textbf{--}$ & $\textbf{--}$ & $\textbf{--}$ & $0.898$ & $0.899$ & $0.900$ & $0.899$ & $0.900$ & $0.900$ & $0.901$ & $0.901$ & $0.901$ & $0.901$\\
        &   & LCP & $\textbf{--}$ & $\textbf{--}$ & $\textbf{--}$ & $0.901$ & $0.899$ & $0.899$ & $0.899$ & $0.900$ & $0.899$ & $0.901$ & $0.901$ & $0.901$ & $0.902$\\
        &   & RLCP & $\textbf{--}$ & $\textbf{--}$ & $\textbf{--}$ & $\textbf{--}$ & $\textbf{--}$ & $\textbf{--}$ & $\textbf{--}$ & $\textbf{--}$ & $\textbf{--}$ & $0.904$ & $0.905$ & $0.905$ & $0.905$\\
        &   & LCP-C & $\textbf{--}$ & $\textbf{--}$ & $\textbf{--}$ & $\textbf{0.931}$ & $\textbf{0.934}$ & $\textbf{0.936}$ & $\textbf{0.938}$ & $\textbf{0.938}$ & $\textbf{0.938}$ & $\textbf{0.938}$ & $\textbf{0.938}$ & $\textbf{0.938}$ & $\textbf{0.938}$\\
        &   & RLCP-C & $\textbf{--}$ & $\textbf{--}$ & $\textbf{--}$ & $\textbf{--}$ & $\textbf{--}$ & $\textbf{--}$ & $\textbf{0.930}$ & $\textbf{0.933}$ & $\textbf{0.935}$ & $\textbf{0.937}$ & $\textbf{0.938}$ & $\textbf{0.938}$ & $\textbf{0.938}$\\
    \bottomrule
    \end{tabular}}}
    \end{table*}
    \begin{table*}[htbp]
        \centering
        \caption{Marginal coverage rates for DGP2 with $\textbf{--}$ indicates cases with more than $30\%$ infinite prediction intervals.\label{table:mardgp2}}
        {\fontsize{8}{4.5}\selectfont{\setlength{\tabcolsep}{3.5pt}
        \begin{tabular}{cccccccccccccccc}
        \toprule
        $n$ & $d$ & $h$ & $0.4$ & $0.6$ & $0.8$ & $1.0$ & $1.2$ & $1.4$ & $1.6$ & $1.8$ & $2.0$ & $2.5$ & $3.0$ & $3.5$ & $4.0$ \vspace{2pt} \\ \midrule
        $100$ & $5$ & ELCP & $0.892$ & $0.890$ & $0.886$ & $0.886$ & $0.889$ & $0.891$ & $0.892$ & $0.892$ & $0.892$ & $0.893$ & $0.894$ & $0.894$ & $0.894$\\
        &   & LCP & $0.897$ & $0.896$ & $0.896$ & $0.894$ & $0.894$ & $0.894$ & $0.893$ & $0.893$ & $0.894$ & $0.894$ & $0.896$ & $0.897$ & $0.897$\\
        &   & RLCP & $\textbf{--}$ & $\textbf{--}$ & $\textbf{--}$ & $0.904$ & $0.904$ & $0.904$ & $0.904$ & $0.904$ & $0.904$ & $0.904$ & $0.904$ & $0.904$ & $0.904$\\
        &   & LCP-C & $\textbf{0.920}$ & $\textbf{0.923}$ & $\textbf{0.927}$ & $\textbf{0.928}$ & $\textbf{0.928}$ & $\textbf{0.928}$ & $\textbf{0.928}$ & $\textbf{0.927}$ & $\textbf{0.926}$ & $\textbf{0.925}$ & $\textbf{0.924}$ & $\textbf{0.923}$ & $\textbf{0.923}$\\
        &   & RLCP-C & $\textbf{--}$ & $\textbf{0.923}$ & $\textbf{0.925}$ & $\textbf{0.926}$ & $\textbf{0.926}$ & $\textbf{0.926}$ & $\textbf{0.926}$ & $\textbf{0.925}$ & $\textbf{0.925}$ & $\textbf{0.924}$ & $\textbf{0.924}$ & $\textbf{0.924}$ & $\textbf{0.924}$\\
        & $10$ & ELCP & $\textbf{--}$ & $0.900$ & $0.898$ & $0.900$ & $0.902$ & $0.901$ & $0.903$ & $0.905$ & $0.905$ & $0.907$ & $0.909$ & $0.909$ & $0.909$\\
        &   & LCP & $\textbf{--}$ & $\textbf{--}$ & $0.907$ & $0.908$ & $0.908$ & $0.908$ & $0.909$ & $0.909$ & $0.909$ & $0.909$ & $0.910$ & $0.911$ & $0.912$\\
        &   & RLCP & $\textbf{--}$ & $\textbf{--}$ & $\textbf{--}$ & $\textbf{--}$ & $\textbf{--}$ & $\textbf{--}$ & $0.912$ & $0.913$ & $0.914$ & $0.915$ & $0.916$ & $0.916$ & $0.916$\\
        &   & LCP-C & $\textbf{--}$ & $\textbf{0.912}$ & $\textbf{0.914}$ & $\textbf{0.917}$ & $\textbf{0.922}$ & $\textbf{0.925}$ & $\textbf{0.927}$ & $\textbf{0.928}$ & $\textbf{0.929}$ & $\textbf{0.930}$ & $\textbf{0.930}$ & $\textbf{0.930}$ & $\textbf{0.929}$\\
        &   & RLCP-C & $\textbf{--}$ & $\textbf{--}$ & $\textbf{--}$ & $\textbf{--}$ & $\textbf{0.920}$ & $\textbf{0.923}$ & $\textbf{0.925}$ & $\textbf{0.926}$ & $\textbf{0.927}$ & $\textbf{0.928}$ & $\textbf{0.929}$ & $\textbf{0.929}$ & $\textbf{0.930}$\\
        & $15$ & ELCP & $\textbf{--}$ & $\textbf{--}$ & $0.899$ & $0.900$ & $0.902$ & $0.905$ & $0.905$ & $0.905$ & $0.904$ & $0.904$ & $0.903$ & $0.904$ & $0.905$\\
        &   & LCP & $\textbf{--}$ & $\textbf{--}$ & $\textbf{--}$ & $0.907$ & $0.908$ & $0.908$ & $0.908$ & $0.908$ & $0.908$ & $0.908$ & $0.908$ & $0.907$ & $0.907$\\
        &   & RLCP & $\textbf{--}$ & $\textbf{--}$ & $\textbf{--}$ & $\textbf{--}$ & $\textbf{--}$ & $\textbf{--}$ & $\textbf{--}$ & $\textbf{--}$ & $0.911$ & $0.913$ & $0.914$ & $0.914$ & $0.914$\\
        &   & LCP-C & $\textbf{--}$ & $\textbf{--}$ & $\textbf{0.909}$ & $\textbf{0.912}$ & $\textbf{0.914}$ & $\textbf{0.917}$ & $\textbf{0.919}$ & $\textbf{0.920}$ & $\textbf{0.921}$ & $\textbf{0.922}$ & $\textbf{0.923}$ & $\textbf{0.923}$ & $\textbf{0.923}$\\
        &   & RLCP-C & $\textbf{--}$ & $\textbf{--}$ & $\textbf{--}$ & $\textbf{--}$ & $\textbf{--}$ & $\textbf{--}$ & $\textbf{0.915}$ & $\textbf{0.917}$ & $\textbf{0.919}$ & $\textbf{0.921}$ & $\textbf{0.922}$ & $\textbf{0.923}$ & $\textbf{0.923}$\\
        & $20$ & ELCP & $\textbf{--}$ & $\textbf{--}$ & $\textbf{--}$ & $0.899$ & $0.900$ & $0.898$ & $0.898$ & $0.900$ & $0.900$ & $0.899$ & $0.898$ & $0.898$ & $0.898$\\
        &   & LCP & $\textbf{--}$ & $\textbf{--}$ & $\textbf{--}$ & $\textbf{--}$ & $0.900$ & $0.900$ & $0.900$ & $0.900$ & $0.901$ & $0.901$ & $0.899$ & $0.899$ & $0.899$\\
        &   & RLCP & $\textbf{--}$ & $\textbf{--}$ & $\textbf{--}$ & $\textbf{--}$ & $\textbf{--}$ & $\textbf{--}$ & $\textbf{--}$ & $\textbf{--}$ & $\textbf{--}$ & $0.905$ & $0.906$ & $0.907$ & $0.908$\\
        &   & LCP-C & $\textbf{--}$ & $\textbf{--}$ & $\textbf{--}$ & $\textbf{0.902}$ & $\textbf{0.903}$ & $\textbf{0.906}$ & $\textbf{0.908}$ & $\textbf{0.910}$ & $\textbf{0.912}$ & $\textbf{0.915}$ & $\textbf{0.917}$ & $\textbf{0.917}$ & $\textbf{0.917}$\\
        &   & RLCP-C & $\textbf{--}$ & $\textbf{--}$ & $\textbf{--}$ & $\textbf{--}$ & $\textbf{--}$ & $\textbf{--}$ & $\textbf{--}$ & $\textbf{0.908}$ & $\textbf{0.910}$ & $\textbf{0.913}$ & $\textbf{0.915}$ & $\textbf{0.917}$ & $\textbf{0.917}$\\
        $150$ & $5$ & ELCP & $0.895$ & $0.897$ & $0.895$ & $0.894$ & $0.893$ & $0.893$ & $0.893$ & $0.892$ & $0.892$ & $0.891$ & $0.890$ & $0.890$ & $0.890$\\
        &   & LCP & $0.899$ & $0.900$ & $0.899$ & $0.897$ & $0.895$ & $0.893$ & $0.893$ & $0.893$ & $0.893$ & $0.892$ & $0.892$ & $0.892$ & $0.892$\\
        &   & RLCP & $\textbf{--}$ & $\textbf{--}$ & $\textbf{--}$ & $0.904$ & $0.903$ & $0.903$ & $0.902$ & $0.901$ & $0.901$ & $0.900$ & $0.899$ & $0.899$ & $0.898$\\
        &   & LCP-C & $\textbf{0.925}$ & $\textbf{0.929}$ & $\textbf{0.932}$ & $\textbf{0.932}$ & $\textbf{0.932}$ & $\textbf{0.931}$ & $\textbf{0.930}$ & $\textbf{0.929}$ & $\textbf{0.928}$ & $\textbf{0.926}$ & $\textbf{0.925}$ & $\textbf{0.924}$ & $\textbf{0.924}$\\
        &   & RLCP-C & $\textbf{--}$ & $\textbf{0.927}$ & $\textbf{0.929}$ & $\textbf{0.929}$ & $\textbf{0.928}$ & $\textbf{0.928}$ & $\textbf{0.927}$ & $\textbf{0.927}$ & $\textbf{0.926}$ & $\textbf{0.925}$ & $\textbf{0.925}$ & $\textbf{0.925}$ & $\textbf{0.924}$\\
        & $10$ & ELCP & $\textbf{--}$ & $0.900$ & $0.899$ & $0.898$ & $0.896$ & $0.898$ & $0.901$ & $0.902$ & $0.903$ & $0.904$ & $0.903$ & $0.903$ & $0.904$\\
        &   & LCP & $\textbf{--}$ & $0.903$ & $0.904$ & $0.905$ & $0.905$ & $0.905$ & $0.904$ & $0.903$ & $0.904$ & $0.905$ & $0.906$ & $0.907$ & $0.907$\\
        &   & RLCP & $\textbf{--}$ & $\textbf{--}$ & $\textbf{--}$ & $\textbf{--}$ & $\textbf{--}$ & $0.908$ & $0.909$ & $0.909$ & $0.910$ & $0.910$ & $0.910$ & $0.910$ & $0.910$\\
        &   & LCP-C & $\textbf{--}$ & $\textbf{0.920}$ & $\textbf{0.923}$ & $\textbf{0.926}$ & $\textbf{0.930}$ & $\textbf{0.932}$ & $\textbf{0.933}$ & $\textbf{0.934}$ & $\textbf{0.934}$ & $\textbf{0.934}$ & $\textbf{0.933}$ & $\textbf{0.933}$ & $\textbf{0.933}$\\
        &   & RLCP-C & $\textbf{--}$ & $\textbf{--}$ & $\textbf{--}$ & $\textbf{0.921}$ & $\textbf{0.926}$ & $\textbf{0.928}$ & $\textbf{0.930}$ & $\textbf{0.931}$ & $\textbf{0.931}$ & $\textbf{0.932}$ & $\textbf{0.933}$ & $\textbf{0.933}$ & $\textbf{0.933}$\\
        & $15$ & ELCP & $\textbf{--}$ & $\textbf{--}$ & $0.901$ & $0.900$ & $0.901$ & $0.903$ & $0.903$ & $0.903$ & $0.903$ & $0.902$ & $0.903$ & $0.903$ & $0.903$\\
        &   & LCP & $\textbf{--}$ & $\textbf{--}$ & $\textbf{--}$ & $0.902$ & $0.903$ & $0.902$ & $0.903$ & $0.905$ & $0.904$ & $0.904$ & $0.904$ & $0.904$ & $0.905$\\
        &   & RLCP & $\textbf{--}$ & $\textbf{--}$ & $\textbf{--}$ & $\textbf{--}$ & $\textbf{--}$ & $\textbf{--}$ & $\textbf{--}$ & $0.905$ & $0.906$ & $0.908$ & $0.908$ & $0.909$ & $0.909$\\
        &   & LCP-C & $\textbf{--}$ & $\textbf{--}$ & $\textbf{0.915}$ & $\textbf{0.918}$ & $\textbf{0.922}$ & $\textbf{0.925}$ & $\textbf{0.926}$ & $\textbf{0.927}$ & $\textbf{0.927}$ & $\textbf{0.927}$ & $\textbf{0.926}$ & $\textbf{0.926}$ & $\textbf{0.926}$\\
        &   & RLCP-C & $\textbf{--}$ & $\textbf{--}$ & $\textbf{--}$ & $\textbf{--}$ & $\textbf{--}$ & $\textbf{0.918}$ & $\textbf{0.921}$ & $\textbf{0.922}$ & $\textbf{0.923}$ & $\textbf{0.925}$ & $\textbf{0.926}$ & $\textbf{0.926}$ & $\textbf{0.926}$\\
        & $20$ & ELCP & $\textbf{--}$ & $\textbf{--}$ & $\textbf{--}$ & $0.899$ & $0.900$ & $0.901$ & $0.903$ & $0.901$ & $0.901$ & $0.902$ & $0.901$ & $0.902$ & $0.902$\\
        &   & LCP & $\textbf{--}$ & $\textbf{--}$ & $\textbf{--}$ & $\textbf{--}$ & $0.905$ & $0.905$ & $0.905$ & $0.905$ & $0.904$ & $0.904$ & $0.903$ & $0.904$ & $0.904$\\
        &   & RLCP & $\textbf{--}$ & $\textbf{--}$ & $\textbf{--}$ & $\textbf{--}$ & $\textbf{--}$ & $\textbf{--}$ & $\textbf{--}$ & $\textbf{--}$ & $\textbf{--}$ & $0.906$ & $0.907$ & $0.908$ & $0.908$\\
        &   & LCP-C & $\textbf{--}$ & $\textbf{--}$ & $\textbf{--}$ & $\textbf{0.913}$ & $\textbf{0.916}$ & $\textbf{0.919}$ & $\textbf{0.922}$ & $\textbf{0.924}$ & $\textbf{0.924}$ & $\textbf{0.925}$ & $\textbf{0.926}$ & $\textbf{0.925}$ & $\textbf{0.925}$\\
        &   & RLCP-C & $\textbf{--}$ & $\textbf{--}$ & $\textbf{--}$ & $\textbf{--}$ & $\textbf{--}$ & $\textbf{--}$ & $\textbf{0.914}$ & $\textbf{0.917}$ & $\textbf{0.919}$ & $\textbf{0.922}$ & $\textbf{0.924}$ & $\textbf{0.925}$ & $\textbf{0.925}$\\
        $200$ & $5$ & ELCP & $0.896$ & $0.897$ & $0.897$ & $0.896$ & $0.896$ & $0.895$ & $0.895$ & $0.896$ & $0.896$ & $0.895$ & $0.895$ & $0.895$ & $0.895$\\
        &   & LCP & $0.898$ & $0.898$ & $0.898$ & $0.898$ & $0.897$ & $0.898$ & $0.897$ & $0.897$ & $0.896$ & $0.896$ & $0.896$ & $0.896$ & $0.896$\\
        &   & RLCP & $\textbf{--}$ & $\textbf{--}$ & $0.904$ & $0.903$ & $0.903$ & $0.903$ & $0.902$ & $0.902$ & $0.901$ & $0.900$ & $0.900$ & $0.900$ & $0.900$\\
        &   & LCP-C & $\textbf{0.929}$ & $\textbf{0.933}$ & $\textbf{0.935}$ & $\textbf{0.936}$ & $\textbf{0.934}$ & $\textbf{0.933}$ & $\textbf{0.932}$ & $\textbf{0.930}$ & $\textbf{0.929}$ & $\textbf{0.927}$ & $\textbf{0.926}$ & $\textbf{0.925}$ & $\textbf{0.925}$\\
        &   & RLCP-C & $\textbf{--}$ & $\textbf{0.930}$ & $\textbf{0.931}$ & $\textbf{0.931}$ & $\textbf{0.930}$ & $\textbf{0.929}$ & $\textbf{0.928}$ & $\textbf{0.927}$ & $\textbf{0.927}$ & $\textbf{0.926}$ & $\textbf{0.925}$ & $\textbf{0.925}$ & $\textbf{0.925}$\\
        & $10$ & ELCP & $\textbf{--}$ & $0.899$ & $0.901$ & $0.901$ & $0.901$ & $0.903$ & $0.904$ & $0.905$ & $0.906$ & $0.906$ & $0.907$ & $0.907$ & $0.906$\\
        &   & LCP & $\textbf{--}$ & $0.903$ & $0.904$ & $0.903$ & $0.904$ & $0.904$ & $0.905$ & $0.906$ & $0.906$ & $0.907$ & $0.907$ & $0.907$ & $0.907$\\
        &   & RLCP & $\textbf{--}$ & $\textbf{--}$ & $\textbf{--}$ & $\textbf{--}$ & $\textbf{--}$ & $0.908$ & $0.909$ & $0.909$ & $0.910$ & $0.910$ & $0.910$ & $0.911$ & $0.911$\\
        &   & LCP-C & $\textbf{--}$ & $\textbf{0.924}$ & $\textbf{0.927}$ & $\textbf{0.931}$ & $\textbf{0.935}$ & $\textbf{0.936}$ & $\textbf{0.936}$ & $\textbf{0.936}$ & $\textbf{0.936}$ & $\textbf{0.935}$ & $\textbf{0.934}$ & $\textbf{0.934}$ & $\textbf{0.934}$\\
        &   & RLCP-C & $\textbf{--}$ & $\textbf{--}$ & $\textbf{--}$ & $\textbf{0.925}$ & $\textbf{0.929}$ & $\textbf{0.931}$ & $\textbf{0.932}$ & $\textbf{0.933}$ & $\textbf{0.933}$ & $\textbf{0.934}$ & $\textbf{0.934}$ & $\textbf{0.934}$ & $\textbf{0.934}$\\
        & $15$ & ELCP & $\textbf{--}$ & $\textbf{--}$ & $0.897$ & $0.897$ & $0.898$ & $0.898$ & $0.900$ & $0.901$ & $0.901$ & $0.902$ & $0.901$ & $0.902$ & $0.902$\\
        &   & LCP & $\textbf{--}$ & $\textbf{--}$ & $0.898$ & $0.899$ & $0.900$ & $0.900$ & $0.902$ & $0.901$ & $0.901$ & $0.901$ & $0.902$ & $0.903$ & $0.902$\\
        &   & RLCP & $\textbf{--}$ & $\textbf{--}$ & $\textbf{--}$ & $\textbf{--}$ & $\textbf{--}$ & $\textbf{--}$ & $\textbf{--}$ & $0.904$ & $0.904$ & $0.905$ & $0.906$ & $0.906$ & $0.906$\\
        &   & LCP-C & $\textbf{--}$ & $\textbf{--}$ & $\textbf{0.919}$ & $\textbf{0.922}$ & $\textbf{0.925}$ & $\textbf{0.928}$ & $\textbf{0.929}$ & $\textbf{0.930}$ & $\textbf{0.930}$ & $\textbf{0.930}$ & $\textbf{0.929}$ & $\textbf{0.929}$ & $\textbf{0.928}$\\
        &   & RLCP-C & $\textbf{--}$ & $\textbf{--}$ & $\textbf{--}$ & $\textbf{--}$ & $\textbf{--}$ & $\textbf{0.922}$ & $\textbf{0.924}$ & $\textbf{0.925}$ & $\textbf{0.926}$ & $\textbf{0.927}$ & $\textbf{0.928}$ & $\textbf{0.928}$ & $\textbf{0.928}$\\
        & $20$ & ELCP & $\textbf{--}$ & $\textbf{--}$ & $\textbf{--}$ & $0.900$ & $0.900$ & $0.900$ & $0.900$ & $0.901$ & $0.901$ & $0.900$ & $0.899$ & $0.899$ & $0.899$\\
        &   & LCP & $\textbf{--}$ & $\textbf{--}$ & $\textbf{--}$ & $0.900$ & $0.902$ & $0.901$ & $0.899$ & $0.899$ & $0.898$ & $0.901$ & $0.900$ & $0.900$ & $0.900$\\
        &   & RLCP & $\textbf{--}$ & $\textbf{--}$ & $\textbf{--}$ & $\textbf{--}$ & $\textbf{--}$ & $\textbf{--}$ & $\textbf{--}$ & $\textbf{--}$ & $\textbf{--}$ & $0.904$ & $0.905$ & $0.905$ & $0.905$\\
        &   & LCP-C & $\textbf{--}$ & $\textbf{--}$ & $\textbf{--}$ & $\textbf{0.918}$ & $\textbf{0.921}$ & $\textbf{0.924}$ & $\textbf{0.926}$ & $\textbf{0.927}$ & $\textbf{0.928}$ & $\textbf{0.928}$ & $\textbf{0.927}$ & $\textbf{0.927}$ & $\textbf{0.927}$\\
        &   & RLCP-C & $\textbf{--}$ & $\textbf{--}$ & $\textbf{--}$ & $\textbf{--}$ & $\textbf{--}$ & $\textbf{--}$ & $\textbf{0.918}$ & $\textbf{0.921}$ & $\textbf{0.923}$ & $\textbf{0.925}$ & $\textbf{0.926}$ & $\textbf{0.926}$ & $\textbf{0.927}$\\
    \bottomrule
    \end{tabular}}}
    \end{table*}
    \begin{table*}[htbp]
        \centering
        \caption{Marginal coverage rates for DGP3 with $\textbf{--}$ indicates cases with more than $30\%$ infinite prediction intervals.\label{table:mardgp3}}
        {\fontsize{8}{4.5}\selectfont{\setlength{\tabcolsep}{3.5pt}
        \begin{tabular}{cccccccccccccccc}
        \toprule
        $n$ & $d$ & $h$ & $0.4$ & $0.6$ & $0.8$ & $1.0$ & $1.2$ & $1.4$ & $1.6$ & $1.8$ & $2.0$ & $2.5$ & $3.0$ & $3.5$ & $4.0$ \vspace{2pt} \\ \midrule
        $100$ & $5$ & ELCP & $0.894$ & $0.891$ & $0.891$ & $0.892$ & $0.893$ & $0.895$ & $0.895$ & $0.896$ & $0.896$ & $0.897$ & $0.897$ & $0.897$ & $0.897$\\
        &   & LCP & $0.902$ & $0.901$ & $0.900$ & $0.897$ & $0.896$ & $0.895$ & $0.895$ & $0.897$ & $0.898$ & $0.899$ & $0.900$ & $0.900$ & $0.901$\\
        &   & RLCP & $\textbf{--}$ & $\textbf{--}$ & $\textbf{--}$ & $0.905$ & $0.905$ & $0.905$ & $0.905$ & $0.905$ & $0.905$ & $0.905$ & $0.905$ & $0.905$ & $0.905$\\
        &   & LCP-C & $\textbf{0.918}$ & $\textbf{0.923}$ & $\textbf{0.927}$ & $\textbf{0.929}$ & $\textbf{0.928}$ & $\textbf{0.927}$ & $\textbf{0.927}$ & $\textbf{0.926}$ & $\textbf{0.925}$ & $\textbf{0.924}$ & $\textbf{0.924}$ & $\textbf{0.924}$ & $\textbf{0.924}$\\
        &   & RLCP-C & $\textbf{--}$ & $\textbf{0.922}$ & $\textbf{0.925}$ & $\textbf{0.926}$ & $\textbf{0.926}$ & $\textbf{0.926}$ & $\textbf{0.926}$ & $\textbf{0.925}$ & $\textbf{0.925}$ & $\textbf{0.925}$ & $\textbf{0.925}$ & $\textbf{0.925}$ & $\textbf{0.925}$\\
        & $10$ & ELCP & $\textbf{--}$ & $0.901$ & $0.902$ & $0.900$ & $0.901$ & $0.901$ & $0.901$ & $0.901$ & $0.902$ & $0.902$ & $0.902$ & $0.902$ & $0.902$\\
        &   & LCP & $\textbf{--}$ & $\textbf{--}$ & $0.906$ & $0.906$ & $0.906$ & $0.906$ & $0.906$ & $0.906$ & $0.905$ & $0.904$ & $0.904$ & $0.904$ & $0.904$\\
        &   & RLCP & $\textbf{--}$ & $\textbf{--}$ & $\textbf{--}$ & $\textbf{--}$ & $\textbf{--}$ & $\textbf{--}$ & $0.909$ & $0.909$ & $0.910$ & $0.910$ & $0.910$ & $0.911$ & $0.910$\\
        &   & LCP-C & $\textbf{--}$ & $\textbf{0.915}$ & $\textbf{0.918}$ & $\textbf{0.921}$ & $\textbf{0.923}$ & $\textbf{0.924}$ & $\textbf{0.924}$ & $\textbf{0.923}$ & $\textbf{0.923}$ & $\textbf{0.922}$ & $\textbf{0.922}$ & $\textbf{0.922}$ & $\textbf{0.922}$\\
        &   & RLCP-C & $\textbf{--}$ & $\textbf{--}$ & $\textbf{--}$ & $\textbf{--}$ & $\textbf{0.919}$ & $\textbf{0.921}$ & $\textbf{0.922}$ & $\textbf{0.923}$ & $\textbf{0.923}$ & $\textbf{0.923}$ & $\textbf{0.923}$ & $\textbf{0.923}$ & $\textbf{0.923}$\\
        & $15$ & ELCP & $\textbf{--}$ & $\textbf{--}$ & $0.899$ & $0.899$ & $0.899$ & $0.901$ & $0.900$ & $0.900$ & $0.900$ & $0.900$ & $0.900$ & $0.900$ & $0.901$\\
        &   & LCP & $\textbf{--}$ & $\textbf{--}$ & $\textbf{--}$ & $0.904$ & $0.904$ & $0.904$ & $0.903$ & $0.901$ & $0.900$ & $0.900$ & $0.901$ & $0.902$ & $0.904$\\
        &   & RLCP & $\textbf{--}$ & $\textbf{--}$ & $\textbf{--}$ & $\textbf{--}$ & $\textbf{--}$ & $\textbf{--}$ & $\textbf{--}$ & $\textbf{--}$ & $0.907$ & $0.908$ & $0.908$ & $0.909$ & $0.909$\\
        &   & LCP-C & $\textbf{--}$ & $\textbf{--}$ & $\textbf{0.914}$ & $\textbf{0.917}$ & $\textbf{0.919}$ & $\textbf{0.920}$ & $\textbf{0.921}$ & $\textbf{0.922}$ & $\textbf{0.922}$ & $\textbf{0.922}$ & $\textbf{0.922}$ & $\textbf{0.922}$ & $\textbf{0.922}$\\
        &   & RLCP-C & $\textbf{--}$ & $\textbf{--}$ & $\textbf{--}$ & $\textbf{--}$ & $\textbf{--}$ & $\textbf{--}$ & $\textbf{0.916}$ & $\textbf{0.919}$ & $\textbf{0.920}$ & $\textbf{0.922}$ & $\textbf{0.923}$ & $\textbf{0.923}$ & $\textbf{0.923}$\\
        & $20$ & ELCP & $\textbf{--}$ & $\textbf{--}$ & $\textbf{--}$ & $0.899$ & $0.898$ & $0.896$ & $0.897$ & $0.896$ & $0.894$ & $0.895$ & $0.896$ & $0.896$ & $0.896$\\
        &   & LCP & $\textbf{--}$ & $\textbf{--}$ & $\textbf{--}$ & $\textbf{--}$ & $0.899$ & $0.898$ & $0.897$ & $0.897$ & $0.898$ & $0.898$ & $0.899$ & $0.899$ & $0.900$\\
        &   & RLCP & $\textbf{--}$ & $\textbf{--}$ & $\textbf{--}$ & $\textbf{--}$ & $\textbf{--}$ & $\textbf{--}$ & $\textbf{--}$ & $\textbf{--}$ & $\textbf{--}$ & $0.904$ & $0.905$ & $0.906$ & $0.906$\\
        &   & LCP-C & $\textbf{--}$ & $\textbf{--}$ & $\textbf{--}$ & $\textbf{0.905}$ & $\textbf{0.907}$ & $\textbf{0.910}$ & $\textbf{0.912}$ & $\textbf{0.913}$ & $\textbf{0.914}$ & $\textbf{0.915}$ & $\textbf{0.914}$ & $\textbf{0.915}$ & $\textbf{0.915}$\\
        &   & RLCP-C & $\textbf{--}$ & $\textbf{--}$ & $\textbf{--}$ & $\textbf{--}$ & $\textbf{--}$ & $\textbf{--}$ & $\textbf{--}$ & $\textbf{0.906}$ & $\textbf{0.909}$ & $\textbf{0.913}$ & $\textbf{0.915}$ & $\textbf{0.915}$ & $\textbf{0.916}$\\
        $150$ & $5$ & ELCP & $0.899$ & $0.898$ & $0.894$ & $0.893$ & $0.893$ & $0.893$ & $0.893$ & $0.894$ & $0.894$ & $0.893$ & $0.893$ & $0.893$ & $0.893$\\
        &   & LCP & $0.896$ & $0.896$ & $0.896$ & $0.895$ & $0.893$ & $0.892$ & $0.892$ & $0.893$ & $0.894$ & $0.894$ & $0.894$ & $0.895$ & $0.895$\\
        &   & RLCP & $\textbf{--}$ & $\textbf{--}$ & $\textbf{--}$ & $0.903$ & $0.902$ & $0.902$ & $0.901$ & $0.901$ & $0.901$ & $0.900$ & $0.900$ & $0.900$ & $0.899$\\
        &   & LCP-C & $\textbf{0.924}$ & $\textbf{0.929}$ & $\textbf{0.931}$ & $\textbf{0.932}$ & $\textbf{0.932}$ & $\textbf{0.930}$ & $\textbf{0.929}$ & $\textbf{0.928}$ & $\textbf{0.927}$ & $\textbf{0.926}$ & $\textbf{0.925}$ & $\textbf{0.925}$ & $\textbf{0.925}$\\
        &   & RLCP-C & $\textbf{--}$ & $\textbf{0.925}$ & $\textbf{0.928}$ & $\textbf{0.928}$ & $\textbf{0.928}$ & $\textbf{0.928}$ & $\textbf{0.927}$ & $\textbf{0.927}$ & $\textbf{0.927}$ & $\textbf{0.926}$ & $\textbf{0.926}$ & $\textbf{0.926}$ & $\textbf{0.926}$\\
        & $10$ & ELCP & $\textbf{--}$ & $0.898$ & $0.898$ & $0.898$ & $0.898$ & $0.897$ & $0.897$ & $0.896$ & $0.897$ & $0.898$ & $0.898$ & $0.898$ & $0.897$\\
        &   & LCP & $\textbf{--}$ & $0.903$ & $0.904$ & $0.904$ & $0.904$ & $0.902$ & $0.901$ & $0.900$ & $0.900$ & $0.900$ & $0.899$ & $0.899$ & $0.899$\\
        &   & RLCP & $\textbf{--}$ & $\textbf{--}$ & $\textbf{--}$ & $\textbf{--}$ & $\textbf{--}$ & $0.904$ & $0.904$ & $0.904$ & $0.904$ & $0.904$ & $0.904$ & $0.904$ & $0.904$\\
        &   & LCP-C & $\textbf{--}$ & $\textbf{0.921}$ & $\textbf{0.923}$ & $\textbf{0.926}$ & $\textbf{0.928}$ & $\textbf{0.929}$ & $\textbf{0.928}$ & $\textbf{0.928}$ & $\textbf{0.927}$ & $\textbf{0.926}$ & $\textbf{0.926}$ & $\textbf{0.926}$ & $\textbf{0.926}$\\
        &   & RLCP-C & $\textbf{--}$ & $\textbf{--}$ & $\textbf{--}$ & $\textbf{0.918}$ & $\textbf{0.922}$ & $\textbf{0.924}$ & $\textbf{0.925}$ & $\textbf{0.926}$ & $\textbf{0.926}$ & $\textbf{0.926}$ & $\textbf{0.926}$ & $\textbf{0.926}$ & $\textbf{0.926}$\\
        & $15$ & ELCP & $\textbf{--}$ & $\textbf{--}$ & $0.901$ & $0.901$ & $0.900$ & $0.900$ & $0.900$ & $0.900$ & $0.901$ & $0.901$ & $0.900$ & $0.900$ & $0.900$\\
        &   & LCP & $\textbf{--}$ & $\textbf{--}$ & $\textbf{--}$ & $0.899$ & $0.900$ & $0.900$ & $0.901$ & $0.901$ & $0.900$ & $0.900$ & $0.901$ & $0.901$ & $0.901$\\
        &   & RLCP & $\textbf{--}$ & $\textbf{--}$ & $\textbf{--}$ & $\textbf{--}$ & $\textbf{--}$ & $\textbf{--}$ & $\textbf{--}$ & $0.903$ & $0.904$ & $0.905$ & $0.905$ & $0.906$ & $0.906$\\
        &   & LCP-C & $\textbf{--}$ & $\textbf{--}$ & $\textbf{0.919}$ & $\textbf{0.922}$ & $\textbf{0.925}$ & $\textbf{0.927}$ & $\textbf{0.928}$ & $\textbf{0.928}$ & $\textbf{0.928}$ & $\textbf{0.927}$ & $\textbf{0.927}$ & $\textbf{0.927}$ & $\textbf{0.927}$\\
        &   & RLCP-C & $\textbf{--}$ & $\textbf{--}$ & $\textbf{--}$ & $\textbf{--}$ & $\textbf{--}$ & $\textbf{0.918}$ & $\textbf{0.922}$ & $\textbf{0.924}$ & $\textbf{0.925}$ & $\textbf{0.926}$ & $\textbf{0.927}$ & $\textbf{0.927}$ & $\textbf{0.927}$\\
        & $20$ & ELCP & $\textbf{--}$ & $\textbf{--}$ & $\textbf{--}$ & $0.899$ & $0.899$ & $0.899$ & $0.898$ & $0.899$ & $0.899$ & $0.898$ & $0.898$ & $0.899$ & $0.899$\\
        &   & LCP & $\textbf{--}$ & $\textbf{--}$ & $\textbf{--}$ & $\textbf{--}$ & $0.899$ & $0.899$ & $0.899$ & $0.899$ & $0.900$ & $0.899$ & $0.899$ & $0.900$ & $0.900$\\
        &   & RLCP & $\textbf{--}$ & $\textbf{--}$ & $\textbf{--}$ & $\textbf{--}$ & $\textbf{--}$ & $\textbf{--}$ & $\textbf{--}$ & $\textbf{--}$ & $\textbf{--}$ & $0.902$ & $0.903$ & $0.903$ & $0.904$\\
        &   & LCP-C & $\textbf{--}$ & $\textbf{--}$ & $\textbf{--}$ & $\textbf{0.914}$ & $\textbf{0.918}$ & $\textbf{0.920}$ & $\textbf{0.922}$ & $\textbf{0.923}$ & $\textbf{0.924}$ & $\textbf{0.923}$ & $\textbf{0.923}$ & $\textbf{0.923}$ & $\textbf{0.923}$\\
        &   & RLCP-C & $\textbf{--}$ & $\textbf{--}$ & $\textbf{--}$ & $\textbf{--}$ & $\textbf{--}$ & $\textbf{--}$ & $\textbf{0.911}$ & $\textbf{0.915}$ & $\textbf{0.918}$ & $\textbf{0.921}$ & $\textbf{0.922}$ & $\textbf{0.923}$ & $\textbf{0.923}$\\
        $200$ & $5$ & ELCP & $0.896$ & $0.896$ & $0.895$ & $0.895$ & $0.895$ & $0.895$ & $0.894$ & $0.894$ & $0.894$ & $0.894$ & $0.894$ & $0.894$ & $0.894$\\
        &   & LCP & $0.895$ & $0.895$ & $0.896$ & $0.895$ & $0.893$ & $0.894$ & $0.895$ & $0.895$ & $0.895$ & $0.895$ & $0.895$ & $0.896$ & $0.896$\\
        &   & RLCP & $\textbf{--}$ & $\textbf{--}$ & $0.903$ & $0.902$ & $0.901$ & $0.901$ & $0.900$ & $0.900$ & $0.900$ & $0.900$ & $0.899$ & $0.899$ & $0.899$\\
        &   & LCP-C & $\textbf{0.926}$ & $\textbf{0.930}$ & $\textbf{0.932}$ & $\textbf{0.932}$ & $\textbf{0.931}$ & $\textbf{0.930}$ & $\textbf{0.929}$ & $\textbf{0.928}$ & $\textbf{0.927}$ & $\textbf{0.926}$ & $\textbf{0.925}$ & $\textbf{0.925}$ & $\textbf{0.925}$\\
        &   & RLCP-C & $\textbf{--}$ & $\textbf{0.927}$ & $\textbf{0.928}$ & $\textbf{0.928}$ & $\textbf{0.928}$ & $\textbf{0.927}$ & $\textbf{0.927}$ & $\textbf{0.926}$ & $\textbf{0.926}$ & $\textbf{0.926}$ & $\textbf{0.926}$ & $\textbf{0.926}$ & $\textbf{0.926}$\\
        & $10$ & ELCP & $\textbf{--}$ & $0.897$ & $0.896$ & $0.895$ & $0.895$ & $0.894$ & $0.894$ & $0.894$ & $0.894$ & $0.895$ & $0.895$ & $0.895$ & $0.895$\\
        &   & LCP & $\textbf{--}$ & $0.897$ & $0.896$ & $0.896$ & $0.896$ & $0.894$ & $0.895$ & $0.896$ & $0.895$ & $0.895$ & $0.895$ & $0.896$ & $0.896$\\
        &   & RLCP & $\textbf{--}$ & $\textbf{--}$ & $\textbf{--}$ & $\textbf{--}$ & $\textbf{--}$ & $0.901$ & $0.901$ & $0.901$ & $0.901$ & $0.901$ & $0.900$ & $0.900$ & $0.900$\\
        &   & LCP-C & $\textbf{--}$ & $\textbf{0.922}$ & $\textbf{0.925}$ & $\textbf{0.928}$ & $\textbf{0.930}$ & $\textbf{0.930}$ & $\textbf{0.929}$ & $\textbf{0.929}$ & $\textbf{0.928}$ & $\textbf{0.928}$ & $\textbf{0.927}$ & $\textbf{0.927}$ & $\textbf{0.927}$\\
        &   & RLCP-C & $\textbf{--}$ & $\textbf{--}$ & $\textbf{--}$ & $\textbf{0.920}$ & $\textbf{0.924}$ & $\textbf{0.926}$ & $\textbf{0.926}$ & $\textbf{0.927}$ & $\textbf{0.927}$ & $\textbf{0.927}$ & $\textbf{0.927}$ & $\textbf{0.927}$ & $\textbf{0.927}$\\
        & $15$ & ELCP & $\textbf{--}$ & $\textbf{--}$ & $0.899$ & $0.899$ & $0.900$ & $0.899$ & $0.899$ & $0.899$ & $0.899$ & $0.899$ & $0.899$ & $0.899$ & $0.899$\\
        &   & LCP & $\textbf{--}$ & $\textbf{--}$ & $0.897$ & $0.899$ & $0.898$ & $0.898$ & $0.899$ & $0.899$ & $0.899$ & $0.898$ & $0.898$ & $0.898$ & $0.899$\\
        &   & RLCP & $\textbf{--}$ & $\textbf{--}$ & $\textbf{--}$ & $\textbf{--}$ & $\textbf{--}$ & $\textbf{--}$ & $\textbf{--}$ & $0.902$ & $0.902$ & $0.903$ & $0.903$ & $0.903$ & $0.903$\\
        &   & LCP-C & $\textbf{--}$ & $\textbf{--}$ & $\textbf{0.922}$ & $\textbf{0.925}$ & $\textbf{0.929}$ & $\textbf{0.930}$ & $\textbf{0.930}$ & $\textbf{0.929}$ & $\textbf{0.929}$ & $\textbf{0.928}$ & $\textbf{0.928}$ & $\textbf{0.927}$ & $\textbf{0.927}$\\
        &   & RLCP-C & $\textbf{--}$ & $\textbf{--}$ & $\textbf{--}$ & $\textbf{--}$ & $\textbf{--}$ & $\textbf{0.921}$ & $\textbf{0.924}$ & $\textbf{0.926}$ & $\textbf{0.927}$ & $\textbf{0.928}$ & $\textbf{0.928}$ & $\textbf{0.928}$ & $\textbf{0.928}$\\
        & $20$ & ELCP & $\textbf{--}$ & $\textbf{--}$ & $\textbf{--}$ & $0.896$ & $0.895$ & $0.897$ & $0.897$ & $0.897$ & $0.897$ & $0.897$ & $0.897$ & $0.896$ & $0.896$\\
        &   & LCP & $\textbf{--}$ & $\textbf{--}$ & $\textbf{--}$ & $0.897$ & $0.897$ & $0.897$ & $0.897$ & $0.896$ & $0.897$ & $0.897$ & $0.898$ & $0.898$ & $0.897$\\
        &   & RLCP & $\textbf{--}$ & $\textbf{--}$ & $\textbf{--}$ & $\textbf{--}$ & $\textbf{--}$ & $\textbf{--}$ & $\textbf{--}$ & $\textbf{--}$ & $\textbf{--}$ & $0.900$ & $0.901$ & $0.901$ & $0.901$\\
        &   & LCP-C & $\textbf{--}$ & $\textbf{--}$ & $\textbf{--}$ & $\textbf{0.915}$ & $\textbf{0.919}$ & $\textbf{0.921}$ & $\textbf{0.922}$ & $\textbf{0.923}$ & $\textbf{0.923}$ & $\textbf{0.923}$ & $\textbf{0.923}$ & $\textbf{0.923}$ & $\textbf{0.923}$\\
        &   & RLCP-C & $\textbf{--}$ & $\textbf{--}$ & $\textbf{--}$ & $\textbf{--}$ & $\textbf{--}$ & $\textbf{--}$ & $\textbf{0.914}$ & $\textbf{0.917}$ & $\textbf{0.919}$ & $\textbf{0.922}$ & $\textbf{0.922}$ & $\textbf{0.923}$ & $\textbf{0.923}$\\
    \bottomrule
    \end{tabular}}}
    \end{table*}

\FloatBarrier
\subsubsection{Effect of $\omega$ and $h$ on test-conditional miscoverage error of ELCP}\label{sec:supp_simu_parameter1}

    First, we examine the effect of $\omega$ and $h$ on the test-conditional miscoverage error of ELCP. 
    Figure~\ref{fig:Impact of varying w} shows how the error changes with $\omega$ when using the optimal bandwidth $h$ for DGP1--DGP3.
    The results indicate that a broad range of $\omega$ values yields similar best performance under these settings, suggesting that ELCP is relatively insensitive to the choice of $\omega$ as long as it falls within an appropriate range (typically larger values for DGP1--DGP3).
    Furthermore, Figure~\ref{fig:Impact of varying h} illustrates the effect of bandwidth $h$ on the test-conditional miscoverage error, considering only configurations where the proportion of infinite prediction intervals remains below 5\%. For DGP2 and DGP3, smaller $h$ values can reduce the test-conditional miscoverage but simultaneously increase the frequency of infinite intervals, leading to trivial prediction sets with limited practical utility.
    Under a fixed sample size, the optimal bandwidth $h$ typically increases as the dimension grows.
    
    \begin{figure}[h]
        \centering
        \includegraphics[width=.9\linewidth]{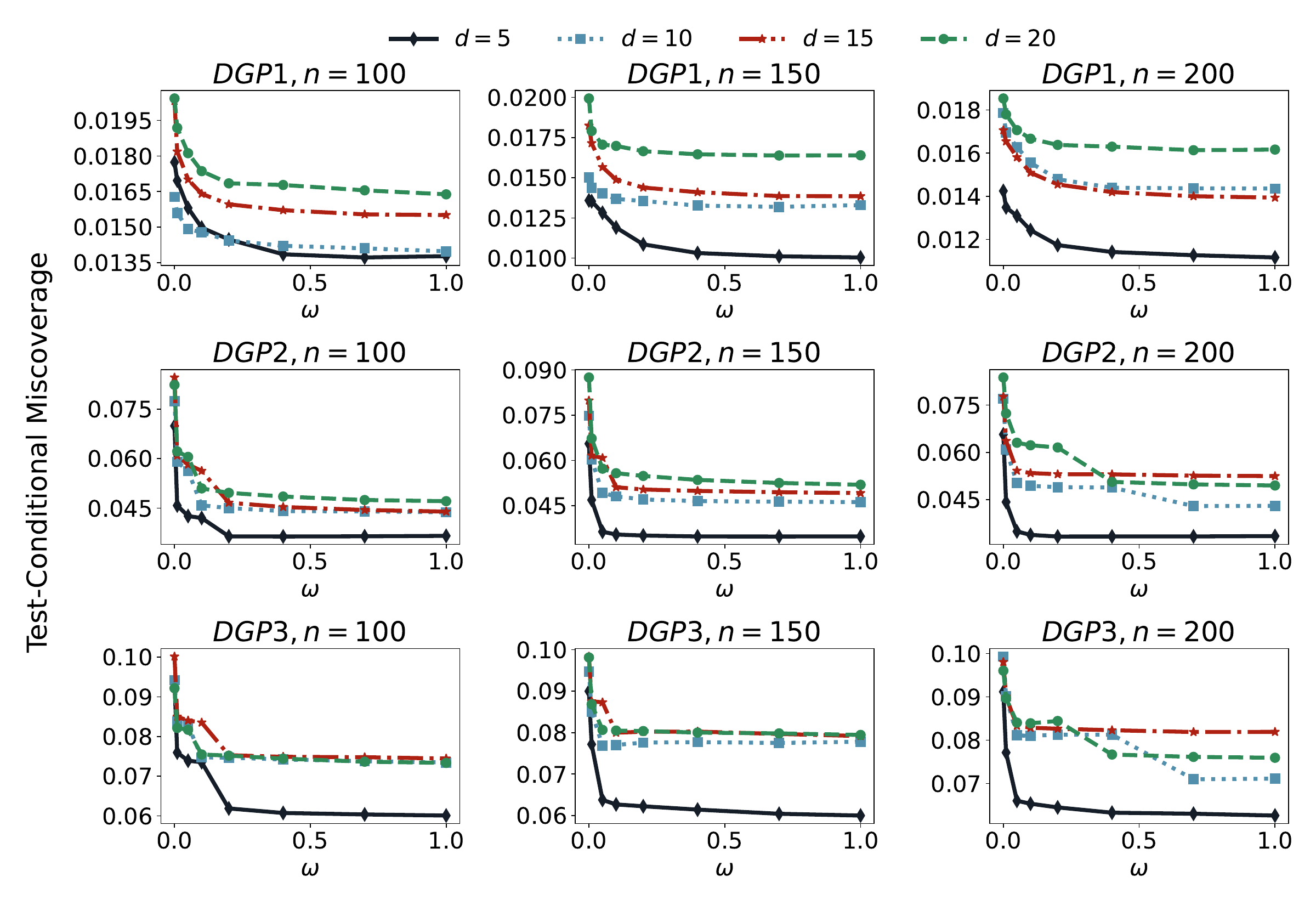}
        \caption{Test-conditional miscoverage of ELCP over $\omega$ under optimal $h$ for DGP1--DGP3.}\label{fig:Impact of varying w}
    \end{figure}
    \begin{figure}[H]
        \centering
        \includegraphics[width=.9\linewidth]{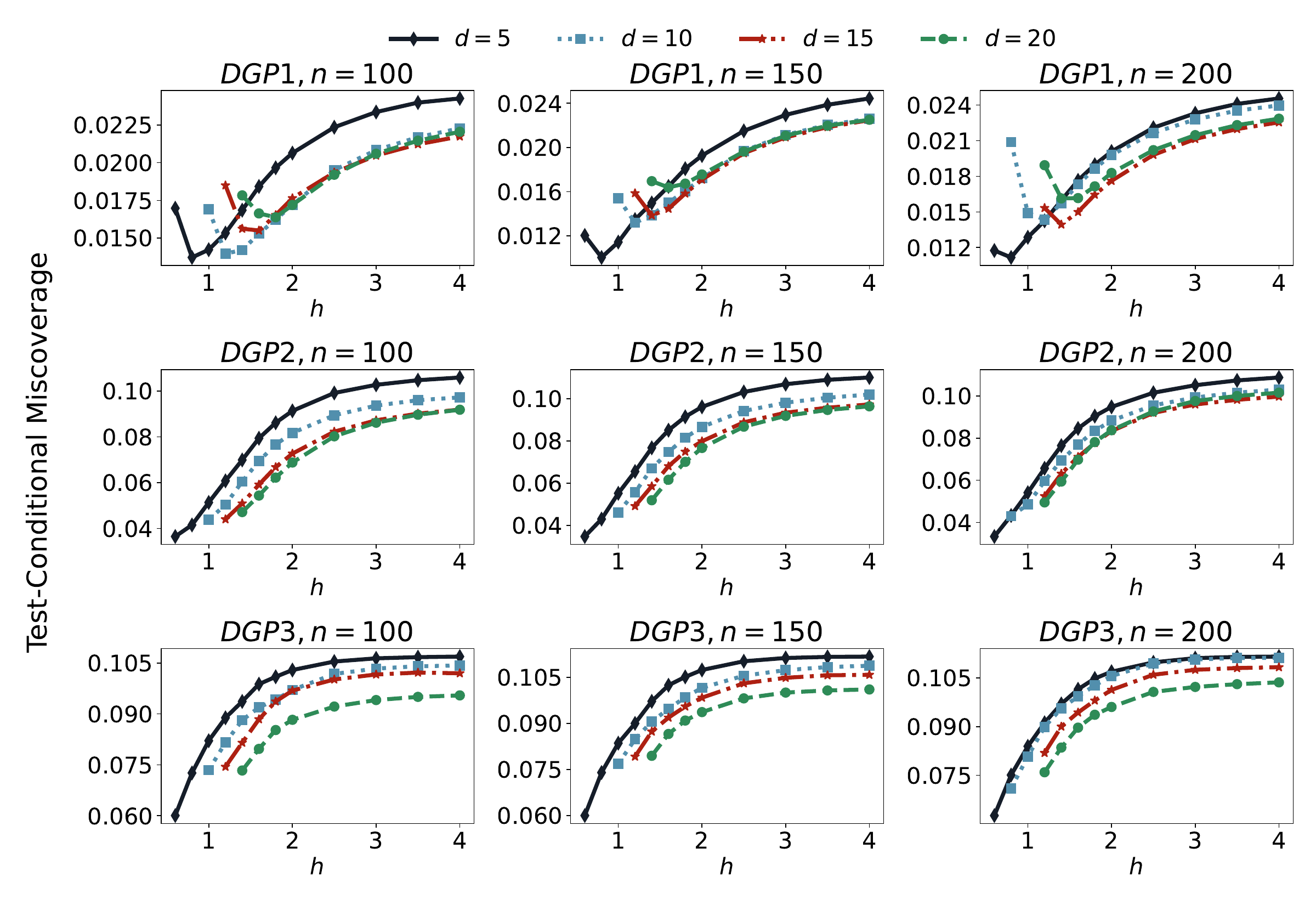}
        \caption{Test-conditional miscoverage of ELCP over $h$ under optimal $\omega$ for DGP1--DGP3.}\label{fig:Impact of varying h}
    \end{figure}

    Second, to assess how the quality of auxiliary data affects the optimal parameters for ELCP, we consider the following DGP: \\
    \noindent\textbf{DGP4:} $\epsilon(X)\sim N\left( 0, \exp(\sum_{i=1}^{5}X_i/2) \right)$, \\
    \indent\indent\indent$\epsilon^\prime(X^\prime)\sim \tau N\left( 0, 1.5\exp(\sum_{i=1}^{5}X_i^\prime/2) \right)+(1-\tau)N\left( 0, 1.5\sum_{i=1}^{5}|\arccos(X_i^\prime/2)| \right)$. \\
    The quality of auxiliary data improves as $\tau$ increases, with $\tau=1$ corresponding to DGP2.
    Figure~\ref{fig:Optimal parameter across different tau for DGP4} presents the optimal parameters and the corresponding test-conditional miscoverage error over $\tau$. 

    In Figure~\ref{fig:Optimal parameter across different tau for DGP4}, as $\tau$ decreases, the distributional discrepancy between auxiliary and target data increases, making the true density ratio more difficult to estimate accurately. Under these conditions, inaccuracies in estimating $\hat{r}$ become substantial, and the performance of ELCP degrades accordingly. 
    For $\omega$, we plot the ratio $\omega m/(n+\omega m)$ to represent the proportion of the effective auxiliary sample size $\omega m$ relative to the total effective size $n+\omega m$.
    The results show that the optimal $\omega$ decreases monotonically as the quality of the auxiliary data deteriorates.
    The optimal $h$ decreases slightly as $\tau$ increases, consistent with the fact that a larger optimal $\omega$ requires a smaller $h$ to balance $K_{0}^{-1}(h^{-d})h$ with $\{(n+\omega m)h^d\}^{1/2}$.
    Finally, the corresponding test-conditional miscoverage error decreases as $\tau$ increases.
    
    \begin{figure}[h]
        \centering
        \includegraphics[width=.8\linewidth]{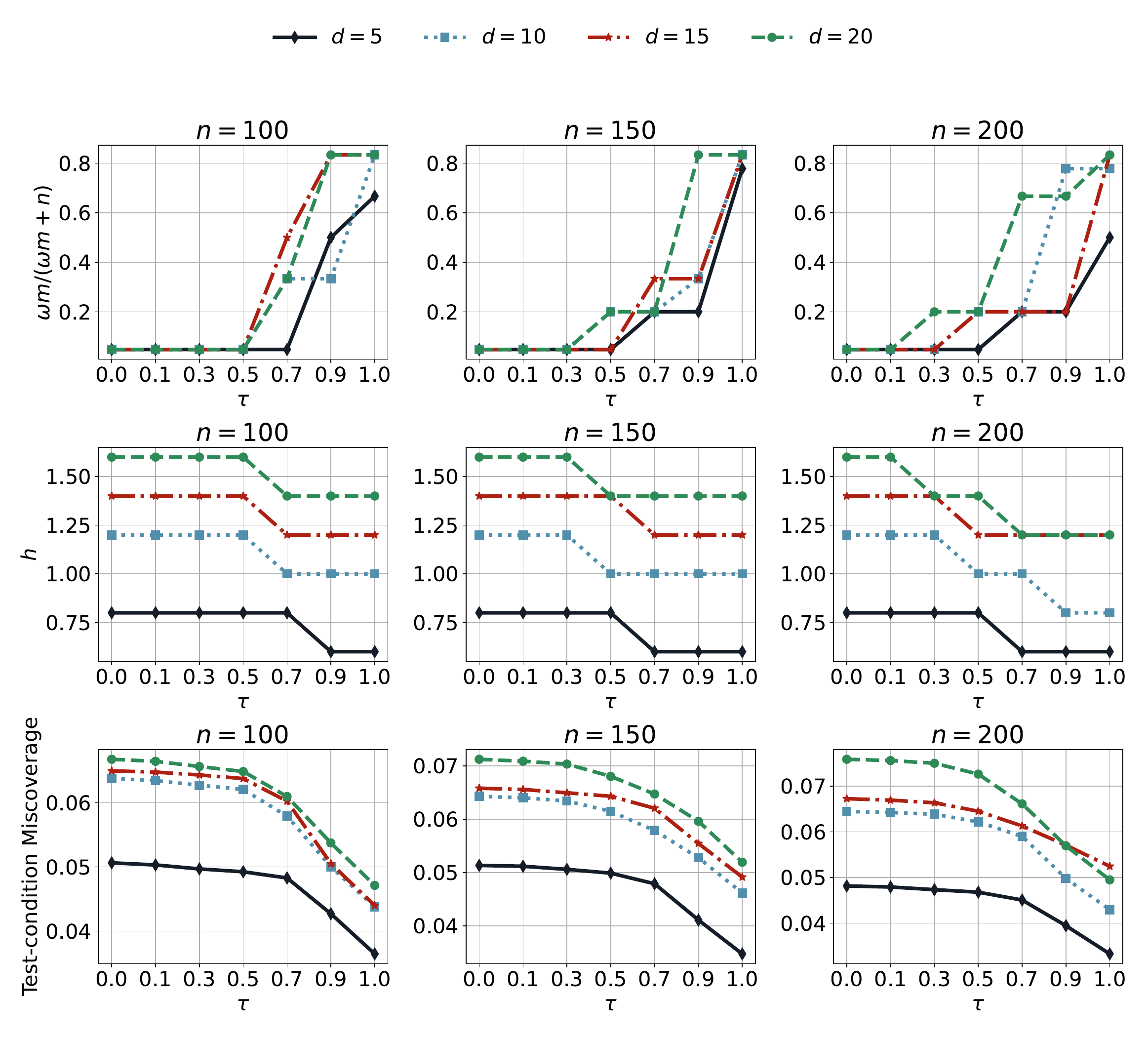}
        \caption{Optimal parameters and corresponding test-conditional miscoverage error over $\tau$ for DGP4.}\label{fig:Optimal parameter across different tau for DGP4}
    \end{figure}

\subsubsection{Results for different density ratio estimators}\label{sec:supp_simu_dre}

    In this section, we examine the performance of ELCP under different density ratio estimators.
    Figure~\ref{fig:Improvement across different DREs for DGP1-3} shows the percentage improvement of ELCP for DGP1--DGP3 over the better-performing method between LCP and RLCP in terms of the smallest test-conditional miscoverage error across all parameter values, with ELCP using three density ratio estimators: kernel-based least-squares importance fitting (KLIEP), random forest (RF) and quadratic discriminant analysis (QDA) with Platt scaling.
    First, ELCP consistently outperforms both LCP and RLCP under all three density ratio estimators.
    Second, the results reveal a clear performance hierarchy: RF delivers the largest improvements, followed by QDA, while KLIEP performs worse. This weaker performance of KLIEP largely reflects its high sensitivity to parameter specification, which often leads to suboptimal density ratio estimation and consequently smaller improvements compared with RF and QDA.
    \begin{figure}[h]
        \centering
        \includegraphics[width=.8\linewidth]{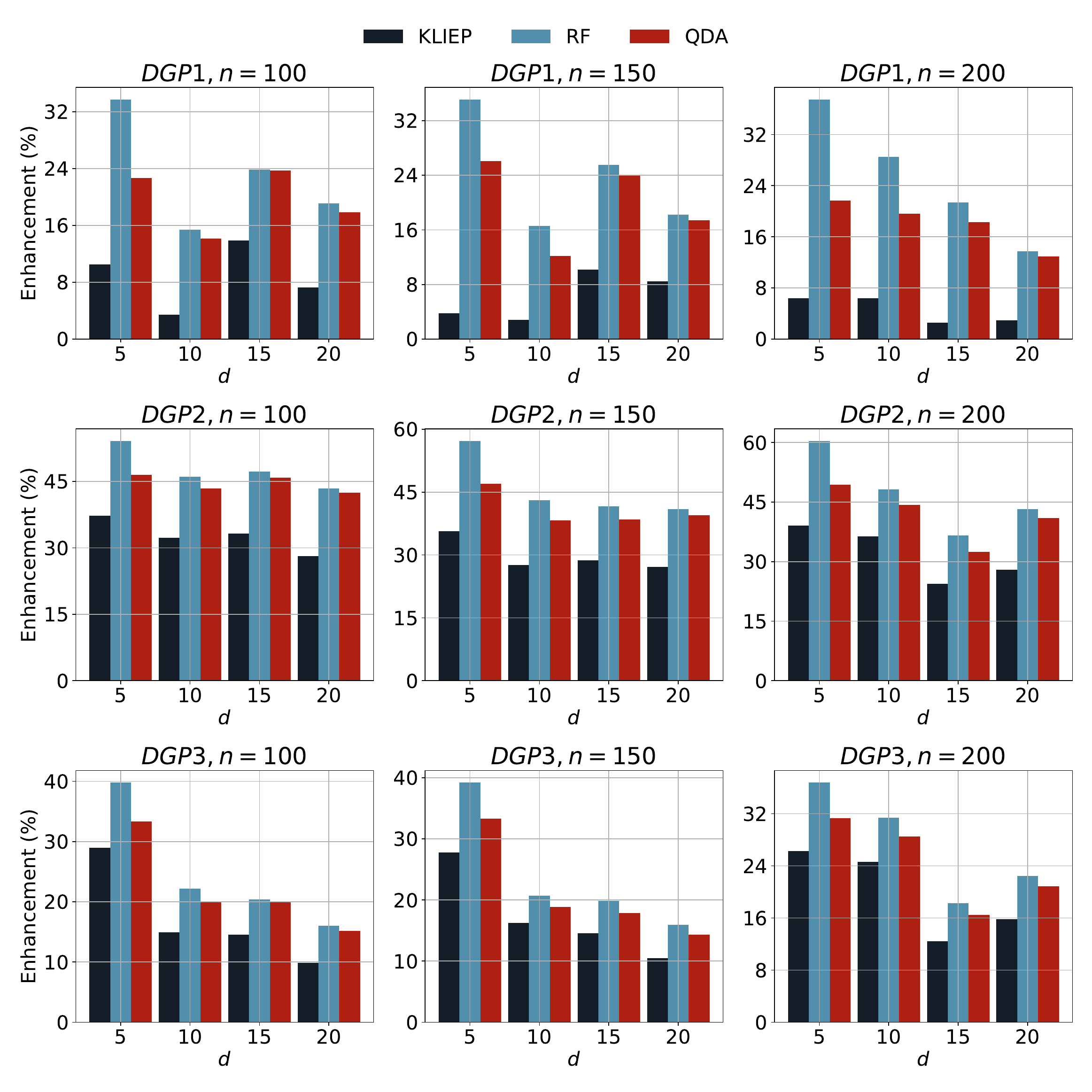}
        \caption{Enhancement ($\%$) of ELCP over LCP and RLCP in test-conditional miscoverage error under different density ratio estimators for DGP1--DGP3.}\label{fig:Improvement across different DREs for DGP1-3}
    \end{figure}

\subsubsection{Impact of auxiliary data size}\label{sec:supp_simu_auxiliary_size}
    We investigate the impact of the auxiliary data size on the performance of ELCP by considering different values of $m/n \in \{0.5, 1, 2, 5, 10\}$. Figure~\ref{fig:Test-conditional coverage for DGP13 d} shows the smallest test-conditional miscoverage error across all parameter values for ELCP (using QDA as the density ratio estimator), LCP, and RLCP. 
    \begin{figure}[h]
        \centering
        \includegraphics[width=.8\linewidth]{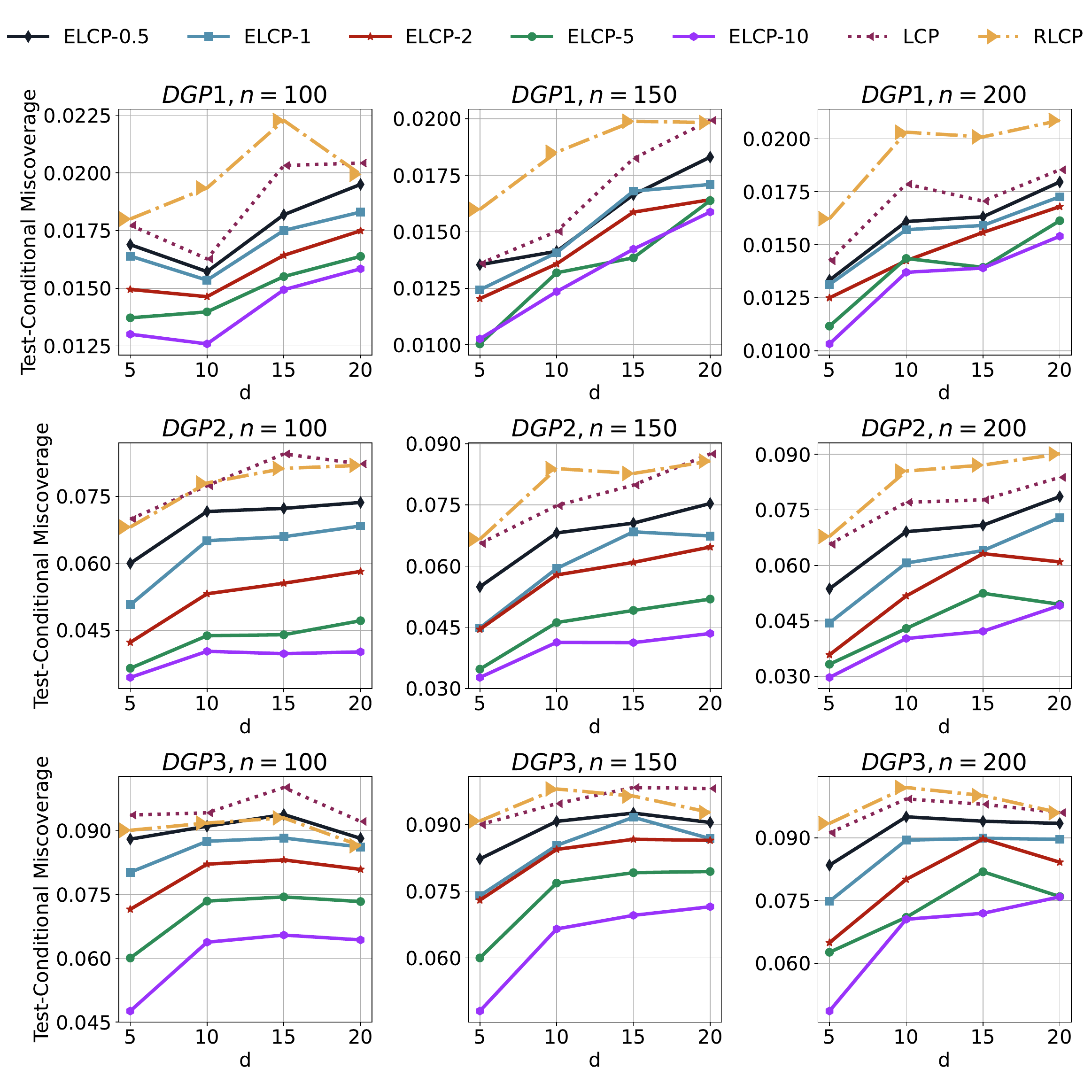}
        \caption{Test-conditional miscoverage error of ELCP with $m/n\in\{0.5, 1, 2, 5, 10\}$, LCP and RLCP for DGP1--DGP3.}\label{fig:Test-conditional coverage for DGP13 d}
    \end{figure}
    The number following ELCP indicates the corresponding $m/n$ ratio. The results show that even when $m/n = 0.5$, meaning the auxiliary data size is only half of the calibration data, ELCP still outperforms LCP and RLCP in most cases. Furthermore, the advantage of ELCP becomes significant for $m/n\geq 2$ and continues to increase as $m/n$ grows.
    
\subsubsection{Experiments on different score functions}\label{sec:supp_simu_score}
    In previous experiments, we employed the residual score $S(x,y) = |y - \hat{\mu}(x)|$, where $\hat{\mu}(\cdot)$ is a prediction model pretrained on the target training dataset. Many other conformal prediction methods that aim to improve test-conditional coverage adopt alternative score formulations, such as scores based on estimated conditional CDF, as in DCP \citep{chernozhukov2021distributional}, or on conditional density estimation (CDE), as in \citet{lei2014distribution, izbicki2019flexible}. 
    
    Moreover, when an auxiliary dataset is available, transfer learning techniques can be leveraged to potentially improve the pretrained model by incorporating additional information \citep{pan2009survey}. Motivated by this, for each type of pretrained model, we also consider a variant that incorporates information transferred from auxiliary data. 
    We begin by introducing a general principle for auxiliary information transfer in model pre-training.

    \noindent\textbf{Pre-training with auxiliary data:}
    Let $\mathcal{D}_{\rm tr} = \{(X_{{\rm tr},i}, Y_{{\rm tr},i})\}_{i=1}^n$ denote the target training dataset, where the samples are i.i.d.~from a distribution $P$ with joint density $f_{\rm raw}(x,y)$.
    Let $\mathcal{D}_{\rm tr}^\prime = \{(X_{{\rm tr},i}^\prime, Y_{{\rm tr},i}^\prime)\}_{i=1}^m$ be an auxiliary dataset with samples i.i.d.~from another distribution $P^\prime$ with joint density $f_{\rm raw}^\prime(x,y)$. 
    Consider a parametric model class $\{\mu_\theta(\cdot) : \theta \in \Theta\}$ for estimating a target function depending on $P$. From the perspective of empirical risk minimization (ERM), the population-level optimal parameter is defined as
    \begin{gather*}
        \theta^* = \underset{\theta \in \Theta}{\arg\min}\ E\left\{ l(X, Y, \theta) \right\},\ (X,Y)\sim P\,,
    \end{gather*}
    where $l(x, y, \theta)$ is a loss function. If $(X,Y)$ is from the auxiliary distribution $P^\prime$, the optimal parameter under importance weighting becomes
    \begin{gather*}
        \theta^* = \underset{\theta \in \Theta}{\arg\min}\ E\left\{ \dfrac{f_{\rm raw}(X,Y)}{f_{\rm raw}^\prime(X,Y)}l(X, Y, \theta) \right\},\ (X,Y)\sim P^\prime\,,
    \end{gather*}
    Given $\mathcal{D}_{\rm tr}$ and $\mathcal{D}_{\rm tr}^\prime$, along with an estimate $\hat{r}_{\rm raw}(x,y)$ of $f_{\rm raw}(x,y)/f_{\rm raw}^\prime(x,y)$, which can be learned from $\mathcal{D}_{\rm tr}$ and $\mathcal{D}_{\rm tr}^\prime$, we estimate $\theta$ by minimizing the empirical counterpart:
    \begin{gather}
        \hat{\theta} = \underset{\theta \in \Theta}{\arg\min} \dfrac{1}{n+m}\left\{ \sum_{i=1}^n l(X_{{\rm tr},i}, Y_{{\rm tr},i}, \theta)+\sum_{i=1}^m \hat{r}_{\rm raw}(X_{{\rm tr},i}^\prime, Y_{{\rm tr},i}^\prime)l(X_{{\rm tr},i}^\prime, Y_{{\rm tr},i}^\prime, \theta) \right\}\,.\label{eq: transfer minimize}
    \end{gather}
    For nonparametric models, the same transfer principle applies. Given any function $g(x,y)$, the transfer method defines its empirical counterpart of $E\{g(X, Y)\}$ as 
    \begin{gather}
        \left\{ n+\sum_{i=1}^m \hat{r}_{\rm raw}(X_{{\rm tr},i}^\prime, Y_{{\rm tr},i}^\prime) \right\}^{-1}\left\{ \sum_{i=1}^n g(X_{{\rm tr},i}, Y_{{\rm tr},i})+\sum_{i=1}^m \hat{r}_{\rm raw}(X_{{\rm tr},i}^\prime, Y_{{\rm tr},i}^\prime)g(X_{{\rm tr},i}^\prime, Y_{{\rm tr},i}^\prime) \right\}\,.\label{eq: transfer expectation}
    \end{gather}

    The following score functions and pre-training schemes are considered in the simulations of this section.
    %Next we describe each score function in detail.
    \begin{itemize}
        \item \noindent\textbf{CDF score}: Following \citet{chernozhukov2021distributional}, we consider scores based on conditional CDF estimation to capture localization properties. Specifically, we train a conditional distribution estimator $\widetilde{F}(y\mid x)$ using the training dataset $\mathcal{D}_{\rm tr}$ and define the score for $(X_i,Y_i), i=1,\ldots,n$ as
        \begin{gather*}
            S_i=|\widetilde{F}(Y_i\mid X_i)-1/2|\,.
        \end{gather*}
        Scores for test and auxiliary data are calculated analogously, and all other procedures remain unchanged. The conditional CDF estimator is implemented using conventional quantile regression (QR) as described in \citet{chernozhukov2021distributional}. We refer to this score as \textbf{CDF}, and to the corresponding method of \citet{chernozhukov2021distributional} as \textbf{DCP}. 
        When the score is based on a CDF estimator trained incorporating information transferred from the \textbf{A}uxiliary \textbf{T}raining dataset via \eqref{eq: transfer minimize}, we denote the resulting score and conformal prediction method as \textbf{CDF-AT} and \textbf{DCP-AT}, respectively. 
        When the transfer in pre-training incorporates \textbf{A}ll \textbf{A}uxiliary data, we denote the corresponding variants as \textbf{CDF-AA} and \textbf{DCP-AA}.

        \item \noindent\textbf{CDE score}: Following \citet{lei2014distribution, izbicki2019flexible}, we consider scores based on conditional density estimation. Specifically, we train a conditional density estimator $\widetilde{f}(y \mid x)$ using the training dataset and define the score for $(X_i,Y_i)$ as
        \begin{gather*}
            S_i = 1 / \widetilde{f}(Y_i \mid X_i)\,.
        \end{gather*}
        The conditional density estimator is obtained via kernel density estimation using a kernel $K(\cdot,\cdot; h)$, where the optimal bandwidth $h$ is selected using the method of \citet{bashtannyk2001bandwidth} implemented in SciPy \citep{2020SciPy-NMeth}. The resulting score and the corresponding conformal method are denoted as \textbf{CDE}. Specifically, the conditional density of $Y\mid X$ estimated using only target training data can be expressed as
        \begin{gather*}
            \widetilde{f}(y\mid x)=\dfrac{\sum_{i=1}^{n}K((x,y),(X_{{\rm tr},i},Y_{{\rm tr},i});h_1)}{\sum_{i=1}^{n}K(x,X_{{\rm tr},i};h_2)}\,,
        \end{gather*}
        where $h_1,h_2$ are the bandwidths chosen for joint density estimation and covariate density estimation, respectively. When the score incorporates information transferred from the auxiliary training dataset by \eqref{eq: transfer expectation}, the conditional density estimator $\widetilde{f}_{\rm AT}(y\mid x)$ can be written as
        \begin{gather*}
            \dfrac{\sum_{i=1}^{n}K((x,y),(X_{{\rm tr},i},Y_{{\rm tr},i});h_3)+\sum_{i=1}^m \hat{r}_{\rm raw}(X_{{\rm tr},i}^\prime, Y_{{\rm tr},i}^\prime)K((x,y),(X_{{\rm tr},i}^\prime,Y_{{\rm tr},i}^\prime);h_3)}{\sum_{i=1}^{n}K(x,X_{{\rm tr},i};h_4)+\sum_{i=1}^m \hat{r}_{\rm raw}(X_{{\rm tr},i}^\prime, Y_{{\rm tr},i}^\prime)K(x,X_{{\rm tr},i}^\prime;h_4)}\,.
        \end{gather*}        
        We denote this resulting variant as \textbf{CDE-AT}. When the transfer uses all auxiliary data, the corresponding score is denoted as \textbf{CDE-AA}.

        \item Additionally, we denote the residual score used in the main text as \textbf{RES}, and the score that incorporates information transferred from the auxiliary training dataset by \eqref{eq: transfer minimize} as \textbf{RES-AT} and from all auxiliary data as \textbf{RES-AA}.
    \end{itemize}

    We continue to consider DGP1--DGP3 with $n=150$ and $m/n=5$, and compare the performance of ELCP, LCP, RLCP, DCP and CDE under the RES, RES-AT, RES-AA, CDF, CDF-AT, CDF-AA, CDE, CDE-AT, and CDE-AA score functions (with DCP implemented only for the CDF-based scores: CDF, CDF-AT, and CDF-AA).
    Figure~\ref{fig:test for DGP13 scores} presents comparisons of CDF, CDE with RES.
    Figure~\ref{fig:test for DGP13 scores 1} compares RES-AT, RES-AA with RES.
    Figure~\ref{fig:test for DGP13 scores 2} presents comparisons of CDF-AT, CDF-AA with CDF.
    Figure~\ref{fig:test for DGP13 scores 3} compares CDE-AT, CDE-AA with CDE.
    In all figures, dashed lines represent the baseline score, while solid lines represent the score function under comparison.

    \begin{figure}[htp]
        \centering
        \begin{subfigure}[b]{0.95\textwidth}
            \includegraphics[width=\linewidth]{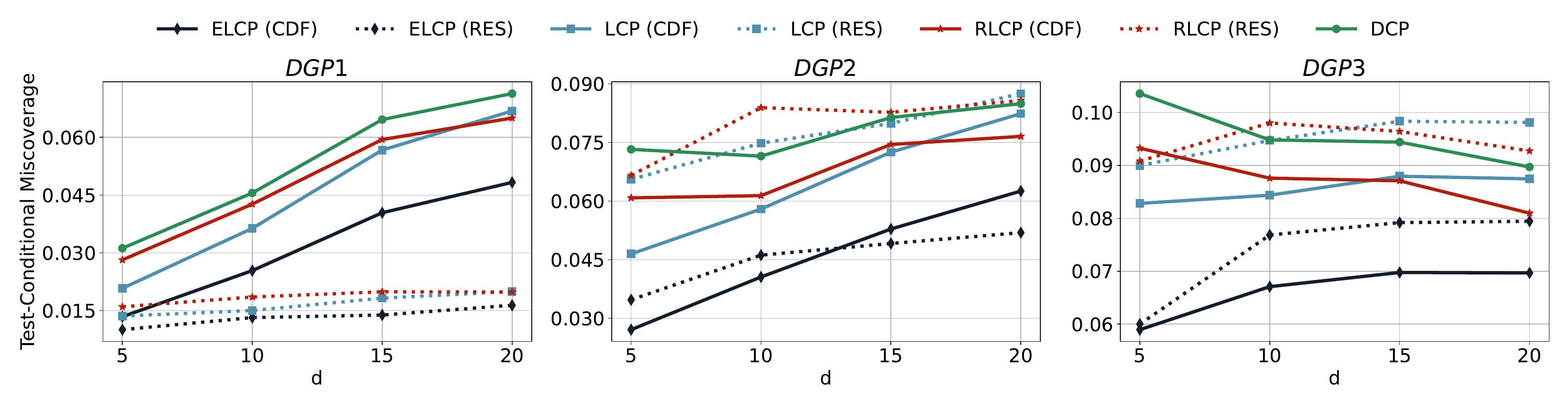}
        \end{subfigure}
        \begin{subfigure}[b]{0.95\textwidth}
            \includegraphics[width=\linewidth]{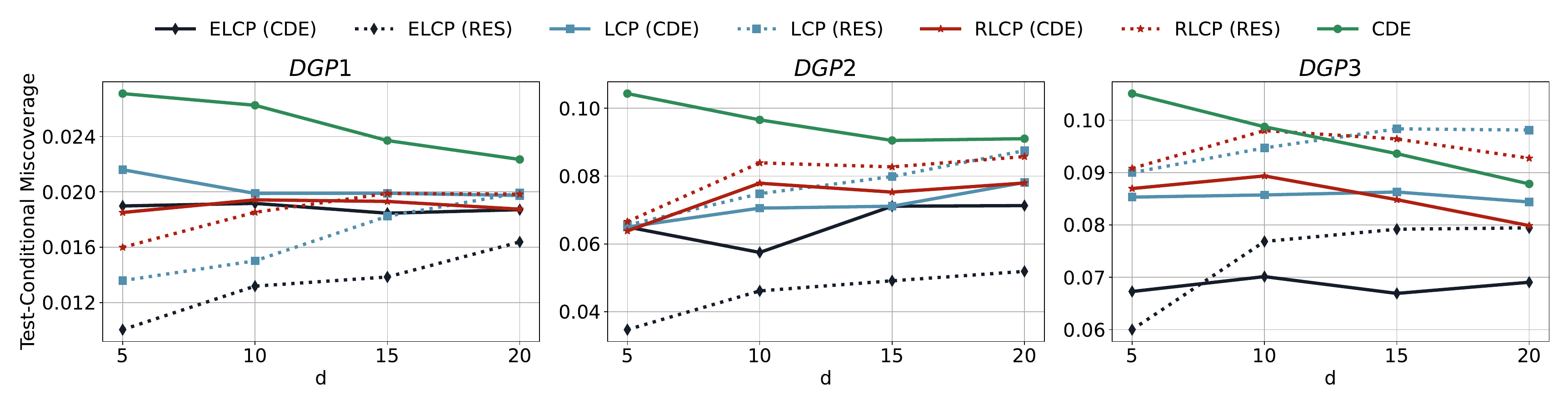}
        \end{subfigure}
        \caption{Test-conditional miscoverage error of CDF/CDE scores for DGP1--DGP3 compared with the RES score.}\label{fig:test for DGP13 scores}
    \end{figure}
    \begin{figure}[htp]
        \centering
        \begin{subfigure}[b]{0.95\textwidth}
            \includegraphics[width=\linewidth]{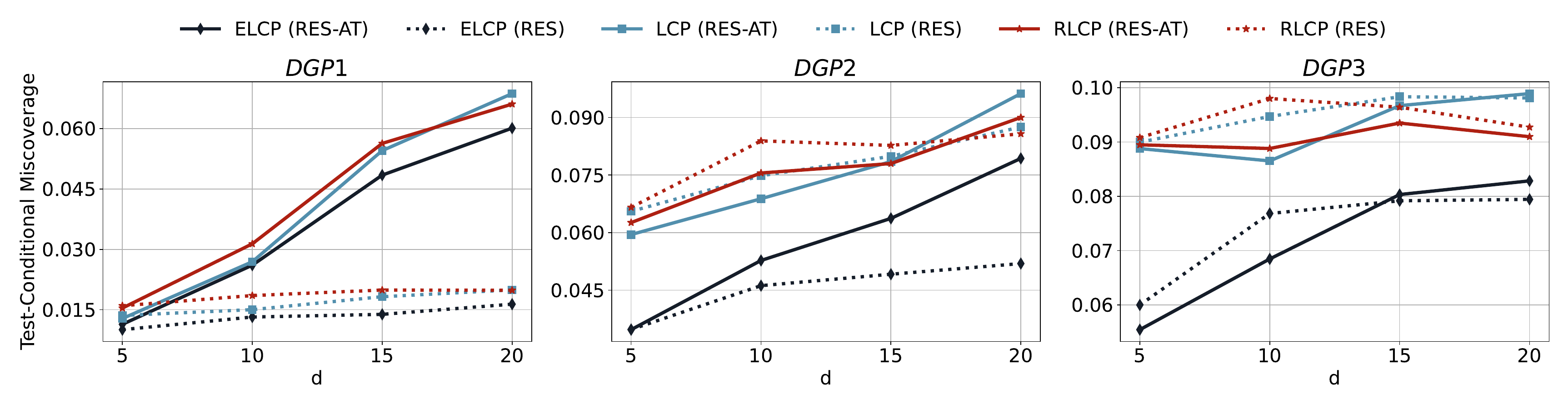}
        \end{subfigure}
        \begin{subfigure}[b]{0.95\textwidth}
            \includegraphics[width=\linewidth]{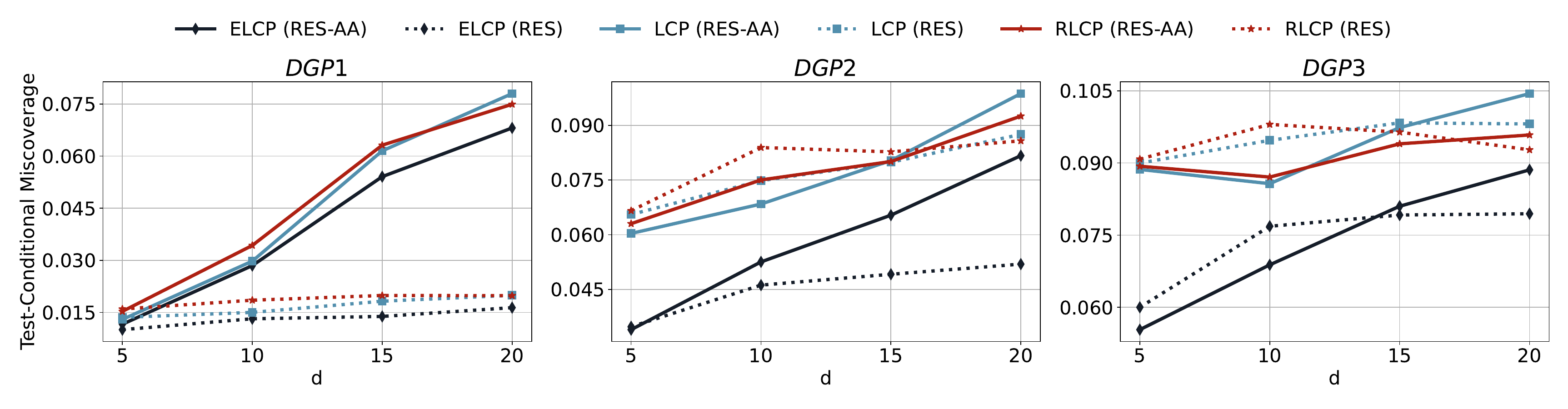}
        \end{subfigure}
        \caption{Test-conditional miscoverage error of RES-AT/RES-AA scores for DGP1--DGP3 compared with the RES score.}\label{fig:test for DGP13 scores 1}
    \end{figure}
    \begin{figure}[htp]
        \centering
        \begin{subfigure}[b]{0.95\textwidth}
            \includegraphics[width=\linewidth]{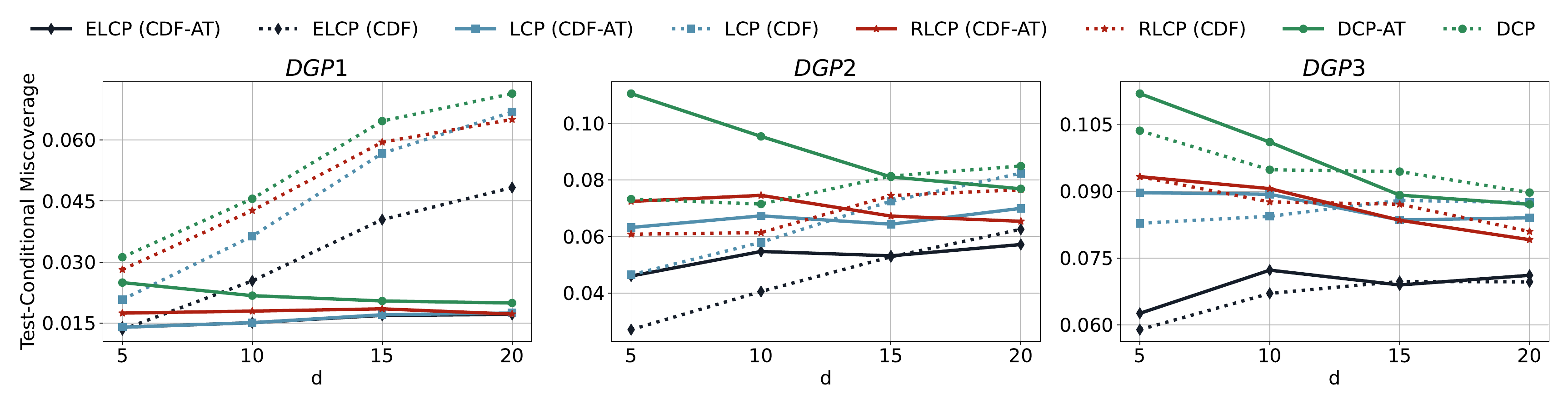}
        \end{subfigure}
        \begin{subfigure}[b]{0.95\textwidth}
            \includegraphics[width=\linewidth]{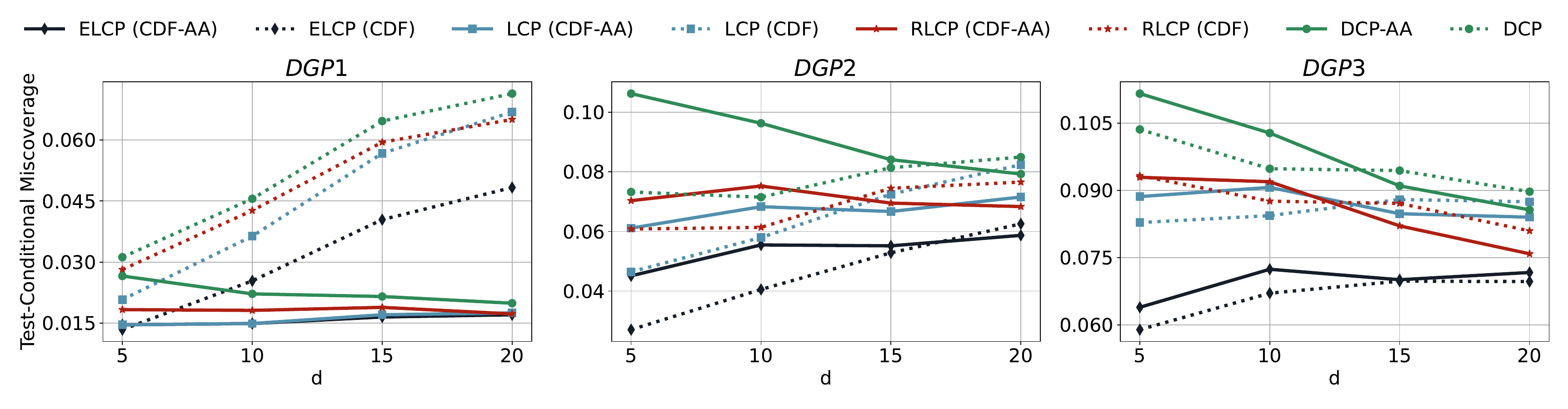}
        \end{subfigure}
        \caption{Test-conditional miscoverage error of CDF-AT/CDF-AA scores for DGP1--DGP3 compared with the CDF score.}\label{fig:test for DGP13 scores 2}
    \end{figure}
    \begin{figure}[htp]
        \centering
        \begin{subfigure}[b]{0.95\textwidth}
            \includegraphics[width=\linewidth]{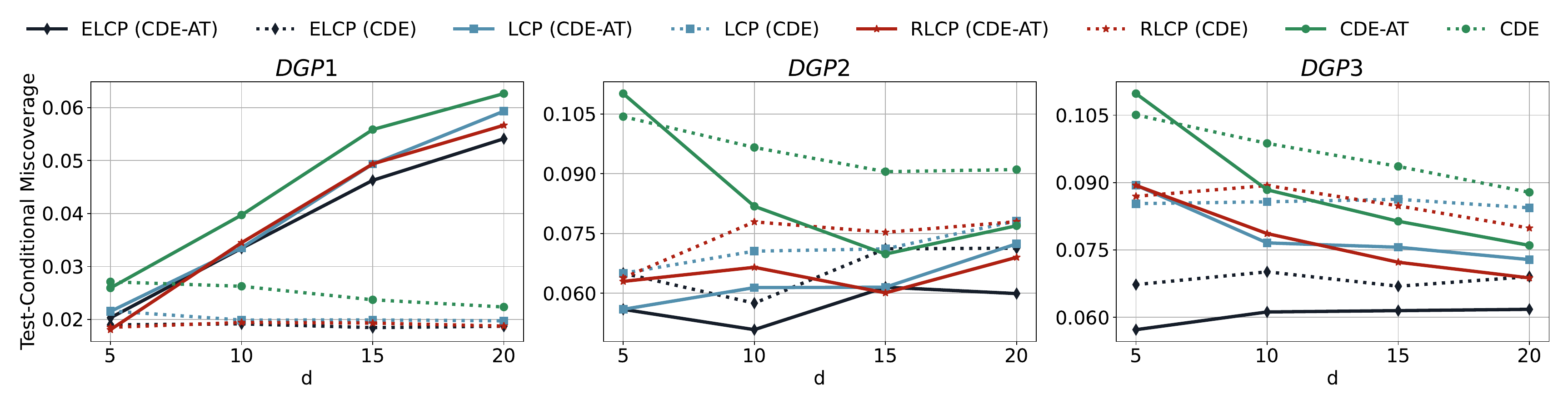}
        \end{subfigure}
        \begin{subfigure}[b]{0.95\textwidth}
            \includegraphics[width=\linewidth]{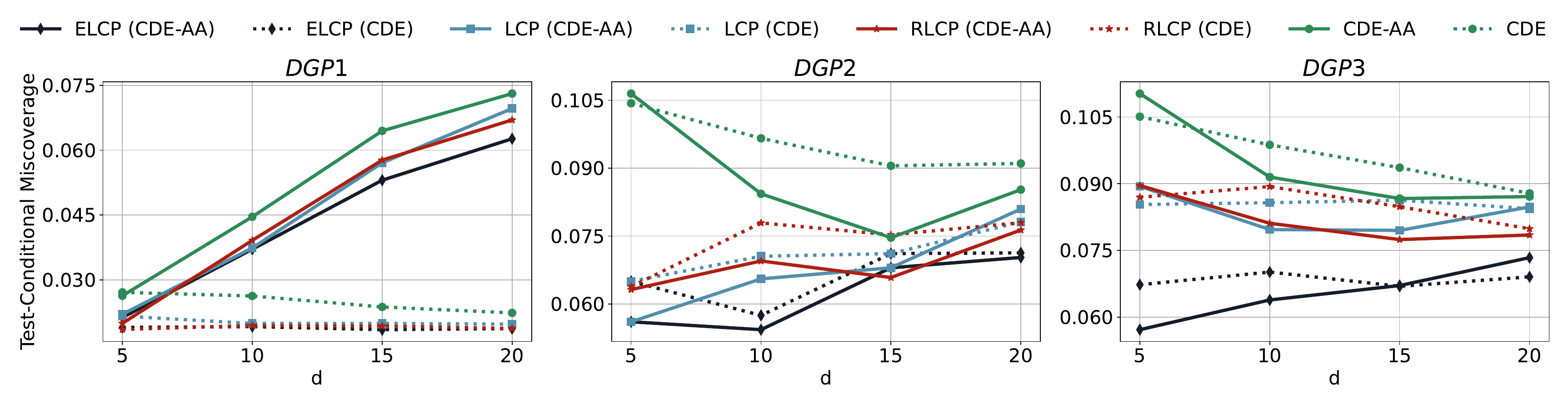}
        \end{subfigure}
        \caption{Test-conditional miscoverage error of CDE-AT/CDE-AA scores for DGP1--DGP3 compared with the CDE score.}\label{fig:test for DGP13 scores 3}
    \end{figure}

    First, we observe that across all DGPs and for any fixed choice of score function and pre-training scheme, ELCP consistently outperforms LCP, RLCP, and DCP.

    Focusing on score functions based on pretrained models that do not incorporate auxiliary data, namely RES, CDF, and CDE (Figure~\ref{fig:test for DGP13 scores}), we see distinct patterns across the three DGPs. Under DGP1, almost all methods perform better with the RES score than with CDF or CDE. Under DGP2, the performance under different scores is more mixed. Under DGP3, however, CDF and CDE often outperform RES, especially when the dimension $d$ is relatively large.

    Next, we examine the effect of incorporating auxiliary data into model pretraining in the score function on test-conditional coverage. Under DGP1, RES and CDE scores based on models pretrained using the AT and AA transfer schemes show degraded performance (Figures \ref{fig:test for DGP13 scores 1} and \ref{fig:test for DGP13 scores 3}), indicating that transferring auxiliary information for model pretraining for score construction does not always lead to improvements.
    In contrast, for CDF type scores under DGP1, transfer from auxiliary data in model pretraining brings clear benefits (Figure \ref{fig:test for DGP13 scores 2}).
    For DGP2 and DGP3, the results are more mixed.
    
    %The results indicate that for DGP1 and DGP2, the RES score generally outperforms the other score functions. For DGP3, however, both CDF and CDE scores yield better performance than RES. In DGP1, score functions using the AT and AA pre-training schemes tend to perform poorly, suggesting that transfer for pretrained score functions does not always help. For DGP2--3, transfer from auxiliary information does bring improvements. Specifically, for DGP2, transfer benefits non-ELCP methods considerably, while in DGP3, all methods see notable gains from transfer. Overall, in all scenarios the best performance is achieved by the ELCP-based results. 

    When auxiliary data are available, end-to-end approaches can be categorized into two classes. The first, exemplified by ELCP, leverages auxiliary information at the calibration stage and can be applied to any pretrained score function. The second incorporates auxiliary information at the pre-training stage of the score function via transfer learning. To highlight the advantages of ELCP, we compare it directly with other approaches using auxiliary data under CDF and CDE type scores in Figure~\ref{fig:test for DGP13 schemes}. 
    For instance, under the CDF score, we consider ELCP with CDF score, DCP with CDF score (DCP), ELCP with CDF-AT score, DCP with CDF-AT score (DCP-AT), ELCP with CDF-AA score, and DCP with CDF-AA score (DCP-AA).

    Figure~\ref{fig:test for DGP13 schemes} shows that in most cases, ELCP outperforms DCP-AT and DCP-AA. The only notable exception occurs under DGP1 with the CDF score, where ELCP (CDF) performs worse than DCP-AT and DCP-AA. Even in this setting, ELCP with the CDF-AT or CDF-AA scores still surpasses DCP-AT and DCP-AA. 
    Overall, these results indicate that while incorporating auxiliary information during pre-training via transfer learning in score construction is often beneficial, further leveraging auxiliary information through ELCP's calibration step can yield additional and consistent gains.

    \begin{figure}[H]
        \centering
        \begin{subfigure}[b]{0.95\textwidth}
            \includegraphics[width=\linewidth]{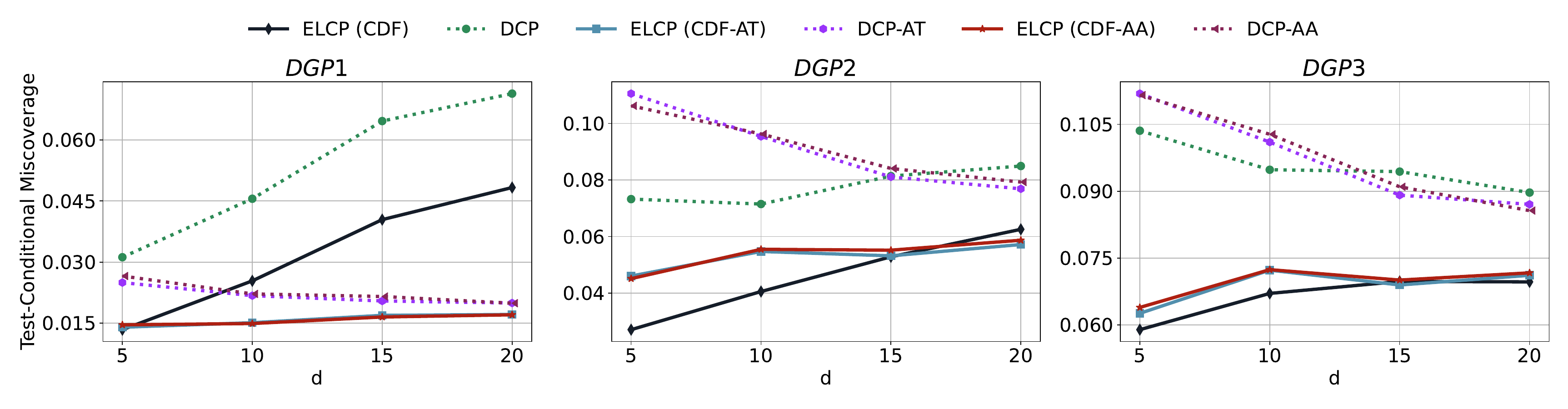}
        \end{subfigure}
        \begin{subfigure}[b]{0.95\textwidth}
            \includegraphics[width=\linewidth]{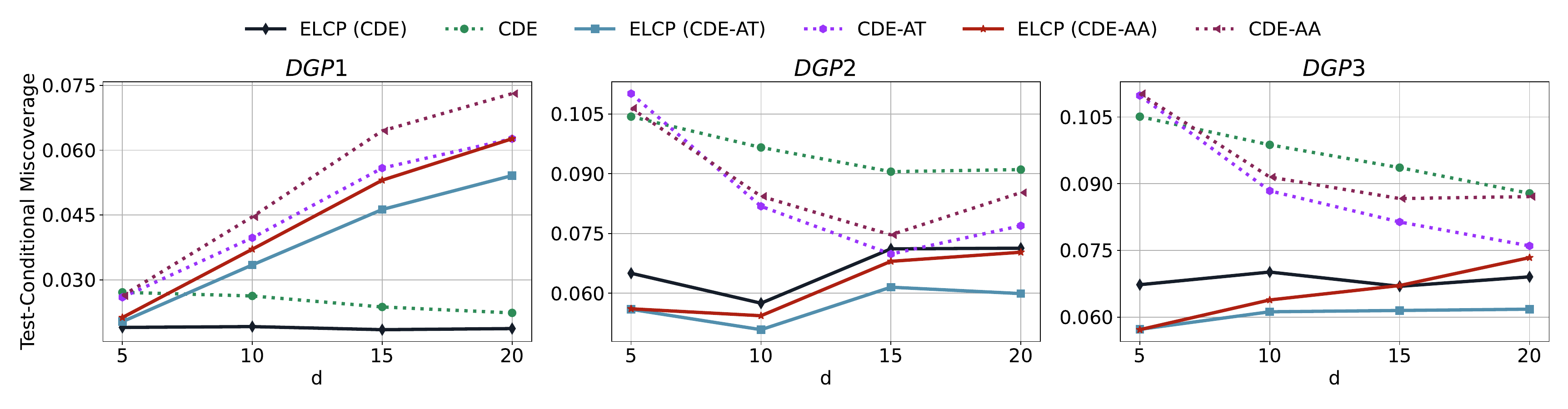}
        \end{subfigure}
        \caption{Test-conditional miscoverage error of different pre-training schemes for DGP1--DGP3.}\label{fig:test for DGP13 schemes}
    \end{figure}

\subsubsection{Experiments on varying nominal coverage level $1-\alpha$}\label{sec:supp_simu_varying_alpha}

    While previous experiments fixed the coverage level at $0.9$, we now consider a range of levels $1-\alpha \in \{0.1, 0.2, 0.3, 0.4, 0.5, 0.6, 0.7, 0.8, 0.9\}$ to evaluate performance across different coverage levels. 
    
    Figure~\ref{fig:test for DGP13 5 alpha} shows the test-conditional miscoverage error for DGP1--DGP3 under varying $1-\alpha$, with ELCP's miscoverage curve consistently lying below those of LCP and RLCP. This confirms that ELCP delivers consistent improvements across all coverage levels, supporting the conclusion of Theorem~\ref{theo:conditional_coverage_error_bound} that auxiliary data enhances the estimation of the entire conditional distribution, rather than just a single quantile.
    When $1-\alpha$ lies in the range $[0.5, 0.7]$, the advantage of ELCP becomes more pronounced.
    \begin{figure}[h]
        \centering
        \includegraphics[width=.8\linewidth]{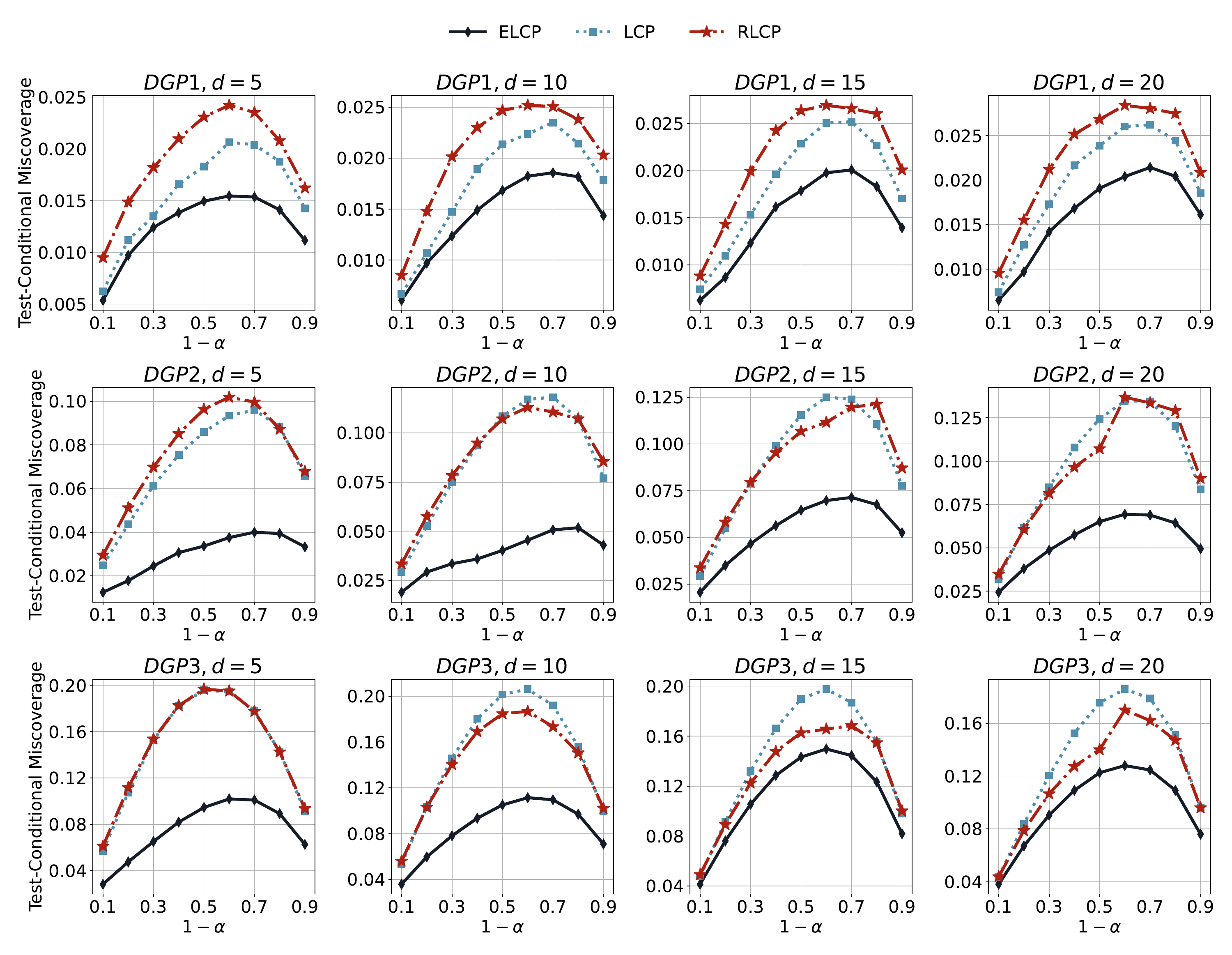}
        \caption{Test-conditional miscoverage error for DGP1--DGP3 with varying $1-\alpha$.}\label{fig:test for DGP13 5 alpha}
    \end{figure}

    Figure~\ref{fig:marginal for DGP13 5 alpha} reports the mean absolute differences between the achieved marginal coverage rate and the target $1-\alpha$ for ELCP, LCP, RLCP, LCP-C, and RLCP-C across all nominal levels. The results show that while LCP-C and RLCP-C fail to maintain marginal coverage guarantees for any $1-\alpha$ in DGP1--DGP3, all other methods successfully preserve the marginal coverage guarantees.
    \begin{figure}[h]
        \centering
        \includegraphics[width=.8\linewidth]{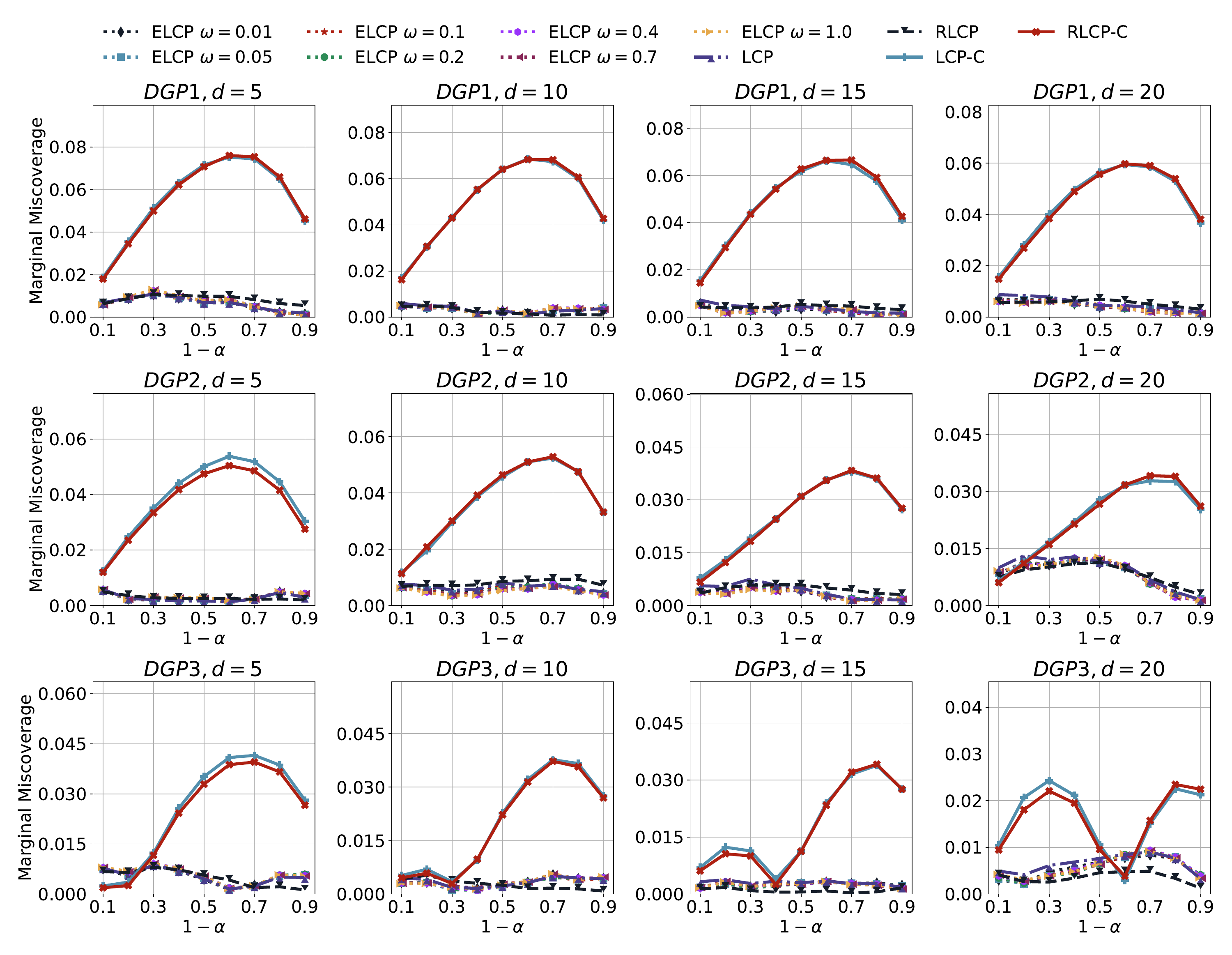}
        \caption{Absolute marginal miscoverage error for DGP1-DGP3 with varying $1-\alpha$.}\label{fig:marginal for DGP13 5 alpha}
    \end{figure}

\subsubsection{Experiments on extremely limited $n$}\label{sec:supp_simu_limited_n}

    While our previous simulations focused on dimensions $d=5,10,15,20$ with training and calibration data sizes $n=100,150,200$, we now investigate a more challenging setting where the training and calibration data are extremely limited. To prevent $1-\alpha+(n+1)^{-1}$ from reaching 1, we specifically consider $n=20$ with $1-\alpha=0.9$ for dimensions $d=3,6,9$.
    Under this scenario, where the scores may become highly inaccurate due to limited data, we retain DGP1–DGP3 and report the resulting test-conditional miscoverage rates in Table~\ref{table:testDGP123smalln}.
    
    In DGP1, where the conditional variance of $Y$ given $X$ is relatively small compared to its mean, the estimated conditional expectation $E(Y\mid X)$ remains reasonably accurate even when $n=20$. 
    Consequently, ELCP achieves substantial test-conditional miscoverage error reductions, exceeding 50\% across all dimensions.

    In contrast, DGP2 and DGP3 yield highly inaccurate estimates of $E(Y \mid X)$. 
    Consequently, the connection between auxiliary and target score distributions becomes extremely weak, making accurate density ratio estimation particularly challenging.
    As a result, ELCP yields only marginal or negligible gains over LCP and RLCP, with reduction ratios often near zero. Nevertheless, ELCP consistently performs at least as well as LCP in all settings.
    \begin{table*}[h]
        \caption{Test-conditional miscoverage error when $n=20$ (with ELCP reduction ratio in parentheses) for DGP1--DGP3.}
        \begin{center}
        {\fontsize{8}{8}\selectfont{\setlength{\tabcolsep}{3pt}\begin{tabular}{ccccccccccccccccc}
        \toprule
        & & \multicolumn{3}{c}{$d=3$} & \multicolumn{3}{c}{$d=6$} & \multicolumn{3}{c}{$d=9$} \vspace{-2pt}\\
         & $n$ & ELCP & LCP & RLCP & ELCP & LCP & RLCP & ELCP & LCP & RLCP \\
        \cmidrule(lr){3-5} \cmidrule(lr){6-8} \cmidrule(lr){9-11}
        DGP1 & $20$ & $\textbf{0.023(54.3\%)}$ & $0.061$ & $0.051$ & $\textbf{0.017(58.7\%)}$ & $0.042$ & $0.041$ & $\textbf{0.019(55.5\%)}$ & $0.046$ & $0.042$\\
        DGP2 & $20$ & $\textbf{0.061(0.0\%)}$ & $0.061$ & $0.079$ & $\textbf{0.059(1.3\%)}$ & $0.060$ & $0.080$ & $\textbf{0.044(26.6\%)}$ & $0.060$ & $0.072$\\
        DGP3 & $20$ & $\textbf{0.071(0.0\%)}$ & $0.071$ & $0.078$ & $\textbf{0.067(0.0\%)}$ & $0.067$ & $0.078$ & $\textbf{0.062(0.0\%)}$ & $0.062$ & $0.073$\\
        \bottomrule
        \end{tabular}}}
        \end{center}
        \label{table:testDGP123smalln}
    \end{table*}

\FloatBarrier
\subsection{Additional results for synthesized data under data-driven selected $\omega$ and $h$}\label{sec:supp_simu_add_selected_parameter}

\subsubsection{Marginal coverage and mean prediction set size under data-driven selected $\omega$ and $h$}\label{sec:supp_marginal_selected}

    We present additional experimental results of marginal coverage under parameter selection in Table~\ref{table:mar for parasel}. We observe that ELCP, LCP, and RLCP all maintain marginal coverage within $n^{-1}$ of the required level when parameters are selected to optimize test-conditional coverage. %For the parameter selection method that optimizes the mean prediction set size, because our implementation excludes the $n+1$ test point, the marginal coverage may be slightly compromised. However, the overall decrease in coverage is bounded within an order of $n^{-1}$. 
    \begin{table*}[htb]
        \caption{Marginal coverage of ELCP, LCP and RLCP under data-driven selected parameters for DGP1--DGP3.}
        \begin{center}
        {\fontsize{10}{8}\selectfont{\setlength{\tabcolsep}{4pt}\begin{tabular}{ccccccccccccccccc}
        \toprule
        & & \multicolumn{4}{c}{$n=100$} & \multicolumn{4}{c}{$n=150$} & \multicolumn{4}{c}{$n=200$} \vspace{-2pt}\\
         & $d$ & 5 & 10 & 15 & 20 & 5 & 10 & 15 & 20 & 5 & 10 & 15 & 20 \\
        \cmidrule(lr){3-6} \cmidrule(lr){7-10} \cmidrule(lr){11-14}
        DGP1 & ELCP & 0.899 & 0.899 & 0.902 & 0.902 & 0.903 & 0.896 & 0.903 & 0.898 & 0.898 & 0.895 & 0.897 & 0.899\\
        & LCP & 0.901 & 0.902 & 0.901 & 0.901 & 0.903 & 0.898 & 0.901 & 0.897 & 0.899 & 0.896 & 0.898 & 0.900\\
        & RLCP & 0.901 & 0.902 & 0.901 & 0.901 & 0.903 & 0.898 & 0.901 & 0.897 & 0.899 & 0.896 & 0.898 & 0.900\\
        DGP2 & ELCP & 0.891 & 0.901 & 0.905 & 0.900 & 0.898 & 0.896 & 0.901 & 0.903 & 0.894 & 0.897 & 0.900 & 0.900\\
        & LCP & 0.895 & 0.906 & 0.906 & 0.902 & 0.900 & 0.904 & 0.901 & 0.905 & 0.895 & 0.901 & 0.900 & 0.900\\
        & RLCP & 0.895 & 0.906 & 0.906 & 0.902 & 0.900 & 0.904 & 0.901 & 0.905 & 0.895 & 0.901 & 0.900 & 0.900\\
        DGP3 & ELCP & 0.891 & 0.902 & 0.901 & 0.902 & 0.897 & 0.896 & 0.900 & 0.895 & 0.894 & 0.895 & 0.897 & 0.893\\
        & LCP & 0.895 & 0.902 & 0.900 & 0.901 & 0.896 & 0.900 & 0.898 & 0.895 & 0.895 & 0.892 & 0.898 & 0.894\\
        & RLCP & 0.895 & 0.902 & 0.900 & 0.901 & 0.896 & 0.900 & 0.898 & 0.895 & 0.895 & 0.892 & 0.898 & 0.894\\
        \bottomrule
        \end{tabular}}}
        \end{center}
        \label{table:mar for parasel}
    \end{table*}
    In Section~\ref{sec: exp selected wh} of the main text, we select the parameters in a data-driven manner to optimize test-conditional coverage. This leads to a significant improvement in test-conditional coverage for ELCP. 
    
    Furthermore, we report the corresponding mean prediction set sizes. As shown in Table~\ref{table:sizeDGP123para}, ELCP attains smaller prediction set sizes in most settings. In other cases, its set sizes are slightly larger; however, they remain very close to the best-performing method among LCP and RLCP, with differences almost always within 1\%.
    We emphasize that the data-driven parameter selection in our framework is guided by a loss function targeting test-conditional coverage rather than directly optimizing prediction set size. Nonetheless, the results indicate that ELCP maintains competitive set efficiency while achieving its primary goal of improved test-conditional coverage.
    In practice, if the prediction set size is the priority, the construction of $\mathcal{L}_1$ in Section~\ref{sec:construct loss fun} can be employed to select parameters with the explicit goal of minimizing the prediction set size.
    %In particular, if one wishes to use auxiliary data specifically to enhance efficiency, the construction of $\mathcal{L}_1$ in Section~\ref{sec:construct loss fun} can be employed to select parameters with the explicit goal of minimizing the prediction set size.
    \begin{table*}[htb]
        \caption{Mean prediction set size under data-driven selected $\omega$ and $h$ (with ELCP reduction ratio in parentheses) for DGP1--DGP3.}
        \begin{center}
        {\fontsize{7}{7}\selectfont{\setlength{\tabcolsep}{1.5pt}\begin{tabular}{ccccccccccccccccc}
        \toprule
        & & \multicolumn{3}{c}{$d=5$} & \multicolumn{3}{c}{$d=10$} & \multicolumn{3}{c}{$d=15$} & \multicolumn{3}{c}{$d=20$}\vspace{-2pt}\\
         & $n$ & ELCP & LCP & RLCP & ELCP & LCP & RLCP & ELCP & LCP & RLCP & ELCP & LCP & RLCP \\
        \cmidrule(lr){3-5} \cmidrule(lr){6-8} \cmidrule(lr){9-11} \cmidrule(lr){12-14}
        DGP1 & $100$ & $9.959(-0.6\%)$ & $9.898$ & $10.148$ & $\textbf{10.016(1.1\%)}$ & $10.132$ & $10.602$ & $10.351(-0.1\%)$ & $10.338$ & $10.784$ & $\textbf{10.656(0.5\%)}$ & $10.706$ & $11.130$\\
        & $150$ & $\textbf{9.938(0.5\%)}$ & $9.991$ & $10.060$ & $\textbf{9.678(1.0\%)}$ & $9.779$ & $10.153$ & $10.164(-0.6\%)$ & $10.106$ & $10.406$ & $\textbf{10.129(0.4\%)}$ & $10.167$ & $10.490$\\
        & $200$ & $\textbf{9.739(1.0\%)}$ & $9.835$ & $9.868$ & $\textbf{9.562(1.1\%)}$ & $9.665$ & $9.887$ & $\textbf{9.740(0.9\%)}$ & $9.825$ & $10.035$ & $\textbf{9.976(0.6\%)}$ & $10.040$ & $10.253$\\
        DGP2 & $100$ & $\textbf{4.300(0.4\%)}$ & $4.319$ & $5.340$ & $\textbf{4.540(0.6\%)}$ & $4.570$ & $5.557$ & $4.951(-0.1\%)$ & $4.947$ & $5.808$ & $\textbf{5.154(0.8\%)}$ & $5.193$ & $6.025$\\
        & $150$ & $4.282(-0.3\%)$ & $4.271$ & $5.098$ & $\textbf{4.292(1.1\%)}$ & $4.341$ & $5.201$ & $4.701(-0.9\%)$ & $4.657$ & $5.471$ & $\textbf{4.811(0.2\%)}$ & $4.821$ & $5.481$\\
        & $200$ & $4.211(-0.8\%)$ & $4.179$ & $5.020$ & $\textbf{4.192(1.0\%)}$ & $4.236$ & $5.037$ & $\textbf{4.487(0.3\%)}$ & $4.499$ & $5.224$ & $4.618(-0.2\%)$ & $4.610$ & $5.298$\\
        DGP3 & $100$ & $\textbf{4.054(3.1\%)}$ & $4.186$ & $4.881$ & $\textbf{4.515(2.2\%)}$ & $4.616$ & $5.137$ & $4.710(-0.9\%)$ & $4.666$ & $5.141$ & $5.014(-1.8\%)$ & $4.925$ & $5.415$\\
        & $150$ & $\textbf{4.063(2.5\%)}$ & $4.167$ & $4.723$ & $\textbf{4.281(3.0\%)}$ & $4.412$ & $4.852$ & $\textbf{4.513(0.9\%)}$ & $4.553$ & $4.947$ & $4.611(-0.2\%)$ & $4.603$ & $4.972$\\
        & $200$ & $\textbf{3.969(1.9\%)}$ & $4.044$ & $4.593$ & $4.184(-0.4\%)$ & $4.165$ & $4.618$ & $\textbf{4.381(1.6\%)}$ & $4.450$ & $4.754$ & $4.509(-0.9\%)$ & $4.470$ & $4.802$\\
        \bottomrule
        \end{tabular}}}
        \end{center}
        \label{table:sizeDGP123para}
    \end{table*}

\subsubsection{Additional results under data-driven selected $\omega$ and $h$}

    We present additional experimental results on parameter selection in this section, including: (i) the absolute difference between the selected bandwidth $\hat{h}$ and the optimal bandwidth $h^*$ for ELCP in Figure~\ref{fig:paraabsmis for DGP13 5}, and (ii) the difference in test-conditional miscoverage error of ELCP between the selected parameters and the optimal parameters in Figure~\ref{fig:paratestmis for DGP13 5}.
    Both metrics are evaluated against the calibration size $n$ and averaged over 100 trials.

    Figure~\ref{fig:paraabsmis for DGP13 5} shows that the absolute difference between the selected bandwidth $\hat{h}$ and the optimal bandwidth $h^*$ tends to decrease as the calibration size $n$ increases. Similarly, Figure~\ref{fig:paratestmis for DGP13 5} indicates that the difference in test-conditional miscoverage error between the selected and optimal parameters is small and further diminishes with larger $n$.
    \begin{figure}[htbp]
        \centering
        \includegraphics[width=.8\linewidth]{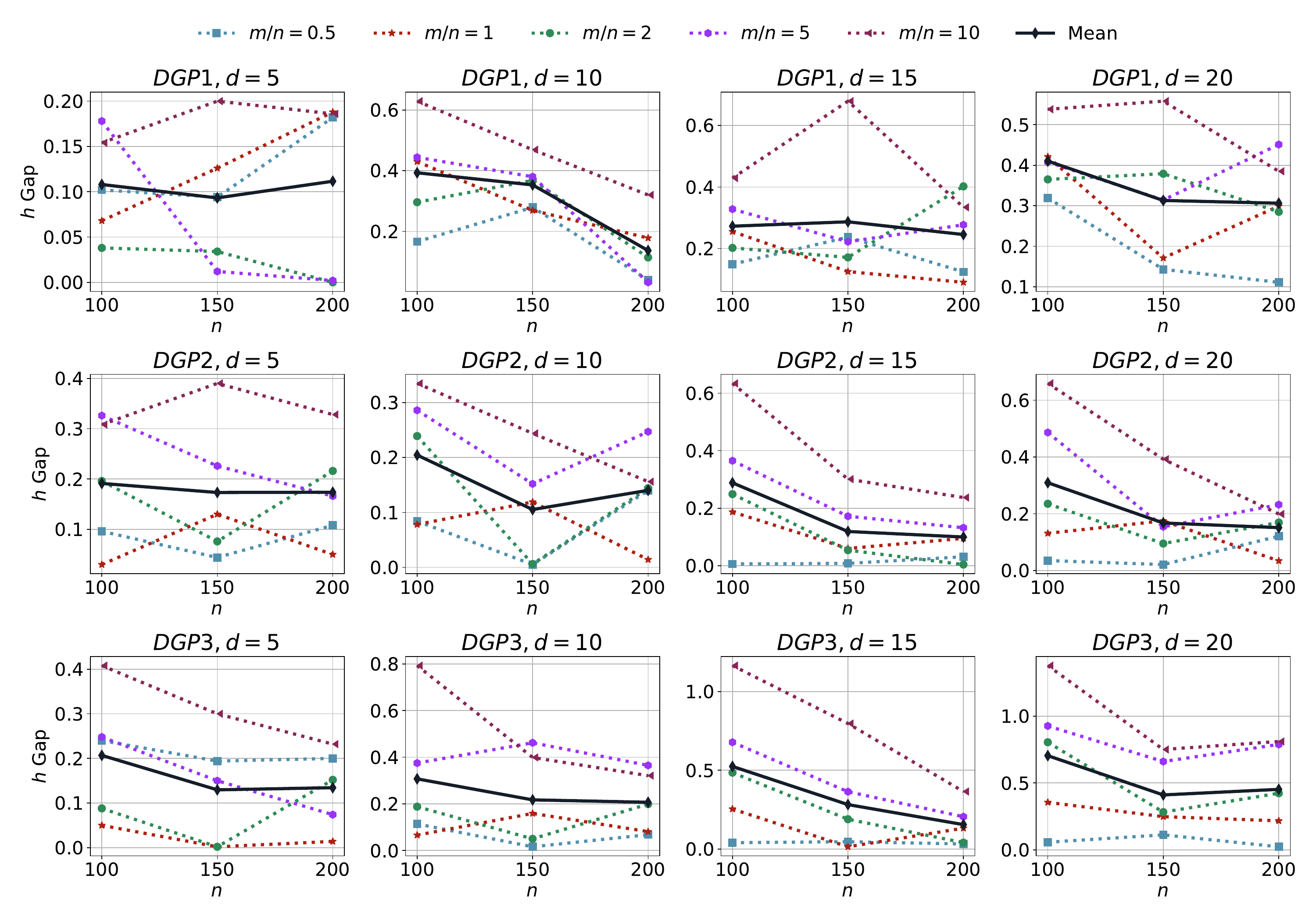}
        \caption{Absolute difference between selected bandwidth $\hat{h}$ and the optimal bandwidth $h^*$ for ELCP for DGP1--DGP3.}\label{fig:paraabsmis for DGP13 5}
    \end{figure}
    \begin{figure}[htbp]
        \centering
        \includegraphics[width=.8\linewidth]{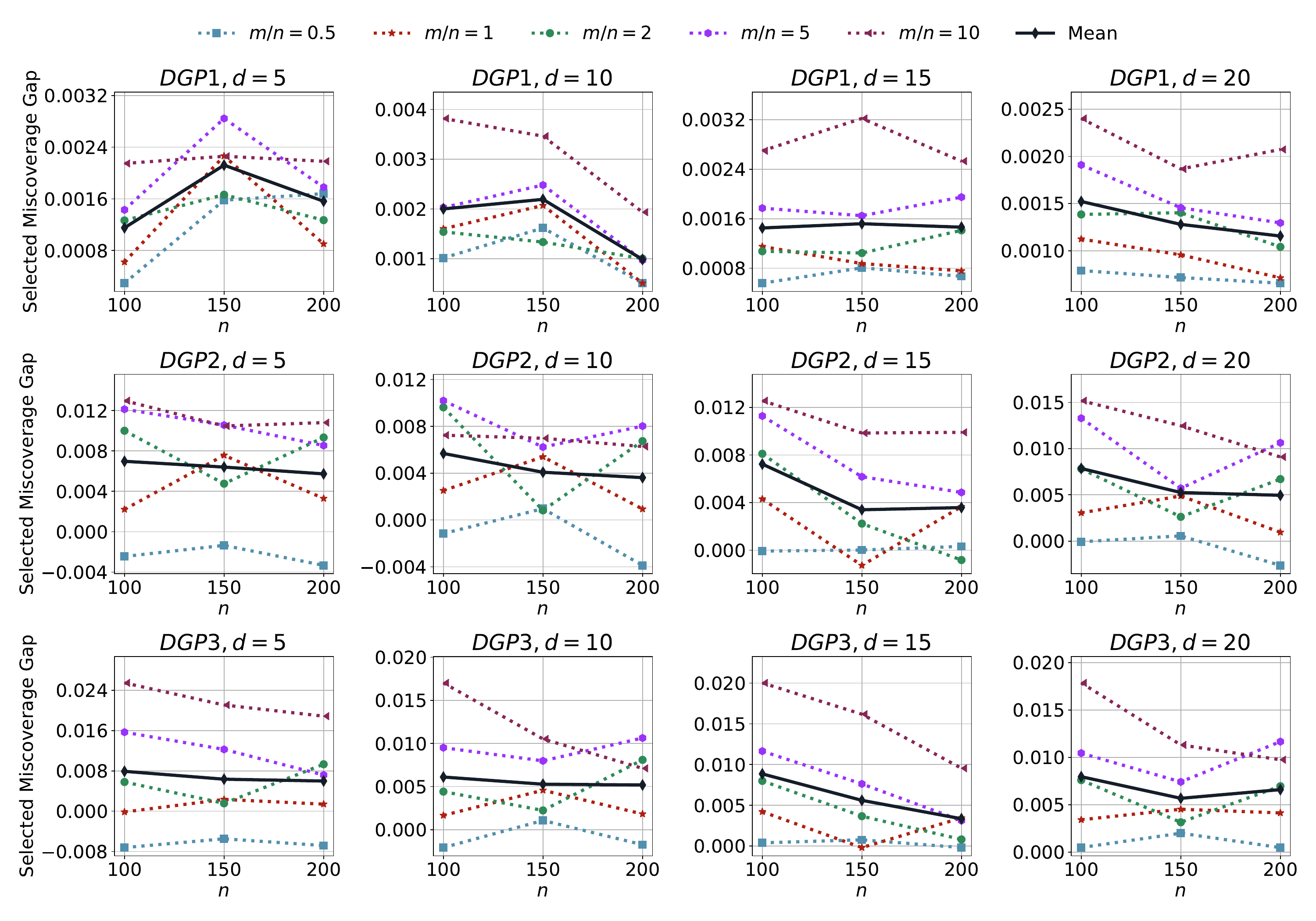}
        \caption{Difference in test-conditional miscoverage error of ELCP between selected bandwidth $\hat{h}$ and the optimal bandwidth $h^*$ for DGP1--DGP3.}\label{fig:paratestmis for DGP13 5}
    \end{figure}

\FloatBarrier
\subsection{Semi-Supervised Setting}\label{sec:simu_semi}
    
    In this section, we conduct simulation studies under the semi-supervised setting, where only the covariates are available, and the response variable is missing in the auxiliary data.
    As discussed in Section \ref{sec:semi}, the covariate $X^{\prime}$ in the auxiliary data is assumed to follow the same distribution as $X$ in the calibration data.
    Specifically, we consider the following DGPs:

    \noindent\textbf{DGP5:} $X\sim \mathrm{Uniform}[-2,2]^d,\ \mu(X)=\overset{d}{\underset{i=1}\sum}X_i,\ \epsilon(X)=\left\{ \sum_{i=1}^{5}|\arccos(X_i/2)| \right\}^{1/2}*\mathcal{M}$,\\
    \noindent\textbf{DGP6:} $X\sim \mathrm{Uniform}[-2,2]^d,\ \mu(X)=\overset{d}{\underset{i=1}\sum}X_i,\ \epsilon(X)=\exp(\sum_{i=1}^{5}X_i/4)*\mathcal{M}$,\\
    \noindent\textbf{DGP7:} $X\sim \mathrm{Uniform}[-2,2]^d,\ \mu(X)=\overset{d}{\underset{i=1}\sum}X_i,\ \epsilon(X)=|\sum_{i=1}^{5}X_i/2|*\mathcal{M}$,\\
    where $\mathcal{M}$ is the Gaussian mixture distribution $0.15N(-5,1)+0.7N(0,1)+0.15N(5,1)$.
    The operator $*$ represents a scaling operation applied to the mixture distributions $\mathcal{M}$. 
    In both cases, the auxiliary data consist only of the covariates $X^{\prime}$ without any corresponding response values.

    We consider calibration and auxiliary sizes $n\in\{100,150,200\}$ and $m/n=5$ for DGP5--DGP7. Two experimental settings are examined:
    \begin{enumerate}
        \item \textbf{With additional information:} Pre-estimated mean and variance are assumed to be available (e.g., when extra samples exist for estimating moments but cannot be used for calibration due to privacy constraints).

        \item \textbf{Without additional information:} Mean and variance must be estimated solely from the training data, as no external moment information is provided.
    \end{enumerate}

    \subsubsection{With additional information}
    
    To predict $E\left( Y\mid X \right)$ and $\Var\left( Y\mid X \right)$, we train models $\nu(X)$ and $\sigma^2(X)$, respectively, using additional samples of size $n_a/n\in\{1,2,5,10\}$ from the calibration distribution. For the auxiliary covariate $X^\prime$, we generate corresponding response $Y^\prime$ by sampling from $N\left( \nu(X^\prime), \sigma^2(X^\prime) \right)$. 

    Figures~\ref{fig:test for DGP5 5 alpha}--\ref{fig:test for DGP7 5 alpha} show the test-conditional miscoverage error of ELCP, LCP, and RLCP for $d \in {5,10,15,20}$. ELCP consistently outperforms both LCP and RLCP, with only minor sensitivity to the additional sample size $n_a$ used for mean and variance estimation. The distribution $\mathcal{M}$ is constructed so that even perfect moment estimates cannot match the calibration distribution.
    
    \begin{figure}[htp]
        \centering
        \includegraphics[width=.85\linewidth]{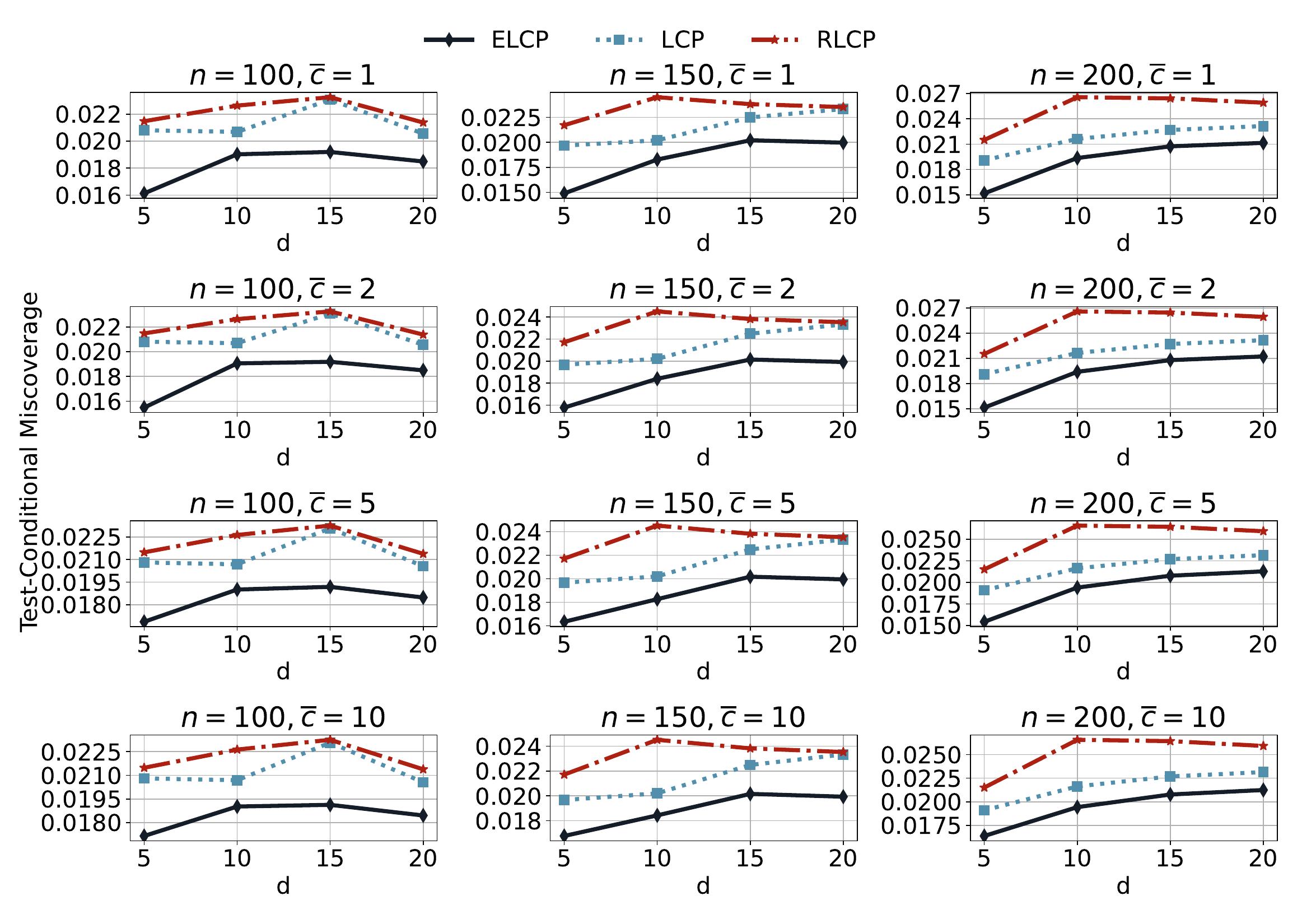}
        \caption{Test-conditional miscoverage error for DGP5 with additional training set.}\label{fig:test for DGP5 5 alpha}
    \end{figure}
    \begin{figure}[htp]
        \centering
        \includegraphics[width=.85\linewidth]{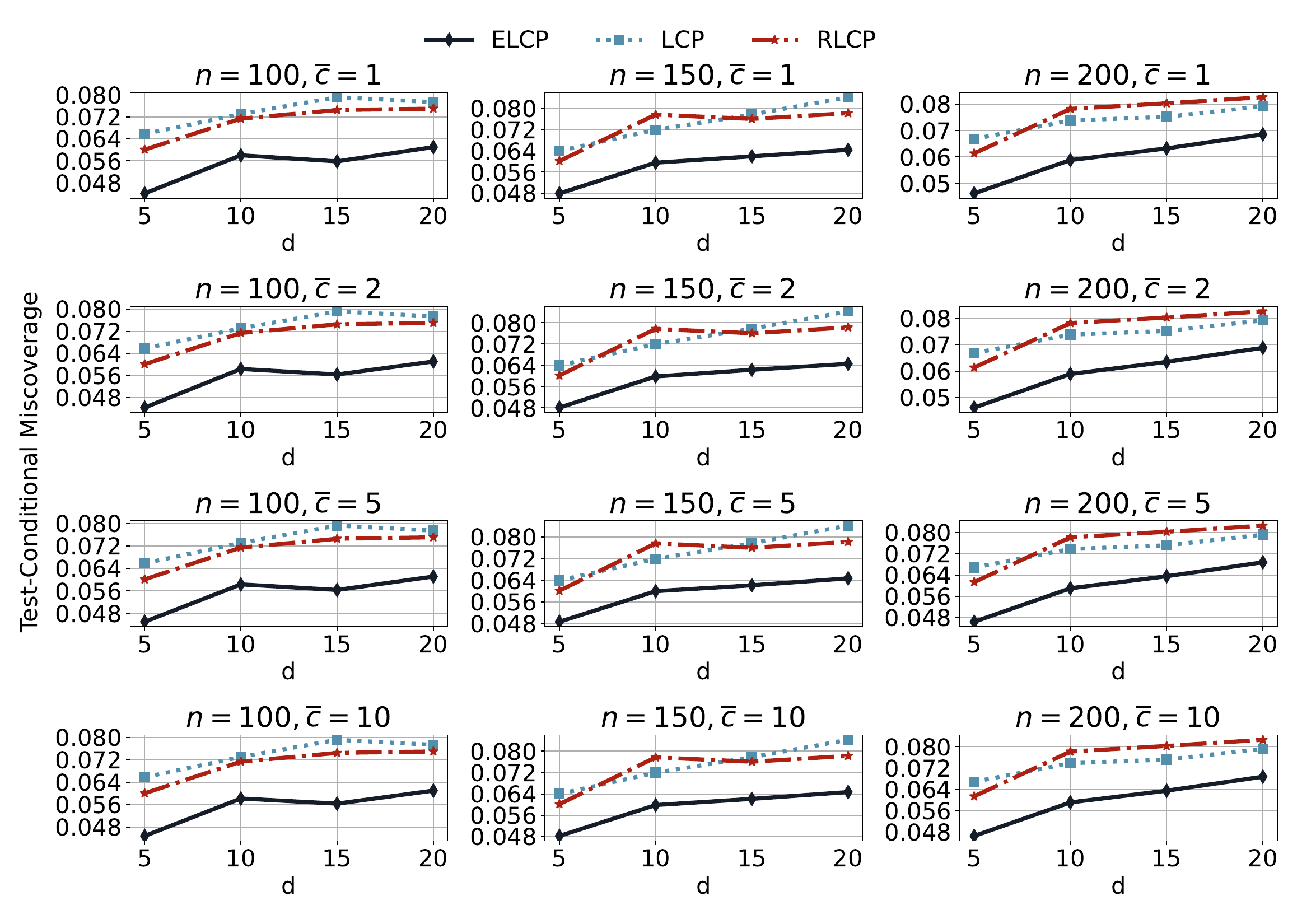}
        \caption{Test-conditional miscoverage error for DGP6 with additional training set.}\label{fig:test for DGP6 5 alpha}
    \end{figure}
    \begin{figure}[htp]
        \centering
        \includegraphics[width=.85\linewidth]{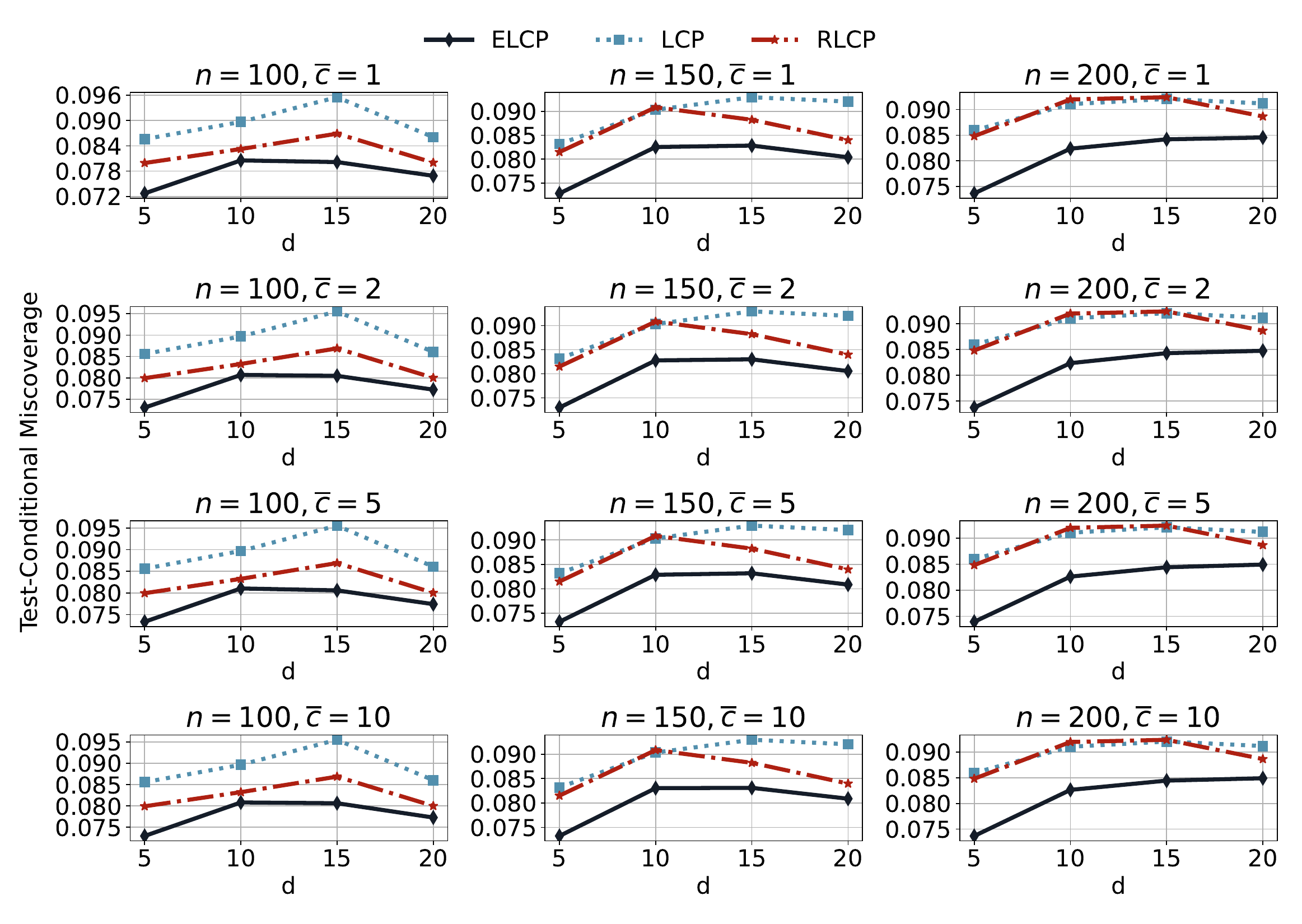}
        \caption{Test-conditional miscoverage error for DGP7 with additional training set.}\label{fig:test for DGP7 5 alpha}
    \end{figure}

    Figures~\ref{fig:marginal for DGP5 5 alpha}--\ref{fig:marginal for DGP7 5 alpha} show the gap between each method’s marginal coverage and the target level $1-\alpha$. For example, with ELCP using $\omega=0.01$, we select bandwidths $h_i \in \{0.4,0.6,0.8,1.0,1.2,1.4,1.6,1.8,2.0,2.5,3.0,3.5,4.0\}$, compute the marginal coverage ${\rm cvrg}_i$ for each $h$, and report the average gap $\sum_{i=1}^{13} |{\rm cvrg}_i - (1-\alpha)|/13$. Results show that LCP-C and RLCP-C have noticeable marginal coverage gaps, though they remain modest because there is no covariate shift and auxiliary labels are generated from the calibration distribution’s estimated moments.
    
    %Figure~\ref{fig:marginal for DGP5 5 alpha}-\ref{fig:marginal for DGP7 5 alpha} presents the gap between each method's marginal coverage and the required level $1-\alpha$. Using ELCP with $\omega=0.01$ as an example, the calculation proceeds as follows: for a given configuration, select bandwidths $h_i \in \{0.4,0.6,0.8,1.0,1.2,1.4,1.6,1.8,2.0,2.5,3.0,3.5,4.0\}$, obtain the corresponding marginal coverage ${\rm cov}_i$, then compute the final gap as $\sum_{i=1}^{13} |{\rm cov}_i - (1-\alpha)|/13$. The results show that both LCP-C and RLCP-C exhibit gaps in their achieved marginal coverage compared to the required level. Since there is no covariate shift and the auxiliary data labels are generated based on estimated means and variances from the calibration distribution, the distributional discrepancy between auxiliary and calibration data remains relatively small. Consequently, while marginal coverage guarantees of LCP-C and RLCP-C cannot be ensured, the observed coverage gaps are not particularly large.
    \begin{figure}[htp]
        \centering
        \includegraphics[width=.95\linewidth]{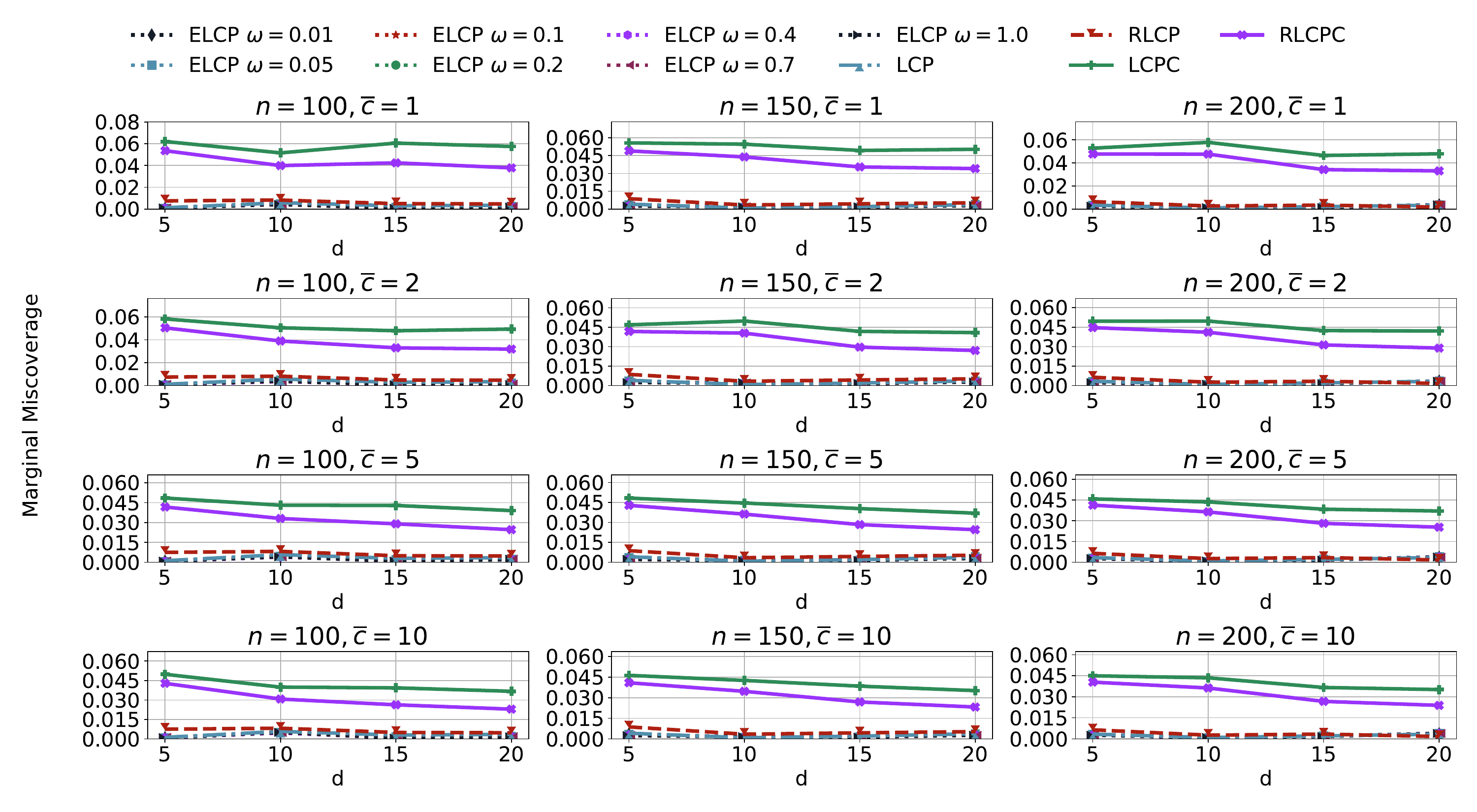}
        \caption{Marginal miscoverage for DGP5 with additional training set.}\label{fig:marginal for DGP5 5 alpha}
    \end{figure}
    \begin{figure}[htp]
        \centering
        \includegraphics[width=.95\linewidth]{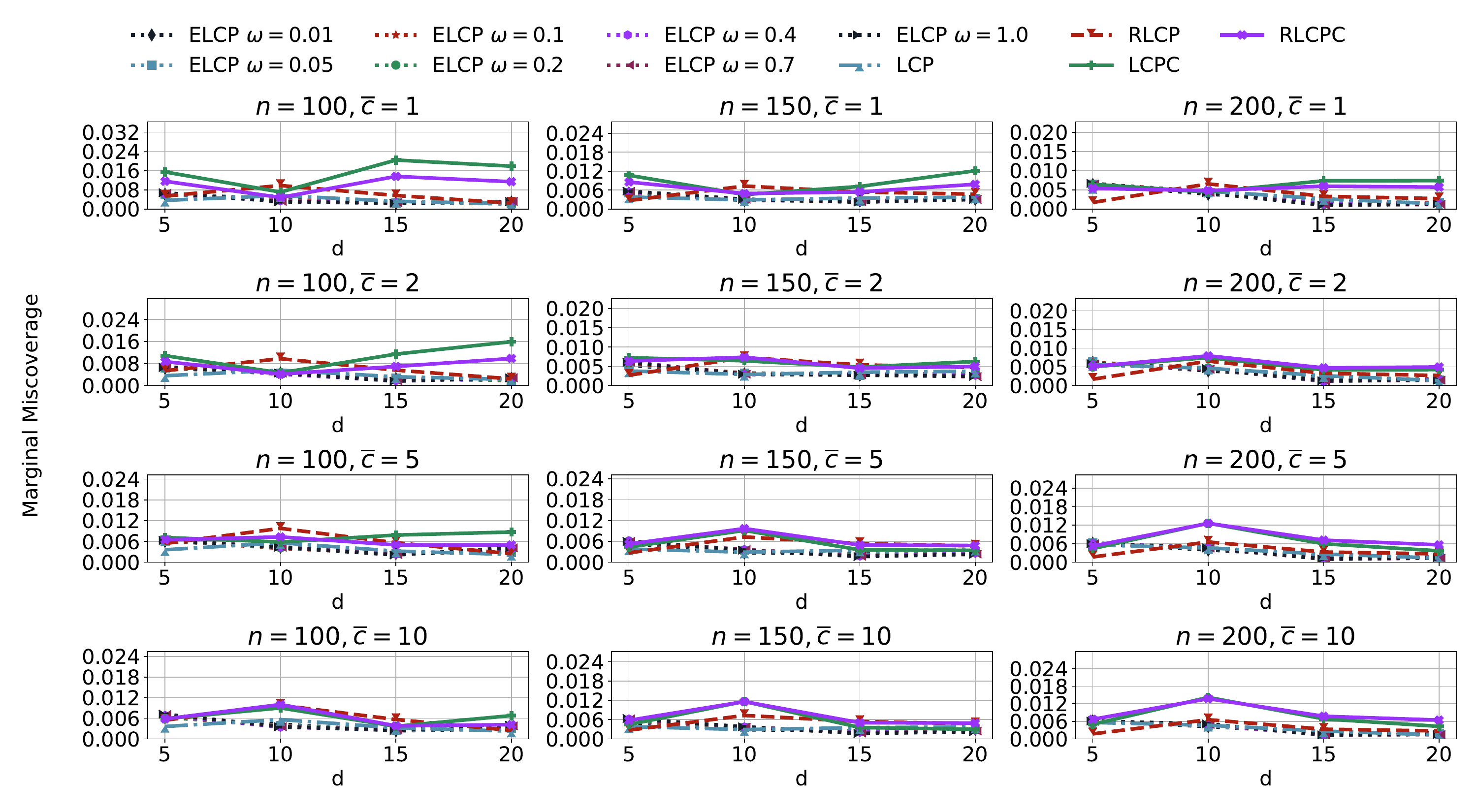}
        \caption{Marginal miscoverage for DGP6 with additional training set.}\label{fig:marginal for DGP6 5 alpha}
    \end{figure}
    \begin{figure}[htp]
        \centering
        \includegraphics[width=.95\linewidth]{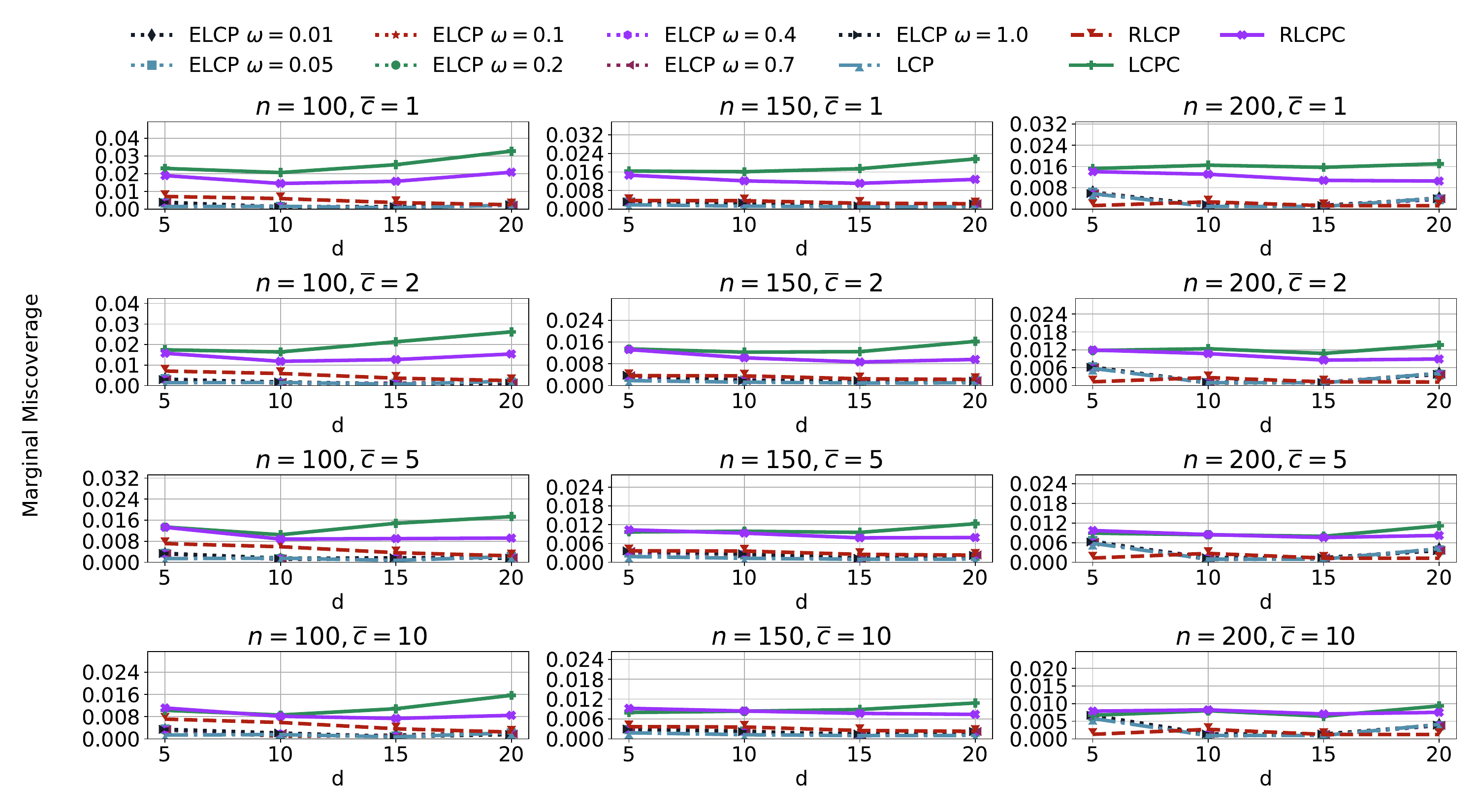}
        \caption{Marginal miscoverage for DGP7 with additional training set.}\label{fig:marginal for DGP7 5 alpha}
    \end{figure}

\subsubsection{Without additional information}

    In more general settings without external mean and variance estimates, we estimate these moments directly from the training dataset $\mathcal{D}_{\rm tr}$. Figure~\ref{fig:test for DGP57 5 alpha without} reports the test-conditional miscoverage error for DGP5--DGP7 under this setup, showing that ELCP consistently outperforms both LCP and RLCP. Figure~\ref{fig:marginal for DGP57 5 alpha without} compares the achieved marginal coverage across all methods with the target level $1-\alpha$. The results indicate that while ELCP, LCP, and RLCP achieve coverage close to the nominal level, LCP-C and RLCP-C exhibit significant deviations from the desired coverage.
    \begin{figure}[h]
        \centering
        \includegraphics[width=.8\linewidth]{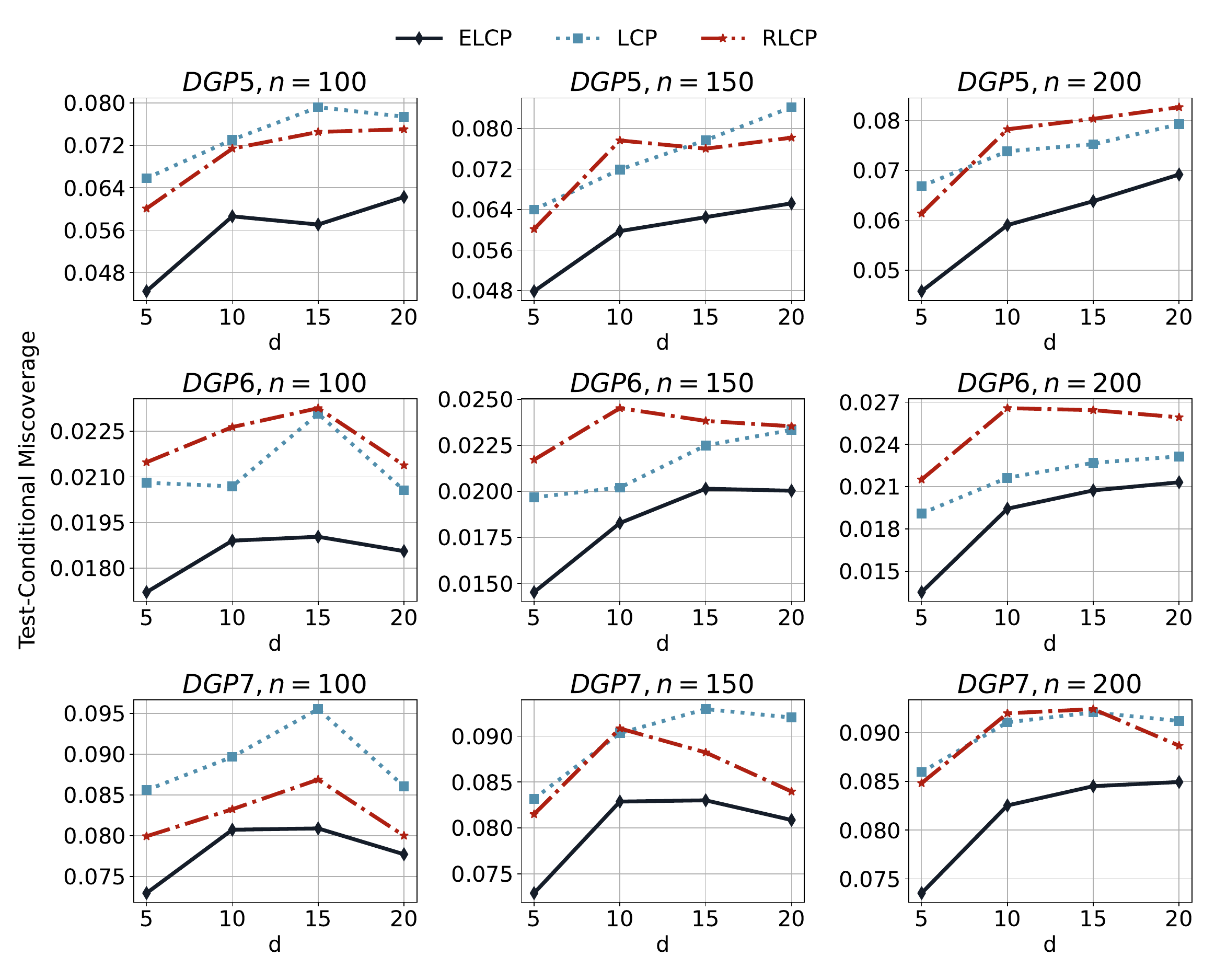}
        \caption{Test-conditional miscoverage error for DGP5--DGP7 without additional training set.}\label{fig:test for DGP57 5 alpha without}
    \end{figure}
    \begin{figure}[H]
        \centering
        \includegraphics[width=.8\linewidth]{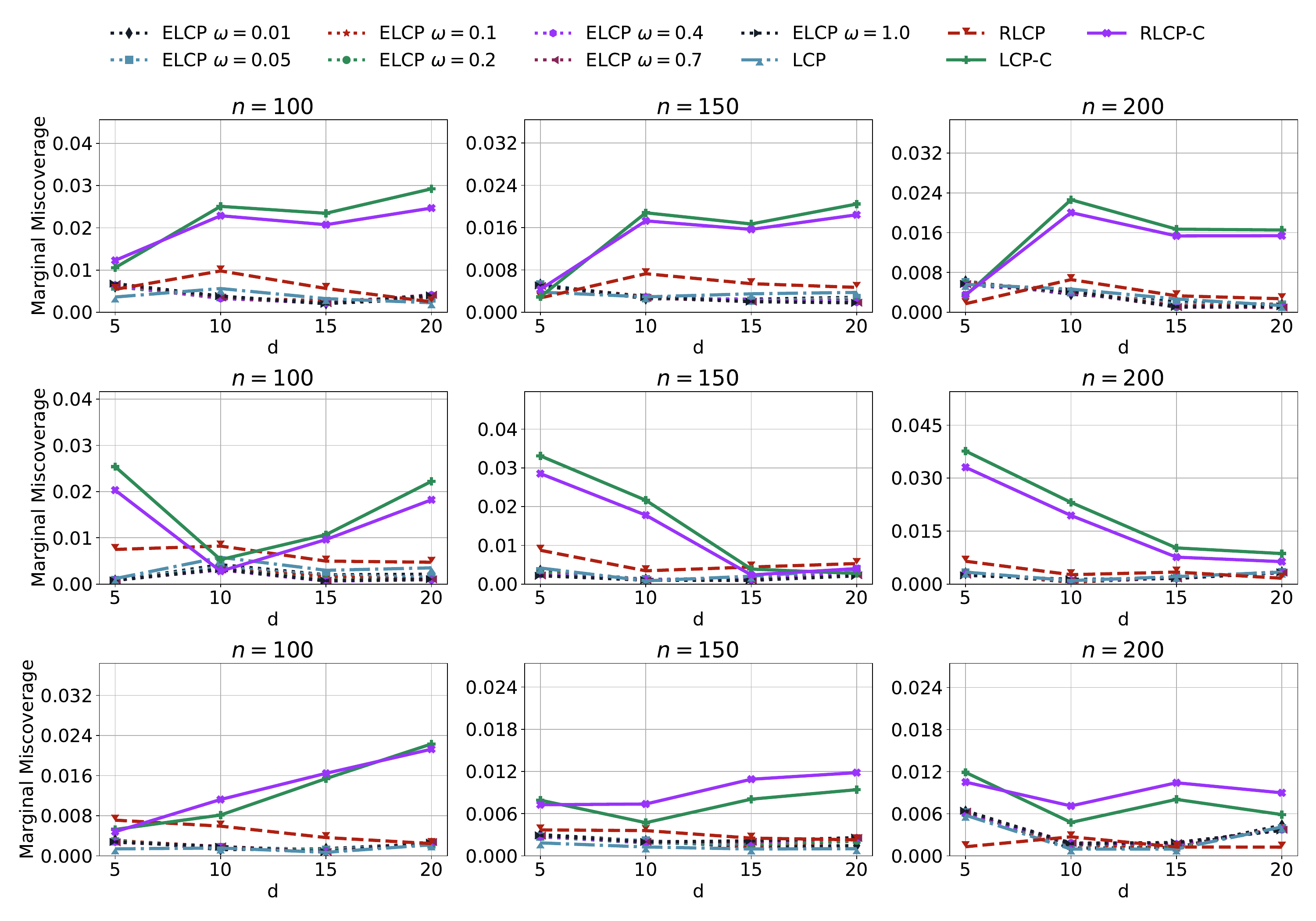}
        \caption{Marginal miscoverage for DGP5--DGP7 without additional training set.}\label{fig:marginal for DGP57 5 alpha without}
    \end{figure}

    \iffalse
    Figure \ref{fig:TestCTADGP45Full} reports the smallest test-conditional miscoverage error across different bandwidths for each method, revealing the advantage of ELCP compared to LCP and RLCP across different scenarios.
    The complete results for all bandwidth values are provided in Table~\ref{table:semi_Test_45}.
    Similar phenomena are observed for calibration-test-conditional miscoverage errors.
    These results highlight that, even in the semi-supervised setting where the auxiliary data contain only covariates, our proposed ELCP method improves local coverage, which is achieved by leveraging predictive models $\nu(X)$ and $\sigma(X)$ to generate synthetic labels $Y^{\prime}$.
    \fi

\FloatBarrier
\subsection{Real data analysis: predicting Moscow housing price -- detailed implementation}\label{sec:supp_housing}

We apply K-means clustering ($k=12$) to subway station coordinates to partition the city into regions. Each property is then assigned to the cluster of its nearest station, ensuring that properties are grouped according to geographic proximity to the transit network.
Figure~\ref{fig:moscow metro} displays the geographical distribution of subway stations and their corresponding cluster assignments based on longitude and latitude coordinates.
\begin{figure}[htbp]
    \centering
    \includegraphics[scale=0.45]{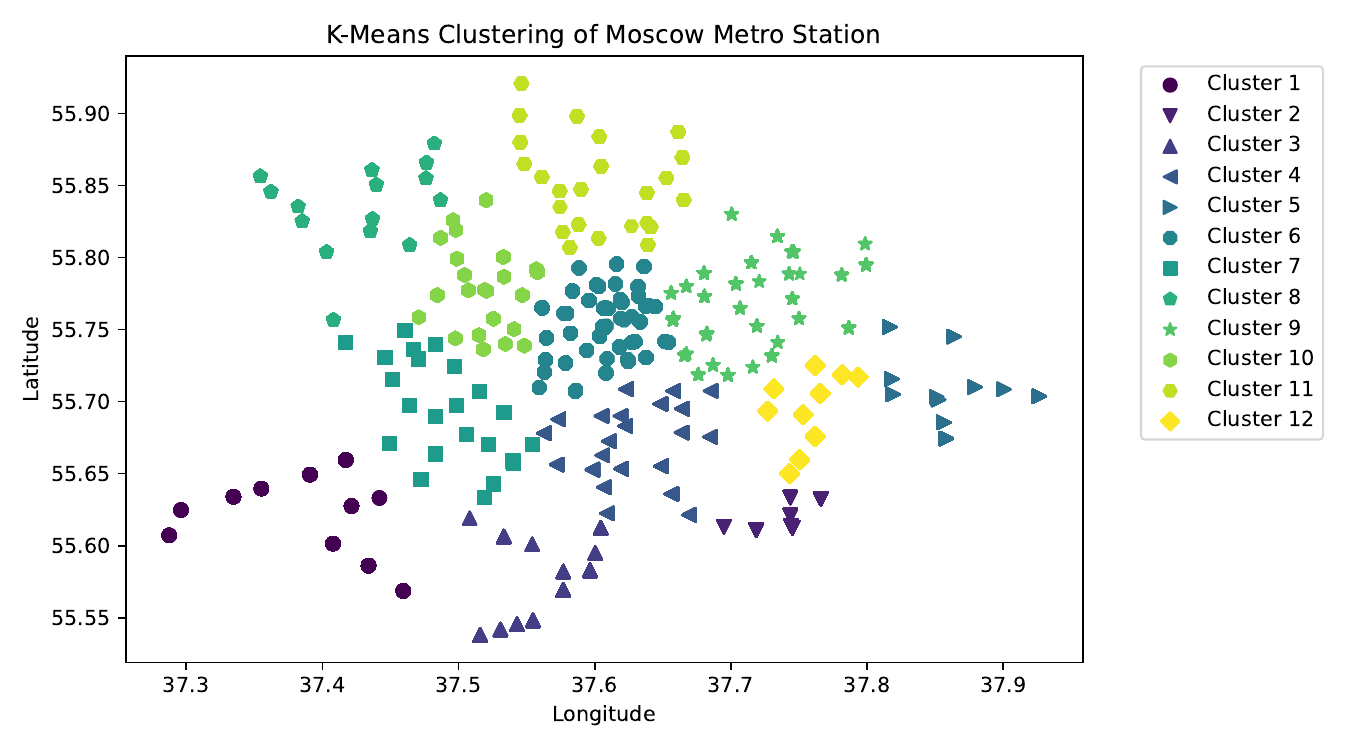}
    \caption{Moscow metro clustering.}\label{fig:moscow metro}
\end{figure}
Cluster 1, representing outlying urban areas with 581 samples, serves as our target distribution. We select adjacent suburban areas from clusters 2 and 3, totaling 1,813 samples, to form the auxiliary dataset.

Next, we supplement the results under various score functions and pre-training schemes. As discussed in Section~\ref{sec:supp_simu_score}, we consider scores based on residuals (RES), conditional distribution functions (CDF), and conditional densities (CDE). Moreover, because auxiliary information can already be transferred during the pre-training stage in constructing the scores, we follow the notation of Section~\ref{sec:supp_simu_score} and denote methods that transfer only the auxiliary training data with the suffix `-AT', and those that use all auxiliary data with the suffix `-AA'. The test-conditional miscoverage errors and mean prediction set size of the different methods are presented in Table~\ref{table:houseT} and Table~\ref{table:houseS}, respectively. The results of ELCP are highlighted in bold. A single asterisk (*) marks the best result among the three pre-training schemes for a given score function, and a double asterisk (**) indicates the best overall result across all settings. 

\begin{table*}[htbp]
    \caption{Weak test-conditional miscoverage errors in predicting Moscow housing price.}\label{table:houseT}
    \begin{center}
    {\fontsize{9}{6.4}\selectfont{\setlength{\tabcolsep}{3.pt}\begin{tabular}{cccccccccccc}
    \toprule
    Interval Index & 1 & 2 & 3 & 4 & 5 & 6 & 7 & 8 & 9 & 10 & Agg\vspace{4pt}\\ \midrule
    Prop. & $0.155$ & $0.120$ & $0.117$ & $0.115$ & $0.107$ & $0.106$ & $0.106$ & $0.091$ & $0.065$ & $0.018$ & $1.000$\\
    ELCP (RES) & $\textbf{0.030}$ & $\textbf{0.001}$ & $\textbf{0.029}$ & $\textbf{0.053}$ & $\textbf{0.027}$ & $\textbf{0.018}$ & $\textbf{0.066}$ & $\textbf{0.070}$ & $\textbf{0.023}$ & $\textbf{0.122}$ & $\textbf{0.0360}$\\
    LCP (RES) & $0.042$ & $0.007$ & $0.034$ & $0.065$ & $0.039$ & $0.021$ & $0.057$ & $0.048$ & $0.068$ & $0.122$ & $0.0422$\\
    RLCP (RES) & $0.028$ & $0.006$ & $0.036$ & $0.070$ & $0.043$ & $0.035$ & $0.071$ & $0.089$ & $0.085$ & $0.012$ & $0.0471$\\
    ELCP (RES-AT) & $\textbf{0.013}$ & $\textbf{0.008}$ & $\textbf{0.031}$ & $\textbf{0.037}$ & $\textbf{0.029}$ & $\textbf{0.008}$ & $\textbf{0.070}$ & $\textbf{0.047}$ & $\textbf{0.022}$ & $\textbf{0.096}$ & $\textbf{0.0296}$\\
    LCP (RES-AT) & $0.024$ & $0.005$ & $0.025$ & $0.055$ & $0.038$ & $0.004$ & $0.058$ & $0.054$ & $0.054$ & $0.088$ & $0.0341$\\
    RLCP (RES-AT) & $0.016$ & $0.004$ & $0.036$ & $0.063$ & $0.037$ & $0.031$ & $0.074$ & $0.093$ & $0.087$ & $0.040$ & $0.0443$\\
    ELCP (RES-AA) & $\textbf{0.009}$ & $\textbf{0.005}$ & $\textbf{0.034}$ & $\textbf{0.021}$ & $\textbf{0.011}$ & $\textbf{0.003}$ & $\textbf{0.064}$ & $\textbf{0.025}$ & $\textbf{0.020}$ & $\textbf{0.107}$ & $\textbf{0.0220}^*$\\
    LCP (RES-AA) & $0.023$ & $0.001$ & $0.030$ & $0.043$ & $0.025$ & $0.003$ & $0.049$ & $0.048$ & $0.047$ & $0.107$ & $0.0296$\\
    RLCP (RES-AA) & $0.002$ & $0.001$ & $0.048$ & $0.064$ & $0.025$ & $0.027$ & $0.067$ & $0.110$ & $0.087$ & $0.005$ & $0.0419$\\
    \midrule
    ELCP (CDF) & $\textbf{0.017}$ & $\textbf{0.006}$ & $\textbf{0.011}$ & $\textbf{0.013}$ & $\textbf{0.023}$ & $\textbf{0.002}$ & $\textbf{0.015}$ & $\textbf{0.017}$ & $\textbf{0.012}$ & $\textbf{0.065}$ & $\textbf{0.0138}^{**}$\\
    LCP (CDF) & $0.012$ & $0.007$ & $0.019$ & $0.012$ & $0.016$ & $0.023$ & $0.023$ & $0.029$ & $0.008$ & $0.031$ & $0.0166$\\
    RLCP (CDF) & $0.010$ & $0.007$ & $0.015$ & $0.021$ & $0.011$ & $0.049$ & $0.002$ & $0.035$ & $0.012$ & $0.083$ & $0.0188$\\
    DCP & $0.007$ & $0.004$ & $0.022$ & $0.028$ & $0.009$ & $0.048$ & $0.027$ & $0.045$ & $0.034$ & $0.100$ & $0.0243$\\
    ELCP (CDF-AT) & $\textbf{0.000}$ & $\textbf{0.008}$ & $\textbf{0.021}$ & $\textbf{0.009}$ & $\textbf{0.026}$ & $\textbf{0.014}$ & $\textbf{0.001}$ & $\textbf{0.048}$ & $\textbf{0.076}$ & $\textbf{0.164}$ & $\textbf{0.0210}$\\
    LCP (CDF-AT) & $0.030$ & $0.004$ & $0.019$ & $0.013$ & $0.024$ & $0.026$ & $0.005$ & $0.040$ & $0.009$ & $0.127$ & $0.0211$\\
    RLCP (CDF-AT) & $0.029$ & $0.015$ & $0.021$ & $0.018$ & $0.023$ & $0.048$ & $0.005$ & $0.046$ & $0.062$ & $0.066$ & $0.0282$\\
    DCP-AT & $0.052$ & $0.007$ & $0.024$ & $0.023$ & $0.019$ & $0.051$ & $0.002$ & $0.053$ & $0.150$ & $0.100$ & $0.0382$\\
    ELCP (CDF-AA) & $\textbf{0.005}$ & $\textbf{0.005}$ & $\textbf{0.020}$ & $\textbf{0.000}$ & $\textbf{0.024}$ & $\textbf{0.012}$ & $\textbf{0.014}$ & $\textbf{0.050}$ & $\textbf{0.102}$ & $\textbf{0.096}$ & $\textbf{0.0222}$\\
    LCP (CDF-AA) & $0.050$ & $0.006$ & $0.017$ & $0.004$ & $0.014$ & $0.027$ & $0.013$ & $0.045$ & $0.040$ & $0.036$ & $0.0241$\\
    RLCP (CDF-AA) & $0.039$ & $0.011$ & $0.019$ & $0.007$ & $0.030$ & $0.052$ & $0.014$ & $0.041$ & $0.053$ & $0.080$ & $0.0294$\\
    DCP-AA & $0.078$ & $0.003$ & $0.020$ & $0.009$ & $0.008$ & $0.051$ & $0.007$ & $0.052$ & $0.088$ & $0.100$ & $0.0351$\\
    \midrule
    ELCP (CDE) & $\textbf{0.003}$ & $\textbf{0.039}$ & $\textbf{0.016}$ & $\textbf{0.029}$ & $\textbf{0.037}$ & $\textbf{0.001}$ & $\textbf{0.002}$ & $\textbf{0.043}$ & $\textbf{0.025}$ & $\textbf{0.099}$ & $\textbf{0.0218}^*$\\
    LCP (CDE) & $0.015$ & $0.034$ & $0.026$ & $0.044$ & $0.034$ & $0.003$ & $0.004$ & $0.045$ & $0.019$ & $0.059$ & $0.0253$\\
    RLCP (CDE) & $0.021$ & $0.032$ & $0.027$ & $0.040$ & $0.034$ & $0.007$ & $0.009$ & $0.061$ & $0.025$ & $0.002$ & $0.0273$\\
    CDE & $0.031$ & $0.047$ & $0.037$ & $0.072$ & $0.055$ & $0.002$ & $0.008$ & $0.092$ & $0.014$ & $0.176$ & $0.0425$\\
    ELCP (CDE-AT) & $\textbf{0.008}$ & $\textbf{0.022}$ & $\textbf{0.004}$ & $\textbf{0.032}$ & $\textbf{0.044}$ & $\textbf{0.013}$ & $\textbf{0.052}$ & $\textbf{0.067}$ & $\textbf{0.038}$ & $\textbf{0.102}$ & $\textbf{0.0300}$\\
    LCP (CDE-AT) & $0.008$ & $0.027$ & $0.006$ & $0.060$ & $0.050$ & $0.011$ & $0.054$ & $0.103$ & $0.089$ & $0.022$ & $0.0399$\\
    RLCP (CDE-AT) & $0.002$ & $0.025$ & $0.002$ & $0.059$ & $0.059$ & $0.021$ & $0.058$ & $0.110$ & $0.090$ & $0.012$ & $0.0410$\\
    CDE-AT & $0.002$ & $0.041$ & $0.010$ & $0.092$ & $0.072$ & $0.020$ & $0.079$ & $0.165$ & $0.097$ & $0.048$ & $0.0574$\\
    ELCP (CDE-AA) & $\textbf{0.003}$ & $\textbf{0.031}$ & $\textbf{0.012}$ & $\textbf{0.044}$ & $\textbf{0.049}$ & $\textbf{0.007}$ & $\textbf{0.049}$ & $\textbf{0.058}$ & $\textbf{0.054}$ & $\textbf{0.105}$ & $\textbf{0.0324}$\\
    LCP (CDE-AA) & $0.009$ & $0.028$ & $0.009$ & $0.068$ & $0.054$ & $0.021$ & $0.051$ & $0.085$ & $0.087$ & $0.025$ & $0.0408$\\
    RLCP (CDE-AA) & $0.001$ & $0.022$ & $0.000$ & $0.060$ & $0.057$ & $0.022$ & $0.054$ & $0.096$ & $0.088$ & $0.029$ & $0.0387$\\
    CDE-AA & $0.004$ & $0.048$ & $0.001$ & $0.101$ & $0.074$ & $0.026$ & $0.082$ & $0.150$ & $0.096$ & $0.130$ & $0.0597$\\
    \bottomrule
    \end{tabular}}}
    \end{center}
\end{table*}
\begin{table*}[htbp]
    \caption{Average prediction set sizes in predicting Moscow housing price.}\label{table:houseS}
    \begin{center}
    {\fontsize{9}{6.4}\selectfont{\setlength{\tabcolsep}{3pt}\begin{tabular}{cccccccccccc}
    \toprule
    Interval Index & 1 & 2 & 3 & 4 & 5 & 6 & 7 & 8 & 9 & 10 & Agg\vspace{4pt}\\ \midrule
    Prop. & $0.155$ & $0.120$ & $0.117$ & $0.115$ & $0.107$ & $0.106$ & $0.106$ & $0.091$ & $0.065$ & $0.018$ & $1.000$\\
    ELCP (RES) & $\textbf{7.93}$ & $\textbf{7.21}$ & $\textbf{7.20}$ & $\textbf{10.75}$ & $\textbf{5.92}$ & $\textbf{7.33}$ & $\textbf{6.20}$ & $\textbf{11.50}$ & $\textbf{4.19}$ & $\textbf{9.15}$ & $\textbf{7.72}^*$\\
    LCP (RES) & $7.68$ & $7.31$ & $7.07$ & $10.48$ & $6.16$ & $7.27$ & $5.77$ & $12.12$ & $5.20$ & $8.71$ & $7.74$\\
    RLCP (RES) & $9.45$ & $8.52$ & $9.52$ & $13.84$ & $7.47$ & $12.49$ & $9.63$ & $13.80$ & $13.14$ & $28.94$ & $10.97$\\
    ELCP (RES-AT) & $\textbf{8.60}$ & $\textbf{7.34}$ & $\textbf{7.36}$ & $\textbf{11.32}$ & $\textbf{6.27}$ & $\textbf{7.79}$ & $\textbf{6.53}$ & $\textbf{12.32}$ & $\textbf{4.57}$ & $\textbf{9.97}$ & $\textbf{8.16}$\\
    LCP (RES-AT) & $8.15$ & $7.25$ & $7.06$ & $10.70$ & $6.45$ & $7.51$ & $5.72$ & $12.11$ & $5.93$ & $9.37$ & $7.94$\\
    RLCP (RES-AT) & $11.05$ & $8.69$ & $10.28$ & $13.95$ & $7.94$ & $12.55$ & $10.43$ & $14.62$ & $13.49$ & $29.70$ & $11.59$\\
    ELCP (RES-AA) & $\textbf{8.87}$ & $\textbf{7.14}$ & $\textbf{7.13}$ & $\textbf{11.38}$ & $\textbf{6.31}$ & $\textbf{7.90}$ & $\textbf{6.69}$ & $\textbf{12.45}$ & $\textbf{4.99}$ & $\textbf{10.83}$ & $\textbf{8.25}$\\
    LCP (RES-AA) & $8.07$ & $7.05$ & $6.82$ & $10.64$ & $6.52$ & $7.49$ & $5.84$ & $11.74$ & $6.35$ & $10.18$ & $7.89$\\
    RLCP (RES-AA) & $11.71$ & $8.77$ & $10.55$ & $14.18$ & $7.97$ & $13.24$ & $10.86$ & $14.76$ & $13.96$ & $31.87$ & $11.96$\\
    \midrule
    ELCP (CDF) & $\textbf{10.88}$ & $\textbf{7.25}$ & $\textbf{6.65}$ & $\textbf{13.68}$ & $\textbf{6.84}$ & $\textbf{6.79}$ & $\textbf{8.84}$ & $\textbf{17.74}$ & $\textbf{12.26}$ & $\textbf{8.70}$ & $\textbf{9.86}$\\
    LCP (CDF) & $8.63$ & $6.59$ & $6.58$ & $13.11$ & $6.24$ & $7.31$ & $9.41$ & $16.10$ & $9.26$ & $10.71$ & $9.10$\\
    RLCP (CDF) & $10.21$ & $7.80$ & $8.13$ & $15.28$ & $7.33$ & $11.58$ & $10.90$ & $17.09$ & $12.58$ & $30.74$ & $11.32$\\
    DCP & $8.38$ & $6.87$ & $6.70$ & $13.82$ & $5.83$ & $8.22$ & $5.60$ & $14.02$ & $7.11$ & $21.00$ & $8.70^*$\\
    ELCP (CDF-AT) & $\textbf{12.72}$ & $\textbf{9.21}$ & $\textbf{8.84}$ & $\textbf{14.71}$ & $\textbf{6.88}$ & $\textbf{10.30}$ & $\textbf{7.71}$ & $\textbf{16.30}$ & $\textbf{12.59}$ & $\textbf{10.68}$ & $\textbf{10.94}$\\
    LCP (CDF-AT) & $11.24$ & $8.86$ & $8.85$ & $14.77$ & $6.68$ & $10.47$ & $7.17$ & $15.71$ & $18.86$ & $11.65$ & $10.99$\\
    RLCP (CDF-AT) & $12.98$ & $9.36$ & $11.18$ & $17.00$ & $7.79$ & $15.58$ & $10.14$ & $17.20$ & $16.46$ & $33.09$ & $13.19$\\
    DCP-AT & $10.52$ & $7.81$ & $8.15$ & $15.01$ & $5.95$ & $9.58$ & $5.04$ & $14.83$ & $3.01$ & $24.24$ & $9.42$\\
    ELCP (CDF-AA) & $\textbf{13.90}$ & $\textbf{9.51}$ & $\textbf{9.70}$ & $\textbf{14.91}$ & $\textbf{6.14}$ & $\textbf{9.63}$ & $\textbf{6.37}$ & $\textbf{16.31}$ & $\textbf{11.64}$ & $\textbf{12.81}$ & $\textbf{10.97}$\\
    LCP (CDF-AA) & $11.99$ & $8.91$ & $9.36$ & $14.89$ & $6.23$ & $9.66$ & $5.82$ & $15.35$ & $13.78$ & $15.63$ & $10.62$\\
    RLCP (CDF-AA) & $13.68$ & $9.55$ & $11.85$ & $17.05$ & $7.34$ & $14.39$ & $9.23$ & $17.53$ & $15.93$ & $33.11$ & $13.13$\\
    DCP-AA & $10.85$ & $8.22$ & $8.36$ & $14.99$ & $6.45$ & $10.33$ & $5.12$ & $14.71$ & $3.13$ & $22.99$ & $9.65$\\
    \midrule
    ELCP (CDE) & $\textbf{13.59}$ & $\textbf{13.67}$ & $\textbf{12.97}$ & $\textbf{17.77}$ & $\textbf{12.52}$ & $\textbf{14.83}$ & $\textbf{11.80}$ & $\textbf{18.10}$ & $\textbf{10.51}$ & $\textbf{12.35}$ & $\textbf{14.02}$\\
    LCP (CDE) & $13.89$ & $13.71$ & $13.19$ & $16.95$ & $12.59$ & $14.46$ & $12.05$ & $17.65$ & $10.62$ & $12.79$ & $13.98$\\
    RLCP (CDE) & $14.50$ & $14.30$ & $14.33$ & $18.40$ & $13.23$ & $16.29$ & $13.04$ & $18.51$ & $13.18$ & $20.81$ & $15.20$\\
    CDE & $14.74$ & $14.84$ & $14.35$ & $16.32$ & $13.66$ & $15.06$ & $12.60$ & $16.61$ & $10.85$ & $12.30$ & $14.45$\\
    ELCP (CDE-AT) & $\textbf{8.24}$ & $\textbf{6.90}$ & $\textbf{7.00}$ & $\textbf{9.71}$ & $\textbf{6.07}$ & $\textbf{6.69}$ & $\textbf{6.08}$ & $\textbf{11.37}$ & $\textbf{5.08}$ & $\textbf{7.09}$ & $\textbf{7.54}$\\
    LCP (CDE-AT) & $6.86$ & $6.75$ & $6.85$ & $8.40$ & $6.13$ & $6.87$ & $5.63$ & $9.12$ & $6.07$ & $7.64$ & $6.98^{**}$\\
    RLCP (CDE-AT) & $8.26$ & $7.45$ & $8.36$ & $10.30$ & $6.84$ & $9.47$ & $7.24$ & $10.73$ & $8.40$ & $16.47$ & $8.66$\\
    CDE-AT & $7.13$ & $7.32$ & $7.22$ & $7.68$ & $6.96$ & $7.17$ & $6.32$ & $7.81$ & $6.33$ & $6.28$ & $7.12$\\
    ELCP (CDE-AA) & $\textbf{8.27}$ & $\textbf{7.13}$ & $\textbf{7.10}$ & $\textbf{10.43}$ & $\textbf{6.14}$ & $\textbf{7.01}$ & $\textbf{6.14}$ & $\textbf{12.27}$ & $\textbf{5.24}$ & $\textbf{10.19}$ & $\textbf{7.86}$\\
    LCP (CDE-AA) & $7.27$ & $6.96$ & $7.09$ & $8.83$ & $6.25$ & $6.99$ & $5.62$ & $9.54$ & $6.01$ & $11.12$ & $7.27$\\
    RLCP (CDE-AA) & $8.95$ & $7.66$ & $8.75$ & $10.85$ & $7.02$ & $9.68$ & $7.67$ & $11.50$ & $8.55$ & $20.95$ & $9.15$\\
    CDE-AA & $7.49$ & $7.61$ & $7.52$ & $8.12$ & $7.20$ & $7.36$ & $6.44$ & $8.22$ & $6.44$ & $6.35$ & $7.40$\\
    \bottomrule
    \end{tabular}}}
    \end{center}
\end{table*}
\FloatBarrier

As can be seen from Table~\ref{table:houseT}, the overall best performance is achieved by ELCP using the CDF score function trained solely on the target data, which is significantly better than all other methods. For the residual score function, transferring auxiliary information in score construction substantially improves test-conditional coverage, whereas for the CDF and CDE scores, no such improvement is observed. Nevertheless, under any fixed choice of score function and pre-training scheme, ELCP yields a significant gain compared to LCP and RLCP.

Finally, we report the marginal coverage of each method. Table~\ref{table:realMar} presents the marginal coverage of ELCP, LCP, RLCP, LCP-C, and RLCP-C under different score functions for two real-data experiments: `House' corresponds to the Moscow housing-price prediction results in this chapter, and `Medical' corresponds to the experimental results in Section~\ref{sec:real_data}. 
The table shows that the marginal coverage of LCP-C and RLCP-C deviates substantially from the nominal level 0.9 in most cases. Moreover, in the Moscow housing-price prediction experiment, RLCP-C exhibits a smaller deviation than LCP-C.

\begin{table*}[htp]
    \caption{Marginal coverage in predicting Moscow housing price and medical insurance cost.}\label{table:realMar}
    \begin{center}
    {\fontsize{10}{10}\selectfont{\setlength{\tabcolsep}{3pt}\begin{tabular}{cccccccccccc}
    \toprule
    & \multicolumn{5}{c}{House} && \multicolumn{5}{c}{Medical}\vspace{-2pt}\\
    \cmidrule(lr){2-6} \cmidrule(lr){8-12}
    & ELCP & LCP & RLCP & LCP-C & RLCP-C && ELCP & LCP & RLCP & LCP-C & RLCP-C\vspace{4pt}\\ \midrule
    RES & $0.898$ & $0.901$ & $0.905$ & $\textbf{0.849}$ & $\textbf{0.882}$ && $0.906$ & $0.904$ & $0.911$ & $\textbf{0.928}$ & $\textbf{0.932}$ \\
    RES-AT & $0.899$ & $0.900$ & $0.907$ & $\textbf{0.846}$ & $\textbf{0.876}$ && $0.908$ & $0.906$ & $0.913$ & $\textbf{0.925}$ & $\textbf{0.929}$ \\
    RES-AA & $0.902$ & $0.900$ & $0.906$ & $\textbf{0.846}$ & $\textbf{0.871}$ && $0.908$ & $0.907$ & $0.909$ & $\textbf{0.923}$ & $\textbf{0.928}$ \\
    CDF & $0.900$ & $0.900$ & $0.906$ & $\textbf{0.872}$ & $\textbf{0.887}$ && $0.905$ & $0.905$ & $0.907$ & $\textbf{0.901}$ & $\textbf{0.903}$ \\
    CDF-AT & $0.898$ & $0.899$ & $0.903$ & $\textbf{0.851}$ & $\textbf{0.884}$ && $0.905$ & $0.907$ & $0.908$ & $\textbf{0.918}$ & $\textbf{0.921}$ \\
    CDF-AA & $0.899$ & $0.899$ & $0.901$ & $\textbf{0.863}$ & $\textbf{0.886}$ && $0.904$ & $0.905$ & $0.908$ & $\textbf{0.917}$ & $\textbf{0.920}$ \\
    CDE & $0.903$ & $0.904$ & $0.907$ & $\textbf{0.818}$ & $\textbf{0.822}$ && $0.907$ & $0.907$ & $0.912$ & $\textbf{0.889}$ & $\textbf{0.893}$ \\
    CDE-AT & $0.904$ & $0.903$ & $0.907$ & $\textbf{0.915}$ & $\textbf{0.925}$ && $0.905$ & $0.906$ & $0.908$ & $\textbf{0.931}$ & $\textbf{0.934}$ \\
    CDE-AA & $0.904$ & $0.904$ & $0.908$ & $\textbf{0.903}$ & $\textbf{0.911}$ && $0.906$ & $0.906$ & $0.908$ & $\textbf{0.931}$ & $\textbf{0.933}$ \\
    \bottomrule
    \end{tabular}}}
    \end{center}
\end{table*}

\subsection{Real Data Analysis: Predicting Medical Insurance Cost}\label{sec:real_data}
In this section, we provide another real data analysis case. We analyze the medical insurance cost dataset (available at \texttt{Kaggle.com}), which contains information on the medical expenses of policyholders.
The primary objective is to predict the medical expenses (\textit{Charges}) incurred by individual policyholders based on a set of features.
Our experiments focus on a subset of the data consisting of 678 female individuals from the western region. 
The remaining 2094 observations, comprising individuals of other genders or from different regions, are utilized as auxiliary data.
For the prediction model, we use the following four features: \textit{Age}, \textit{Sex}, \textit{BMI} (body mass index), \textit{\# Children} (number of dependents), \textit{Smoker} (a Boolean variable indicating smoking status), \textit{Region} (categorical variable of 4 locations). All categorical features (sex, smoker, region) were converted to one-hot encoded representations. The predictor was trained using all 9 resulting variables, while only the continuous variables were utilized for calibration purposes.

During preprocessing, the response variable \textit{Charges} is log-transformed. The numerical variables are standardized to have zero mean and unit variance.

For each configuration, 100 repeated experiments are conducted. 
The calibration, training, and test datasets each consist of $226$ samples, respectively, while the auxiliary and auxiliary training sets include $1047$ samples each.
The point predictor is trained using a random forest regression algorithm. We evaluate the performance of LCP, RLCP, LCP-C, RLCP-C, and ELCP across various bandwidths, specifically $h$ is selected from $\{0.1, 0.15, 0.2, 0.4, 0.6, 0.8, 1.0, 1.4, 1.8, 2.2, 2.6, 3.0\}$ and $\omega$ from $\{0.0, 0.1, 0.4, 0.7, 1.0\}$. 
% In particular, LCP, RLCP, LCP-C, and RLCP-C employ three different conformity scores: (i) residual score; (ii) score constructed from its conditional distribution used by DCP; (iii) and score based on conditional density estimate used by CDE. Since the pretrained score function can also leverage auxiliary training data, we test three pretraining schemes: (i) use only the target training data; (ii) transfer knowledge from auxiliary training (AT) data; (iii) transfer knowledge from all auxiliary (AA) data. These pretraining schemes follow the details described in Section \ref{sec:supp_simu_score}.
%The evaluation metrics include marginal coverage, mean set size, and weak test-conditional coverage, consistent with the approach described in Section \ref{sec:real_data_housing}.

After 100 repeated experiments, the computed marginal coverage rates were 0.906 (ELCP), 0.904 (LCP) and 0.911 (RLCP), while the directly combined version LCP-C and RLCP-C achieved 0.928 and 0.931 respectively. The failure of LCP-C and RLCP-C to maintain nominal coverage demonstrates the distributional heterogeneity between calibration and auxiliary data.

Similar to Section~\ref{sec:real_data_housing}, we split the data space of $678$ female individuals into 10 non-overlapping subsets. The results for weak test-conditional miscoverage errors and average prediction set sizes are presented in Table~\ref{table:medTS}. Compared to LCP and RLCP, ELCP achieves the smallest weak test-conditional miscoverage error in more than half of the 10 subsets. Furthermore, ELCP provides the shortest prediction set size in the majority of these 10 subsets. The aggregated results further emphasize that, on average, ELCP outperforms both LCP and RLCP.
\begin{table*}[h]
    \caption{Weak test-conditional miscoverage errors and average set sizes (in parentheses) in predicting medical expenses.}\label{table:medTS}
    \begin{center}
    {\fontsize{10}{10}\selectfont{\setlength{\tabcolsep}{3pt}\begin{tabular}{cccccccccccc}
    \toprule
    Subset Index & 1 & 2 & 3 & 4 & 5 & 6 & 7 & 8 & 9 & 10 & Agg\vspace{4pt}\\ \midrule
    Prop. & $0.065$ & $0.092$ & $0.023$ & $0.112$ & $0.059$ & $0.099$ & $0.142$ & $0.172$ & $0.129$ & $0.107$ & $1.000$\\
    ELCP & $\textbf{0.047}$ & $\textbf{0.030}$ & $\textbf{0.028}$ & $\textbf{0.031}$ & $\textbf{0.089}$ & $\textbf{0.019}$ & $\textbf{0.060}$ & $\textbf{0.013}$ & $\textbf{0.022}$ & $\textbf{0.018}$ & $\textbf{0.0325}$\\
    & $(1.01)$ & $(1.13)$ & $(0.65)$ & $(0.70)$ & $(1.09)$ & $(0.76)$ & $(1.26)$ & $(0.74)$ & $(0.88)$ & $(1.36)$ & $(0.97)$\\
    LCP & $\textbf{0.050}$ & $\textbf{0.045}$ & $\textbf{0.005}$ & $\textbf{0.027}$ & $\textbf{0.092}$ & $\textbf{0.015}$ & $\textbf{0.078}$ & $\textbf{0.007}$ & $\textbf{0.026}$ & $\textbf{0.014}$ & $\textbf{0.0346}$\\
    & $(1.05)$ & $(1.27)$ & $(0.91)$ & $(0.73)$ & $(1.20)$ & $(0.94)$ & $(1.29)$ & $(0.78)$ & $(1.03)$ & $(1.40)$ & $(1.05)$\\
    RLCP & $\textbf{0.071}$ & $\textbf{0.022}$ & $\textbf{0.040}$ & $\textbf{0.046}$ & $\textbf{0.097}$ & $\textbf{0.029}$ & $\textbf{0.115}$ & $\textbf{0.025}$ & $\textbf{0.039}$ & $\textbf{0.034}$ & $\textbf{0.0506}$\\
    & $(0.97)$ & $(1.18)$ & $(1.06)$ & $(0.82)$ & $(1.07)$ & $(1.02)$ & $(1.20)$ & $(0.88)$ & $(1.11)$ & $(1.21)$ & $(1.04)$\\
    \bottomrule
    \end{tabular}}}
    \end{center}
\end{table*}

Next, we also consider different score functions and pre-training schemes. The corresponding test-conditional miscoverage and mean set size results are provided in Tables~\ref{table:medT}--\ref{table:medS}, and the marginal coverage is provided in Table~\ref{table:realMar}.
\begin{table*}[htbp]
    \caption{Weak test-conditional miscoverage errors in predicting medical expenses.}\label{table:medT}
    \begin{center}
    {\fontsize{9}{6.4}\selectfont{\setlength{\tabcolsep}{3.pt}\begin{tabular}{cccccccccccc}
    \toprule
    Interval Index & 1 & 2 & 3 & 4 & 5 & 6 & 7 & 8 & 9 & 10 & Agg\vspace{4pt}\\ \midrule
    Prop. & $0.172$ & $0.142$ & $0.129$ & $0.112$ & $0.107$ & $0.099$ & $0.092$ & $0.065$ & $0.059$ & $0.023$ & $1.000$\\
    ELCP (RES) & $\textbf{0.013}$ & $\textbf{0.060}$ & $\textbf{0.022}$ & $\textbf{0.031}$ & $\textbf{0.018}$ & $\textbf{0.019}$ & $\textbf{0.030}$ & $\textbf{0.047}$ & $\textbf{0.089}$ & $\textbf{0.028}$ & $\textbf{0.0325}$\\
    LCP (RES) & $0.007$ & $0.078$ & $0.026$ & $0.027$ & $0.014$ & $0.015$ & $0.045$ & $0.050$ & $0.092$ & $0.005$ & $0.0346$\\
    RLCP (RES) & $0.025$ & $0.115$ & $0.039$ & $0.046$ & $0.034$ & $0.029$ & $0.022$ & $0.071$ & $0.097$ & $0.040$ & $0.0506$\\
    ELCP (RES-AT) & $\textbf{0.005}$ & $\textbf{0.046}$ & $\textbf{0.029}$ & $\textbf{0.012}$ & $\textbf{0.001}$ & $\textbf{0.031}$ & $\textbf{0.029}$ & $\textbf{0.052}$ & $\textbf{0.099}$ & $\textbf{0.022}$ & $\textbf{0.0281}$\\
    LCP (RES-AT) & $0.001$ & $0.071$ & $0.027$ & $0.013$ & $0.000$ & $0.017$ & $0.036$ & $0.058$ & $0.096$ & $0.004$ & $0.0298$\\
    RLCP (RES-AT) & $0.025$ & $0.100$ & $0.050$ & $0.038$ & $0.020$ & $0.011$ & $0.028$ & $0.073$ & $0.099$ & $0.021$ & $0.0458$\\
    ELCP (RES-AA) & $\textbf{0.003}$ & $\textbf{0.028}$ & $\textbf{0.017}$ & $\textbf{0.009}$ & $\textbf{0.002}$ & $\textbf{0.015}$ & $\textbf{0.009}$ & $\textbf{0.051}$ & $\textbf{0.099}$ & $\textbf{0.019}$ & $\textbf{0.0198}^{**}$\\
    LCP (RES-AA) & $0.006$ & $0.062$ & $0.021$ & $0.011$ & $0.003$ & $0.009$ & $0.025$ & $0.047$ & $0.098$ & $0.004$ & $0.0260$\\
    RLCP (RES-AA) & $0.022$ & $0.093$ & $0.044$ & $0.028$ & $0.023$ & $0.013$ & $0.017$ & $0.053$ & $0.099$ & $0.021$ & $0.0408$\\
    \midrule
    ELCP (CDF) & $\textbf{0.046}$ & $\textbf{0.101}$ & $\textbf{0.016}$ & $\textbf{0.031}$ & $\textbf{0.007}$ & $\textbf{0.011}$ & $\textbf{0.056}$ & $\textbf{0.034}$ & $\textbf{0.046}$ & $\textbf{0.011}$ & $\textbf{0.0401}$\\
    LCP (CDF) & $0.039$ & $0.112$ & $0.014$ & $0.030$ & $0.016$ & $0.003$ & $0.070$ & $0.054$ & $0.059$ & $0.044$ & $0.0442$\\
    RLCP (CDF) & $0.050$ & $0.159$ & $0.037$ & $0.044$ & $0.012$ & $0.006$ & $0.069$ & $0.039$ & $0.058$ & $0.085$ & $0.0570$\\
    DCP & $0.052$ & $0.236$ & $0.069$ & $0.049$ & $0.010$ & $0.040$ & $0.067$ & $0.003$ & $0.069$ & $0.096$ & $0.0746$\\
    ELCP (CDF-AT) & $\textbf{0.038}$ & $\textbf{0.038}$ & $\textbf{0.006}$ & $\textbf{0.033}$ & $\textbf{0.004}$ & $\textbf{0.019}$ & $\textbf{0.046}$ & $\textbf{0.109}$ & $\textbf{0.041}$ & $\textbf{0.032}$ & $\textbf{0.0331}^*$\\
    LCP (CDF-AT) & $0.037$ & $0.067$ & $0.003$ & $0.042$ & $0.003$ & $0.029$ & $0.066$ & $0.123$ & $0.066$ & $0.025$ & $0.0426$\\
    RLCP (CDF-AT) & $0.040$ & $0.087$ & $0.016$ & $0.044$ & $0.002$ & $0.029$ & $0.062$ & $0.114$ & $0.065$ & $0.064$ & $0.0476$\\
    DCP-AT & $0.044$ & $0.127$ & $0.034$ & $0.048$ & $0.011$ & $0.046$ & $0.055$ & $0.115$ & $0.051$ & $0.079$ & $0.0583$\\
    ELCP (CDF-AA) & $\textbf{0.035}$ & $\textbf{0.030}$ & $\textbf{0.004}$ & $\textbf{0.039}$ & $\textbf{0.001}$ & $\textbf{0.022}$ & $\textbf{0.044}$ & $\textbf{0.154}$ & $\textbf{0.041}$ & $\textbf{0.047}$ & $\textbf{0.0351}$\\
    LCP (CDF-AA) & $0.033$ & $0.050$ & $0.004$ & $0.042$ & $0.003$ & $0.022$ & $0.061$ & $0.142$ & $0.057$ & $0.017$ & $0.0391$\\
    RLCP (CDF-AA) & $0.039$ & $0.076$ & $0.024$ & $0.047$ & $0.001$ & $0.031$ & $0.061$ & $0.152$ & $0.057$ & $0.064$ & $0.0494$\\
    DCP-AA & $0.040$ & $0.103$ & $0.036$ & $0.050$ & $0.013$ & $0.038$ & $0.054$ & $0.146$ & $0.041$ & $0.079$ & $0.0556$\\
    \midrule
    ELCP (CDE) & $\textbf{0.046}$ & $\textbf{0.021}$ & $\textbf{0.016}$ & $\textbf{0.029}$ & $\textbf{0.032}$ & $\textbf{0.004}$ & $\textbf{0.054}$ & $\textbf{0.119}$ & $\textbf{0.101}$ & $\textbf{0.184}$ & $\textbf{0.0428}^*$\\
    LCP (CDE) & $0.046$ & $0.020$ & $0.019$ & $0.028$ & $0.031$ & $0.006$ & $0.054$ & $0.117$ & $0.103$ & $0.175$ & $0.0429$\\
    RLCP (CDE) & $0.047$ & $0.018$ & $0.043$ & $0.036$ & $0.029$ & $0.029$ & $0.052$ & $0.139$ & $0.118$ & $0.150$ & $0.0506$\\
    CDE & $0.051$ & $0.014$ & $0.037$ & $0.040$ & $0.026$ & $0.029$ & $0.048$ & $0.152$ & $0.129$ & $0.194$ & $0.0522$\\
    ELCP (CDE-AT) & $\textbf{0.055}$ & $\textbf{0.011}$ & $\textbf{0.024}$ & $\textbf{0.041}$ & $\textbf{0.015}$ & $\textbf{0.014}$ & $\textbf{0.042}$ & $\textbf{0.160}$ & $\textbf{0.079}$ & $\textbf{0.233}$ & $\textbf{0.0459}$\\
    LCP (CDE-AT) & $0.053$ & $0.010$ & $0.026$ & $0.042$ & $0.017$ & $0.028$ & $0.051$ & $0.176$ & $0.084$ & $0.239$ & $0.0500$\\
    RLCP (CDE-AT) & $0.048$ & $0.010$ & $0.042$ & $0.046$ & $0.012$ & $0.028$ & $0.044$ & $0.189$ & $0.080$ & $0.143$ & $0.0486$\\
    CDE-AT & $0.051$ & $0.011$ & $0.047$ & $0.050$ & $0.000$ & $0.026$ & $0.036$ & $0.211$ & $0.107$ & $0.164$ & $0.0518$\\
    ELCP (CDE-AA) & $\textbf{0.057}$ & $\textbf{0.005}$ & $\textbf{0.020}$ & $\textbf{0.043}$ & $\textbf{0.023}$ & $\textbf{0.005}$ & $\textbf{0.046}$ & $\textbf{0.141}$ & $\textbf{0.078}$ & $\textbf{0.237}$ & $\textbf{0.0443}$\\
    LCP (CDE-AA) & $0.055$ & $0.007$ & $0.029$ & $0.040$ & $0.020$ & $0.025$ & $0.053$ & $0.171$ & $0.092$ & $0.263$ & $0.0509$\\
    RLCP (CDE-AA) & $0.048$ & $0.008$ & $0.042$ & $0.042$ & $0.015$ & $0.025$ & $0.049$ & $0.187$ & $0.081$ & $0.164$ & $0.0490$\\
    CDE-AA & $0.051$ & $0.011$ & $0.045$ & $0.050$ & $0.001$ & $0.026$ & $0.036$ & $0.200$ & $0.108$ & $0.171$ & $0.0511$\\
    \bottomrule
    \end{tabular}}}
    \end{center}
\end{table*}
\begin{table*}[htbp]
    \caption{Average prediction set sizes in predicting medical expenses.}\label{table:medS}
    \begin{center}
    {\fontsize{9}{6.4}\selectfont{\setlength{\tabcolsep}{3pt}\begin{tabular}{cccccccccccc}
    \toprule
    Interval Index & 1 & 2 & 3 & 4 & 5 & 6 & 7 & 8 & 9 & 10 & Agg\vspace{4pt}\\ \midrule
    Prop. & $0.172$ & $0.142$ & $0.129$ & $0.112$ & $0.107$ & $0.099$ & $0.092$ & $0.065$ & $0.059$ & $0.023$ & $1.000$\\
    ELCP (RES) & $\textbf{0.74}$ & $\textbf{1.26}$ & $\textbf{0.88}$ & $\textbf{0.70}$ & $\textbf{1.36}$ & $\textbf{0.76}$ & $\textbf{1.13}$ & $\textbf{1.01}$ & $\textbf{1.09}$ & $\textbf{0.65}$ & $\textbf{0.97}^*$\\
    LCP (RES) & $0.78$ & $1.29$ & $1.03$ & $0.73$ & $1.40$ & $0.94$ & $1.27$ & $1.05$ & $1.20$ & $0.91$ & $1.05$\\
    RLCP (RES) & $0.88$ & $1.20$ & $1.11$ & $0.82$ & $1.21$ & $1.02$ & $1.18$ & $0.97$ & $1.07$ & $1.06$ & $1.04$\\
    ELCP (RES-AT) & $\textbf{0.75}$ & $\textbf{1.24}$ & $\textbf{0.93}$ & $\textbf{0.73}$ & $\textbf{1.37}$ & $\textbf{0.81}$ & $\textbf{1.10}$ & $\textbf{1.01}$ & $\textbf{1.13}$ & $\textbf{0.67}$ & $\textbf{0.98}$\\
    LCP (RES-AT) & $0.81$ & $1.29$ & $1.09$ & $0.79$ & $1.40$ & $1.03$ & $1.26$ & $1.06$ & $1.24$ & $0.95$ & $1.08$\\
    RLCP (RES-AT) & $0.86$ & $1.15$ & $1.12$ & $0.81$ & $1.17$ & $1.03$ & $1.15$ & $0.94$ & $1.05$ & $1.05$ & $1.03$\\
    ELCP (RES-AA) & $\textbf{0.76}$ & $\textbf{1.25}$ & $\textbf{0.93}$ & $\textbf{0.74}$ & $\textbf{1.38}$ & $\textbf{0.85}$ & $\textbf{1.10}$ & $\textbf{1.03}$ & $\textbf{1.13}$ & $\textbf{0.68}$ & $\textbf{0.99}$\\
    LCP (RES-AA) & $0.80$ & $1.24$ & $1.10$ & $0.79$ & $1.37$ & $1.05$ & $1.24$ & $1.03$ & $1.22$ & $0.99$ & $1.07$\\
    RLCP (RES-AA) & $0.85$ & $1.11$ & $1.13$ & $0.79$ & $1.13$ & $1.04$ & $1.13$ & $0.90$ & $1.02$ & $1.05$ & $1.01$\\
    \midrule
    ELCP (CDF) & $\textbf{0.61}$ & $\textbf{2.35}$ & $\textbf{0.84}$ & $\textbf{0.43}$ & $\textbf{2.09}$ & $\textbf{0.66}$ & $\textbf{1.75}$ & $\textbf{2.25}$ & $\textbf{2.38}$ & $\textbf{1.55}$ & $\textbf{1.37}$\\
    LCP (CDF) & $0.66$ & $2.43$ & $0.93$ & $0.44$ & $2.16$ & $0.71$ & $1.92$ & $2.30$ & $2.41$ & $1.71$ & $1.43$\\
    RLCP (CDF) & $0.84$ & $2.28$ & $1.11$ & $0.62$ & $2.10$ & $0.83$ & $1.93$ & $2.31$ & $2.43$ & $2.03$ & $1.50$\\
    DCP & $0.93$ & $2.07$ & $1.19$ & $0.72$ & $1.91$ & $0.94$ & $1.82$ & $2.42$ & $2.53$ & $2.26$ & $1.51$\\
    ELCP (CDF-AT) & $\textbf{0.65}$ & $\textbf{2.04}$ & $\textbf{0.84}$ & $\textbf{0.51}$ & $\textbf{1.91}$ & $\textbf{0.70}$ & $\textbf{1.53}$ & $\textbf{1.84}$ & $\textbf{1.97}$ & $\textbf{1.22}$ & $\textbf{1.24}$\\
    LCP (CDF-AT) & $0.77$ & $2.05$ & $0.97$ & $0.60$ & $1.93$ & $0.78$ & $1.68$ & $1.89$ & $2.00$ & $1.37$ & $1.33$\\
    RLCP (CDF-AT) & $0.80$ & $2.01$ & $1.04$ & $0.61$ & $1.92$ & $0.81$ & $1.71$ & $1.85$ & $1.95$ & $1.52$ & $1.34$\\
    DCP-AT & $0.80$ & $1.93$ & $0.98$ & $0.63$ & $1.83$ & $0.80$ & $1.69$ & $1.79$ & $1.85$ & $1.56$ & $1.30$\\
    ELCP (CDF-AA) & $\textbf{0.67}$ & $\textbf{1.94}$ & $\textbf{0.85}$ & $\textbf{0.51}$ & $\textbf{1.86}$ & $\textbf{0.70}$ & $\textbf{1.48}$ & $\textbf{1.75}$ & $\textbf{1.86}$ & $\textbf{1.17}$ & $\textbf{1.21}^*$\\
    LCP (CDF-AA) & $0.78$ & $1.97$ & $0.97$ & $0.60$ & $1.92$ & $0.75$ & $1.59$ & $1.81$ & $1.95$ & $1.28$ & $1.29$\\
    RLCP (CDF-AA) & $0.82$ & $1.96$ & $1.06$ & $0.63$ & $1.92$ & $0.84$ & $1.66$ & $1.78$ & $1.88$ & $1.45$ & $1.33$\\
    DCP-AA & $0.81$ & $1.90$ & $0.98$ & $0.64$ & $1.83$ & $0.80$ & $1.66$ & $1.72$ & $1.77$ & $1.48$ & $1.28$\\
    \midrule
    ELCP (CDE) & $\textbf{1.17}$ & $\textbf{1.38}$ & $\textbf{1.15}$ & $\textbf{1.10}$ & $\textbf{1.35}$ & $\textbf{1.16}$ & $\textbf{1.32}$ & $\textbf{1.15}$ & $\textbf{1.06}$ & $\textbf{0.97}$ & $\textbf{1.21}$\\
    LCP (CDE) & $1.16$ & $1.37$ & $1.14$ & $1.09$ & $1.33$ & $1.16$ & $1.33$ & $1.14$ & $1.06$ & $0.99$ & $1.20$\\
    RLCP (CDE) & $1.14$ & $1.24$ & $1.20$ & $1.08$ & $1.22$ & $1.20$ & $1.23$ & $1.05$ & $1.00$ & $1.03$ & $1.16$\\
    CDE & $1.17$ & $1.17$ & $1.14$ & $1.13$ & $1.15$ & $1.16$ & $1.16$ & $1.05$ & $0.96$ & $1.00$ & $1.13$\\
    ELCP (CDE-AT) & $\textbf{0.85}$ & $\textbf{0.95}$ & $\textbf{0.83}$ & $\textbf{0.78}$ & $\textbf{0.99}$ & $\textbf{0.81}$ & $\textbf{0.92}$ & $\textbf{0.90}$ & $\textbf{0.84}$ & $\textbf{0.62}$ & $\textbf{0.87}$\\
    LCP (CDE-AT) & $0.77$ & $0.95$ & $0.77$ & $0.71$ & $0.96$ & $0.79$ & $0.95$ & $0.87$ & $0.83$ & $0.65$ & $0.84$\\
    RLCP (CDE-AT) & $0.70$ & $0.80$ & $0.83$ & $0.66$ & $0.84$ & $0.79$ & $0.84$ & $0.73$ & $0.75$ & $0.74$ & $0.77$\\
    CDE-AT & $0.73$ & $0.68$ & $0.75$ & $0.72$ & $0.71$ & $0.72$ & $0.73$ & $0.69$ & $0.69$ & $0.69$ & $0.71$\\
    ELCP (CDE-AA) & $\textbf{0.91}$ & $\textbf{1.03}$ & $\textbf{0.85}$ & $\textbf{0.84}$ & $\textbf{1.04}$ & $\textbf{0.84}$ & $\textbf{0.96}$ & $\textbf{0.96}$ & $\textbf{0.87}$ & $\textbf{0.63}$ & $\textbf{0.92}$\\
    LCP (CDE-AA) & $0.80$ & $1.00$ & $0.79$ & $0.75$ & $0.98$ & $0.82$ & $0.97$ & $0.90$ & $0.84$ & $0.67$ & $0.86$\\
    RLCP (CDE-AA) & $0.70$ & $0.80$ & $0.83$ & $0.66$ & $0.85$ & $0.79$ & $0.85$ & $0.73$ & $0.75$ & $0.74$ & $0.77$\\
    CDE-AA & $0.72$ & $0.67$ & $0.75$ & $0.71$ & $0.71$ & $0.72$ & $0.72$ & $0.70$ & $0.69$ & $0.70$ & $0.71^{**}$\\
    \bottomrule
    \end{tabular}}}
    \end{center}
\end{table*}

Table~\ref{table:medT} shows that the best performance is achieved when ELCP uses the RES-AA score function.
Moreover, ELCP significantly outperforms all other compared methods under the same score function and pre-training scheme.
For RES and CDF scores, performing transfer during the pre-training stage in score construction can substantially improve test-conditional coverage, while this is not the case for the CDE score.

\FloatBarrier

\end{appendices}
\end{document}